%% file: Oritatami-shapes.tex
\documentclass[runningheads]{llncs}
\usepackage[utf8]{inputenc}

\usepackage{geometry}%% NS: INCLUDE GEOMETRY TO ALLOW APPENDIX IN FULL PAGE 
\usepackage{appendix}
\usepackage{amsmath}
\usepackage{amssymb}
\usepackage{ifpdf}
\usepackage{algorithm}
\usepackage[noend]{algpseudocode}
\usepackage{comment}
\usepackage{complexity} % See http://cse.unl.edu/~cbourke/latex/complexity.pdf
\usepackage{subcaption}
\usepackage{wrapfig}
\usepackage{cite}
\usepackage{colortbl}
\usepackage{tikz}
\usetikzlibrary{
	arrows,%
	decorations.pathmorphing,%
	decorations.pathreplacing,%
}

\usepackage{morefloats}
\usepackage{gensymb}
\usepackage[textsize=scriptsize]{todonotes}

\algnewcommand{\LineComment}[1]{\State \(\triangleright\) #1}

\usepackage{commands-tam}

\input{shapes-macros}

\ifpdf

  \usepackage[pdftex]{epsfig}
  \usepackage[pdftex]{hyperref}

\else

    \usepackage[dvips]{epsfig}
    \newcommand{\href}[2]{#2}

\fi

\newif\ifabstract
\newif\iffull

\ifabstract
	\fullfalse
\else
	\fulltrue
\fi

\newtoks\magicAppendix
\magicAppendix={}
\newtoks\magictoks
\newif\iflater
\laterfalse
\ifabstract
\long\def\later#1{\magictoks={#1}%
  \edef\magictodo{\noexpand\magicAppendix={\the\magicAppendix \par
    \the\magictoks%
  }}
  \magictodo}
\long\def\both#1{\magictoks={#1}%
  \edef\magictodo{\noexpand\magicAppendix={\the\magicAppendix \par
    \noexpand\setcounter{theorem-preserve}{\noexpand\arabic{theorem}}%
    \noexpand\setcounter{theorem}{\arabic{theorem}}%
    \noexpand\setcounter{section-preserve}{\noexpand\arabic{section}}%
    \noexpand\setcounter{section}{\arabic{section}}%
	\noexpand\let\noexpand\oldsection=\noexpand\thesection
	\noexpand\def\noexpand\thesection{\thesection}
	\noexpand\let\noexpand\oldlabel=\noexpand\label
	\noexpand\let\noexpand\label=\noexpand\blank
    \the\magictoks%
    \noexpand\setcounter{theorem}{\noexpand\arabic{theorem-preserve}}%
    \noexpand\setcounter{section}{\noexpand\arabic{section-preserve}}%
	\noexpand\let\noexpand\thesection=\noexpand\oldsection
	\noexpand\let\noexpand\label=\noexpand\oldlabel
  }}
  \magictodo
  \the\magictoks}
\else
\long\def\later#1{#1}
\long\def\both#1{#1}
\fi
\long\def\magicappendix{
	\latertrue%
	\the\magicAppendix%
}

\newtheorem{observation}[theorem]{Observation}

\newcommand{\calF}{\mathcal{F}}
\newcommand{\calPB}{\mathcal{P} \mathcal{B}}
\newcommand{\calNB}{NB}
\newcommand{\todoi}[1]{\todo[inline]{#1}}

\newcommand{\urltt}[1]{{\ttfamily\protect\url{#1}}}
\newcommand{\hreftt}[2]{\protect\href{#1}{\ttfamily #2}}

\title{Know When to Fold 'Em: \\ Self-Assembly of Shapes by Folding in Oritatami} % (extended abstract)}

\authorrunning{E.D., J.H., M.O., M.J.P., T.A.R., N.S., S.S.\/ and H.T.}
\titlerunning{Self-Assembly of Shapes by Folding in Oritatami}

\institute{%
	CSAIL, Massachusetts Institute of Technology, USA. \urltt{edemaine@mit.edu}
    \and
    Department of Computer Science and Information Systems, University of Wisconsin - River Falls, River Falls, WI, USA.
    \urltt{jacob.hendricks@uwrf.edu}
    \and
    Department of Computer Science and Computer Engineering, University of Arkansas, Fayetteville, AR, USA.
    \hreftt{mailto:mo015@uark.edu}{\{mo015,} \hreftt{mailto:patitz@uark.edu}{patitz,} \hreftt{mailto:tar003@uark.edu}{tar003\}@uark.edu}.
    \and
    Supported in part by NSF Grant CCF-1422152 and CAREER-1553166.
    \and
    CNRS, \'Ecole Normale Sup\'erieure de Lyon (LIP, UMR 5668) \& IXXI, U. Lyon. \hreftt{http://perso.ens-lyon.fr/nicolas.schabanel}{perso.ens-lyon.fr/first.last} Supported by Moprexprogmol CNRS MI grant.
    \and
	University of Electro-Communications, Tokyo, Japan.
	\urltt{s.seki@uec.ac.jp}. Supported in part by JST Program to Disseminate Tenure Tracking System, MEXT, Japan, No.~6F36, JSPS Grant-in-Aid for Young Scientists (A) No.~16H05854, and JSPS Bilateral Program No.~YB29004
 	\and
	Colorado School of Mines, Golden, CO, USA.
 	\urltt{hadleythomas88@gmail.com}
}%institute

\author{
 	Erik D. Demaine%
	\inst{1}
\and
	Jacob Hendricks%
 	\inst{2}
\and
	 Meagan Olsen\inst{3}%
\and
 	Matthew J. Patitz%
 	\inst{3,4}%
\and
	Trent A. Rogers
 	\inst{3,4}%
\and
	Nicolas Schabanel
 	\inst{5}%
\and
 	Shinnosuke Seki%
	\inst{6}
\and
 	Hadley Thomas%
	\inst{7}
}

\date{}

\begin{document}

\maketitle

%\thispagestyle{empty}
%\addtocounter{page}{-1}

\begin{abstract}
An oritatami system (OS) is a theoretical model of self-assembly via co-transcriptional folding.  It consists of a growing chain of beads which can form bonds with each other as they are transcribed. During the transcription process, the $\delta$ most recently produced beads  dynamically fold so as to maximize the number of bonds formed, self-assemblying into a shape incrementally. The parameter $\delta$ is called the \emph{delay} and is related to the transcription rate in nature.  

This article initiates the study of shape self-assembly using oritatami. 
A shape is a connected set of points in the triangular lattice.
We first show that oritatami systems differ fundamentally from tile-assembly systems by exhibiting a family of infinite shapes that can be tile-assembled but cannot be folded by any OS.
As it is NP-hard in general to determine whether there is an OS that folds into (self-assembles) a given finite shape, we explore the folding of upscaled versions of finite shapes. We show that any shape can be folded from a constant size seed, at any scale $n\geq 3$, by an OS with delay $1$. We also show that any shape can be folded at the smaller scale $2$ by an OS with \emph{unbounded} delay. 
This leads us to investigate the influence of delay and to prove that, for all $\delta > 2$, there are shapes that can be folded  (at scale~$1$) with delay $\delta$ but not with delay $\delta'<\delta$. 
%
%We conclude with a conjecture saying that no OS can fold a shape at scale lower than $3$ with constant delay for the considering scaling scheme.

These results serve as a foundation for the study of shape-building in this new model of self-assembly, and have the potential to provide better understanding of cotranscriptional folding in biology, as well as improved abilities of experimentalists to design artificial systems that self-assemble via this complex dynamical process.

\end{abstract}

%\input{content-for-later.tex}

\input{introduction2}

%\input{definitions-informal}

\input{shapes-def}
\input{shapes-finitely-cut}

\input{finite-shapes-short.tex}
\input{shapes-scaling-algo}

\input{smallDelayWeak-short}

\withOrWithoutAppendix{%
	\input{d1a1_det_finiteness_short}

}{%
}

\enlargethispage*{1cm}

\bibliographystyle{amsplain}
\bibliography{scale2}

%%%
%%% APPENDIX
%%%

\withOrWithoutAppendix{ 
    %% WITH APPENDIX
    \input{shape-appendix.tex}
}{%% WITHOUT APPENDIX
}

\end{document}

%% file: shapes-macros.tex
\graphicspath{{./figs/}{./images/}{./patterns/scale_x.0/}{./patterns/scale_x.0+/}{./patterns/scale_x.5/}{./movies/}}

\usepackage{pifont,bbding}
\usepackage{ifthen}

\usepackage{wasysym}

\usepackage{booktabs, tabularx}
\usepackage{enumitem}

\usepackage{caption,subcaption}

\usepackage{rotating}
\usepackage{afterpage}

\newenvironment{omittedproof}[2]{
	\begin{proof}[of {#1}~\ref{#2}]
	\label{proof:#2}
}{
	\end{proof}
}

\newcommand{\rotateClockwise}{\ensuremath{\operatorname{RotateClockwise}}}
\newcommand{\rotateClockwiseAA}[1]{\ensuremath{\rotateClockwise\left(#1,120^\circ\right)}}

\DeclareSymbolFont{extraup}{U}{zavm}{m}{n}
\DeclareMathSymbol{\varheart}{\mathalpha}{extraup}{86}
\newcommand{\heart}{\ensuremath{\mathbin{\textcolor{red}{\ensuremath\varheart}}}}%{\ensuremath{\mathbin{\raisebox{-.25ex}{\includegraphics[height=1.75ex]{Heart-light.pdf}}}}}%\ensuremath{\heartsuit}}
\renewcommand{\mod}{\mathbin{\operatorname{mod}}}

\usepackage{tikz}
\usetikzlibrary{
  arrows,% 
  decorations.pathmorphing,%
  decorations.pathreplacing,%
}

\newcommand{\elong}[2]{{#1}^{\hspace*{.05em}\triangleright{#2}}}

\usepackage{mathrsfs}

\newcommand{\Dynamics}{\operatorname{\mathscr D}}

\newcommand{\Ccal}{{\mathcal C}}
\newcommand{\Tlat}{{\mathbb T}}

\newcommand{\delay}{\delta}
\renewcommand{\leq}{\leqslant}
\renewcommand{\geq}{\geqslant}

\newcommand{\TM}[1]{\ensuremath{\mathcal{#1}}}
\newcommand{\TMO}{{\TM O}}

\newcommand{\scaling}[2]{\ensuremath{\mathscr{#1}_{#2}}}

\newcommand{\indexscaling}[3]{\ensuremath{\operatorname{#1}_{\scaling{#2}{#3}}}}
\newcommand{\linscaling}[2]{{\indexscaling\lambda{#1}{#2}}}
\newcommand{\muscaling}[2]{{\indexscaling\mu{#1}{#2}}}
\newcommand{\Linscaling}[2]{{\indexscaling\Lambda{#1}{#2}\!}}

\newcommand{\cell}{\Linscaling{}{}}
\newcommand{\ccenter}{\linscaling{}{}}

\newcommand{\figPatNoSign}[4]{%
	\includegraphics[scale=#1]{{#2scale_#3-#4}.pdf}%
}

\newcommand{\NSpattern}[1]{\texttt{#1}}

\newcommand{\figPatSign}[5]{%
	\begin{tabular}[t]{c}%
	{\small\NSpattern{#5}}\\[-0.5mm]%
	\figPatNoSign{#1}{#2}{#3}{#4}
	\end{tabular}
}

\newcommand{\figPat}[4]{%
	\figPatSign{#1}{#2}{#3}{#4}{#4}
}

\newcommand{\tablePatsXFive}[2]{%
    \begin{tabular}{c@{\hspace*{5mm}}c@{\hspace*{5mm}}c}
        \figPatNoSign{#1}{patterns/scale_x.5/}{#2.5}{0}
        & \figPatNoSign{#1}{patterns/scale_x.5/}{#2.5}{1}
        & \figPatNoSign{#1}{patterns/scale_x.5/}{#2.5}{10}
    \end{tabular}%
}

\newcommand{\tablePatsXEasy}[2]{%
    \begin{tabular}{c@{\hspace*{5mm}}c}
        \figPatNoSign{#1}{patterns/scale_x.0+/}{#2.0+}{0}
        & \figPatNoSign{#1}{patterns/scale_x.0+/}{#2.0+}{1}
    \end{tabular}%
}

\newcommand{\tablePatsXZero}[2]{%
        \begin{tabular}{cccc}
        	\figPat{#1}{patterns/scale_x.0/scale_#2.0/}{#2.0}{0}
        &	\figPat{#1}{patterns/scale_x.0/scale_#2.0/}{#2.0}{1}
        &	\figPat{#1}{patterns/scale_x.0/scale_#2.0/}{#2.0}{101}
        &	\figPat{#1}{patterns/scale_x.0/scale_#2.0/}{#2.0}{1001}
        \\
        	\figPat{#1}{patterns/scale_x.0/scale_#2.0/}{#2.0}{1101}
        &	\figPat{#1}{patterns/scale_x.0/scale_#2.0/}{#2.0}{10101}
        &	\figPat{#1}{patterns/scale_x.0/scale_#2.0/}{#2.0}{11001}
        &	\figPat{#1}{patterns/scale_x.0/scale_#2.0/}{#2.0}{11101}
        \\
        	\figPat{#1}{patterns/scale_x.0/scale_#2.0/}{#2.0}{100001}
        &	\figPat{#1}{patterns/scale_x.0/scale_#2.0/}{#2.0}{101101}
        &	\figPat{#1}{patterns/scale_x.0/scale_#2.0/}{#2.0}{110001}
        &	\figPat{#1}{patterns/scale_x.0/scale_#2.0/}{#2.0}{111001}
        \\
        &	\figPat{#1}{patterns/scale_x.0/scale_#2.0/}{#2.0}{111101}
        &	\figPat{#1}{patterns/scale_x.0/scale_#2.0/}{#2.0}{111111}
        \end{tabular}%
}

\newcommand{\tablePatsXZeroSmall}[2]{%
        \begin{tabular}{ccccc}
        	\figPat{#1}{patterns/scale_x.0/scale_#2.0/}{#2.0}{0}
        &	\figPat{#1}{patterns/scale_x.0/scale_#2.0/}{#2.0}{1}
        &	\figPat{#1}{patterns/scale_x.0/scale_#2.0/}{#2.0}{101}
        &	\figPat{#1}{patterns/scale_x.0/scale_#2.0/}{#2.0}{1001}
        &	\figPat{#1}{patterns/scale_x.0/scale_#2.0/}{#2.0}{1101}
        \\
            \figPat{#1}{patterns/scale_x.0/scale_#2.0/}{#2.0}{10101}
        &	\figPat{#1}{patterns/scale_x.0/scale_#2.0/}{#2.0}{11001}
        &	\figPat{#1}{patterns/scale_x.0/scale_#2.0/}{#2.0}{11101}
        &	\figPat{#1}{patterns/scale_x.0/scale_#2.0/}{#2.0}{100001}
        &	\figPat{#1}{patterns/scale_x.0/scale_#2.0/}{#2.0}{101101}
        \\
        	\figPat{#1}{patterns/scale_x.0/scale_#2.0/}{#2.0}{110001}
        &	\figPat{#1}{patterns/scale_x.0/scale_#2.0/}{#2.0}{111001}
        &	\figPat{#1}{patterns/scale_x.0/scale_#2.0/}{#2.0}{111101}
        &	\figPat{#1}{patterns/scale_x.0/scale_#2.0/}{#2.0}{111111}
        \end{tabular}%
}

\newcommand{\fittablePatsXZero}[1]{%
    \resizebox{\textwidth}{!}{\tablePatsXZero{1}{#1}}%
}

\newcommand{\fittablePatsXZeroSmall}[1]{%
    \resizebox{\textwidth}{!}{\tablePatsXZeroSmall{1}{#1}}%
}

\newcommand{\fitfigurePatsXZero}[1]{
        \begin{figure}[t]
        \centering
        \fittablePatsXZero{#1}
        \caption{\captionpar{The 14 self-supported tight routings at scale \protect\scaling A{#1}:} in purple, the clean edge used to extend the routing in this cell; in green, the sides already covered by an occupied neighboring cell; in red, the new clean edges at the second clockwise-most position on every available side; highlighted in orange, the seed.}
        \label{fig:scaling:algo:routing:A#1}
        \end{figure}
}

\newcommand{\fitfigurePatsXZeroSmall}[1]{
        \begin{figure}[t]
        \centering
        \fittablePatsXZeroSmall{#1}
        \caption{\captionpar{The 14 self-supported tight routing extensions at scale \protect\scaling A{#1}:} in purple, the clean edge used to extend the path in this cell; in red, the new clean edges at the second clockwise-most position on every available side; highlighted in orange, the seed.}
        \label{fig:scaling:algo:routing:A#1}
        \end{figure}
}

\newcommand{\tabFigFourZero}[1]{
       \begin{tabular}{ccccc}
        	\figPat{#1}{patterns/scale_x.0/scale_4.0/}{4.0}{0}
        &	\figPat{#1}{patterns/scale_x.0/scale_4.0/}{4.0}{1}
        &	\figPat{#1}{patterns/scale_x.0/scale_4.0/}{4.0}{11}
        &	\figPat{#1}{patterns/scale_x.0/scale_4.0/}{4.0}{101}
        &	\figPat{#1}{patterns/scale_x.0/scale_4.0/}{4.0}{111}
       \\	\figPat{#1}{patterns/scale_x.0/scale_4.0/}{4.0}{1001}
        &	\figPat{#1}{patterns/scale_x.0/scale_4.0/}{4.0}{1011}
        &	\figPat{#1}{patterns/scale_x.0/scale_4.0/}{4.0}{1101}
        &	\figPat{#1}{patterns/scale_x.0/scale_4.0/}{4.0}{1111}
        &	\figPat{#1}{patterns/scale_x.0/scale_4.0/}{4.0}{10001}
        \\	\figPat{#1}{patterns/scale_x.0/scale_4.0/}{4.0}{10011}
        &	\figPat{#1}{patterns/scale_x.0/scale_4.0/}{4.0}{10101}
        &	\figPat{#1}{patterns/scale_x.0/scale_4.0/}{4.0}{10111}
        &	\figPat{#1}{patterns/scale_x.0/scale_4.0/}{4.0}{11001}
        &	\figPat{#1}{patterns/scale_x.0/scale_4.0/}{4.0}{11011}
        \\	\figPat{#1}{patterns/scale_x.0/scale_4.0/}{4.0}{11101}
        &	\figPat{#1}{patterns/scale_x.0/scale_4.0/}{4.0}{11111}
        &	\figPat{#1}{patterns/scale_x.0/scale_4.0/}{4.0}{100001}
        &	\figPat{#1}{patterns/scale_x.0/scale_4.0/}{4.0}{100011}
        &	\figPat{#1}{patterns/scale_x.0/scale_4.0/}{4.0}{100101}
        \\	\figPat{#1}{patterns/scale_x.0/scale_4.0/}{4.0}{100111}
        &	\figPat{#1}{patterns/scale_x.0/scale_4.0/}{4.0}{101001}
        &	\figPat{#1}{patterns/scale_x.0/scale_4.0/}{4.0}{101011}
        &	\figPat{#1}{patterns/scale_x.0/scale_4.0/}{4.0}{101101}
        &	\figPat{#1}{patterns/scale_x.0/scale_4.0/}{4.0}{101111}
        \\	\figPat{#1}{patterns/scale_x.0/scale_4.0/}{4.0}{110001}
        &	\figPat{#1}{patterns/scale_x.0/scale_4.0/}{4.0}{110011}
        &	\figPat{#1}{patterns/scale_x.0/scale_4.0/}{4.0}{110101}
        &	\figPat{#1}{patterns/scale_x.0/scale_4.0/}{4.0}{110111}
        &	\figPat{#1}{patterns/scale_x.0/scale_4.0/}{4.0}{111001}
        \\	\figPat{#1}{patterns/scale_x.0/scale_4.0/}{4.0}{111011}
        &	\figPat{#1}{patterns/scale_x.0/scale_4.0/}{4.0}{111101}
        &	\figPat{#1}{patterns/scale_x.0/scale_4.0/}{4.0}{111111}
        \end{tabular}
 }

\newcommand{\figPatThreeZero}[2]{%
    \figPat{#1}{patterns/scale_x.0/scale_3.0/}{3.0}{#2}%
}

\newcommand{\tabFigBasicThreeZeroCompact}[1]{
      \begin{tabular}{ccccccccc}
    	\figPatThreeZero{#1}{0}
    	&   \figPatThreeZero{#1}{1}
    	&   \figPatThreeZero{#1}{100001}
    	&   \figPatThreeZero{#1}{110001}
    	&   \figPatThreeZero{#1}{111001}
        &  	\figPatThreeZero{#1}{111101}
    	&   \figPatThreeZero{#1}{111111}
    	&   \figPatThreeZero{#1}{101}
    	&   \figPatThreeZero{#1}{1001}
    	\\
    	\figPatThreeZero{#1}{10001}
    	&   \figPatThreeZero{#1}{100101}
    	&   \figPatThreeZero{#1}{101001}
    	&   \figPatThreeZero{#1}{1101}
    	&   \figPatThreeZero{#1}{11001}
    	&   \figPatThreeZero{#1}{101101}
    	&   \figPatThreeZero{#1}{110101}
    	&   \figPatThreeZero{#1}{11101}
    	&   \figPatThreeZero{#1}{10101}
   \end{tabular}
}

\newcommand{\figPatAThree}[2]{%
     \figPatSign{\NSpatScale}{patterns/scale_x.0/scale_3.0/}{3.0}{#1}{#2}%
}

\tikzstyle{NSpat}=[]
\tikzstyle{NSarr}=[->,>=latex,thick]
\tikzstyle{NSdir}=[midway,above,sloped]%fill=white]
\tikzstyle{NSdirB}=[near end,above,sloped]
\newcommand{\NSpatScale}{1}
\newcommand{\NSpatDist}{3.2cm}
\newcommand{\NSpatDistV}{2.5cm}
\newcommand{\NSnx}{\raisebox{.1em}{{\scalebox{0.5}{\,\ensuremath{\rangle\!\rangle}\,}}}}
\newcommand{\mathopsf}[1]{\ensuremath{\operatorname{\textsf{#1}}}}
\newcommand{\vertexColor}{{\mathopsf{color}}}
\newcommand{\dirStyle}{\mathopsf}
\newcommand{\dirNW}{{\dirStyle{nw}}}
\newcommand{\dirN}{{\dirStyle{n}}}
\newcommand{\dirNE}{{\dirStyle{ne}}}
\newcommand{\dirE}{{\dirStyle{e}}}
\newcommand{\dirSE}{{\dirStyle{se}}}
\newcommand{\dirS}{{\dirStyle{s}}}
\newcommand{\dirSW}{{\dirStyle{sw}}}
\newcommand{\dirW}{{\dirStyle{w}}}
\newcommand{\dirSet}{\ensuremath{\mathcal{D}}}
\newcommand{\allDir}{\ensuremath{\{\dirNW, \dirNE, \dirE, \dirSE, \dirSW, \dirW\}}}
\newcommand{\makecell}[1]{\ensuremath{{#1}^{\scalebox{0.75}{$\hexagon$}}}}
\newcommand{\mc}[1]{\makecell{(#1)}}
\newcommand{\makelat}[1]{\ensuremath{{#1}^{\scalebox{0.6}{$\bigtriangleup$}}}}
\newcommand{\ml}[1]{\makelat{(#1)}}

\newcommand{\cdirNW}{{\makecell{\dirStyle{nw}}}}
\newcommand{\cdirN}{{\makecell{\dirStyle{n}}}}
\newcommand{\cdirNE}{{\makecell{\dirStyle{ne}}}}
\newcommand{\cdirSE}{{\makecell{\dirStyle{se}}}}
\newcommand{\cdirSW}{{\makecell{\dirStyle{sw}}}}
\newcommand{\cdirS}{{\makecell{\dirStyle{s}}}}
\newcommand{\dirCell}{\ensuremath{\makecell{\mathcal{D}}}}
\newcommand{\allCellDir}{\ensuremath{\{\cdirNW, \cdirN, \cdirNE, \cdirSE, \cdirS, \cdirSW\}}}
\newcommand{\diropp}[1]{\ensuremath{\bar{#1}}}

\newcommand{\captionpar}[1]{\textit{#1}}

\renewcommand{\vec}{\overrightarrow}

\newcommand{\NSfct}[1]{\ensuremath{\operatorname{\mathsf{#1}}}}

\newcommand{\inter}[1]{\ensuremath{[\![#1]\!]}}
\newcommand{\rtime}{\NSfct{rtime}}

\newcommand{\CW}{\NSfct{cw}}
\newcommand{\CCW}{\NSfct{ccw}}

\newcommand{\sig}{\NSfct{sig}}

\newcommand{\NSalgoName}[1]{\textup{\scshape #1}}
\newcommand{\algoCoverHex}{\NSalgoName{CoverPseudoHexagon}}

\newcommand{\fitbox}[1]{\resizebox{\textwidth}{!}{#1}}

\newcommand{\subfigB}[1]{
    \begin{subfigure}{.48\textwidth}
        \fitbox{\tablePatsXEasy{1}{#1}}
        \caption{Extensions for \protect\scaling B{#1}}
        \label{fig:scaling:algo:routing:B#1}
    \end{subfigure}
}

\newcommand{\subfigC}[1]{
    \begin{subfigure}{\textwidth}
        \fitbox{\tablePatsXFive{1}{#1}}
        \caption{Extensions for \protect\scaling C{#1}}
        \label{fig:scaling:algo:routing:C#1}
    \end{subfigure}
}

\newtheorem{invariant}{Invariant}

\makeatletter
\newcommand{\ALGOmultiline}[1]{%
  \begin{tabularx}{\dimexpr\linewidth-\ALG@thistlm}[t]{@{}X@{}}
    #1
  \end{tabularx}
}
\makeatother

\newcommand{\incMovie}[3]{\includegraphics[width=#1]{./movies/#2-#3.pdf}}

\newcommand{\tabMovieThirteenFrames}[2]{
    \begin{tabular}{cccc}
        \incMovie{#1}{#2}{0000}
    &   \incMovie{#1}{#2}{0001}
    &   \incMovie{#1}{#2}{0002}
    &   \incMovie{#1}{#2}{0003}
    \\
        \incMovie{#1}{#2}{0004}
    &   \incMovie{#1}{#2}{0005}
    &   \incMovie{#1}{#2}{0006}
    &   \incMovie{#1}{#2}{0007}
    \\
        \incMovie{#1}{#2}{0008}
    &   \incMovie{#1}{#2}{0009}
    &   \incMovie{#1}{#2}{0010}
    &   \incMovie{#1}{#2}{0011}
    \\
        \incMovie{#1}{#2}{0012}
    \end{tabular}
}

\newcommand{\tabMovieFourteenFrames}[2]{
    \begin{tabular}{cccccc}
        \incMovie{#1}{#2}{0000}
    &   \incMovie{#1}{#2}{0001}
    &   \incMovie{#1}{#2}{0002}
    &   \incMovie{#1}{#2}{0003}
    \\
        \incMovie{#1}{#2}{0004}
    &    \incMovie{#1}{#2}{0005}
    &   \incMovie{#1}{#2}{0006}
    &   \incMovie{#1}{#2}{0007}
    \\
        \incMovie{#1}{#2}{0008}
    &   \incMovie{#1}{#2}{0009}
    &   \incMovie{#1}{#2}{0010}
    &   \incMovie{#1}{#2}{0011}
    \\
        \incMovie{#1}{#2}{0012}
    &   \incMovie{#1}{#2}{0013}
    \end{tabular}
}

\newcommand{\NSomitted}[1]{}

\newcommand{\figPatAAA}{Fig.~\ref{fig:scaling:algo:A3:routing}}

\newcommand{\refInvAAA}[1]{\mbox{(\ref{inv:algo:A3}.\ref{#1})}}

\newboolean{withAppendix}

\setboolean{withAppendix}{true} %% true = with appendix; false = without appendix (DNA camera ready version)

\newcommand{\withOrWithoutAppendix}[2]{\ifthenelse{\boolean{withAppendix}}{#1}{#2}}

\newcommand{\omittedSeeFullArticle}{omitted, see~\cite{shape2018DNAfull}}

%% file: introduction2.tex
%%%
%%% INTRODUCTION
%%%

\section{Introduction}

Transcription is the process in which an RNA polymerase enzyme (colored in orange in Fig.~\ref{fig:rna_origami}) synthesizes the temporal copy (blue) of a gene (gray spiral) out of  ribonucleotides of four types {\tt A}, {\tt C}, {\tt G}, and {\tt U}.
The copied sequence is called the \textit{transcript}.

\input{fig-RNA-origami.tex}

The transcript starts folding upon itself into intricate tertiary structures immediately after it emerges from the RNA polymerase.
Fig.~\ref{fig:rna_origami} (Left) illustrates \textit{cotranscriptional folding} of a transcript into a rectangular RNA tile structure while being synthesized out of an artificial gene engineered by Geary, Rothemund, and Andersen \cite{GeRoAn2014}.
The RNA tile is provided with a kissing loop (KL) structure, which yields a $120^{\circ}$ bend, at its four corners, and sets of six copies of it self-assemble into hexagons and further into a hexagonal lattice.
\emph{Structure} is almost synonymous to \emph{function} for RNA complexes since they are highly correlated, as exemplified by various natural and artificial RNAs \cite{MolBiolRNA}.
Cotranscriptional folding plays significant roles in determining the structure (and hence function) of RNAs. 
To give a few examples, introns along a transcript cotranscriptionally fold into a loop recognizable by spliceosome and get excised \cite{MerkhoferHuJohnson2014}, and riboswitches make a decision on gene expression by folding cotranscriptionally into one of two mutually exclusive structures: an intrinsic terminator hairpin and a pseudoknot, as a function of specific ligand concentration \cite{WaStYuLiLu2016}.

What is folded is affected by various environmental factors including transcription rate.
Polymerases have their own transcription rate: e.g., bacteriophage 3ms/nucleotide (nt) and eukaryote 200ms/nt \cite{Isambert2009} (less energy would be dissipated at slower transcription \cite{FeynmanLecComp}).
Changing the natural transcription rate, by adjusting, e.g., NTP concentration \cite{ReWiRaScRiSt1999}, can impair cotranscriptional processes \cite{ChaoKanChao1995, LeMaReNi1993} (note that polymerase pausing can also facilitate efficient folding \cite{WongSosnickPan2007} but it is rather a matter of gene design).
Given a target structure, it is hence necessary to know not only what to fold but at what rate to fold, that is, to know when to fold `em.

The primary goal of both natural and artificial self-assembling systems is to form predictable structures, i.e. shapes grown from precisely placed components, because the form of the products is what yields their functions. Mathematical models have proven useful in developing an understanding of how shapes may self-assemble, and self-assembling finite shapes is one of the fundamental goals of theoretical modeling of systems capable of self-assembly.
e.g. in tile-based self-assembly~\cite{RNAPods,DDFIRSS07,SolWin07} as well as other models of programmable matter~\cite{derakhshandeh2016universal, Nubots}.

An oritatami system (abbreviated as OS) is a novel mathematical model of cotranscriptional folding,  introduced by \cite{GeMeScSe2016}. It abstracts an RNA tertiary structure as a triple of 1) a sequence of abstract molecules (of finite types) called \emph{bead types}, 2) a directed path over a triangular lattice of \emph{beads} (i.e. a location/bead type pair), and 3) a set of pairs of adjacent beads that are considered to interact with each other via hydrogen bonds. Such a triple is called a \textit{configuration}. An abstraction of the RNA tile from~\cite{GeRoAn2014} as a configuration is shown in Fig.~\ref{fig:rna_origami} (Right). In the figure, each bead (represented as a dot) abstractly represents a sequence of 3-4 nucleotides, whose type is not stated explicitly but retrievable from the transcript's sequence of the tile (available in \cite{GeRoAn2014}); moreover, the interactions (or bonds) between pairs of beads are represented by dashed lines. An OS is provided with a finite alphabet $B$ of bead types, a sequence $w$ of beads over $B$ called its \textit{transcript}, and a rule $\heart$, which specifies between which types of beads interactions are allowed. The OS cotranscriptionally folds its transcript $w$, beginning from its initial configuration (\textit{seed}), over the triangular lattice by stabilizing beads of $w$ from the beginning one by one. Two parameters of OS govern the bead stabilization: \textit{arity} and \textit{delay}; arity models valence (maximum number of bonds per bead). Delay models the transcription rate in the sense that the system stabilizes the next bead in such a way that the sequence of the next bead and the $\delta-1$ succeeding beads is folded so as to form as many bonds as possible.%; hence, delay~1 models slowest possible transcription.

Using this model, researchers have mainly explored the computational power of cotranscriptional folding (see \cite{GeMeScSe2016} and the recent surveys \cite{RogersSeki2017, Seki2017}).
In contrast, little has been done on self-assembly of shapes.
Elonen in \cite{Elonen2016} informally sketched how an OS can fold a transcript whose beads are all of distinct types (hardcodable transcript) into a finite shape using a provided Hamiltonian path. 
Masuda et al.~implemented an OS that folds its periodic transcript into a finite portion of the Heighway dragon fractal \cite{MasudaSekiUbukata}.

\paragraph{Our results.}
We initiate a systematic study of shape self-assembly by oritatami systems. We start with the formal definitions of OS and shapes in Section~\ref{sec:def}. As it is NP-hard to decide if a given connected shape of the triangular lattice contains a Hamiltonian path \cite{Arkin}, it is also NP-hard to decide if there is an OS that folds into (self-assembles) a given finite shape. We thus explore the folding of upscaled versions of finite shapes. We introduce three upscaling schemes \scaling An, \scaling Bn and \scaling Cn, where $n$ is the scale factor (see Fig.~\ref{fig:def:upscaling:ABC}). 
We first show that oritatami systems differ fundamentally from tile-assembly systems by exhibiting a family of infinite shapes that can be tile-assembled but cannot be folded by any OS (Theorem~\ref{thm:infinite-shapes-short}, Section~\ref{sec:finite:cut}).
We then show that any shape can be folded at scale factor $2$ by an OS with \emph{unbounded} delay (Theorem~\ref{thm:hard-coded}, Section~\ref{sec:finite:unbounded}). 
In section~\ref{sec:scaling:algo},
we present various incremental algorithms that produce a delay-$1$ arity-$4$ OS that folds any shape from a seed of size $3$, at any scale $n\geq 3$ (Theorems~\ref{thm:scaling:algo:routing:Bn} and \ref{thm:scaling:algo:A3}, Section~\ref{sec:scaling:algo}). For this purpose, we introduce a universal set of $114$ bead types suitable for folding any delay-$1$ tight OS (Theorem~\ref{thm:scaling:algo:universal:bead:type}) that can be used in other oritatami designs.
We then show that the delay impacts our ability to build shapes: we prove that there are shapes that can be folded  (at scale~$1$) with delay $\delta$ but not with delay $\delta'<\delta$ (Theorem~\ref{thm:small-delay-weak}, Section~\ref{sec:small-delay-weak}). 
\withOrWithoutAppendix{%
}{%
    Omitted proofs may be found in \cite{shape2018DNAfull}.
}
%
%Arity is also an important parameter as we proved that delay-$1$ arity-$1$ OS cannot fold any shape significantly larger than the seed  deterministically, regardless of the size of the transcript (Theorem~\ref{thm:d1a1_det_finiteness}, Section~\ref{sec:d1a1}).

These results serve as a foundation for the study of shape-building in this new model of self-assembly, and have the potential to provide better understanding of cotranscriptional folding in biology, as well as improved abilities of experimentalists to design artificial systems that self-assemble via this complex dynamical process. 

\paragraph{Note that} in \cite{HanKim2018DNA} in the present proceedings, the authors study a slightly different problem: they show that one can design an oritatami transcript that folds an upscaled version of a \emph{non-self-intersecting path} (instead of a shape). The initial path may come from the triangular grid or from the square grid. The scale of the resulting path is somewhere in between our scales 3 and 4 according to our definition. Note that the cells are only partially covered by their scheme. Combining their result with our theorem~\ref{thm:hard-coded}, their algorithm provides an oritatami transcript partially covering the upscaled version of any shape at scale 6.

%This implies computational weakness, which might explain natural transcription rate being set not so slow even at the cost of higher energy consumption. 

%% file: fig-RNA-origami.tex
\begin{figure}[b]
%\centering
\scalebox{.93}{%
\begin{minipage}{0.6\linewidth}
\includegraphics[width=\linewidth]{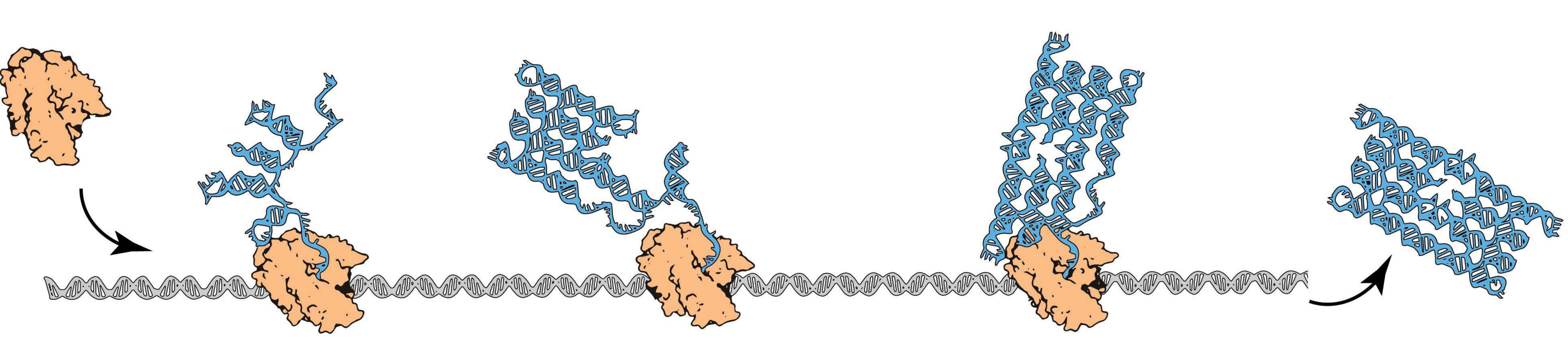}
\end{minipage}
{\Large $\Rightarrow$}
\begin{minipage}{0.3\linewidth}
\scalebox{0.3}{\begin{tikzpicture}
\tikzstyle{mol} = [fill,circle,inner sep=2pt]

\draw[-triangle 90, very thick] (3, 0) node[mol] {}
-- ++(180:1) node[mol] {}
-- ++(180:1) node[mol] {}
-- ++(180:1) node[mol] {}
-- ++(180:1) node[mol] {}
-- ++(180:1) node[mol] {}
-- ++(180:1) node[mol] {}
-- ++(180:1) node[mol] {}
-- ++(180:1) node[mol] {}
-- ++(180:1) node[mol] {}
-- ++(60:1) node[mol] {}
-- ++(0:1) node[mol] {}
-- ++(0:1) node[mol] {}
-- ++(60:1) node[mol] {}
-- ++(180:1) node[mol] {}
-- ++(180:1) node[mol] {}
-- ++(60:1) node[mol] {}
-- ++(0:1) node[mol] {}
-- ++(0:1) node[mol] {}
-- ++(60:1) node[mol] {}
-- ++(180:1) node[mol] {}
-- ++(180:1) node[mol] {}
-- ++(60:1) node[mol] {}
-- ++(0:1) node[mol] {}
-- ++(0:1) node[mol] {}
-- ++(0:1) node[mol] {}
-- ++(0:1) node[mol] {}
-- ++(0:1) node[mol] {}
-- ++(0:1) node[mol] {}
-- ++(240:1) node[mol] {}
-- ++(180:1) node[mol] {}
-- ++(180:1) node[mol] {}
-- ++(180:1) node[mol] {}
-- ++(240:1) node[mol] {}
-- ++(0:1) node[mol] {}
-- ++(0:1) node[mol] {}
-- ++(0:1) node[mol] {}
-- ++(240:1) node[mol] {}
-- ++(180:1) node[mol] {}
-- ++(180:1) node[mol] {}
-- ++(180:1) node[mol] {}
-- ++(240:1) node[mol] {}
-- ++(0:1) node[mol] {}
-- ++(0:1) node[mol] {}
-- ++(0:1) node[mol] {}
-- ++(0:1) node[mol] {}
-- ++(0:1) node[mol] {}
-- ++(0:1) node[mol] {}
-- ++(0:1) node[mol] {}
-- ++(60:1) node[mol] {}
-- ++(180:1) node[mol] {}
-- ++(180:1) node[mol] {}
-- ++(180:1) node[mol] {}
-- ++(60:1) node[mol] {}
-- ++(0:1) node[mol] {}
-- ++(0:1) node[mol] {}
-- ++(0:1) node[mol] {}
-- ++(60:1) node[mol] {}
-- ++(180:1) node[mol] {}
-- ++(180:1) node[mol] {}
-- ++(180:1) node[mol] {}
-- ++(60:1) node[mol] {}
-- ++(0:1) node[mol] {}
-- ++(0:1) node[mol] {}
-- ++(0:1) node[mol] {}
-- ++(0:1) node[mol] {}
-- ++(0:1) node[mol] {}
-- ++(0:1) node[mol] {}
-- ++(240:1) node[mol] {}
-- ++(180:1) node[mol] {}
-- ++(180:1) node[mol] {}
-- ++(240:1) node[mol] {}
-- ++(0:1) node[mol] {}
-- ++(0:1) node[mol] {}
-- ++(240:1) node[mol] {}
-- ++(180:1) node[mol] {}
-- ++(180:1) node[mol] {}
-- ++(240:1) node[mol] {}
-- ++(0:1) node[mol] {}
-- ++(0:1) node[mol] {}
-- ++(240:1) node[mol] {}
-- ++(180:1) node[mol] {}
-- ++(180:1) node[mol] {}
-- ++(180:1)
;

\draw[dashed, very thick]
(-5, 0) -- ++(120:1)
(-4, 0) -- ++(120:1)
(-3, 0) -- ++(120:3)
(-2, 0) -- ++(120:3)
(-1, 0) -- ++(120:5)
(0, 0) -- ++(120:1) ++(120:1) -- ++(120:3)
(1, 0) -- ++(120:1) ++(120:1) -- ++(120:3)
(2, 0) -- ++(120:1) ++(120:1) -- ++(120:1) ++(120:1) -- ++(120:1)
(3, 0) -- ++(120:1) ++(120:3) -- ++(120:1)
(4, 0) -- ++(120:1) ++(120:1) -- ++(120:1) ++(120:1) -- ++(120:1)
(5, 0) -- ++(120:3)
(6, 0) -- ++(120:3) ++(120:1) -- ++(120:1)
(7, 0) -- ++(120:5)
(8, 0) ++(120:2) -- ++(120:3)
(9, 0) ++(120:2) -- ++(120:3)
(10, 0) ++(120:4) -- ++(120:1)
(11, 0) ++(120:4) -- ++(120:1)
(3, 0) -- ++(0:1)
;

\foreach \y in {3, 5} {
\draw[dashed,very thick] (0, 0)++(60:\y) -- ++(300:1);	
}

\end{tikzpicture}}

\end{minipage}
}%scalebox
\caption{
(Left) RNA Origami \cite{GeRoAn2014}. 
(Right) An abstraction of the resulting RNA tile in the oritatami system, where a dot $\bullet$ represents a sequence of 3-4 nucleotides, and the solid arrow and dashed lines represent its transcript and interactions based on hydrogen bonds between nucleotides, respectively.
}
\label{fig:rna_origami}
\end{figure}

%% file: shapes-def.tex
%!TEX root = Oritatami-shapes.tex

%%% --------------------------------------------------------------------------------
\section{Definitions}
%%% --------------------------------------------------------------------------------

\label{sec:def}

%\paragraph{Notations.}
%We index the letters of every word $u=u_0\ldots u_{|u|-1}$ from $0$ to $|u|-1$. Given two words $u$ and $v$, we denote by $u\cdot v$ their concatenation: ${u\cdot v = u_0\ldots u_{|u|-1} v_0\ldots v_{|v|-1}}$. We denote by $u^\infty$ the oneway infinite periodic word $u\cdot u\cdots$. For all $i\leq j$, we denote by $u_{i..j}$ the (possibly empty) factor $u_{\max(0,i)}\ldots u_{\min(j,|u|-1)}$. The empty word is denoted by~$\epsilon$. The indices in the notation $O_{L}()$ where $L$ is a list of variables (for instance $L=A,B$) indicates that the constant in the $O()$ only depends on the variable in $L$ (for instance $A$ and $B$) and on no other values.  

% --------------------------------------------------------------------------------
\subsection{Oritatami System}
% --------------------------------------------------------------------------------
\label{sec:def:OS}

Let $B$ be a finite set of \emph{bead types}.
A \emph{routing} $r$ of a bead type sequence $w \in B^* \cup B^{\mathbb{N}}$ is a directed self-avoiding path in the triangular lattice $\Tlat$,\footnote{The triangular lattice is defined as $\mathbb{T} = (\mathbb{Z}^2, \sim)$, where $(x, y) \sim (u, v)$ if and only if $(u, v) \in \cup_{\epsilon = \pm1}\{{(x+\epsilon, y)},  {(x, y+\epsilon)}, {(x+\epsilon, y+\epsilon)}\}$. Every position $(x,y)$ in $\Tlat$ is mapped in the euclidean plane to $x\cdot X + y \cdot Y$ using the vector basis $X = (1,0)$ and $Y = \rotateClockwiseAA{X} = (-\frac12, -\frac{\sqrt3}2)$.} where for all integer $i$, vertex $r_i$ of $r$ is labelled by $w_i$. $r_i$ is the \emph{position} in $\Tlat$ of the $(i+1)$th bead, of type $w_i$, in routing $r$. A \emph{partial routing} of a sequence $w$ is a routing of a prefix of $r$. %In a configuration, a \emph{bead} has a type and a position. %We write the position of a bead $b$ in configuration $c$ as $\pos_c(b)$.

\paragraph{An Oritatami system $\TMO=(B, w,\heart,\delay,\alpha)$} is composed of (1) a set of bead types $B$, (2) a (possibly infinite) bead type sequence $w$, called the \emph{transcript}, (3) an \emph{attraction rule}, which is a symmetric relation $\heart\subseteq B^2$, (4) a parameter $\delay$ called the \emph{delay}, and (5) a parameter $\alpha$ called the \emph{arity}. %$\TMO$ is said \emph{periodic} if $x$ is infinite and periodic. Periodicity ensures that the ``program'' $w$ embedded in the oritatami system is finite (does not hardcode any specific behavior) and at the same time allows arbitrary long computation. 

We say that two bead types $a$ and $b$ \emph{attract} each other when $a\heart b$. Given a (partial) routing $r$ of a bead type sequence $w$, we say that there is a \emph{potential (symmetric) bond $r_ir_j$} between two adjacent positions $r_i$ and $r_j$ of $r$ in $\Tlat$ if  $w_i\heart w_j$ and $|i-j|>1$. A \emph{set of bonds} $H$ for a (partial) routing $r$ is a subset of its potential bonds. A couple $c=(r,H)$ is called a (partial) \emph{configuration} of $w$.  The \emph{arity} $\alpha_i(c)$ of position $r_i$ in the partial configuration $c=(r,H)$ is the number of bonds in $H$ involving $r_i$, i.e. ${\alpha_i(c) = \#\{j\,: r_ir_j\in H\}|}$.  A (partial) configuration $c$ is \emph{valid} if each position $r_i$ is involved in at most $\alpha$ bonds in $H$, i.e. if $(\forall i)~ \alpha_i(c)\leq \alpha$. We denote by $h(c)=|H|$ the number of bonds in configuration~$c$.

For any partial valid configuration $c=(r,H)$ of some sequence $w$, an \emph{elongation} of $c$ by $k$ beads (or \emph{$k$-elongation}) is a partial valid configuration $c'=(r',H')$ of $w$ of length $|c|+k$ where $r'$ extends the self-avoiding path $r$ by $k$ positions and such that $H\subseteq H'$ . We denote by $\Ccal_w$ the set of all partial configurations of $w$ (the index $w$ will be omitted when the context is clear). We denote by $\elong{c}{k}$ the set of all  $k$-elongations of a partial configuration~$c$ of sequence~$w$.

\paragraph{Oritatami dynamics.}
The folding of an oritatami system is controlled by the delay $\delta$ and the arity $\alpha$. Informally, the configuration grows from a \emph{seed configuration}, one bead at a time. This new bead adopts the position(s) that maximise the  number of valid bonds the configuration can make when elongated by $\delta$ beads in total.  This dynamics is \emph{oblivious} as it keeps no memory of the previously preferred positions; it differs thus slightly from the hasty dynamics studied in~\cite{GeMeScSe2016} but is more prevailing in the OS research \cite{GeMeScSe2017,HanKim2017,MasudaSekiUbukata,OtaSeki2017,RogersSeki2017} because it seems closer to experimental conditions such as in \cite{GeRoAn2014}.

Formally,  given an oritatami system $\TMO = (B, w, \heart, \delay, \alpha)$ and a \emph{seed configuration} $\sigma$ of the $|\sigma|$-prefix of $w$, we denote by $\Ccal_{\sigma,w}$ the set of all partial  configurations of the sequence $w$ elongating the seed configuration $\sigma$. The considered \emph{dynamics} ${\Dynamics:2^{\Ccal_{\sigma,w}}\rightarrow2^{\Ccal_{\sigma,w}}}$ maps every subset $S$ of partial configurations of length~$\ell$, elongating $\sigma$, of the sequence $w$ to the subset $\Dynamics(S)$ of partial configurations of length~$\ell+1$ of $w$ as follows:\\
\centerline{
$
\displaystyle{\Dynamics(S) = \bigcup_{\mbox{\footnotesize $c\in S$}}\, \underset{\mbox{\footnotesize$\gamma\in\elong{c}{1}$}}{\arg\max} \left(\, 
%\max\left\{H(\eta)\,:\, \eta\in\elong{\gamma}{(\delta-1)}\right\}\right)
\max_{\mbox{\footnotesize$\eta\in\elong{\gamma}{\min(\delta-1,\, |w|-|\gamma|)}$}} h(\eta)\,\right)
}%centerline
$}

We say that a (partial) configuration $c$ \emph{produces} a configuration $c'$ over $w$, denoted ${c\vdash c'}$, if $c'\in\Dynamics(\{c\})$. We write $c\vdash^* c'$ if there is a sequence of configurations $c=c^0,\ldots, c^t = c'$, for some $t\geq 0$, such that $c^0 \vdash \cdots \vdash c^t$. A sequence of configurations $c=c^0\vdash\cdots \vdash c^t=c'$ is called a \emph{foldable sequence  over $w$ from configuration $c$ to configuration $c'$}. The \emph{foldable configurations} in $t$~steps of $\TMO$ are the elongations of the seed configuration~$\sigma$ by $t$ beads in the set $\Dynamics^t(\{\sigma\})$. We denote by $\prodasm{\TMO} = \cup_{t\geq0} \Dynamics^t(\{\sigma\})$ the set of all foldable configurations. A configuration $c\in\prodasm{\TMO}$ is \emph{terminal} if $\Dynamics(\{c\}) =\varnothing$. We denote by $\termasm{\TMO}$ the set of all terminal foldable configurations of $\TMO$. A finite foldable sequence $\sigma = c^0\vdash \cdots\vdash c^t$ \emph{halts} at $c^t$ after $t$ steps if $c^t$ is terminal; then, $c^t$ is called the \emph{result} of the foldable sequence. A foldable sequence may halt after $|w|-|\sigma|$ steps or earlier if  the growth is geometrically obstructed  (i.e., if no more elongation is possible because the configuration is trapped in a closed area). An infinite foldable sequence $\sigma = c^0\vdash \cdots\vdash c^t\vdash \cdots$ admits a \emph{unique limiting configuration} $c^\infty = \sqcup_{t}\, c^t$ (the superposition of all the configurations $(c^t)$), which is called the \emph{result} of the foldable sequence. 

We say that the oritatami system is \emph{deterministic} if at all time~$t$, $\Dynamics^t(\{\sigma\})$ is either a singleton or the empty set. In this case, we denote by $c^t$ the configuration at time~$t$, such that: $c^0 = \sigma$ and $\Dynamics^t(\{\sigma\}) = \{c^t\}$ for all $t>0$; we say that the partial configuration $c^t$ \emph{folds (co-transcriptionally) into} the partial configuration $c^{t+1}$ deterministically. In this case, at time $t$, the $(t+1)$-th bead of $w$ is placed in $c^{t+1}$ at the position that maximises the number of valid bonds that can be made in a $\min(\delta, |w|-t-|\sigma|)$-elongation of $c^t$. Note that when $\alpha\geq4$ the arity constraint vanishes (as a vertex may bond to at most $4$ neighbors, $5$ if the growth is at a dead end) and then, there is only one maximum-size bond set for every routing, consisting of all its potential bonds.  

%---------------------
\subsection{Shape folding and scaling}
%---------------------
\label{sec:def:shape:upscaling}

The goal of this article is to study how to fold shapes. A \emph{shape} is a connected set of points in $\Tlat$. The \emph{shape} associated to a configuration $c=(r,H)$ of an OS $\TMO$ is the set of the points $S(c) = \cup_i \{r_i\}$ covered by the routing of $c$. A shape $S$ is \emph{foldable} from a seed of size $s$ if there is a deterministic OS $\TMO$ and a seed configuration $\sigma$ with $|\sigma| = s$, whose terminal configuration has shape $S$.  

Note that every shape admitting a Hamiltonian path is trivially foldable from a seed of size $|S|$, whose routing is a Hamiltonian path of the shape itself. The challenge is to design an OS folding into a given shape whose seed size is an \emph{absolute} constant.  
%As we will show next, there are shapes that are not foldable in this sense. 
One classic approach in self-assembly is then to try to fold an \emph{upscaled} version of the shape. The goal is then to minimize the scale at which an upscaled version of every shape can be folded.     

From now on, we denote by $(i,j)\in\mathbb N^2$ the point $i\cdot X + j \cdot Y$ of $\Tlat$ in $\mathbb R^2$ where $X = (1,0)$ (east) and $Y = (-\frac12, -\frac{\sqrt3}2)$ (south west) in the canonical basis.

As it turns out, there are different possible upscaling schemes for shapes in~$\Tlat$. A \emph{scaling scheme} $\Lambda = (\lambda, \mu)$ of $\Tlat$ is defined by a homothetic linear map $\lambda$ from $\Tlat$ to $\Tlat$, and a shape $\mu$ containing the point $(0,0)$, called the \emph{cell mold}. For all $p\in\Tlat$, the \emph{cell} associated to $p$ by $\Lambda$ is the set $\cell(p) = \lambda(p)+\mu = \{\lambda(p)+q: q\in\mu\}$, i.e. the translation of the cell mold by $\lambda(p)$. $\lambda(p)$ is called the \emph{center} of the cell $\cell(p)$. The $\cell$-scaling of a shape $S$ is then the set of points $\cell(S) = \cup_{p\in S} \cell(p)$. We say that two cells $\cell(p)$ and $\cell(q)$ are neighbors, denoted by $\cell(p)\sim\cell(q)$, if they intersect or have neighboring points, i.e. if $\cell(p)\cap\cell(q)\neq\varnothing$ or there are two points $p'\in\cell(p)$ and $q'\in\cell(q)$ such that $p'\sim q'$. We require upscaling schemes to preserve the topology of~$S$, in particular that $\cell(p)\sim\cell(q)$ iff $p\sim q$.
%(see appendix page~\pageref{sec:def:shape:upscaling:app}). 
%
%A scaling $\Lambda$ is \emph{valid} if it preserves the topology of any shape, that is if:  (1) for all shape $S$ and $p\in\Tlat$, $p\in S$ if and only if $\lambda(p)\in \Lambda(S)$ (we do not allow a cell to be fully covered by others); (2) for all $p,q\in\Tlat$, we have $p\sim q$ if and only if $\Lambda(p) \sim \Lambda(q)$ (cells are neighbors if and only if their associated points in the original shape are neighbors). We say that a scaling $\Lambda$ is \emph{fully covering} if every shape $S$ without hole is mapped to a shape $\Lambda(S)$ without hole.%
%
%\footnote{Recall that a hole of a shape $S$ is a finite non-empty connected component of $\Tlat\smallsetminus S$.}
%
%
%The \emph{compactness} $\compactness_\Lambda$ of a valid scaling $\Lambda$ is defined as the limit of the ratio $|H_n|/|\Lambda(H_n)|$ where $H_n = \{(i,j)\in\Tlat: |i|< n, |j|\leq -1, |i-j|< n\}$ is the (filled) hexagon of radius $n-1$ with $n$ vertices on each side. 
%\todoi{using $\bar\mu$ is incorrect for generic valid scaling scheme... Need to find an other way to say that. Compactness of each scaling is correct}
%But, for a valid scaling, if we let $\bar\mu = \mu\smallsetminus (\mu+\lambda(\{(1,0), (0,1), (1,1)\})$, $\bar\Lambda(p) = \lambda(p)+\bar\mu$ and $\bar\Lambda(S) = \cup_{p\in S} \bar\Lambda(p)$, then for all shape $S$, 
%%
%$$
%\bar\Lambda(S) \subseteq \Lambda(S) \subseteq \bar\Lambda(S) + \lambda(\{(1,0),(0,1),(1,1)\}).
%$$
%%
%A simple calculation shows that $\compactness_\Lambda = 1/|\bar\mu|$ and is thus one over an integer. 
%
We consider the following upscaling schemes (see Fig.~\ref{fig:def:upscaling:ABC}):
%, all of them are valid and fully covering: 
%\todoi{NOTE: \scaling An = ex-$n.0$; \scaling Bn = ex-$n.0+$; \scaling Cn = ex-$n.5$)}
%
\begin{description}
\item[Scaling \scaling An:] $\linscaling An (i,j) = i\cdot(n-1,1-n) + j\cdot(n-1,2n-2)$ and $\muscaling An = H_n$   
\item[Scaling \scaling Bn:] $\linscaling Bn (i,j) = i\cdot(n-1,-n) + j\cdot(n,2n-1)$ and $\muscaling Bn  = H_n$   
\item[Scaling \scaling Cn:] $\linscaling Cn (i,j) = i\cdot(n,-n) + j\cdot(n,2n)$ and $\muscaling Cn  = H'_n$   
\end{description}
where ${H_n = \{(i,j)\in\Tlat: |i|< n, |j|< n, |i-j|< n\}}$ is the (filled) hexagon of radius $n-1$ with $n$ vertices on each side, and ${H'_n = \{(i,j)\in\Tlat: -n< i \leq n, -n <  j \leq n, -n\leq i-j < n\}}$ is the irregular hexagon whose sides are of alternating sizes $n$ and $n+1$. Note that ${H_n \subset H'_n \subset H_{n+1}}$.
\begin{figure}[t]
    \resizebox{\textwidth}{!}{
        \includegraphics[height=3cm]{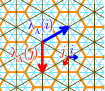}
        \includegraphics[height=3cm]{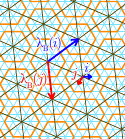}
        \includegraphics[height=3cm]{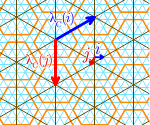}
        \includegraphics[height=3cm]{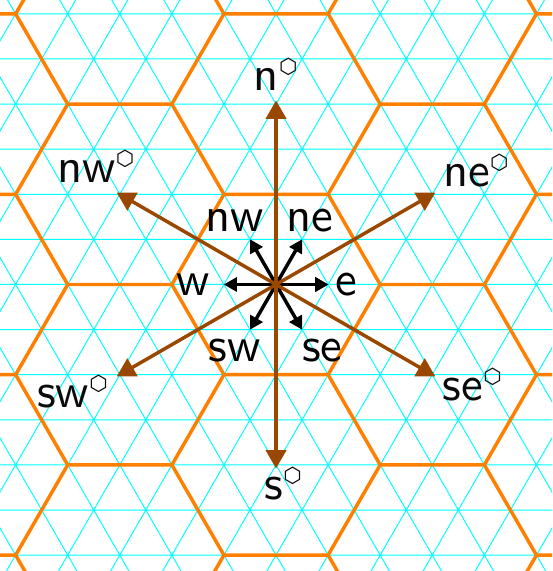}
    }  
    \caption{The three upscaling schemes \scaling A3, \scaling B3 and \scaling C3 (cell boundaries are represented in orange and the upscaled triangular grid in brown); to the right: the lattice directions $\protect\dirSet = \protect\allDir$ in $\Tlat$, and the cell directions $\protect\dirCell = \protect\allCellDir$.}
    \label{fig:def:upscaling:ABC}
\end{figure}
Each of these upscaling schemes have their ups and downs:
\begin{itemize}
\item Every cell in \scaling An is a regular hexagon. It is the most compact but, as the sides of the cells overlap, the area of $\Linscaling An(S)$ scales linearly only asymptotically with the size of the original shape $S$. In particular empty cells are smaller than occupied cell. 
\item Every cell in \scaling Bn is a regular hexagon. It is less compact than \scaling An and twisted, but the edges of neighboring cells never overlap so the area of $\Linscaling An(S)$ scales linearly with the size of the original shape $S$.
\item \scaling Cn can be considered as a non-overlapping version of \scaling A{n+1} where the \cdirNW-, \cdirN- and \cdirNE-sides of each cell have been trimmed by $1$. It is isotropic as its cells are irregular hexagons, but it is untwisted and $\Linscaling Cn(S)$ scales linearly with the size of the original shape $S$. One can also see the irregular hexagons as concentric spheres growing from the center of the triangles in lattice $\Tlat$. 
\end{itemize}

In terms of the resulting size of $\Linscaling{}{}(S)$, \scaling An is strictly more compact than \scaling Bn which is strictly more compact than \scaling Cn which is as compact as \scaling A{n+1} for all $n\geq 2$. $n$ is referred as the \emph{scale} for each scheme.
%as it is the radius of each cell (note that the radius for cell in \scaling Cn is closer to $n+\frac12$). 
Our goal is to find an OS with constant seed size for each of these schemes that can fold any shape at the smallest scale $n$. 

\NSomitted{
\begin{figure}[H]
\hfill
\includegraphics[height=3cm]{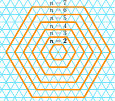}
\hfill
\includegraphics[height=3cm]{scale-3_0.pdf}
\hfill
\,
\caption{\captionpar{Left:} Scaling \scaling An cell shapes. \captionpar{Right:}  the cells at scale \scaling C3 (in orange) and the underlying  rotated triangular lattice (in brown), by $-30^\circ$, whose vertices are located at the center vertices of the hexagons.}
\end{figure}

\begin{figure}[H]
\hfill
\includegraphics[height=3cm]{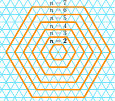}
\hfill
\includegraphics[height=3cm]{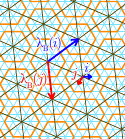}
\hfill
\,
\caption{\captionpar{Left:} Scaling \scaling Bn cell shapes. \captionpar{Right:}  the cells at scale \scaling C3 (in orange) and the underlying  rotated triangular lattice (in brown), by $-30^\circ+\epsilon$, whose vertices are located at the center vertices of the hexagons.}
\end{figure}

\begin{figure}[H]
\hfill
\includegraphics[height=3cm]{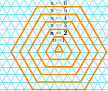}
\hfill
\includegraphics[height=3cm]{scale-3_5.pdf}
\hfill
\,
\caption{\captionpar{Left:} Scaling \scaling Cn cell shapes. \captionpar{Right:}  the cells at scale \scaling Cn (in orange) and the underlying rotated triangular lattice (in brown), by $-30^\circ$, whose vertices are located at the center of the triangle at the center of each cell in the original triangular lattice (the orange dot in the figure to the left).}
\end{figure}
}%NSomitted

Before we give our algorithms, we note the importance of scaling the shape in order to self-assemble it. Figure~\ref{fig:impossible-shape-unscaled} shows an example of a shape which cannot be self-assembled by any OS (at scale~$1$), as it does not contain any Hamiltonian path. In fact, \cite{Arkin} proves that it is NP-hard to decide if a shape in $\Tlat$ has a Hamiltonian path. Note that, if we are given a Hamiltonian path, there is a (hard-coding) OS that ``folds'' it, by simply using this path as the seed with no transcript. The existence of an OS (with unbounded seed) self-assembling a shape is thus equivalent to the existence of an Hamiltonian path. It follows that:

\begin{observation}\label{obs:NP-hard}
Given an arbitrary shape $S$, it is NP-hard to decide if there is an oritatami system (with unbounded seed) which self-assembles it.
\end{observation}

In Section~\ref{sec:scaling:algo}, we will present three algorithms building delay-1 OS that fold into arbitrary shapes at any of the scales \scaling An, \scaling Bn, and \scaling Cn with $n\geq 3$. 

%% file: shapes-finitely-cut.tex
\section{Infinite shapes with finite cut}
%------------------------------

\label{sec:finite:cut}

The self-assembly of shapes in oritatami systems is fundamentally different from the self-assembly of shapes in the Tile Assembly Model due to the fact that every configuration in an OS has a routing that is a linear path of beads. To illustrate this difference, let us say an infinite shape has a \emph{finite cut}  if there is a finite subset of points $K$ in $S$ such that $S\smallsetminus K$ contains at least two \emph{infinite} connected components, $S_1$ and $S_2$. As every path going between $S_1$ and $S_2$ has to pass through the cut $K$ of finite size, after a finite number of back and forth passes it will no longer be possible and the routing will not be able to fill at least one of $S_1$ or $S_2$.
\NSomitted{
Thus,

\begin{theorem}\label{thm:infinite-shapes-short}
For any infinite shape $S$ having a finite cut, there is no OS that folds into $S$.
\end{theorem}

Details of the proof can be found in Section~\ref{sec:finite-cut-append}. Furthermore, since any scaling of $S$ has also a finite cut, scaling cannot help here and we conclude that:

\begin{corollary}\label{cor:infinite-shapes-short}
Let $S$ be an infinite shape having a finite cut. Then for any scaling scheme $\Lambda$ and any OS $\TMO$, $\Lambda(S)$ is not foldable in $\TMO$.
\end{corollary}
}%NSomitted
Furthermore, since any scaling of $S$ has also a finite cut, scaling cannot help here and we conclude that:
%(proof omitted, see section~\ref{sec:finite-cut-append}):

\begin{theorem}\label{thm:infinite-shapes-short}
Let $S$ be an infinite shape having a finite cut. Then for any scaling scheme $\Lambda$ and any OS $\TMO$, $\Lambda(S)$ is not foldable in $\TMO$.
\end{theorem}

%% file: finite-shapes-short.tex
\section{Self-assembling finite shapes at scale 2 with linear delay}

\label{sec:finite:unbounded}

\NSomitted{
\begin{wrapfigure}{r}{1.1in}
  \begin{center}
  \vspace{-20pt}
    \includegraphics[width=1.0in]{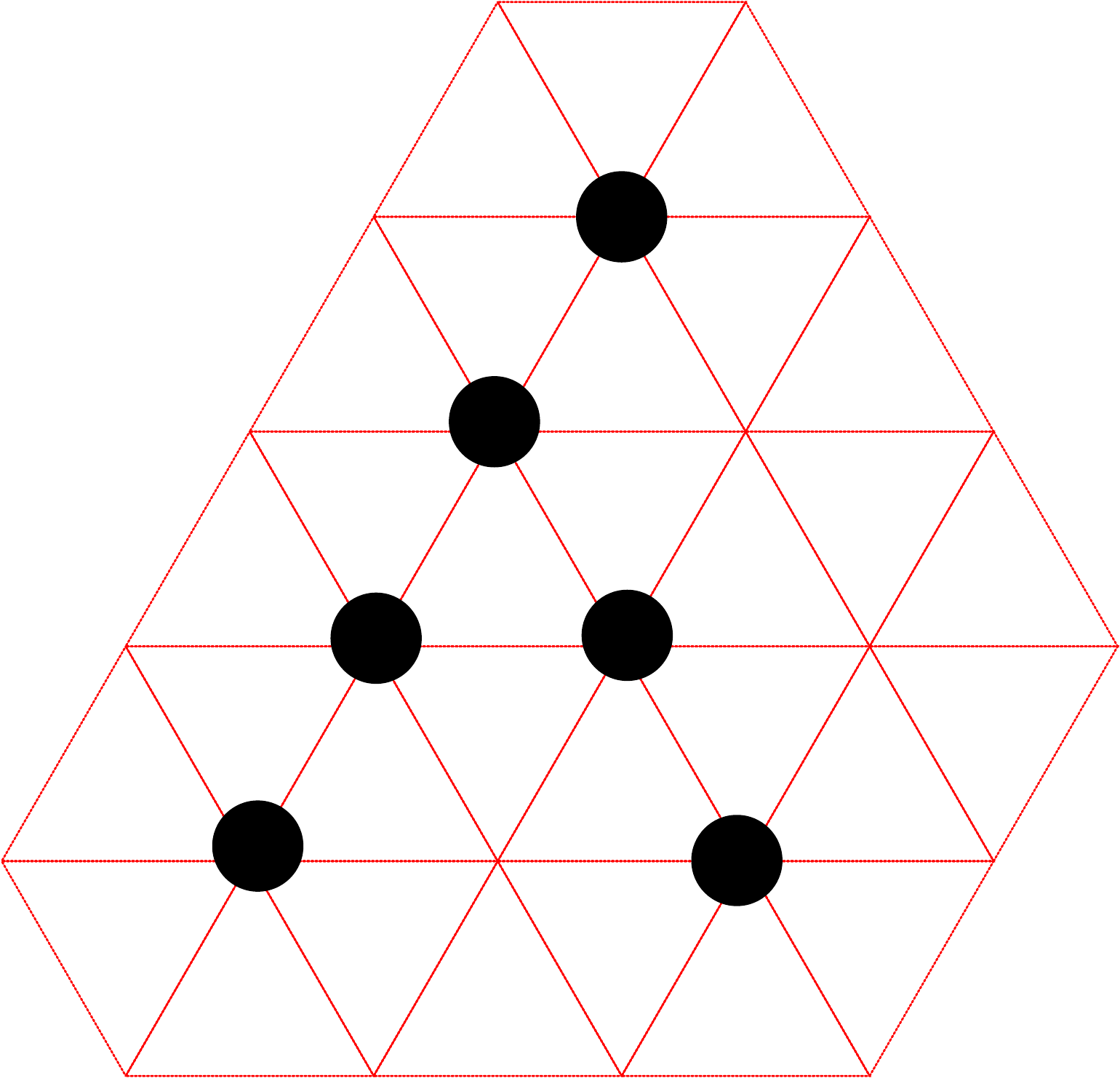}
  \end{center}
  \caption{An example shape which cannot be self-assembled by an oritatami system without being scaled\vspace{-12pt}}
  \label{fig:impossible-shape-unscaled}
\end{wrapfigure}
}%NSomitted

In this section, we show how to create an oritatami system for building an arbitrary finite shape $S$ at scales~\scaling A2, \scaling B2, and \scaling C2, with a delay equal to $|S|$.

\NSomitted{Before we give our algorithm, we note the importance of scaling the shape in order to self-assemble it. Figure~\ref{fig:impossible-shape-unscaled} shows an example of a shape which cannot be self-assembled by any OS (at scale~$1$), as it does not contain any Hamiltonian path. In fact it is NP-hard to decide if a shape has a Hamiltonian path~\cite{Arkin}; and given a Hamiltonian path, there is a (hard-coding) OS that folds it, which simply uses this path as the seed with no transcript. The existence of an OS (with unbounded seed) self-assembling a shape is thus equivalent to the existence of an Hamiltonian path. It follows that:

%The reason for this is clear by inspection and observation of the fact that, for a shape to be self-assembled by an OS, there must be a Hamiltonian path through that shape (which is necessary for the routing of the beads).  With three points which each have just one neighbor, there is no way to enter or exit each of those neighbors exactly once.

%\begin{figure}[htp]
%\centering
%\includegraphics[width=1.0in]{images/impossible-shape-unscaled}
%\caption{An example shape which cannot be self-assembled by an oritatami system without being scaled}
%\label{fig:impossible-shape-unscaled}
%\end{figure}

%We note that for an arbitrary finite shape $S$ in the triangular grid, it is NP-hard to decide if there is an oritatami system which self-assembles it.

\begin{observation}\label{obs:NP-hard}
Given an arbitrary shape $S$, it is NP-hard to decide if there is an oritatami system which self-assembles it.
\end{observation}
}%NSomitted

%Although it is easy to determine that some shapes in the triangular grid graph have no Hamiltonian path, in general it is NP-complete \cite{Arkin}.  For a shape $S$ in the triangular grid, we can define an oritatami system such that the seed configuration has the shape of $S$ and the transcript is empty. Therefore, there exists an oritatami system which self-assembles $S$ if and only if $S$ contains a Hamiltonian path. Hence it is NP-hard to determine if there is an oritatami system which self-assembles $S$.
%
%Given an arbitrary finite shape $S$ and a constant $c\in \N$ that is independent of $S$, the complexity of the problem of determining if there is an oritatami system with seed size at most $c$ which self-assembles $S$ is an open problem. 
%
%For our next result, we show that for any finite shape $S$ there are oritatami systems with seed size $3$, $|S|$ bead types, delay $|S|$, and arity $4$ which deterministically self-assemble $S$ at scale factors \scaling A2, \scaling B2, or \scaling C2..

The theorem below proves that: every \scaling A2-, \scaling B2- and \scaling C2-upscaled version of a given shape $S$ has a Hamiltonian cycle (HC); and furthermore, presents an algorithm that outputs an OS with delay $|\Linscaling{}{}(S)|= O(|S|)$ that folds into this cycle from a seed of size~$3$. The OS relies on set of beads following the HC and custom designed to bind to all of their neighboring beads. Using a delay factor equivalent to the size of the shape, all beads after the first three of the seed are transcribed before they then all lock into their optimal placements along the HC which allows them to form the maximum number of bonds. A schematic overview of the scaling, HC, and bead path is shown in Fig.~\ref{fig:small-series}.
%(details of the proof can be found in Section~\ref{sec:hard-coded-shapes-long}.)

\begin{theorem}\label{thm:hard-coded}
Let $S$ be a finite shape.  For each scale $s \in \{\scaling A2, \scaling B2, \scaling C2\}$, there is an OS $\mathcal{O}_S$ with delay $|\Linscaling{}{s}(S)| = O(|S|)$ and seed size~$3$ that self-assembles $S$ at scale $s$.
\end{theorem}

\begin{figure}[t]
\resizebox{\textwidth}{!}{
	\begin{subfigure}[t]{0.6in}
	\centering
	\includegraphics[width=\textwidth]{images/impossible-shape-unscaled}
	\caption{}
	\label{fig:impossible-shape-unscaled}
	\end{subfigure}
\quad    
	\begin{subfigure}[t]{0.6in}
	\centering
	\includegraphics[width=0.6in]{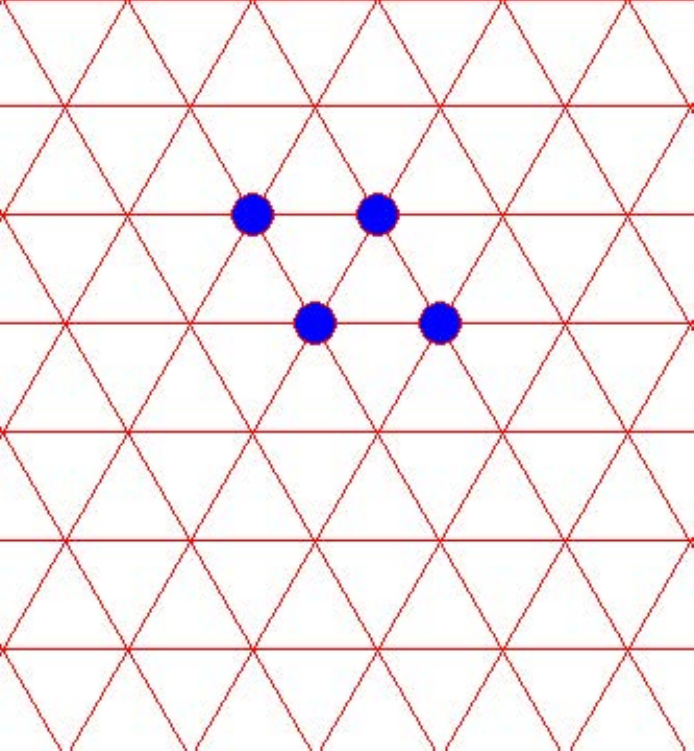}
	\caption{}
	\label{fig:small-series-original}
	\end{subfigure}
\quad
	\begin{subfigure}[t]{0.6in}
	\centering
	\includegraphics[width=0.6in]{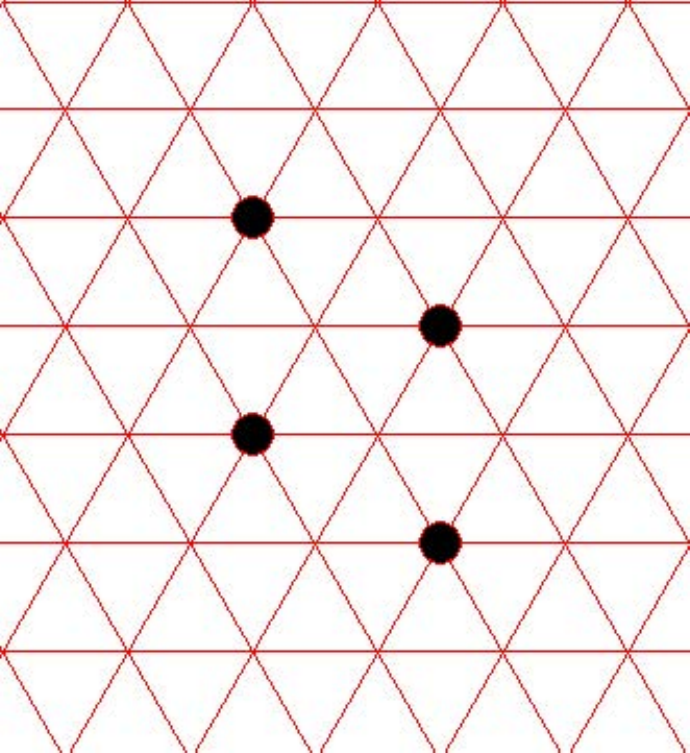}
	\caption{}
	\label{fig:small-series-scaled}
	\end{subfigure}
\quad
	\begin{subfigure}[t]{0.6in}
	\centering
	\includegraphics[width=0.6in]{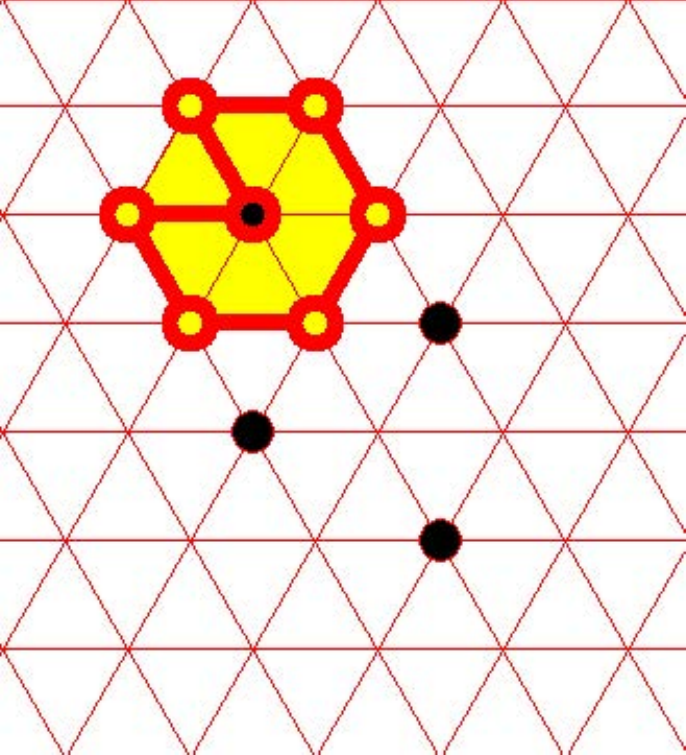}
	\caption{}
	\label{fig:small-series-first}
	\end{subfigure}
\quad
	\begin{subfigure}[t]{0.6in}
	\centering
	\includegraphics[width=0.6in]{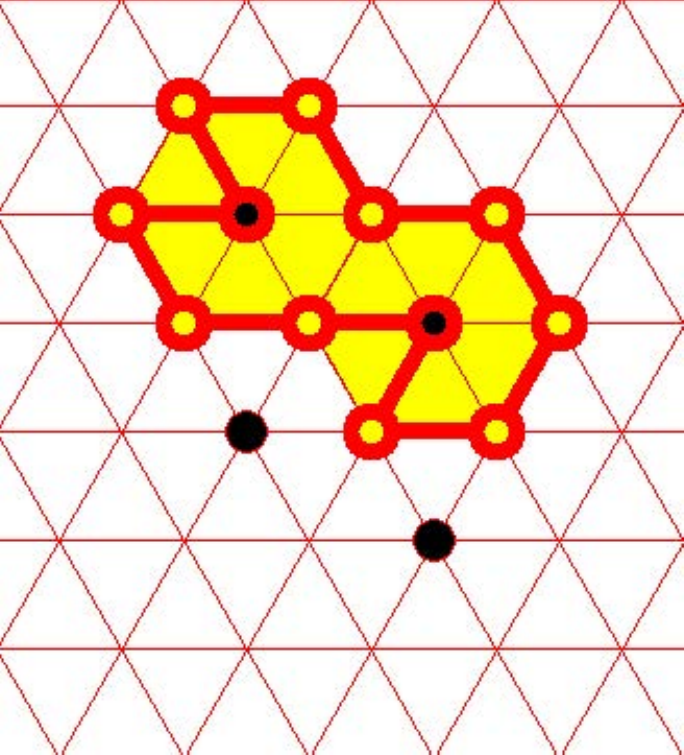}
	\caption{}
	\label{fig:small-series-second}
	\end{subfigure}
\quad
	\begin{subfigure}[t]{0.6in}
	\centering
	\includegraphics[width=0.6in]{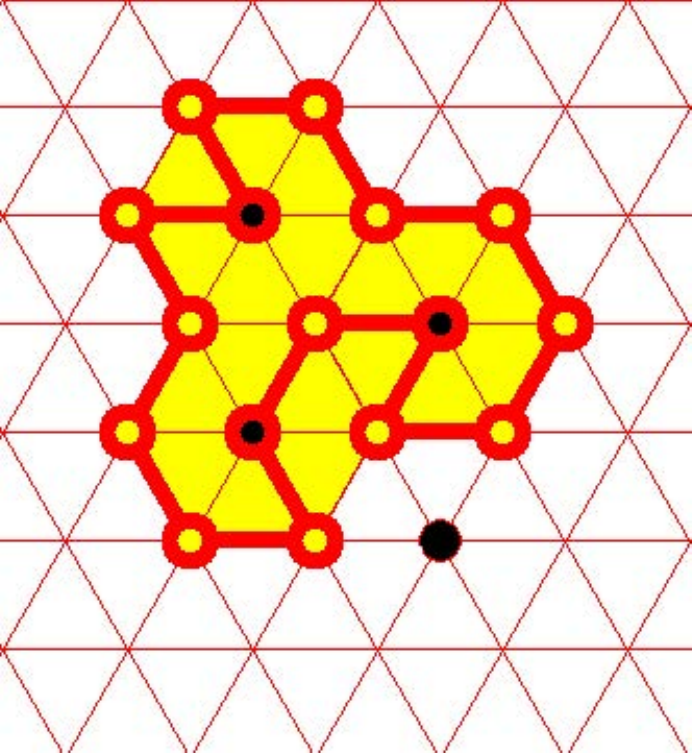}
	\caption{}
	\label{fig:small-series-third}
	\end{subfigure}
\quad
	\begin{subfigure}[t]{0.6in}
	\centering
	\includegraphics[width=0.6in]{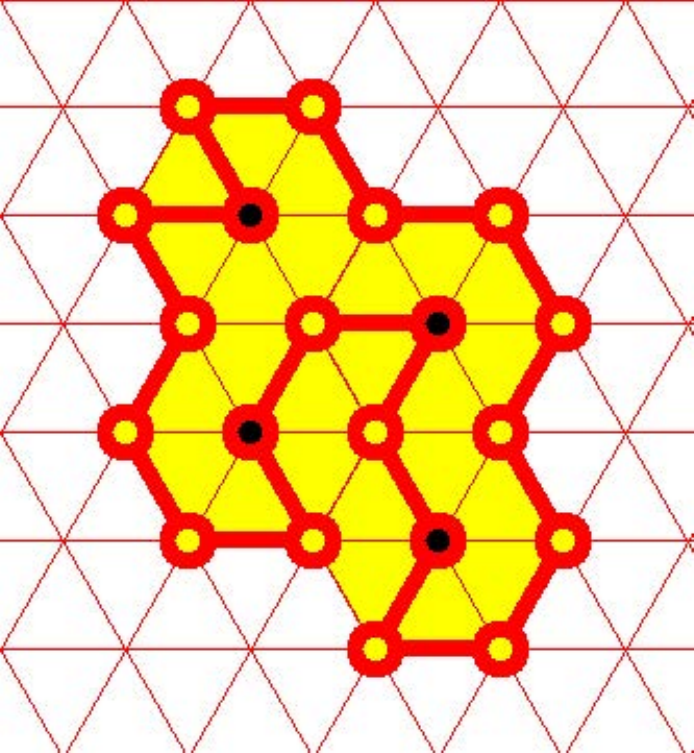}
	\caption{}
	\label{fig:small-series-fourth}
	\end{subfigure}
}
    \vspace*{-3mm}
	\caption{(a) An example shape which cannot be self-assembled by an oritatami system without being scaled (b)  Small example shape, (c) scaled to \scaling A2 and rotated version, (d) after addition of first gadget, (e) after second gadget, (f) after third gadget, (g) after fourth gadget and completion of HC.}
\label{fig:small-series}
\end{figure}

%% file: shapes-scaling-algo.tex
\section{Self-assembling finite shapes at scale $\geq$3 with delay 1}
\label{sec:scaling:algo}
%%% --------------------------------------------------------------------------------
All our algorithms are \emph{incremental} and proceed by extending the foldable routing at each step, to cover a new cell, neighboring the already covered cells. They proceed by maintaining a set of \emph{"clean edges"} in the routing, one on every \emph{"available side"} of each cell, from which we can extend the routing. Predictably, this is getting harder and harder as the scale gets smaller and as the edges of the cells overlap. We will present our different scaling algorithms by increasing difficulty: \scaling Bn for $n\geq 3$, then \scaling Cn for $n\geq 3$, then \scaling An for $n\geq 5$, then \scaling A4 and finally our most compact scaling \scaling A3. 

All the scaling algorithms presented in this section have been implemented in Swift on iOS.%
\footnote{Our app \hreftt{https://itunes.apple.com/us/app/id1335581323}{Scary~Pacman} can be freely downloaded from the app store at \hreftt{https://itunes.apple.com/us/app/id1335581323}{https://apple.co/2qP9aCX} and its source code can be downloaded and compiled from the public Darcs repository at \hreftt{https://hub.darcs.net/nikaoOoOoO/OritatamiScaling}{https://bit.ly/2qQjzy6}.}
All the figures in this section have been generated by this program and reflect its actual implementation.

%%% --------------------------------------------------------------------------------
\subsection{Universal tight oritatami system with delay~$1$}
%%% --------------------------------------------------------------------------------
\label{sec:scaling:algo:univ:OS}

\begin{definition}
We say that an OS is \emph{tight} if (1)~its delay is~$1$, (2)~every bead makes only one bond when it is placed by the folding and there is only one location where it can make a bond at the time it is placed during the folding.
\end{definition}

All the OS presented in this section are tight. Tight OS can be conveniently implemented using the following result:
%(proof deferred to page~\pageref{proof:thm:scaling:algo:universal:bead:type}):

\begin{theorem} \label{thm:scaling:algo:universal:bead:type}
Every tight OS can be implemented using a universal set of $114 = 19\times 6$ bead types together with a universal rule, from a seed of size~$3$.
\end{theorem}

% NSomitted

In the next subsections, all oritatami systems are tight. We will thus focus on designing routing with a single tight bond per bead, and rely on Theorem~\ref{thm:scaling:algo:universal:bead:type} for generating the transcript from the designed routing in linear time.

%%%
%%%
%%% --------------------------------------------------------------------------------

% --------------------------------------------------------------------------------
\subsection{Key definitions}
% --------------------------------------------------------------------------------

\NSomitted{
    In this section, we define all the tools used to design our algorithms. 
}
Consider a shape $S$ and $p_1,\ldots,p_{|S|}$ a \emph{search} of $S$, i.e. a sequence of distinct points covering $S$ such that for all $i\geq 2$, there is a $j<i$ such that $p_i\sim p_j$. W.l.o.g., we require that the \dirNW-neighbor of $p_1$ does not belong to~$S$ so that the $\dirN$-neighboring cell of $\cell(p)$ is empty in $\cell(S)$.

Starting from a tight routing covering the cell $\Linscaling{}{}(p_1)$, our algorithms cover each other cell $\Linscaling{}{}(p_i)$ in order $i=2\ldots|S|$, one by one, by extending the tight routing from a previously covered cell.

\paragraph{Lattice and cell directions.} 
We denote by $\dirSet = \allDir$ the set of all \emph{lattice directions} in $\Tlat$, and by $\dirCell = \allCellDir$ the set of all \emph{cell directions}, joining the centers of two neighboring cells (see Fig.~\ref{fig:def:upscaling:ABC}).
\NSomitted{
(see Fig.~\ref{fig:scaling:algo:lat:cell:dir}).
}
We denote by $\diropp{d}$ the direction opposite to $d$. We denote by $\CW(d)$ and $\CCW(d)$ the next direction in $\dirSet$ if $d\in\dirSet$ (or in $\dirCell$ if $d\in\dirCell$), in clockwise and counterclockwise order respectively. For $d\in\dirSet$ (resp. $d\in\dirCell$), we denote by $\mc d$ (resp. $\ml d$) the cell direction (resp. lattice direction) next to $d$ in counterclockwise order, e.g. $\mc\dirW = \cdirSW$ and $\ml\cdirNE = \dirNE$. 
\NSomitted{We denote by $\vec d$ the unit vector pointing in $d$ direction.
}

\NSomitted{
\begin{figure}[t]
    \centering
    \includegraphics[scale=0.5]{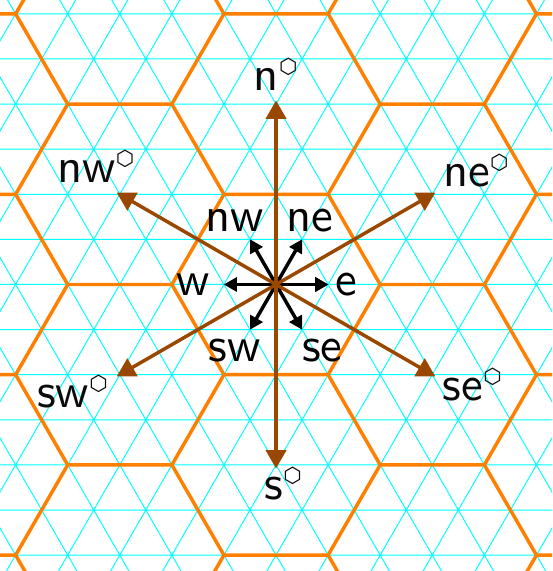}
    \caption{Lattice and cell directions.}
    \label{fig:scaling:algo:lat:cell:dir}
\end{figure}
}

A cell $\cell(p)$ is \emph{occupied} if the current routing covers it, otherwise it is \emph{empty}. Each cell has six \emph{sides}, its \cdirNW-, \cdirN-, \cdirNE-, \cdirSE-, \cdirS-, and \cdirSW-sides, connecting each of its six \dirW-, \dirNW-, \dirNE-, \dirE-, \dirSE-, and \dirSW-\emph{corners}. Given a cell, its neighboring cell in direction $d\in\dirCell$ is called its \emph{$d$-neighboring cell}. At scale \scaling An, the $d$-side of a cell is the $\diropp d$-side of its $d$-neighboring cell. At scales \scaling Bn and \scaling Cn, we say that the $d$-side of a cell and the $\diropp d$-side of its $d$-neighboring cell are \emph{neighboring sides}. 

The \emph{clockwise-most} and \emph{second clockwise-most} edges of the $d$-side of a cell are the two last edges in $\Tlat$ of this side in the direction $d'=\CCW(\ml d)$, e.g., if $d=\cdirNW$, the two \dirSW-most edges of the \cdirNW-side of the cell.

\paragraph{Routing Time.} At each step of our algorithms, the routing defines a \emph{total order over the vertices} of the currently occupied cells. For every vertex $p$ covered by the routing, we denote by $\rtime(p)$ its rank (from $0$ to $|r|-1$) in the current routing~$r$. We say of two occupied vertices $p$ and $q$, that $p$ is \emph{earlier} (resp. \emph{later}) than $q$ if ${\rtime(p)< \rtime(q)}$ (resp. ${\rtime(p)>\rtime(q)}$). 
\NSomitted{
Note that if \emph{the routing does not define any proper order over the cells}: it can happen that two cells $\cell(p)$ and $\cell(q)$ are such that some vertices of $\cell(p)$ are earlier than all vertices of $\cell(q)$ while some other vertices of $\cell(p)$ are later than all vertices of $\cell(q)$.
}%NSomitted

\paragraph{Clean edge.} The $d$-side of an occupied cell $\cell(p_i)$ is \emph{available} if its $d$-neighboring cell is empty. Consider an edge $uv$ of $\Tlat$ which belongs to an available  $d$-side of an occupied cell $\cell(a)$. Let $\cell(b)$ be the empty $d$-neighboring cell of $\cell(a)$. We say that edge $uv$ is \emph{clean} if: (1)~it belongs to the current routing; (2)~$uv$'s orientation $d'$ in the routing is clockwise with respect to the center $\ccenter(b)$ of $\cell(b)$, i.e. $d' = \CCW(\ml d)$  (e.g., $d'= \dirE$ if $d=\cdirS$); and (3)~the $\diropp{d'}$- and $\CW(\diropp {d'})$-neighbors $p$ and $q$ of its origin $u$ are both occupied and earlier than $u$ (e.g., the \dirW- and \dirNW-neighbors of $u$ if $d = \cdirS$). $p$ and $q$ are resp. called the \emph{support} and the \emph{bouncer} of the clean edge $uv$. Fig.~\ref{fig:scaling:algo:clean:example} gives examples of clean edges for the different scaling schemes. 
\begin{figure}[t]
    \resizebox{\textwidth}{!}{
    \includegraphics[height=3cm]{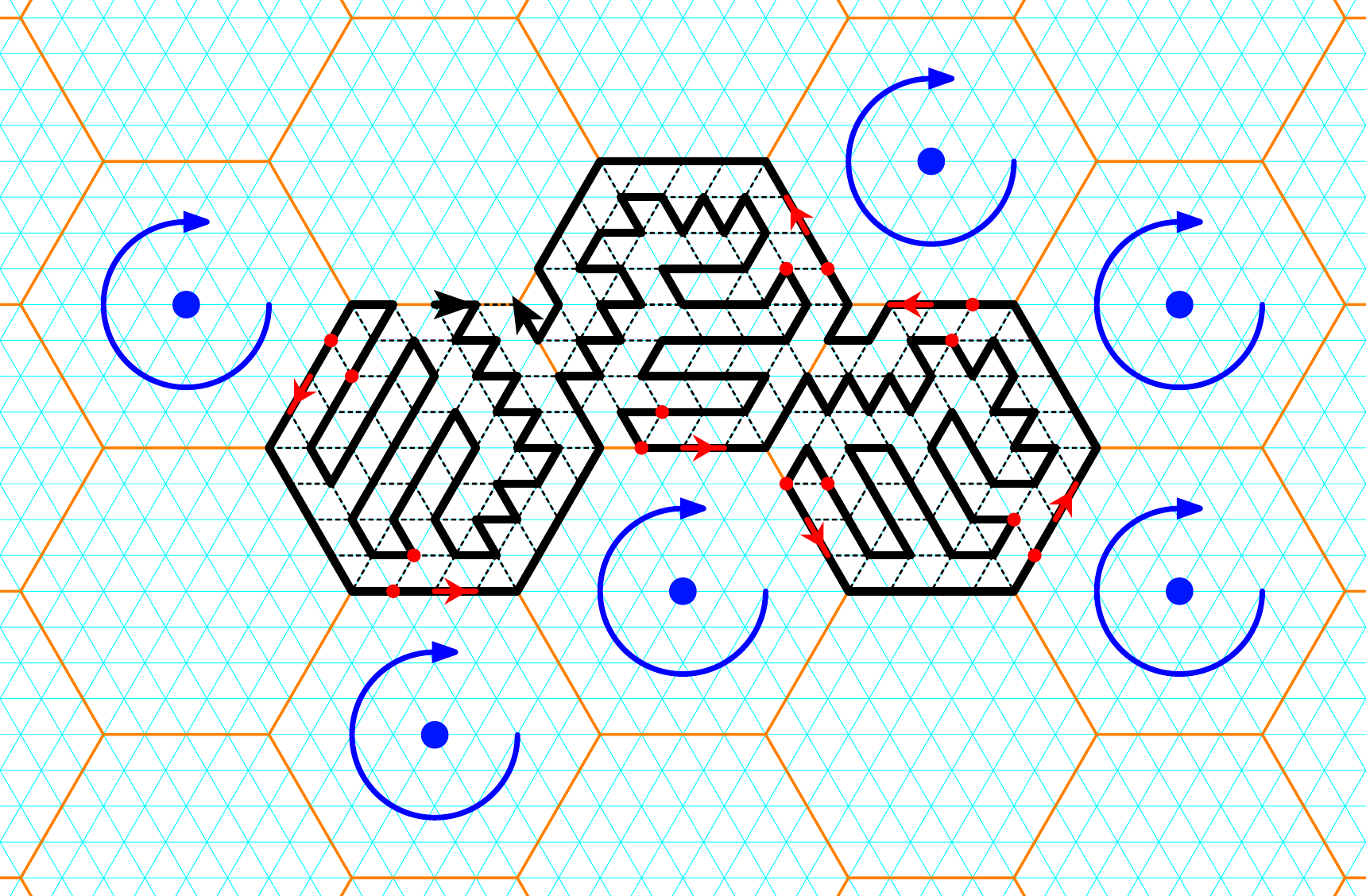}
    \includegraphics[height=3cm]{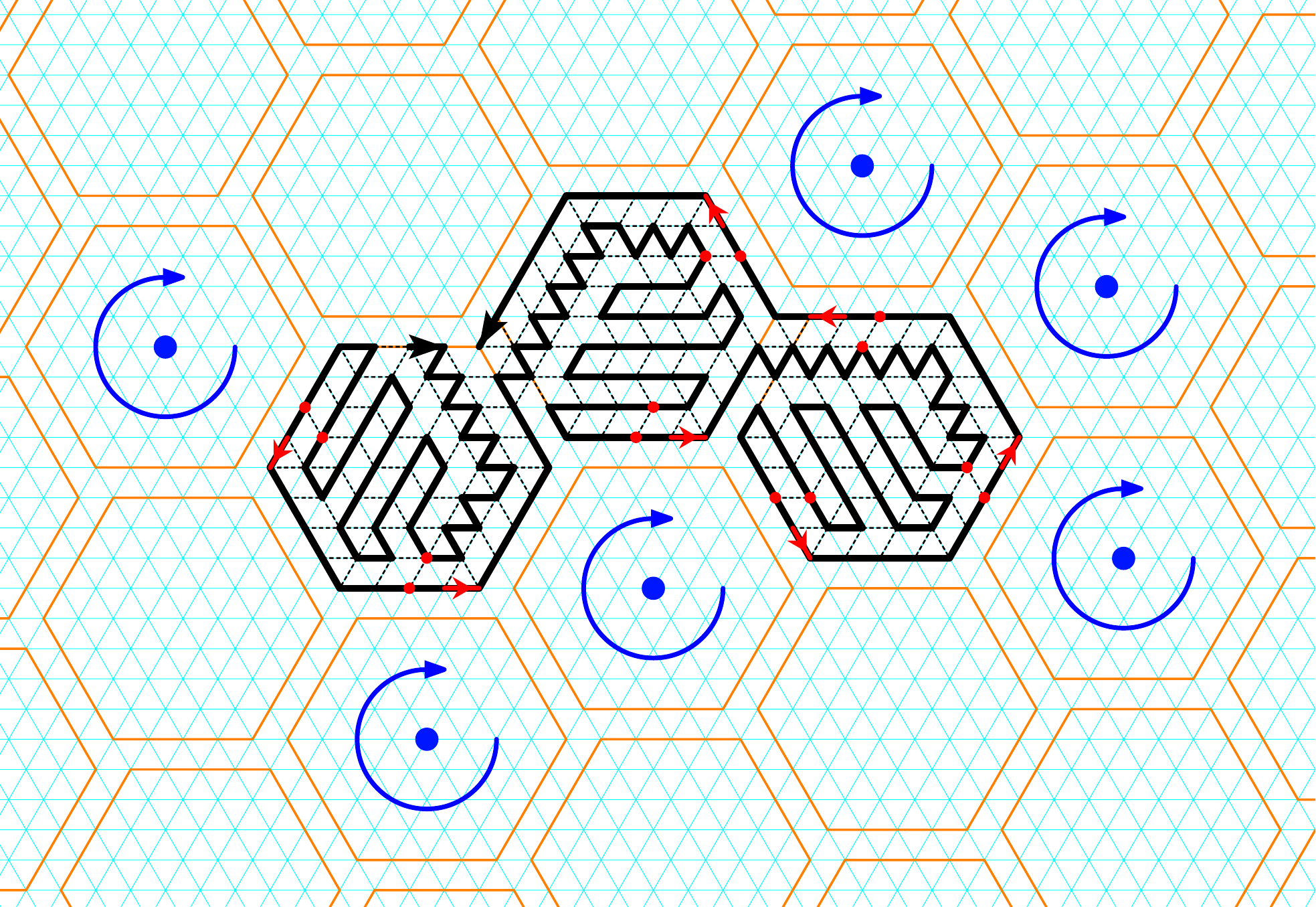}
    \includegraphics[height=3cm]{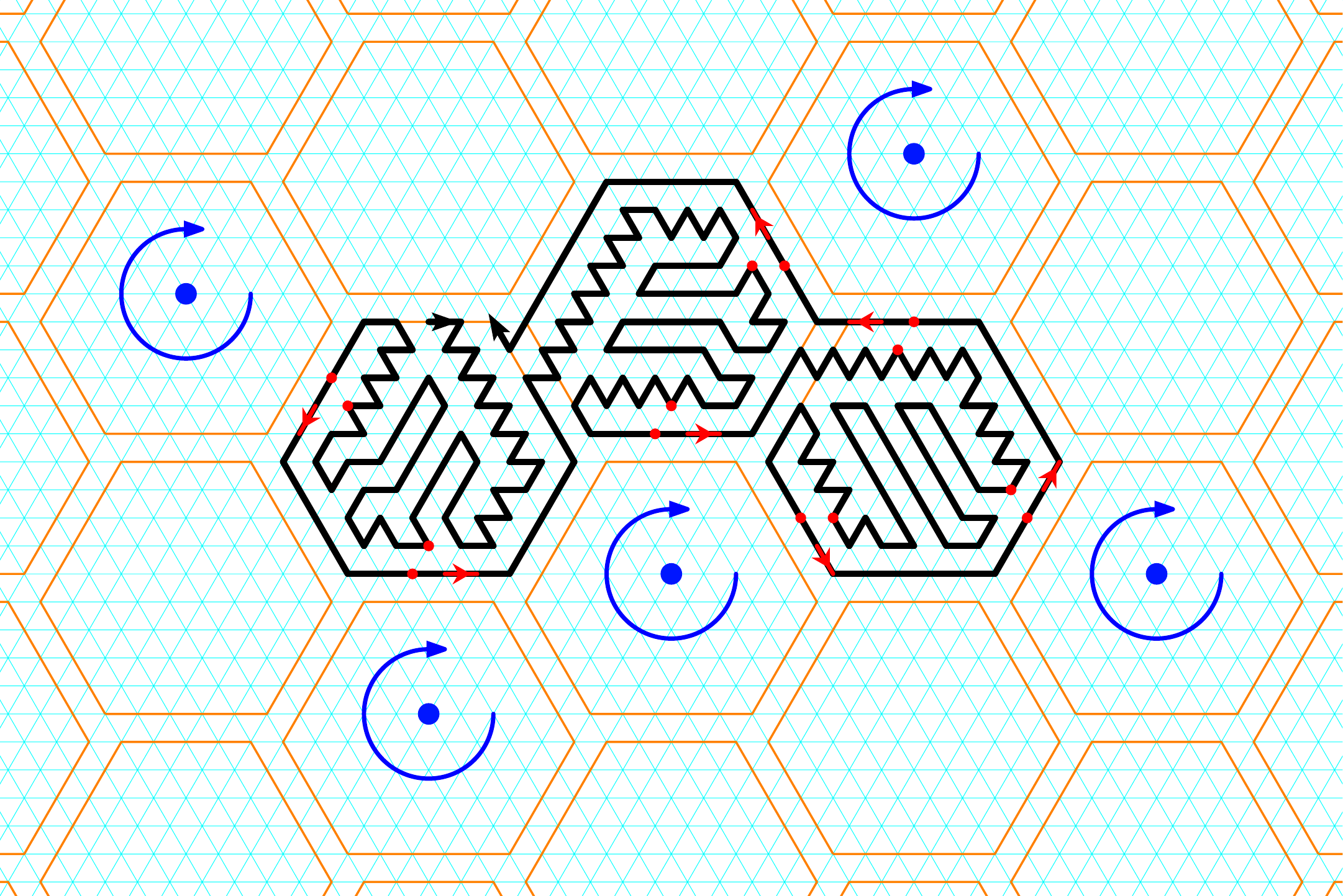}
    \includegraphics[height=3cm]{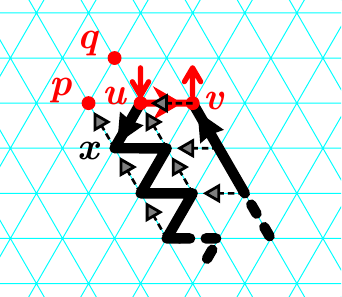}
    }
    \caption{\captionpar{Left: Examples of clean edges at scales \scaling A5, \scaling B5 and \scaling C5 --} the current routing is displayed in black; some clean edges are highlighted in red together with the two vertices required to be occupied, and earlier than the origin of the edge; the centers of some empty cells are highlighted in blue together with their clockwise orientation. %
    \captionpar{Right: Extending the routing from a clean edge --} the extension, drawn in black together with its tight bonds, replaces the clean edge $u\rightarrow v$ of the current routing $r$ (in red); because $p$ and $q$ are occupied and earlier than $u$ in $r$, the first bead of the extension is deterministically placed at $x$ by the folding and the zigzag pattern grows south-eastwards, self-supportedly; the way back to $v$ folds by bonding to the initial zigzag; note that all bonds are tight.}
    \label{fig:scaling:algo:clean:example}
    \label{fig:scaling:algo:grow:from:clean}
\end{figure}
Clean edges are a key component for our algorithms because they are the edges from which the routing is extended to cover a new empty cell. Indeed it is easy to grow a tight path from a clean edge as shown in Fig.~\ref{fig:scaling:algo:grow:from:clean}.

\NSomitted{
\begin{figure}
    \centering
    \includegraphics[scale=0.8]{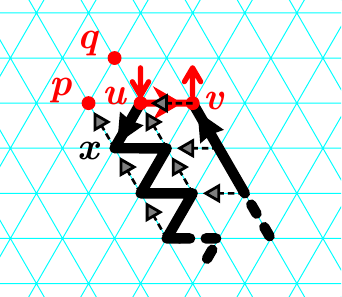}
    \caption{\captionpar{Extending the routing from a clean edge:} the extension is drawn in black together with its tight bonds and replaces the clean edge $u\rightarrow v$ of the current routing $r$ (in red); because $p$ and $q$ are occupied and earlier than $u$ in $r$, the first bead of the extension is deterministically placed at $x$ by the folding and the zig-zag pattern grows south-eastwards self-supportedly; the way back to $v$ folds by bonding to the initial zig-zag; note that all bonds are tight.}
    \label{fig:scaling:algo:grow:from:clean}
\end{figure}
}%NSomitted

\paragraph{Self-supported extension.}
We say that a path $\rho$ extending a routing from a clean edge $uv$ with support $p$ is \emph{self-supported} if all its bond are tight and made only with the beads at $u$, $p$ or within $\rho$. Self-supported extensions are convenient because they fold correctly independently on their surrounding.

%\paragraph{Segments.}

%\paragraph{Clockwise clean routing and time anomalies.}

%\paragraph{Signature.}

%-----
\subsection{Design of self-supported tight paths covering pseudo-hexagons}
%-----
\label{sec:scaling:algo:cover:hex}

A \emph{$(a,b,c,d,e,f)$-pseudohexagon} is an hexagonal shape whose sides have length $a,b,c,d,e$ and $f$ respectively from the \cdirNE- to the \cdirN-side in clockwise order, i.e. is the convex shape in $\Tlat$ encompassed in a path consisting in $a$ steps to \dirSE, $b$ to \dirSW, $c$ to \dirW, $d$ to \dirNW, $e$ to \dirNE\/ and $f$ to \dirE. 
%
%The proof of the following theorem may be found in appendix p.~\pageref{proof:thm:scaling:algo:cover:hex}.
%
\NSomitted{
Note that we must have: ${a+b = d+e}$, ${b+c = e+f}$ and ${a+f = c+d}$, since ${a\,\vec\dirSE} + {b\,\vec\dirSW} + {c\,\vec\dirW} + {d\,\vec\dirNW} + {e\,\vec\dirNE} + {f\,\vec\dirW} = {(a-c-d+f) \vec\dirE + (a+b-d-e) \vec\dirSW} =  0$. 
}

\begin{theorem}%
%[proof deferred to page~\pageref{proof:thm:scaling:algo:cover:hex}]
%[Tight routing for pseudohexagons]
\label{thm:scaling:algo:cover:hex}
Let $H$ be a $(a,b,c,d,e,f)$-pseudohexagon with ${a,b,c,d,e,f\geq 5}$.
There is an algorithm \algoCoverHex\/ that outputs in linear time a self-supported tight routing covering $H$ from a clean edge placed on either of the two eastmost edges above its \cdirN-side, and such that it ends with a counterclockwise tour covering the \cdirNW-, \cdirSW-, \cdirS-, \cdirSE- and finally \cdirNE-sides. 
\end{theorem}

By Theorem~\ref{thm:scaling:algo:universal:bead:type}, we conclude that all large enough pseudo-hexagons can be self-supportedly folded by a tight OS.

%NSomitted
%

% --------------------------------------------------------------------------------
\subsection{Scale \scaling Bn and \scaling Cn with $n\geq 3$}
% --------------------------------------------------------------------------------
\label{sec:scaling:algo:Bn}

Cells in scaling \scaling Bn and \scaling Cn do not overlap. It is then enough to find one routing extension for the cell (with a clean edge on all of its all available side) from every possible neighboring clean edge. 

\paragraph{Scale \scaling Bn}
\label{sec:scaling:algo:B}
is isotropic. Thus, there are only two cases to consider up to rotations: either the cell is the first, or it will plug onto a neighboring clean edge. For \scaling Bn, the clean edges that we plug onto, are the \emph{counterclockwise-most} of each side of an neighboring occupied cell. For $n\geq 7$, we rely on Theorem~\ref{thm:scaling:algo:cover:hex} to construct such a routing. The two routings for $n=3$ are given in Fig.~\ref{fig:scaling:algo:routing:B3}.
%(see  Fig.~\ref{fig:scaling:algo:routing:B4-9} for $n=4\ldots9$). 
We have then:

\begin{lemma}
At every step, the computed routing is self-supported and tight, covers all the cells inserted, and contains a clean edge on every available side with the exception of the \cdirN-side of the initial cell $\cell(p_1)$. 
\end{lemma}

\begin{proof}
This is immediate by induction on the size of the cell insertion sequence by noticing that all the routing extensions are self-supported and tight and that every available side (but the \cdirN-side of the root cell) bears a clean edge. \qed
\end{proof}

Note that no insertion will occur on the \cdirN-neighboring cell of $\cell(p_1)$ because it is assumed w.l.o.g. to be empty. Theorem~\ref{thm:scaling:algo:universal:bead:type} thus applies and outputs, in linear time, a corresponding OS with $114$ bead types and a seed of size $3$.
\begin{figure}[t]
    \centering
    \tablePatsXEasy{.9}{3}
    \caption{\captionpar{The self-supported tight routing extensions for scale \protect\scaling B3:} in light purple, the clean edge used to extend the routing in this cell; in red, the ready-to-use new clean edges in every direction; highlighted in orange, the seed.}
    \label{fig:scaling:algo:routing:B3}
\end{figure}
\NSomitted{
\begin{figure}%[tb]
    \subfigB4
    \hfill 
    \subfigB5
    \\[2mm]
    \subfigB6
    \hfill 
    \subfigB7
    \\[2mm]
    \subfigB8
    \hfill 
    \subfigB9
    \caption{\captionpar{The self-supported tight routing extensions for scales \scaling B4 to \scaling B9:} in purple, the clean edge used to extend the routing in this cell; in red, the ready-to-use new clean edges in every direction; highlighted in orange, the seed.}
    \label{fig:scaling:algo:routing:B4-9}
\end{figure}
}%NSomitted
The same technique applies at scale \scaling Cn with $n\geq 3$ 
%\omittedSeeFullArticle.
\withOrWithoutAppendix{%
    (see appendix p.~\pageref{sec:scaling:algo:C}). Fig.~\ref{fig:scaling:algo:movie:B3} and~\ref{fig:scaling:algo:movie:C3} (p.~\pageref{fig:scaling:algo:movie:C3}) present a step-by-step execution of the routing extension algorithm at scales \scaling B3 and \scaling C3 respectively.%
}{%
    (\omittedSeeFullArticle).
}
%Fig.~\ref{fig:scaling:algo:movie:B3} and~\ref{fig:scaling:algo:movie:C3} (p.~\pageref{fig:scaling:algo:movie:C3}) present a step-by-step execution of the routing extension algorithm at scales \scaling B3 and \scaling C3. 
It follows that:

\begin{theorem}
\label{thm:scaling:algo:routing:Bn}
\label{thm:scaling:algo:routing:Cn}
Any shape $S$ can be folded by a tight OS at all scales \scaling Bn and \scaling Cn with $n\geq 3$.
\end{theorem}

%NSomitted
%
%As for \scaling B2, we conjecture that there is the same family of finite shapes cannot be folded by any finite-delay OS at scale \scaling C2. 

% ------------------------------
\subsection{Scale \scaling An with $n\geq4$} % $n\geq 5$}
% ------------------------------
\label{sec:scaling:algo:An>=5}

Scale \scaling An is the most compact considered in this article. It is isotropic but its cells do overlap. For this reason, we need to provide more extension in order to manage all the cases. The cases $n\geq 5$ are the easiest because we can provide a routing for each situation with a clean edge on every available side. Scale \scaling A4 is trickier because only one available side (the latest) may contain a clean edge. Scale \scaling A3 requires a careful management of time and geometry in the routing to ensure that a clean edge can be exposed when needed. 
\NSomitted{
Scales \scaling A4 and \scaling A3 are presented separately in the next subsections.
}
Scale \scaling A3 is presented separately in the next subsection. Scale \scaling A4 is 
\withOrWithoutAppendix{%
    deferred to the appendix p.~\pageref{sec:scaling:algo:A4}.
}{%
\omittedSeeFullArticle.
}

\paragraph{At scale \scaling An with $n\geq 5$,}
the clean edges are located at the \emph{second counterclockwise-most} edge on all of the available sides of every occupied cell
%(e.g., see Fig.~\ref{fig:scaling:algo:routing:A5}).
(e.g., see leftmost figure on Fig.~\ref{fig:scaling:algo:clean:example}).
Our design guarantees this property for every possible empty cell shape. As every occupied cell covers the $d$-side of all its $\diropp d$-neighboring empty cells, there are a priori $33 = 1+2^5$ different shapes to consider: the completely empty cell, for the first cell inserted; plus the $2^5$ possible shapes corresponding to the five possible states occupied/empty for the neighboring cells on which we do not plug. For \scaling An with $n\geq 5$, our design can extend the routing from any clean edge, regardless of its time or location. This reduces the number of shapes to consider to 14 cases, by rotating the configuration. The following definition allows to identify conveniently the various cases.

\paragraph{Segment and signature.} The \emph{signature rooted on $d\in\dirCell$} of an empty cell $\cell(p)$ is the  integer (written in binary) ${\sig_d(p) = \sum_{i=0}^5 s_i 2^i}$ where $s_i=1$ if the $\CW^i(d)$-neighboring cell of $\cell(p)$ is occupied, and $=0$ otherwise. $\sig_d(p) = 0$ if and only if all the neighboring cells of $\cell(p)$ are empty; $\sig_d(p)$ is odd if and only if the $d$-neighboring cell of $\cell(p)$ is occupied. A \emph{segment} of an empty cell $\cell(p)$ is a maximal sequence of consecutive sides already covered by its neighboring cells. \emph{We will always root the signature of an empty cell on the clockwise-most side of a segment}. With this convention, the two least significant bits of the signature of an empty cell with at least one and at most 5 neighboring occupied cells is always \texttt{01}.
\withOrWithoutAppendix{%
    We are then left with the following possible signatures for an empty cell, sorted by the number of segments around this cell (see Fig.~\ref{fig:scaling:algo:all:CWmost:signatures} p.~\pageref{fig:scaling:algo:all:CWmost:signatures}):
    \begin{description}
        \item[No segment:] \NSpattern{0}
        \item[1 segment:] \NSpattern{1}, \NSpattern{100001}, \NSpattern{110001}, \NSpattern{111001}, \NSpattern{111101}, \NSpattern{111111}.
        \item[2 segments:] \NSpattern{101}, \NSpattern{1001}, \NSpattern{10001}, when both have length $1$; \NSpattern{100101}, \NSpattern{101001}, \NSpattern{1101}, \NSpattern{11001}, when their lengths are $1$ and $2$; \NSpattern{101101} when both have length $2$; \NSpattern{110101}, \NSpattern{11101} when one has length $3$.
        \item[3 segments:] \NSpattern{10101}.
    \end{description}
    \begin{figure}[t]
        \centering
        \includegraphics[scale=1]{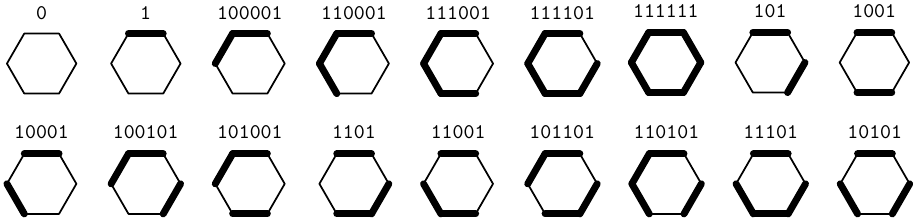}
        \caption{List of the 18 possible signatures rooted on the clockwise-most side of a segment, placed on the \protect\cdirN-side.}
        \label{fig:scaling:algo:all:CWmost:signatures}
    \end{figure}
    Now, note that the following signatures define an identical cell shape up to a rotation:  $\NSpattern{10001} \equiv \NSpattern{101}$, $\NSpattern{101001} \equiv \NSpattern{1101}$, $\NSpattern{100101} \equiv \NSpattern{11001}$, and $\NSpattern{110101} \equiv \NSpattern{11101}$. (note that symmetries are not allowed because they do not preserve ``clockwisevity''). 
}{%
}%withorwthoutappendix
By rotating the patterns, we are then left with designing self-supported tight routings for 14 shapes with clean edges at the second clockwise position of every available side. For $n\geq 8$, the 14 pseudo-hexagons are large enough for Theorem~\ref{thm:scaling:algo:cover:hex} to provide the desired routings. 
\withOrWithoutAppendix{%
    The routing extensions for $n=5,\ldots,8$ are given in Fig.~\ref{fig:scaling:algo:routing:A5}  to~\ref{fig:scaling:algo:routing:A8} in appendix.
    Scale \scaling A4 is handled similarly (see appendix p.~\pageref{sec:scaling:algo:A4}).
}{
    The routing extensions for $n=5,\ldots,8$ may be found in \cite{shape2018DNAfull}. Scale \scaling A4 is handled similarly (\omittedSeeFullArticle).
}
We can thus conclude by an immediate induction on the size of the cell insertion sequence, as for scale \scaling Bn, that:

\begin{theorem}
\label{thm:scaling:algo:An>=4}
Any shape $S$ can be folded by a tight OS at scale \scaling An, for $n\geq 4$.
\end{theorem}

\NSomitted{
\fitfigurePatsXZeroSmall5 
\fitfigurePatsXZero5
\fitfigurePatsXZero6
\fitfigurePatsXZero7
\fitfigurePatsXZero8
\fitfigurePatsXZero9
\fitfigurePatsXZero{10}
\fitfigurePatsXZero{11}
}%NSomitted

%NSomitted

%%%
%%%
%%%

% ------------------------------
\subsection{Scale \scaling A3}
% ------------------------------
\label{sec:scaling:algo:A3}

At scale \scaling A3, the sides of each cell have length $2$, and no edge can fit in if both neighboring cells are already occupied. We must then pay extra attention to the order of self-assembly, i.e. to time. 
\NSomitted{
As pointed out earlier, we must be extra careful when dealing with time, because \emph{time defines a total order only on occupied vertices, and no order at all on occupied cells}. 
}%NSomitted
We define the \emph{time of an occupied side} as the routing time of its middle vertex (its rank in the current routing). In \scaling A3, the clean edges are located at the \emph{counterclockwise-most} position of the available sides of the occupied cells. Our routing algorithm maintains, before each insertion, an invariant for the routing that \emph{combines time and geometry} as follows:

\begin{invariant}[insertion]
\label{inv:scaling:algo:A3:CW:last}
Around an empty cell, the clockwise-most side of any segment is always the latest of that segment, and its clockwise-most edge is clean.
\end{invariant}

As it turns out, we cannot maintain this invariant for every empty cell at every step. The middle vertex of a side violating this invariant is called  a \emph{time-anomaly}.

The anomalies around an empty cell are fixed \emph{only} at the step the empty cell is covered by the algorithm. Because fixing anomalies consists in freeing the corresponding side (as if the neighboring cell was empty), without actually freeing the cell, we define the signature rooted on side~$d$ of an empty $\cell(p)$ slightly differently here, as: ${\sig_d(p) = \sum_{i=0}^5 s_i 2^i}$ where $s_i=1$ if the \emph{vertex at the middle} of the $\CW^i(d)$-side is occupied, and $=0$ otherwise.

The routing algorithm is described in Algorithm~\ref{algo:scaling:algo:ext:A3} and uses two series of routing extensions: the \emph{basic} patterns in Fig~\ref{fig:scaling:algo:A3:basic:ext}, and the \emph{anomaly-fixing} patterns in Fig.~\ref{fig:scaling:algo:A3:fix:anomalies}(b-d). There are two types of anomalies: \emph{path-anomalies} (marked as yellow dots) only require a local rerouting inside the cell to become clean; \emph{time-anomalies} (marked as red dots) cannot be turned into clean edge and must be freed according to the diagram in~Fig.~\ref{fig:scaling:algo:A3:fix:anomalies}(b-d). 
Fig.~\ref{fig:scaling:algo:movie:A3:anomalyA:main} gives a step-by-step construction of a shape which involves fixing several anomalies.
%
%For instance, to insert a cell $\cell(p)$ to the south of a cell $\cell(q)$ currently covered with the routing pattern \NSpattern{101\NSnx\dirSW}, first, fix $\cell(q)$'s \cdirS-time-anomaly by replace its pattern with the routing pattern \NSpattern{101\NSnx\dirSW\NSnx\dirS}.
%
\begin{algorithm}[t]
\caption{Incremental routing algorithm for scale \scaling A3}
\label{algo:scaling:algo:ext:A3}
\small
\begin{algorithmic}[1]
\Procedure{FillEmptyCell}{centered at: $\lambda(p)$}
   \If{$\cell(p)$ has no occupied neighboring cell}
        \State \ALGOmultiline{Fill $\cell(p)$ with routing \NSpattern0 from Fig.~\ref{fig:scaling:algo:A3:basic:ext}, mark the \cdirN-cell as \emph{forbidden} and \textbf{return}.}
\EndIf
    \While{the latest side of $\cell(p)$ is an anomaly}
        \State \ALGOmultiline{Fix this anomaly in the corresponding neighboring cell according to the diagram in Fig.~\ref{fig:scaling:algo:A3:fix:anomalies}.}
    \EndWhile
    \State \label{algo:scaling:algo:A3:ext:compute:sign}
    \ALGOmultiline{Compute the $\cell(p)$'s signature rooted on the latest side and extend the path according to the corresponding basic pattern in Fig.~\ref{fig:scaling:algo:A3:basic:ext}.}
\EndProcedure
\end{algorithmic}
\end{algorithm}
\begin{figure}[t]
    \begin{subfigure}{\textwidth}
    \resizebox{\textwidth}{!}{\tabFigBasicThreeZeroCompact{1.5}}
    \caption{The 18 basic routing extensions.}
    \label{fig:scaling:algo:A3:basic:ext}
    \end{subfigure}
\\
    \begin{subfigure}{0.55\textwidth}
    \centering
    \scalebox{0.53}{\input{shapes-scaling-algo-fix-101}}
    \vspace*{-6mm}
    \caption{Fixing anomalies in \NSpattern{101}, \NSpattern{1101} and \NSpattern{11101}.}
    \label{fig:scaling:algo:A3:fix:101}
\end{subfigure}
\begin{tabular}{c}
    \begin{subfigure}{.43\textwidth}
        \centering
        \scalebox{0.53}{\input{shapes-scaling-algo-fix-1001}}
        \vspace*{-3mm}
        \caption{Fixing anomalies in \NSpattern{1001} and \NSpattern{11001}.}
        \label{fig:scaling:algo:A3:fix:1001}
    \end{subfigure}
    \\[7mm]
    \begin{subfigure}{.49\textwidth}
        \centering
        \scalebox{0.53}{\input{shapes-scaling-algo-fix-100101}}
        \vspace*{-3mm}
        \caption{Fixing in \NSpattern{100101} and \NSpattern{101101}.}
        \label{fig:scaling:algo:fix:100101}
    \end{subfigure}
\end{tabular}
\caption{\captionpar{Routing extensions at \scaling A3:} in purple, the latest (clockwise-most) clean edge used to extend the routing; in green, the sides already covered, earlier in the routing; in yellow, the side shared with the newly covered neighboring cell after fixing a path-anomaly; the red arrows are the new potential clean edges available to extend the routing; time- and path-anomalies, that need to be fixed to allow extension on that side, are indicated resp. by red and yellow dots; the seed is highlighted in orange in signature \NSpattern0.}
\label{fig:scaling:algo:A3:fix:anomalies}
\label{fig:scaling:algo:A3:routing}
\end{figure}

The following key topological lemma and corollary ensure that time- and path-anomalies are very limited and can be handled locally 
\withOrWithoutAppendix{%
    (proofs may be found in Section~\ref{sec:scaling:algo:A3:apx}). 
}{%
    (\omittedSeeFullArticle).
}
And the theorem follows by immediate induction:

\begin{lemma}[Key topological lemma]
\label{lem:scaling:algo:one:boundary}
At every step of the algorithm, the boundary of each empty area contains exactly one time-anomaly vertex.
\end{lemma}

\begin{corollary}
\label{cor:scaling:algo:CW:clean}
The \textbf{\upshape while} loop is executed at most twice, and it fixes: at most one time-anomaly, and at most one path-anomaly. After these fixes, the latest edge around the empty cell is always the clockwise-most of a segment and clean.
\end{corollary}

\begin{theorem}
\label{thm:scaling:algo:A3}
Any shape $S$ can be folded by a tight OS at scale \scaling A3.
\end{theorem}

%Fig.~\ref{fig:scaling:algo:movie:A3:anomaly} and~\ref{fig:scaling:algo:movie:A3:anomalyB} in appendix show two step-by-step constructions of a shape which involve fixing several anomalies. %Fig.~\ref{fig:tribute:stacking:A3} to~\ref{fig:tribute:stacking:C3} concluding the appendix show how the same shape gets folded at the different scales \scaling A3, \scaling B3 and \scaling C3. 

%\begin{figure}[tbh]
%    \resizebox{\textwidth}{!}{
%        \begin{tabular}{@{}ccccc@{}}
%            \incMovie{.15\textwidth}{anomalyA2}{0008}
%        &   \incMovie{.15\textwidth}{anomalyA2}{0009}
%        &   \incMovie{.15\textwidth}{anomalyA2}{0010}
%        &   \incMovie{.15\textwidth}{anomalyA2}{0011}
%        &   \incMovie{.15\textwidth}{anomalyA2}{0012}
%        \end{tabular}
%    }
%    \caption{The step-by-step construction of a routing folding into a shape at scale \scaling A3 according to Algorithm~\ref{algo:scaling:algo:ext:A3}, involving solving anomalies $\NSpattern{101} \rightarrow \NSpattern{101\protect\NSnx\dirSW} \rightarrow \NSpattern{101\protect\NSnx\dirSW\protect\NSnx\dirS}$ and $\NSpattern{11101} \rightarrow \NSpattern{11101\protect\NSnx\dirNW} = \NSpattern{111101}$ (rotated clockwise) in the three last steps.}
%    \label{fig:scaling:algo:movie:A3:anomalyB:main}
%\end{figure}

\begin{figure}[tbh]
    \resizebox{\textwidth}{!}{
        \begin{tabular}{@{}ccccc@{}}
            \incMovie{.15\textwidth}{anomalyA1}{0008}
        &   \incMovie{.15\textwidth}{anomalyA1}{0009}
        &   \incMovie{.15\textwidth}{anomalyA1}{0010}
        &   \incMovie{.15\textwidth}{anomalyA1}{0011}
        &   \incMovie{.15\textwidth}{anomalyA1}{0012}
        \end{tabular}
    }
    \caption{The step-by-step construction of a routing folding into a shape at scale \scaling A3 according to Algorithm~\ref{algo:scaling:algo:ext:A3}, involving fixing anomalies $\NSpattern{101} \rightarrow \NSpattern{101\protect\NSnx\dirNW} \rightarrow \NSpattern{101\protect\NSnx\dirNW\protect\NSnx\dirS}\rightarrow \NSpattern{101\protect\NSnx\dirNW\protect\NSnx\dirS\protect\NSnx\dirSW}$ in the four last steps.}
    \label{fig:scaling:algo:movie:A3:anomalyA:main}
\end{figure}

%%%
%%% Discussion
%%%

\NSomitted{
\subsection{Discussion}

\todoi{missing blabla + fig} no hamitlonian path.

\begin{proposition}
The 3-arms star family cannot be folded by any OS at scale \scaling B1.
\end{proposition}

\begin{conjecture}
The 3-arms star family cannot be folded by any absolutely bounded delay OS at any of the scales \scaling A2, \scaling B2 and \scaling C2.
\end{conjecture}

Our routing designs produce OS with delay $1$ and maximal arity $4$. One may ask if there is a bound on the minimal arity required to fold arbitrary shapes. The following result shows that any deterministic OS with delay~$1$ and minimal arity~$1$ has very limited capacities as it may only not fold any shape significantly larger than the seed, independently of the size of its transcript:

\begin{theorem}\label{thm:d1a1_det_finiteness_short}
	Let $\TMO$ be an OS of delay 1 and arity 1 whose seed consists of $n$ beads, and let $w$ be the transcript of $\TMO$. 
	If $\TMO$ is deterministic, then $|w| \le 9n$. 
\end{theorem}
}%NSomitted

%NSomitted

%%%%
%%%%%  TRANSITION
%%%%%%

%\smallskip

%The next section explores how delay impacts the capacity of an OS to fold a given shape.

%% file: shapes-scaling-algo-fix-101.tex
%!TEX root = Oritatami-shapes.tex

%% 
%% FIX 101
%%

\begin{tikzpicture}

\node[NSpat] (p101) at (0,0){\figPatAThree{101}{101}};
\node[NSpat] (p1101) at (\NSpatDist,0){\figPatAThree{1101}{101\NSnx\dirS\/ = 1101}};
\node[NSpat] (p11101) at (2.1*\NSpatDist,0){\figPatAThree{11101}{1101\NSnx\dirSW\/ = 11101}};
\node[NSpat] (p111101) at (3.3*\NSpatDist,0){\figPatAThree{111101}{11101\NSnx\dirNW\/ = 111101}};

\node[NSpat] (p1101nw) at (2*\NSpatDist,-\NSpatDistV){\figPatAThree{1001101}{1101\NSnx\dirNW}};
\node[NSpat] (p1101nw_sw) at (3*\NSpatDist,-\NSpatDistV){\figPatAThree{10001101}{1101\NSnx\dirNW\NSnx\dirSW}};

\node[NSpat] (p101sw) at (\NSpatDist,-2*\NSpatDistV){\figPatAThree{1000101}{101\NSnx\dirSW}};
\node[NSpat] (p101sw_s) at (2*\NSpatDist,-2*\NSpatDistV){\figPatAThree{10000101}{101\NSnx\dirSW\NSnx\dirS}};

\node[NSpat] (p101nw) at (\NSpatDist,-3*\NSpatDistV){\figPatAThree{11000101}{101\NSnx\dirNW}};
\node[NSpat] (p101nw_sw) at (2*\NSpatDist,-3*\NSpatDistV){\figPatAThree{100000101}{101\NSnx\dirNW\NSnx\dirSW}};
\node[NSpat] (p101nw_sw_s) at (3*\NSpatDist,-3*\NSpatDistV){\figPatAThree{101000101}{101\NSnx\dirNW\NSnx\dirSW\NSnx\dirS}};

\node[NSpat] (p101nw_s) at (2*\NSpatDist,-4*\NSpatDistV){\figPatAThree{110000101}{101\NSnx\dirNW\NSnx\dirS}};
\node[NSpat] (p101nw_s_sw) at (3*\NSpatDist,-4*\NSpatDistV){\figPatAThree{111000101}{101\NSnx\dirNW\NSnx\dirS\NSnx\dirSW}};

%\node[NSpat] (p101101) at (\NSpatDist,0){\figPatAThree{101101}{101\NSnx\dirS\/ = 101101}};

\draw[NSarr] (p101)--(p1101) node[NSdir]{\cdirS};
\draw[NSarr] (p1101)--(p11101) node[NSdir]{\cdirSW};
\draw[NSarr] (p11101)--(p111101) node[NSdir]{\cdirNW};

\draw[NSarr] (p1101)to[bend right] node[NSdirB]{\cdirNW} (p1101nw);
\draw[NSarr] (p1101nw)--(p1101nw_sw) node[NSdir]{\cdirSW};

\draw[NSarr] (p101)to[bend right] node[NSdirB]{\cdirSW}(p101sw);
\draw[NSarr] (p101sw)--(p101sw_s) node[NSdir]{\cdirS};

\draw[NSarr] (p101)to[bend right] node[NSdirB]{\cdirNW} (p101nw);
\draw[NSarr] (p101nw)--(p101nw_sw) node[NSdir]{\cdirSW};
\draw[NSarr] (p101nw_sw)--(p101nw_sw_s) node[NSdir]{\cdirS};

\draw[NSarr] (p101nw)to[bend right] node[NSdirB]{\cdirS} (p101nw_s);
\draw[NSarr] (p101nw_s)--(p101nw_s_sw) node[NSdir]{\cdirSW};

\end{tikzpicture}

%% file: shapes-scaling-algo-fix-1001.tex
%!TEX root = Oritatami-shapes.tex

%% 
%% FIX 1001
%%

\begin{tikzpicture}

\node[NSpat] (p1001) at (0,0){\figPatAThree{1001}{1001}};
\node[NSpat] (p11001) at (\NSpatDist,0){\figPatAThree{11001}{1001\NSnx\dirSW\/ = 11001}};
\node[NSpat] (p111001) at (2.1*\NSpatDist,0){\figPatAThree{111001}{11001\NSnx\dirNW\/ = 111001}};

\node[NSpat] (p1001nw) at (\NSpatDist,-\NSpatDistV){\figPatAThree{1001001}{1001\NSnx\dirNW}};
\node[NSpat] (p1001nw_sw) at (2*\NSpatDist,-\NSpatDistV){\figPatAThree{10001001}{1001\NSnx\dirNW\NSnx\dirSW}};

%\node[NSpat] (p101101) at (\NSpatDist,0){\figPatAThree{101101}{101\NSnx\dirS\/ = 101101}};

\draw[NSarr] (p1001)--(p11001) node[NSdir]{\cdirSW};
\draw[NSarr] (p11001)--(p111001) node[NSdir]{\cdirNW};

\draw[NSarr] (p1001)to[bend right] node[NSdirB]{\cdirNW} (p1001nw);
\draw[NSarr] (p1001nw)--(p1001nw_sw) node[NSdir]{\cdirSW};

\end{tikzpicture}

%% file: shapes-scaling-algo-fix-100101.tex
%!TEX root = Oritatami-shapes.tex

%% 
%% FIX 100101
%%

\begin{tikzpicture}

\node[NSpat] (p100101) at (0,0){\figPatAThree{100101}{100101}};
\node[NSpat] (p101101) at (\NSpatDist,0){\figPatAThree{101101}{100101\NSnx\dirS\/ = 101101}};
\node[NSpat] (p111101) at (2.1*\NSpatDist,0){\figPatAThree{111101}{101101\NSnx\dirSW\/ = 111101}};

\node[NSpat] (p100101sw) at (\NSpatDist,-\NSpatDistV){\figPatAThree{1100101}{100101\NSnx\dirSW}};
\node[NSpat] (p100101sw_s) at (2*\NSpatDist,-\NSpatDistV){\figPatAThree{10100101}{100101\NSnx\dirSW\NSnx\dirS}};

%\node[NSpat] (p101101) at (\NSpatDist,0){\figPatAThree{101101}{101\NSnx\dirS\/ = 101101}};

\draw[NSarr] (p100101)--(p101101) node[NSdir]{\cdirS};
\draw[NSarr] (p101101)--(p111101) node[NSdir]{\cdirSW};

\draw[NSarr] (p100101)to[bend right] node[NSdirB]{\cdirSW} (p100101sw);
\draw[NSarr] (p100101sw)--(p100101sw_s) node[NSdir]{\cdirS};

\end{tikzpicture}

%% file: smallDelayWeak-short.tex
\withOrWithoutAppendix{}{
\vspace{-15pt}}
\section{A shape which can be assembled at delay $\delta$ but not $<\delta$}
\label{sec:small-delay-weak}
\withOrWithoutAppendix{}{
\vspace{-5pt}}

This section contains the statement of
Theorem~\ref{thm:small-delay-weak} and a high-level description of its proof. For full details, see
\withOrWithoutAppendix{%
    Section~\ref{sec:sdw}.
}{%
    \cite{shape2018DNAfull}.
}

\vspace{-5pt}
\begin{figure}[htp]
\centering
\includegraphics[width=4.0in]{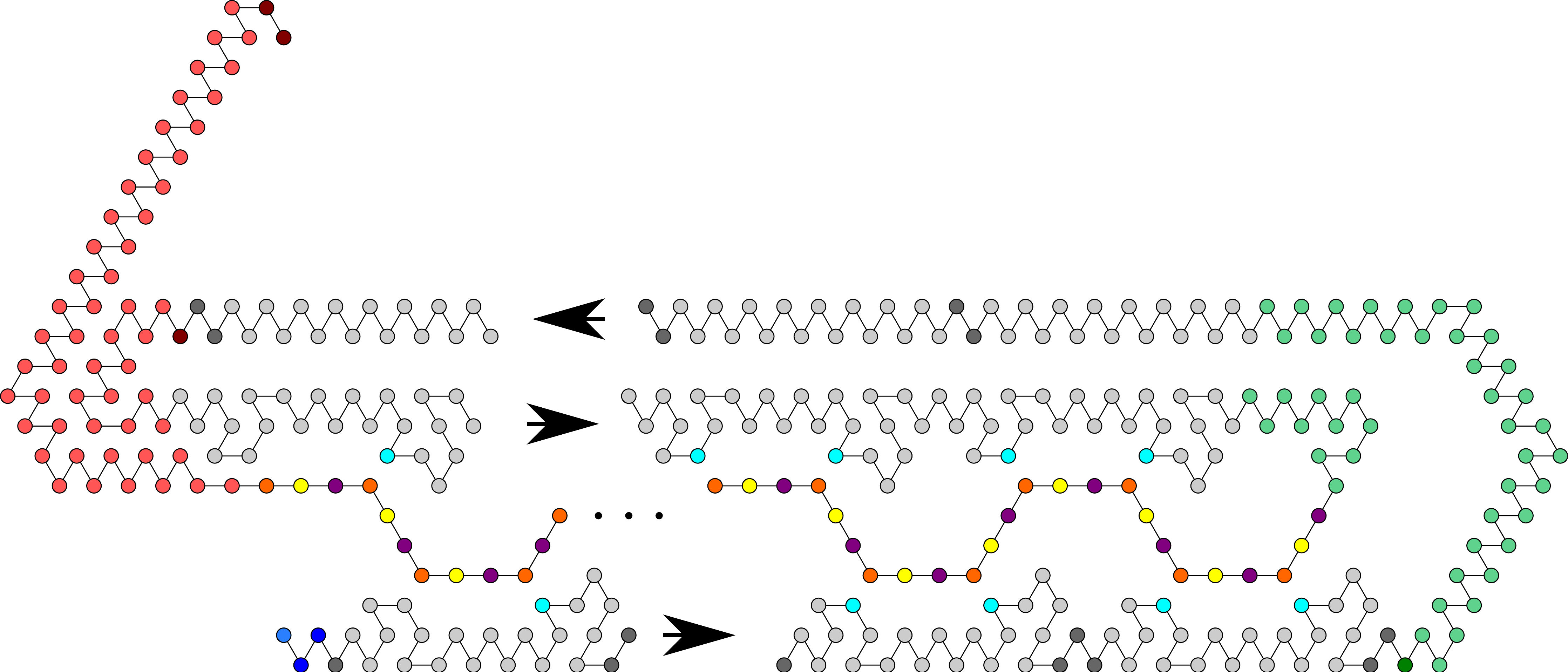}
\caption{A depiction of shape $S_{\delta}$ and a routing $R'_{\delta}$ for $\delta = 4$. This can be thought of as a ``slice'' of the shape (along with a forced routing) which cannot be self-assembled by an oritatami system with delay $< 4$, but can be assembled by an OS with delay $4$.  The arrows represent the direction of the directed path in the routing and the different colored beads represent the different gadgets in the routing.}
\label{fig:sw-main-full-example}
\end{figure}

\begin{theorem} \label{thm:small-delay-weak}
Let $\delta > 2$.  There exists a shape $S_{\delta}$ such that $S_{\delta}$ can be self-assembled by some OS $\TMO_{\delta}$ at delay $\delta$, but no OS with delay $\delta'$ self-assembles $S_{\delta}$ where $\delta' < \delta$.
\end{theorem}
\withOrWithoutAppendix{}{
\vspace{-5pt}}

We prove Theorem~\ref{thm:small-delay-weak} by constructing a deterministic OS $\TMO_{\delta}$ for every $\delta >2$, and we define $S_{\delta} = \dom(C_{\delta})$ where $C_{\delta} \in \termasm{\TMO_{\delta}}$.  It then immediately follows that there exists a system at delay $\delta$ which assembles $S_{\delta}$, and we complete the proof by showing that there exists no OS with delay less than $\delta$ which can assemble $S_{\delta}$.
A schematic depiction of the shape $S_{\delta}$ (for $\delta = 4$) can be seen in Fig.~\ref{fig:sw-main-full-example}. $\TMO_{\delta}$ forms the shape as follows. First a ``cave'' is formed where the distance between the top and the bottom is $\delta$ at specified points.  At regular intervals along the top and bottom, blue beads are placed. Once the cave is complete, a single-bead-wide path grows through it from right to left, and every $\delta$ beads is a red bead which interacts with the blue.  To optimize bonds, each red binds to a blue, which is possible since the spacing between locations adjacent to blue beads is exactly $\delta$, allowing the full transcription length to ``just barely'' discover the binding configuration. The geometry of $S_{\delta}$ ensures that any oritatami system forming it must have single-stranded portions that reach all the way across the cave. So, in any system with $\delta' < \delta$, since the minimal distance at which beads can form a bond across the cave is $\delta$, when the transcription is occurring from a location adjacent to one of the sides, no configuration can be possible which forms a bond with a bead across the cave. Thus, the beads must stabilize without a bond across the cave forcing their orientation and so can stabilize in incorrect locations, meaning $S_{\delta}$ isn't deterministically formed.

%The intuition behind how this proof works is the following.  Periodically throughout the shape, it must be the case that a path of beads is forced to ``reach across'' a gap and have beads which bind with beads on the opposite side.  This requirement is enforced by careful design of the shape.  Furthermore, the specific distance across the gap is tuned to the specific $\delta$ value.  This ensures that at a given delay factor $\delta$, a line of beads can ``just barely'' reach across the gap and stabilize in the correct location.  However, for any $\delta' < \delta$, the delay factor is too small, which allows a path of beads to stabilize before being able to reach the opposite side of the gap and therefore in incorrect locations.

\withOrWithoutAppendix{}{
%\vspace{-10pt}
}

%% file: d1a1_det_finiteness_short.tex
\section{Finiteness of delay-1, arity-1 deterministic oritatami systems}

\tikzstyle{mol} = [fill, circle, inner sep=1.25pt]
\tikzstyle{point} = [fill, circle, inner sep=0.5pt]

\begin{figure}[h]
\centering
\begin{minipage}{0.3\linewidth}
\centering
\scalebox{0.6}{\begin{tikzpicture}

\draw[red, thick] (0, 0) node[mol]{} node[above] {$a$}
-- ++(300:1) node[mol]{} 
-- ++(240:1) node[mol]{} node[below] {$\overline{b}$}
-- ++(0:1) node[mol]{} node[below] {$b$}
-- ++(60:1) node[mol]{} 
-- ++(120:1) node[mol]{} node[above] {$\overline{a}$}
;
\draw[-latex, thick] (1, 0) -- ++(0:1) node[mol]{} node[above] {$a$}
-- ++(300:1) node[mol]{} 
-- ++(240:1) node[mol]{} node[below] {$\overline{b}$}
-- ++(0:1) node[mol]{} node[below] {$b$}
-- ++(60:1) node[mol]{} 
-- ++(120:1) node[mol]{} node[above] {$\overline{a}$}
-- ++(0:1)
;
\draw[dotted, thick, red] (0, 0) -- ++(0:1); 
\draw[dotted, thick] (2, 0) -- ++(0:1) ++(240:2) -- ++(180:1);

\draw(4,0)++(300:1) node {\Large $\cdots$};

\end{tikzpicture}}
\end{minipage}
\begin{minipage}{0.025\linewidth}
\ \\
\end{minipage}
\begin{minipage}{0.3\linewidth}
\centering
\scalebox{0.6}{\begin{tikzpicture}

\draw[red, thick] (0, 0) node[mol]{} node[above] {$a$}
-- ++(300:1) node[mol]{} node[left] {$\overline{b}$}
-- ++(240:1) node[mol]{} node[below] {$\overline{b}$}
-- ++(0:1) node[mol]{} node[below] {$b$}
-- ++(60:1) node[mol]{} node[left] {$\overline{a}$}
-- ++(120:1) node[mol]{} node[above] {$\overline{a}$}
;
\draw[-latex, thick] (1, 0) -- ++(0:1) node[mol]{} node[above] {$a$}
-- ++(300:1) node[mol]{} node[left] {$\overline{b}$}
-- ++(240:1) node[mol]{} node[below] {$\overline{b}$}
-- ++(0:1) node[mol]{} node[below] {$b$}
-- ++(60:1) node[mol]{} node[left] {$\overline{a}$}
-- ++(120:1) node[mol]{} node[above] {$\overline{a}$}
-- ++(0:1)
;
\draw[dotted, thick, red] (0, 0) -- ++(0:1) ++(240:1) -- ++(300:1);
\draw[dotted, thick] (0,0) ++(300:2) -- ++(0:1) ++(120:1) -- ++(60:1) -- ++(0:1) ++(240:1) -- ++(300:1);

\draw(4,0)++(300:1) node {\Large $\cdots$};

\end{tikzpicture}}
\end{minipage}
\begin{minipage}{0.025\linewidth}
\ \\
\end{minipage}
\begin{minipage}{0.3\linewidth}
\centering
\scalebox{0.6}{\begin{tikzpicture}

\draw[red, thick] (0, 0) node[mol]{} node[below] {$a$}
-- ++(60:1) node[mol]{} node[above] {$b$}
-- ++(300:1) node[mol]{} node[below] {$\overline{a}$}
;
\draw[-latex, thick] (1, 0)
-- ++(60:1) node[mol]{} node[above] {$\overline{b}$}
-- ++(300:1) node[mol]{} node[below] {$a$}
-- ++(60:1) node[mol]{} node[above] {$b$}
-- ++(300:1) node[mol]{} node[below] {$\overline{a}$}
-- ++(60:1) node[mol]{} node[above] {$\overline{b}$}
-- ++(300:1) node[mol]{} node[below] {$a$}
-- ++(60:1) node[mol]{} node[above] {$b$}
-- ++(300:0.5)
;

\draw(5,0)++(60:0.5) node {\Large $\cdots$};

\draw[dotted, thick] (0, 0) -- ++(0:4) (0, 0)++(60:1) -- ++(0:4);
\end{tikzpicture}}
\end{minipage}

\caption{Deterministically foldable infinite shapes: 
(Left) A glider at delay-3 and arity-1; 
(Middle) A glider at delay-2 and arity-2, and 
(Right) A zigzag at delay-1 and arity-2.
Seeds are colored in red. 
The rule set used in common is complementary: $a$ with $\overline{a}$ and $b$ with $\overline{b}$. 
}
\label{fig:infinite_shapes_short}
\end{figure}

In this section, we briefly argue that oritatami systems cannot yield any infinite conformation at delay~1 and arity~1 deterministically. For more detail, see Section~\ref{sec:d1a1}.
The finiteness stems essentially from the particular settings of these parameters. 
The \textit{glider} is a well-known infinite conformation foldable deterministically at delay~3 and arity~1, shown in Figure~\ref{fig:infinite_shapes_short} (Left), and can be ``widened'' to adapt to longer delays. 
The glider can be ``reinforced'' with more bonds to fold at delay~2 and arity~2 as shown in Figure~\ref{fig:infinite_shapes_short} (Middle). 
Even at delay~1, arity being 2 allows for the infinite zigzag conformation shown in Figure~\ref{fig:infinite_shapes_short} (Right). 
Thus, we are left with just two possible settings of delay and arity under which infinite deterministic folding is impossible: arity~1 and delay at most 2. 
Here we set delay to 1 and leave the problem open at the other setting. 
Note that infinite \textit{nondeterministic} folding is always possible at arbitrary delay and arity, as exemplified by an infinite transcript of inert beads, which can fold into an arbitrary non-self-interacting path. 

\begin{figure}[tb]
\centering
\scalebox{0.65}{\begin{tikzpicture}

\foreach \x in {5, 9, 13} {
	\draw (\x, 0)++(60:1) node[mol] {} node[above] {$a_{j_1}$};
	\draw (\x, 0)++(120:1) node[mol] {} node[above] {$a_{j_4}$};
	\draw (\x, 0)++(240:1) node[mol] {} node[below] {$a_{j_3}$};
	\draw (\x, 0)++(300:1) node[mol] {} node[below] {$a_{j_2}$};
}

\draw (4, 0) node[point] {} ++(0:1) node[point]{} node[above] {$p$} ++(0:1) node[point]{};

\draw (7, 0) node {$\Rightarrow$};

\draw[-latex, thick] (8, 0) node[mol] {} node[below] {$a_{i-2}$} 
-- ++(0:1) node[mol] {} node[below] {$a_{i-1}$};
\draw (10, 0) node[point] {} node[above] {};

\draw (11, 0) node {$\Rightarrow$};

\draw[-latex, thick] (12, 0) node[mol] {} node[below] {$a_{i-2}$} 
-- ++(0:1) node[mol] {} node[below] {$a_{i-1}$}
-- ++(0:1) node[mol] {} node[below] {$a_i$}
;

\draw (16, 0) ++(60:1) node[point]{};
\draw (16, 0) ++(300:1) node[point]{};

\foreach \x in {17, 18, 19, 20} {
\draw (\x, 0) ++(60:1) node{$\times$};
\draw (\x, 0) ++(300:1) node{$\times$};
}
\draw (17,0)++(60:1) -- ++(0:3);
\draw (17,0)++(300:1) -- ++(0:3);

\draw[-latex, thick] (16, 0) node[mol]{} node[below] {$a_{i-5}$} -- ++(0:1) node[mol]{} node[below] {$a_{i-4}$} -- ++(0:1) node[mol]{} node[below] {$a_{i-3}$} -- ++(0:1) node[mol]{} node[below] {$a_{i-2}$} -- ++(0:1) node[mol] {} node[above] {$a_{i-1}$} -- ++(0:1) node[mol]{} node[right] {$a_i$} -- ++(60:1) node[mol]{} node[above left] {$a_{i+1}$};

\draw[very thick, dotted] (18, 0) -- ++(60:1);
\draw[very thick, dotted] (20, 0) -- ++(300:1);
\draw[very thick, dotted] (20, 0)++(60:1) -- ++(0:1);

\draw (21, 0)++(300:1) node[point]{} ++(60:1) node[point]{} ++(60:1) node[point]{}; % ++(120:1) node[point]{} ++(180:1);

\draw[-latex, dashed] (21, 0) -- ++(300:0.75);
%\draw[-latex, dashed] (5, 0)++(60:1) -- ++(300:0.75);
%\draw[-latex, dashed] (5, 0)++(60:1) -- ++(0:0.75);
%\draw[-latex, dashed] (5, 0)++(60:1) -- ++(60:0.75);

\end{tikzpicture}}
\caption{Stabilization of the bead $a_i$ through a tunnel section formed by the four beads $a_{j_1}, a_{j_2}, a_{j_3}, a_{j_4}$. 
}
\label{fig:tunnels_short}
\end{figure}

At delay~1, a bead cannot collaborate with its successors so that it has to bind to as many (other) beads as possible for stabilization. 
It can however get stabilized without binding to any other bead only when just one point left unoccupied around.  
Such a non-binding stabilization requires four beads already stabilized around one common point; see Figure~\ref{fig:tunnels_short}, where four beads $a_{j_1}, a_{j_2}, a_{j_3}, a_{j_4}$ are at neighbors of the point $p$. 
Once the $i-2$-th bead of a transcript, say $a_{i-2}$, is stabilized at one of the two free neighbors of $p$ and also the next bead $a_{i-1}$ is stabilized at $p$, then the next bead $a_i$ cannot help but be put at the sole free neighbor of $p$ and the stabilization does not require any binding. 
Such a structure of four beads around one point is called a \textit{tunnel section}. 
Tunnel sections can be concatenated into a longer tunnel, as shown in Figure~\ref{fig:tunnels_short} (Right). 
Tunnels and unbound beads, or more precisely, their \textit{one-time} binding capabilities are resources for an oritatami system to fold deterministically at delay~1 and arity~1. 
Once bound, a bead cannot bind to any other bead. 
One tunnel consumes two binding capabilities to guide the transcript into it and to decide which way to lead the transcript to, while it can create only one new binding capability; in Figure~\ref{fig:tunnels_short}, $a_i$ does. 
Thus, intuitively, we can see that the number of binding capabilities is monotonically decreasing, and once they are used up, the system cannot stabilize beads deterministically any more. 
Formalizing this intuition brings the following theorem. 

\begin{theorem}\label{thm:d1a1_det_finiteness_short}
	Let $\Xi$ be an OS of delay 1 and arity 1 whose seed consists of $n$ beads, and let $w$ be the transcript of $\Xi$. 
	If $\Xi$ is deterministic, then $|w| \le 9n$. 
\end{theorem}

%% file: shape-appendix.tex
\newpage
\newgeometry{margin={1in,1in}} %% NS: APPENDIX IS NOW IN FULLPAGE (YAHOO!)

\appendix

\input{shapes-def-appendix}

%%%
\afterpage\clearpage
%%%

%\input{definitions-formal}
\input{shapes-finitely-cut-appendix.tex}
%\input{infinite-shapes}

%%%
\afterpage\clearpage
%%%

\input{finite-shapes}

%%%
\afterpage\clearpage
%%%

\input{shapes-scaling-algo-appendix}

%%%
\afterpage\clearpage
%%%

\newpage

\input{smallDelayWeak}
\input{d1a1_det_finiteness}

%%%
\afterpage\clearpage 
%%%

\section{Tribute galery}

This section displays the oritatami foldings of the same iconic shape at the three scales \scaling A3, \scaling B3 and \scaling C3, and it is a kind of a tribute to the field.

\begin{figure}
\centering
\rotatebox{90}{\includegraphics[trim={0cm 7cm 0cm 3cm}, clip, width=.9\textheight, height=\textwidth,keepaspectratio]{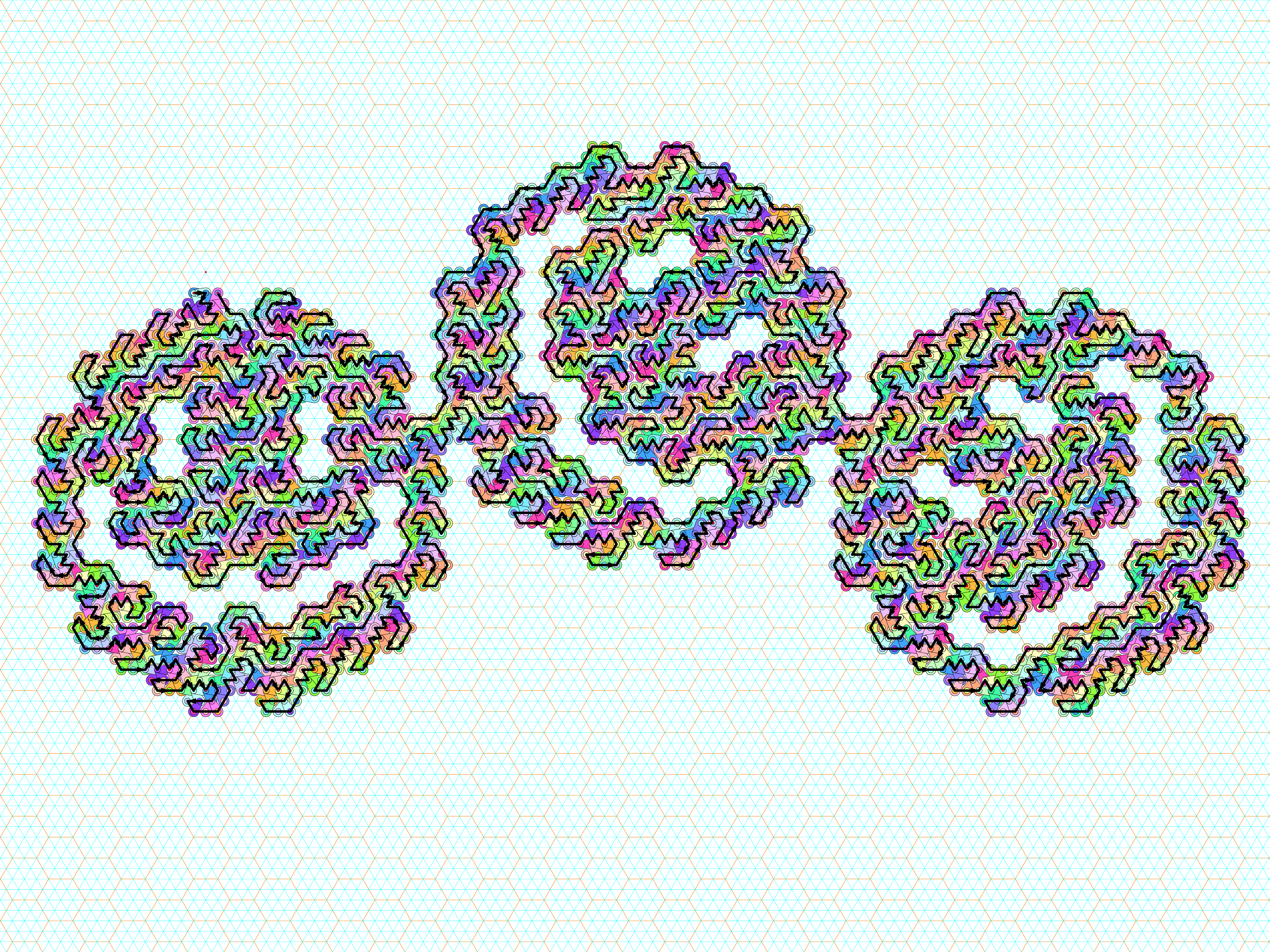}}
\caption{Oritatami ``stacking smileys'' at scale \scaling A3.}
\label{fig:tribute:stacking:A3}
\end{figure}

\begin{figure}
\centering
\rotatebox{90}{\includegraphics[trim={0cm 5cm 0cm 4cm}, clip, width=.9\textheight, height=\textwidth,keepaspectratio]{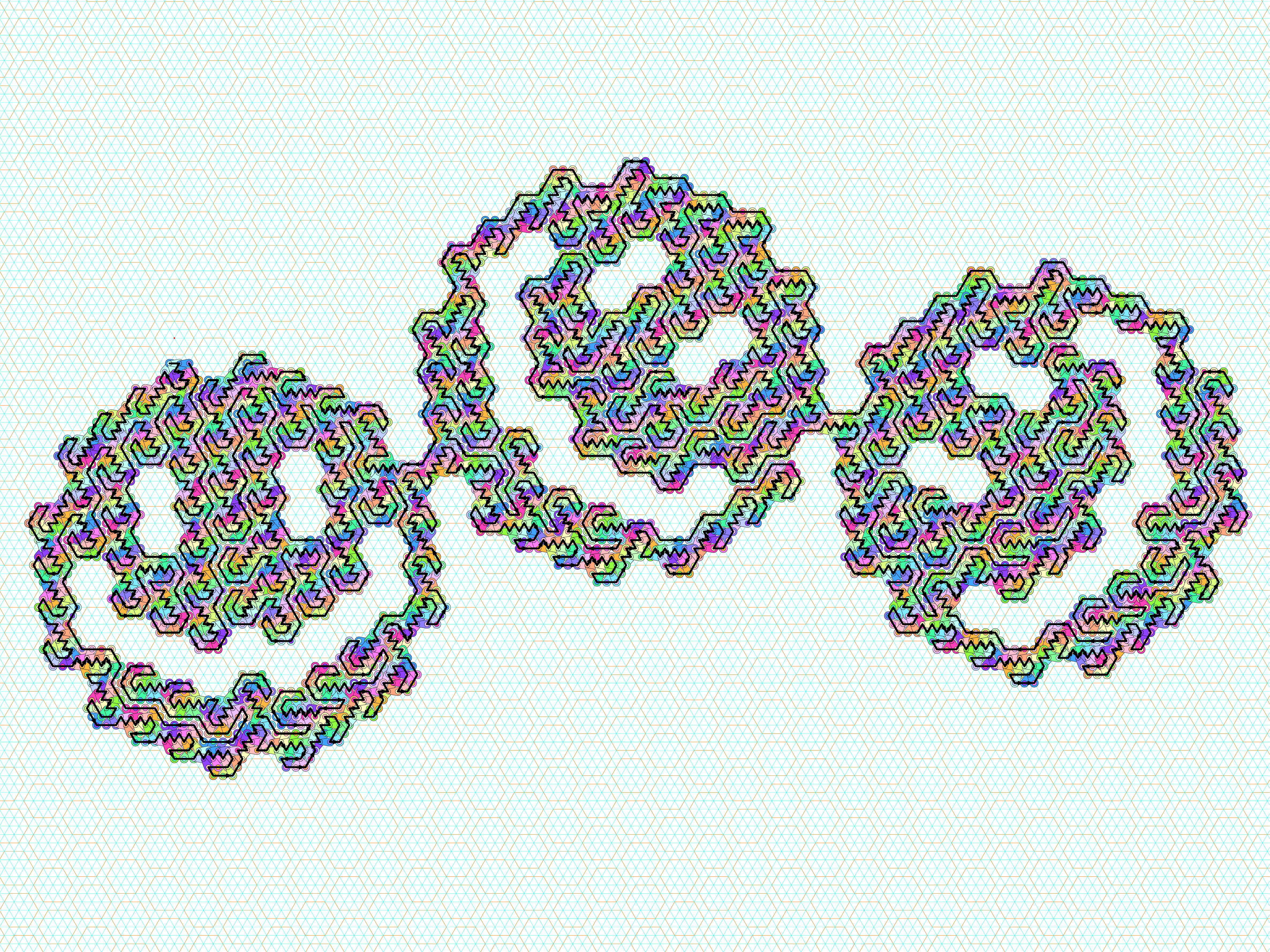}}
\caption{Oritatami ``stacking smileys'' at scale \scaling B3.}
\label{fig:tribute:stacking:B3}
\end{figure}

\begin{figure}
\centering
\rotatebox{90}{\includegraphics[trim={0cm 6cm 0cm 4cm}, clip, width=.9\textheight, height=\textwidth,keepaspectratio]{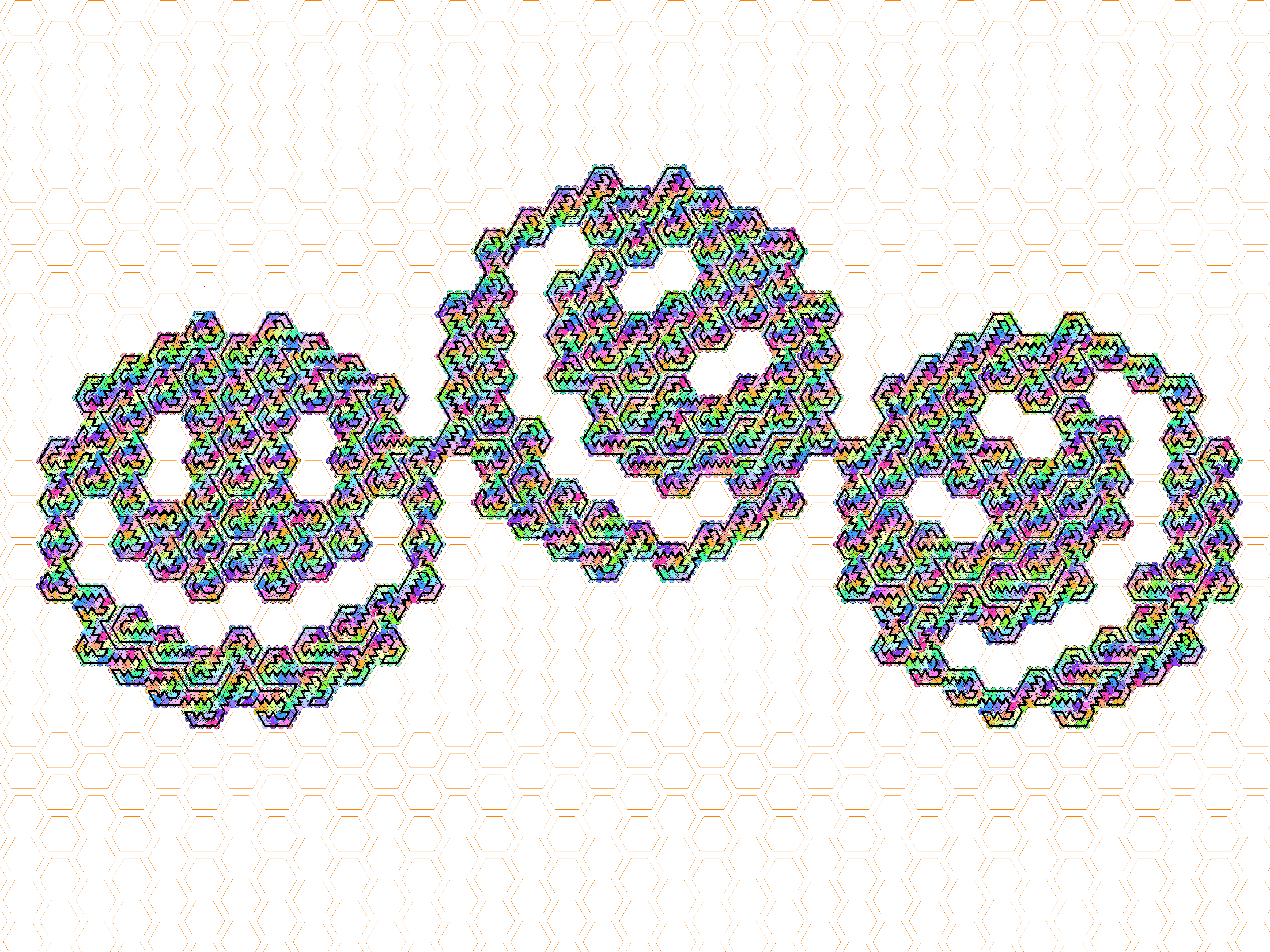}}
\caption{Oritatami ``stacking smileys'' at scale \scaling C3.}
\label{fig:tribute:stacking:C3}
\end{figure}

%\pagebreak

\ifabstract
\newpage
\appendix
\magicappendix
\fi

%% file: shapes-def-appendix.tex
\section{Omitted contents for Section~\ref{sec:def}}
\label{sec:def:appendix}

\subsection{Omitted contents for Subsection~\ref{sec:def:shape:upscaling}}
\label{sec:def:shape:upscaling:app}

A scaling $\Lambda$ is \emph{valid} if it preserves the topology of any shape, that is if:  (1) for all shape $S$ and $p\in\Tlat$, $p\in S$ if and only if $\lambda(p)\in \Lambda(S)$ (we do not allow a cell to be fully covered by others); (2) for all $p,q\in\Tlat$, we have $p\sim q$ if and only if $\Lambda(p) \sim \Lambda(q)$ (cells are neighbors if and only if their associated points in the original shape are neighbors). We say that a scaling $\Lambda$ is \emph{fully covering} if every shape $S$ without hole is mapped to a shape $\Lambda(S)$ without hole.%
\footnote{Recall that a hole of a shape $S$ is a finite non-empty connected component of $\Tlat\smallsetminus S$.}
%
%The \emph{compactness} $\compactness_\Lambda$ of a valid scaling $\Lambda$ is defined as the limit of the ratio $|H_n|/|\Lambda(H_n)|$ where $H_n = \{(i,j)\in\Tlat: |i|< n, |j|<n, |i-j|< n\}$ is the (filled) hexagon of radius $n-1$ with $n$ vertices on each side. 
%\todoi{using $\bar\mu$ is incorrect for generic valid scaling scheme... Need to find an other way to say that. Compactness of each scaling is correct}
%But, for a valid scaling, if we let $\bar\mu = \mu\smallsetminus (\mu+\lambda(\{(1,0), (0,1), (1,1)\})$, $\bar\Lambda(p) = \lambda(p)+\bar\mu$ and $\bar\Lambda(S) = \cup_{p\in S} \bar\Lambda(p)$, then for all shape $S$, 
%%
%$$
%\bar\Lambda(S) \subseteq \Lambda(S) \subseteq \bar\Lambda(S) + \lambda(\{(1,0),(0,1),(1,1)\}).
%$$
%%
%A simple calculation shows that $\compactness_\Lambda = 1/|\bar\mu|$ and is thus one over an integer. 
%
Upscaling schemes \scaling An, \scaling Bn and \scaling Cn are all of them are valid and fully covering.

\begin{figure}[H]
\hfill
\includegraphics[height=3cm]{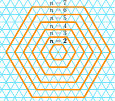}
\hfill
\includegraphics[height=3cm]{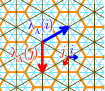}
\hfill
\,
\caption{\captionpar{Left:} Scaling \scaling An cell shapes; they are the concentric balls centered at a vertex in $\Tlat$. \captionpar{Right:}  the cells at scale \scaling C3 (in orange) and the underlying  rotated triangular lattice (in brown), by $-30^\circ$, whose vertices are located at the center vertices of the hexagons.}
\end{figure}

\begin{figure}[H]
\hfill
\includegraphics[height=3cm]{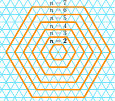}
\hfill
\includegraphics[height=3cm]{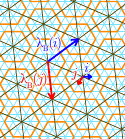}
\hfill
\,
\caption{\captionpar{Left:} Scaling \scaling Bn cell shapes; they are the concentric balls centered at a vertex in $\Tlat$. \captionpar{Right:}  the cells at scale \scaling C3 (in orange) and the underlying  rotated triangular lattice (in brown), by $-30^\circ+\epsilon$, whose vertices are located at the center vertices of the hexagons.}
\end{figure}

\begin{figure}[H]
\hfill
\includegraphics[height=3cm]{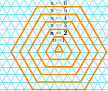}
\hfill
\includegraphics[height=3cm]{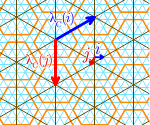}
\hfill
\,
\caption{\captionpar{Left:} Scaling \scaling Cn cell shapes; they are the concentric balls, in $\Tlat$, centered on the vertex at the center of the triangles of $\Tlat$. \captionpar{Right:}  the cells at scale \scaling Cn (in orange) and the underlying rotated triangular lattice (in brown), by $-30^\circ$, whose vertices are located at the center of the triangle at the center of each cell in the original triangular lattice (the orange dot in the figure to the left).}
\end{figure}

%% file: shapes-finitely-cut-appendix.tex
\section{Infinite shapes with finite cut technical details}\label{sec:finite-cut-append}

In this section, we provide the details of the proof of Theorem~\ref{thm:infinite-shapes-short}.

\begin{proof}
Let $K$ be a finite subset of $S$ such that $S\setminus K$ contains two disjoint infinite connected components $S_1$ and $S_2$. For the sake of contradiction, suppose that $\TMO$ is an OS and $S$ is foldable in $\TMO$. 
As $S$ is foldable in $\TMO$ by assumption, there must exist a foldable sequence, $\vec{C} = \{C_i\}_{i=1}^\infty$ say, with result $S$ and each $C_i$ a valid foldable configuration in $\TMO$. Note that, since both $S_1$ and $S_2$ are infinite, we can find a sequence of points in $S$ $\{q_i\}_{i=1}^{|K|+1}$ such that the following properties hold. (1) $q_i$ is in $S_1$ for $i$ odd and $q_i$ is in $S_2$ for $i$ even, and (2) for some $j\in \N$, $q_i$ is a location for a bead in $C_j$ but not a bead in $C_{j-1}$. Then, the routing of $C_j$ must contain a path from a bead at location $q_{i-1}$ to a bead at location $q_i$ as a subpath for each $i$ between $1$ and $|K|+1$ inclusive. This subpath must must contain a point in $K$, for otherwise we arrive at a contradiction of the assumption that $S$ is weakly connected, with $S_1$ and $S_2$ connected solely by $K$. Moreover, one can show that $C_j$ must contain beads at $i$ many distinct points in $K$. Therefore, for $i = |K|+1$, such a configuration $C_j$ contains $|K|+1$ distinct points of $K$. Hence we arrive at a contradiction. Therefore, $S$ is not foldable in $\TMO$.
\end{proof}

%% file: finite-shapes.tex
\section{Self-Assembling Finite Shapes At Small Scale And Linear Delay: Technical Details}\label{sec:hard-coded-shapes-long}

In this section, we give details for the proof of Theorem~\ref{thm:hard-coded}. We do this by first proving the case for scaling \scaling A2, and then the cases for scalings \scaling B2 and \scaling C2 are straightforward extensions.  For the case of scaling \scaling A2 we first show how to construct Hamiltonian cycles in the scaled shapes.

\subsection{Details for construcing Hamiltonian cycles in scaled shapes}\label{sec:HC-long}

\begin{lemma}\label{lem:scaling}
For any finite shape $S$ in the triangular grid graph, there exists a scaling \scaling A2 of $S$, say $S'$, such that there exists a Hamiltonian cycle through the points of $S'$.
\end{lemma}

To prove Lemma~\ref{lem:scaling}, we give a polynomial time algorithm which, given an arbitrary shape $S$ (i.e. a set of connected points) in the triangular grid, creates a version of $S$ scaled by \scaling A2, $S'$, and a Hamiltonian cycle through $S'$.  (See Figure~\ref{fig:example-shape-original} for an example shape.) A program performing this algorithm has been implemented in Python and can be downloaded from \url{http://self-assembly.net/wiki/index.php?title=OritatamiShapeMaker}.  Please note that in order to remain consistent with that code, in this section we present the scalings as rotated $90\degree$ clockwise from the formulation given in Section~\ref{sec:def:shape:upscaling}, which clearly results in the same shapes, just at a different rotation.

\begin{figure}[ht]
\centering
    \begin{subfigure}[b]{.22\textwidth}
	%\begin{subfigure}[t]{2.5in}
	\centering
	\includegraphics[width=1.6in]{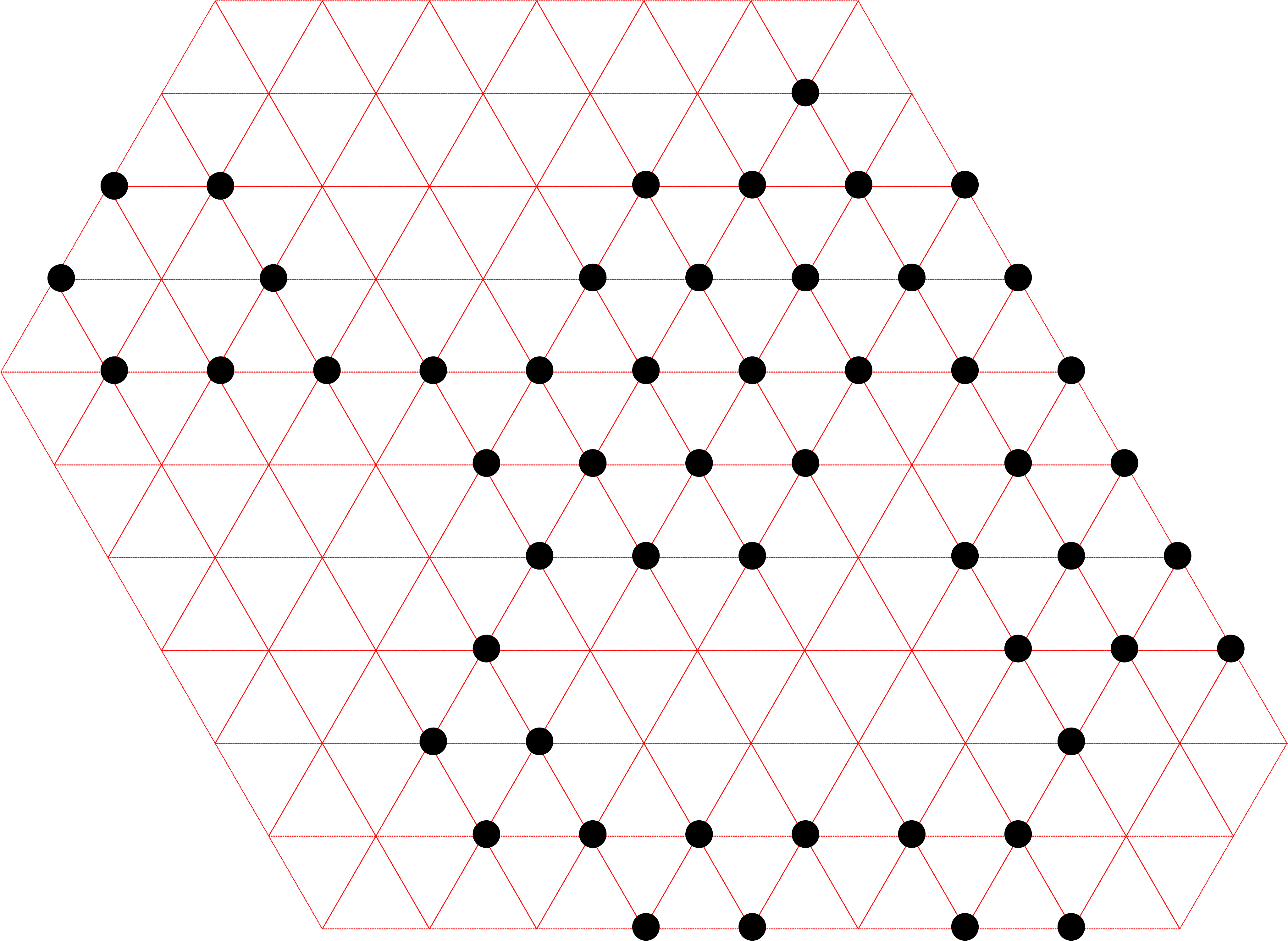}
	\caption{}
    \label{fig:example-shape-original}
	%\end{subfigure}
    \end{subfigure}
\quad
    \begin{subfigure}[b]{.22\textwidth}
	%\begin{subfigure}[t]{3.5in}
	\centering
	\includegraphics[width=1.5in]{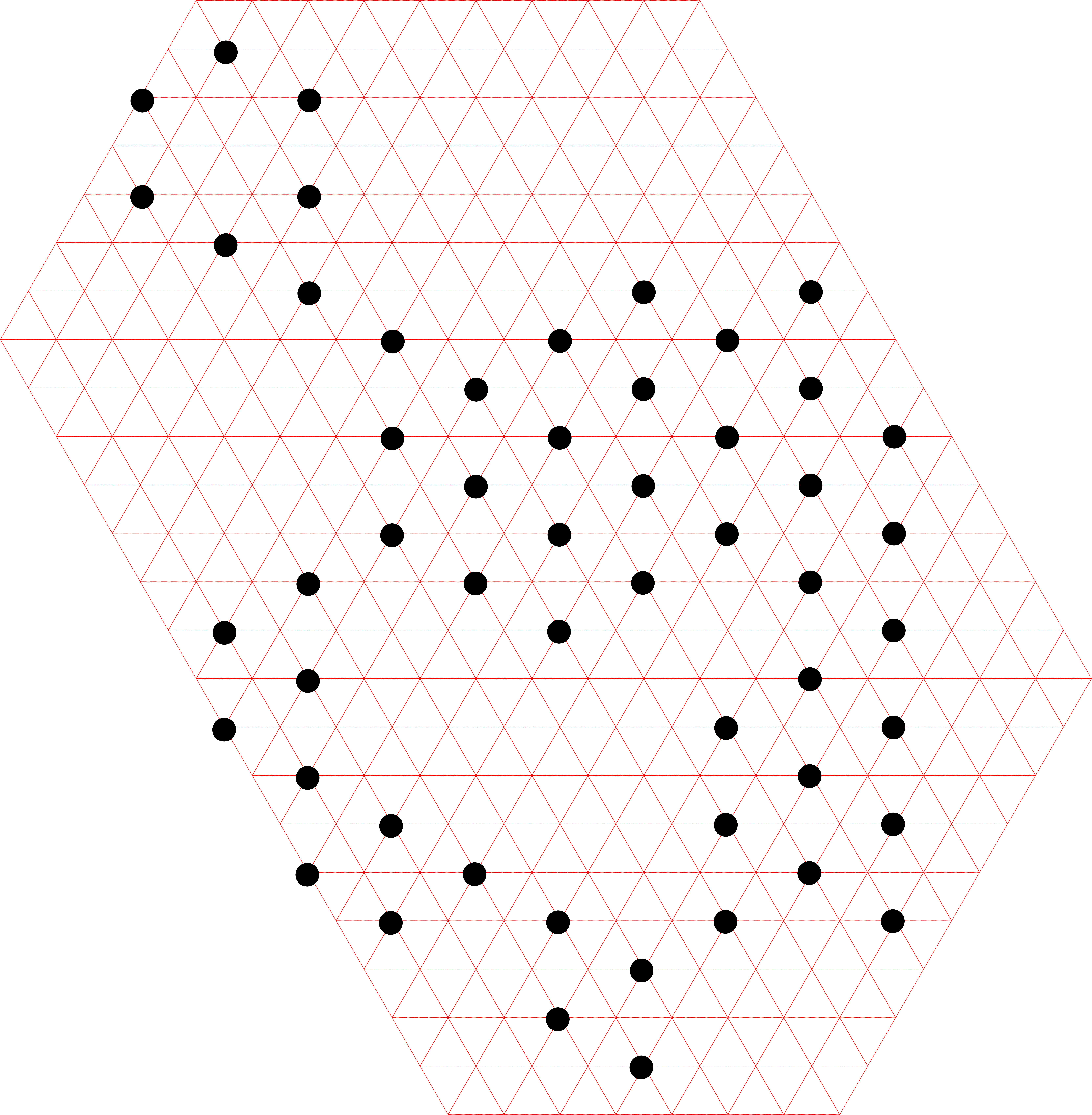}
	\caption{}
	\label{fig:example-shape-scaled-points}
	%\end{subfigure}
    \end{subfigure}
\quad
\begin{subfigure}[b]{.22\textwidth}
    	\centering
        \includegraphics[width=1.6in]{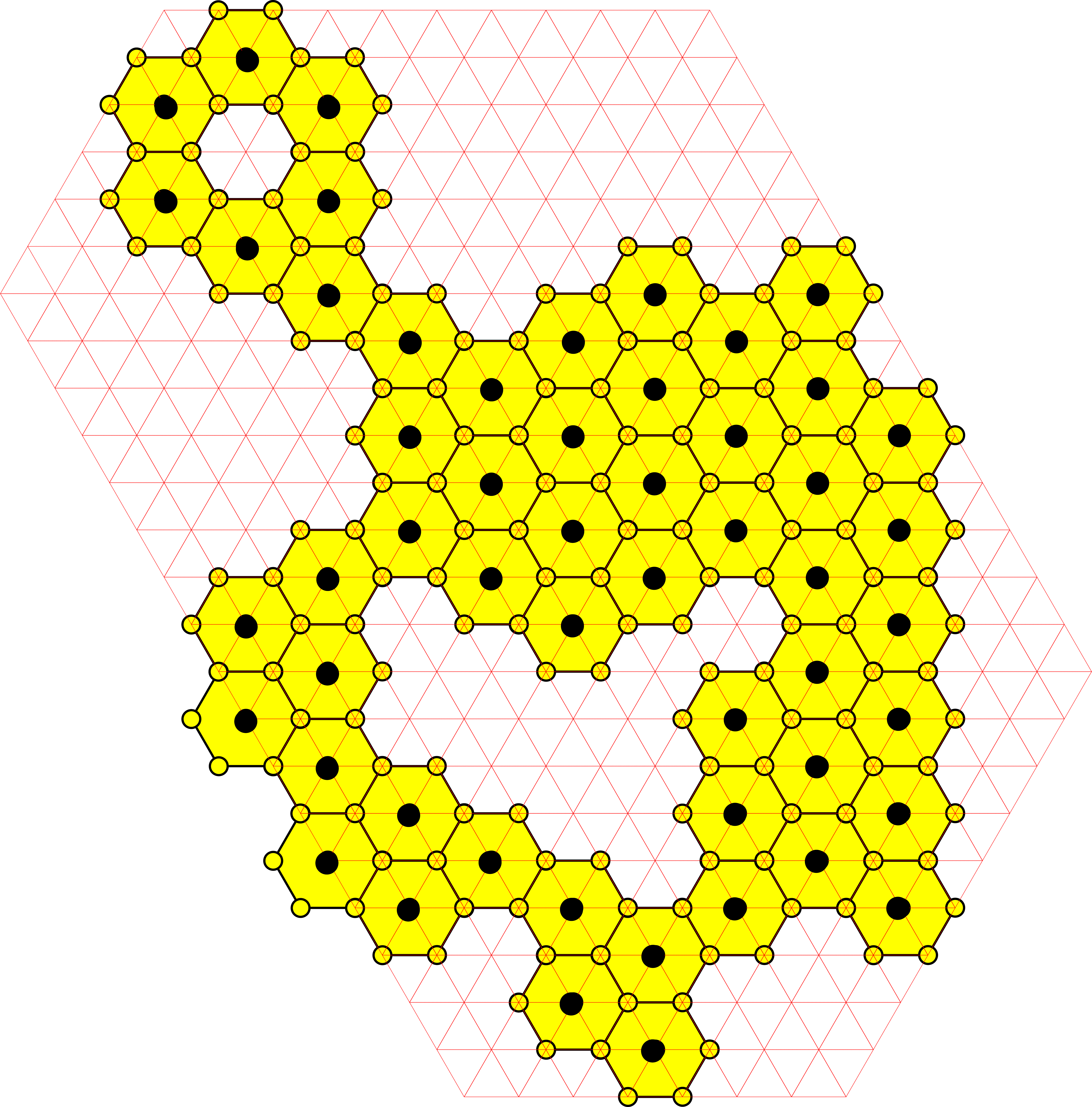}
        \caption{}
        \label{fig:example-shape-pixel-widgets}
    \end{subfigure}
\quad
    \begin{subfigure}[b]{.22\textwidth}
        \centering
        \includegraphics[width=1.4in]{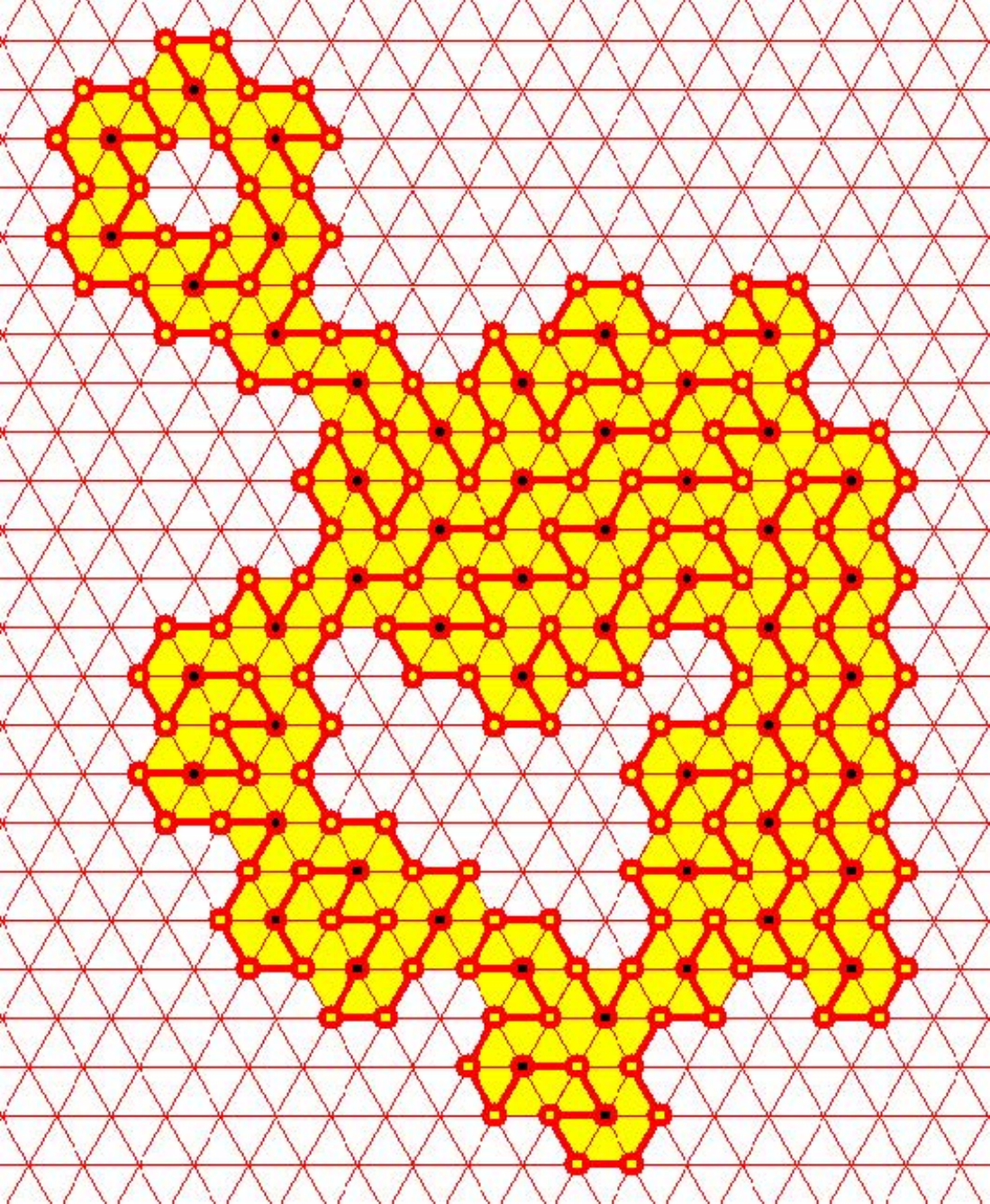}
        \caption{}
        \label{fig:example-shape-cycle}
    \end{subfigure}
	\caption{(a) Example shape $S$, (b) the points of example shape $S$ rotated and scaled. It should be noted that the figure shown in (b) is not the shape $S'$. $S'$ is given in Figure~\ref{fig:example-shape-pixel-widgets}. (c) Example $S'$ consisting of points of $S$ rotated and scaled, then replaced with scaled points, i.e. pixel gadgets. $S'$ is a $2$-scaling of $S$. (d) A Hamiltonian cycle drawn through the points of the gadgets}
	\label{fig:example-shape}
\end{figure}

Given a finite shape $S$, the algorithm to obtain $S'$ and a HC in $S'$ is composed of sub-algorithms which we outline here. For detailed algorithms, see Section~\ref{sec:finite-shapes-algorithms}.

First, the $\texttt{ORDER-POINTS}$ sub-algorithm takes a shape $S$ as input and outputs an ordered list $L$ of all of the points in $S$. $\texttt{ORDER-POINTS}$ is as defined in Algorithm~\ref{alg:order-points} (and the subroutines it utilizes are defined in Algorithms~\ref{alg:get-top-left-point}-\ref{alg:get-bottom-nbrs}) in Section~\ref{sec:finite-shapes-algorithms}.  After the completion of $L = \texttt{ORDER-POINTS}(S)$, the ordered list $L$ contains all of the points in $S$ (this is a standard breadth-first search).

Next, the $\texttt{SCALE-AND-ROTATE-POINTS}$ sub-algorithm takes an ordered list $L$ of all of the points in $S$ and outputs $L'$, which is a scaled and rotated version of $L$. Algorithm~\ref{alg:scale-and-rotate-points} in Section~\ref{sec:finite-shapes-algorithms} formally describes this algorithm. The scaling and rotation is essentially equivalent to expanding the distances between pairs of adjacent points from $1$ to $2\cos(30\degree)$ and rotating points $30\degree$ clockwise relative to the top-left point (which is the first point in both $L$ and then $L'$ by definition of $\texttt{ORDER-POINTS}$). See Figure~\ref{fig:example-shape-scaled-points} for an example shape and the scaled and rotated shape. While $L'$ is an ordered list of these points, the shape $S'$ is defined to simply be the set of points in $L'$.

%Let $L = \texttt{ORDER-POINTS}(S)$, where $\texttt{ORDER-POINTS}$ is as defined in Algorithm~\ref{alg:order-points} (and the subroutines it utilizes are defined in Algorithms~\ref{alg:get-top-left-point}-\ref{alg:get-bottom-nbrs}).  After the completion of $L = \texttt{ORDER-POINTS}(D)$, the ordered list $L$ contains all of the points in $S$ (this is a standard breadth-first search).%The ordering of the points in $L$ can be seen in Figure~\ref{fig:point-ordering}.
%Now let $L' = \texttt{SCALE-AND-ROTATE-POINTS}(L)$ be the ordered list of the scaled and rotated points of $L$.  The scaling and rotation is essentially equivalent to expanding the distances between pairs of adjacent points from $1$ to $2\cos(30\degree)$ and rotating points $30\degree$ clockwise relative to the top-left point (which is the first point in both $L$ and then $L'$). (See Figure~\ref{fig:example-shape-scaled-points} for an example shape after rotation and scaling.)  While $L'$ is an ordered list of these points, let $S'$ simply be the set of points in $L'$.

\begin{figure}[htp]
\centering
\includegraphics[width=5.5in]{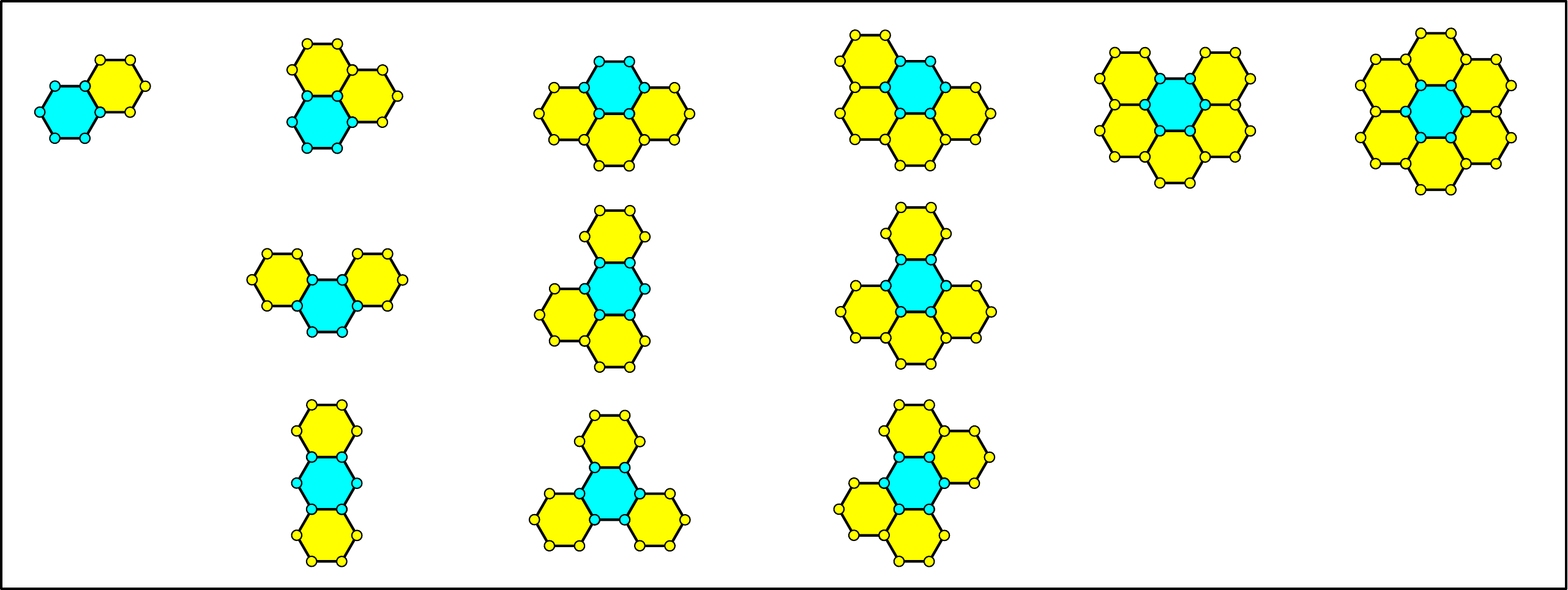}
\caption{All possible neighborhoods (i.e. sets of adjacent neighbors) for a newly added point gadget (blue), up to rotation and after the first.  In columns from left to right, $1,2,3,4,5,$ then $6$ neighbors.}
\label{fig:point-scaling-nbrhoods}
\end{figure}

Now, a new shape $S_G$, and the Hamiltonian cycle (HC), are created by calling $(S_G,HC) = \texttt{ADD-GADGETS}(L')$.  $S_G$ contains a ``point gadget'' for each point in $S'$, which is simply the point and its $6$ adjacent neighbor points (see Figure~\ref{fig:example-shape-pixel-widgets} for an example), and they are added in the order specified by $L'$.  Note that adjacent point gadgets share boundary points. The HC is created by first creating a cycle through the points of the first gadget to be added, and then by extending it to include the points of each subsequently added gadget, one by one.  As each gadget is added, we first note its neighborhood, which is the arrangement of any neighboring gadgets which were previously added to $S_G$ and the edges of the HC which run through them and along the boundary of the newly added gadget.  Modulo rotation and reflection, there are only 12 possible arrangements of neighboring gadgets after the placement of the first.  See Figure~\ref{fig:point-scaling-nbrhoods} for depictions of each. We then locate the specific neighborhood scheme (again, modulo rotation and reflection) from the top rows in Figures~\ref{fig:point-scaling-1nbr} and Figures\ref{fig:point-scaling-2nbrs}-\ref{fig:point-scaling-6nbrs}, and then apply the depicted addition and modification of edges in the HC to extend it to cover all new points of the added gadget, while still covering all previously added points. % For more detail see Section~\ref{sec:HC-long}.

%The HC is created by first creating a cycle through the points of the first gadget to be added, and then by extending it to include the points of each subsequently added gadget, one by one.  An example of this process for a small shape can be seen in Figure~\ref{fig:small-series}. As each gadget is added, we first note its neighborhood, which is the arrangement of any neighboring gadgets which were previously added to $S_G$ and the edges of the HC which run through them and along the boundary of the newly added gadget.  Modulo rotation and reflection, there are only 12 possible arrangements of neighboring gadgets after the placement of the first.  All of these arrangements are presented in Section~\ref{sec:HC-long} in Figure~\ref{fig:point-scaling-nbrhoods}. For each of these arrangements we give rules for extending $S_G$ and the edges of the HC for each newly added gadget.

\begin{figure}[ht]
\centering
	\begin{subfigure}[t]{0.6in}
	\centering
	\includegraphics[width=0.5in]{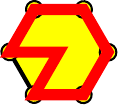}
	\caption{}
	\label{fig:point-scaling-first}
	\end{subfigure}
\quad
	\begin{subfigure}[t]{1.5in}
	\centering
	\includegraphics[width=1.4in]{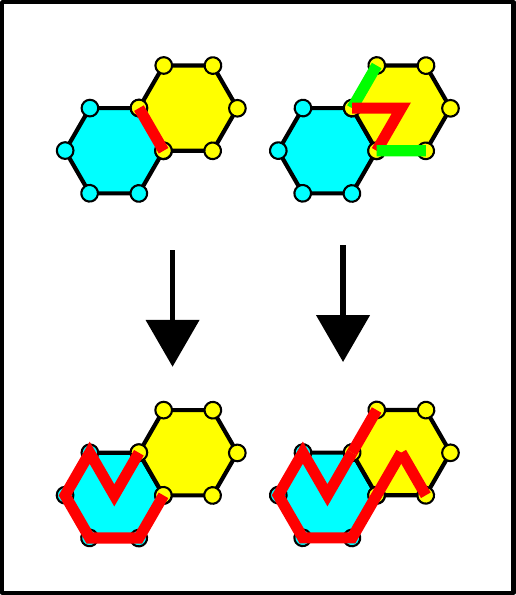}
	\caption{}
	\label{fig:point-scaling-1nbr}
	\end{subfigure}
	\caption{(a) The Hamiltonian cycle (HC) drawn through the first point gadget.  Note that all locations adjacent to this gadget have adjacent to them a straight edge or $V$ of the HC, thus maintaining the boundary invariant for any gadget that could be added in one of those locations, (b) The extension of the Hamiltonian cycle through a newly added point gadget (blue) which only has a single neighboring point gadget (yellow) when it is added. (In this and subsequent figures, only the edges of the existing HC which need to be observed and/or manipulated to extend it into the new gadget are shown, in red and/or green.)  In the left case, the newly added gadget has a straight edge of the HC adjacent (shown in red, on the top), which can be extended into the new gadget as shown (in red, on the bottom left). Due to the boundary invariant, we know that the only other possible scenario is that shown on the right, in which a $V$ is adjacent to the newly added gadget.  Additionally, since we know that the adjacent locations in neighbor positions $0$ and $2$ are empty (because we're in the case with only a single occupied adjacent location relative to the glue gadget), then we also know that the additional edges colored green must be present in the HC (on the top right), because otherwise with the $V$ present, the points which are shared with the new (blue) gadget couldn't be included in the existing HC.  Therefore, the existing red and green HC edges (top) can be replaced with those on the bottom while still including all previously covered points in yellow and now covering all new points in blue.  Note that both extensions result in the same previously covered points and same end points for the line segments, thus not disrupting any other portion of the HC, while covering all new points, and also maintaining the boundary invariant by exposing straight edges or $V$'s on the boundary of the newly added gadget.}
\label{fig:point-scaling-0-and-1}
\end{figure}

We now prove that $S_G$ must have an HC and that one is correctly generated by this procedure.  The specific methods for extending the HC into each new gadget are shown in Figures~\ref{fig:point-scaling-1nbr}-\ref{fig:point-scaling-6nbrs}, and the correctness is maintained due to the following facts:
\begin{enumerate}
    \item Every gadget (after the first) is added in a location adjacent to at least one existing gadget (and this is guaranteed by the ordering of $L$ created in Algorithm~\ref{alg:order-points}).
    \item As each gadget is added, it is guaranteed to have on its boundary an existing edge of the HC or a ``V'' (an example ``V'' can be seen in Figure~\ref{fig:point-scaling-first} in the direction which would be facing a neighbor in position $4$, as numbered in Figure~\ref{fig:point-nbrs}) which includes two of its exterior points.  We will refer to this as the \emph{boundary invariant}, and it will be maintained throughout the addition of new gadgets, as will be shown.
    \item Figure~\ref{fig:point-scaling-nbrhoods} shows all possible neighborhood configurations, modulo rotation and reflection, into which a new point gadget can be added. This is clear by inspection. Figures~\ref{fig:point-scaling-1nbr}-\ref{fig:point-scaling-6nbrs} depict all possible scenarios, modulo rotation, for a point gadget addition. (It is important to note that, in the scenario of each figure, the full set of gadgets adjacent to the newly added gadgets are shown, i.e. empty gadget locations adjacent to the new gadgets are guaranteed to be empty, as there is a figure depicting each scenario where they are filled, up to rotation and reflection).
\end{enumerate}

For each extension of the existing HC into a new point gadget, the necessity is for the replaced edge(s) to be replaced in such a way that the new series of segments (1) has the same end points as the replaced edge(s), (2) all points previously covered by the original edge(s) are covered by the new series of edges, and (3) all points of the newly added gadget which weren't already included in the HC are now included.  The methods for extending the HC while doing that, while also maintaining the boundary invariant, are shown for each possible point gadget addition in Figures~\ref{fig:point-scaling-1nbr}-\ref{fig:point-scaling-6nbrs}.

We prove the correctness of the generation of the HC through the points of $S_G$ using induction.  Our induction hypothesis is the following:

After $n$ points from $L'$ have been added to $S_G$, then
\begin{enumerate}
    \item the HC at that time is a valid Hamiltonian cycle through all points in $S_G$, and
    \item for every location $l$ adjacent to $S_G$ into which a point gadget could be validly placed (i.e. at the correct offset for a neighboring gadget), there is an adjacent gadget already in $S_G$ such that the boundary it shares with $l$ consist of either a straight edge occupying both of the shared points, or a ``V'' (as previously defined) which occupies both shared points. (Note that this is the previously mentioned boundary invariant.)
\end{enumerate}

For our base case, we simply inspect the single gadget and simple HC in Figure~\ref{fig:point-scaling-first}) which exist after the addition of the gadget for the first point from $L'$, and note that this is a Hamiltonian cycle through all 7 points (the outer 6 and the center 1), and that for the valid locations for neighbor gadgets in positions $1,2,3,5,$ and $6$ (as numbered in Figure~\ref{fig:gadget-nbrs}) the HC through the existing gadget has a straight edge through the potential shared points, and for that in position $4$ it shares a ``V'' through those points.  Thus, it holds for the base case.

To prove that if the induction hypothesis holds after $n$ points from $L'$ have been added, it must also hold after $n+1$, we rely on inspection of the scenarios depicted in Figures~\ref{fig:point-scaling-1nbr}-\ref{fig:point-scaling-6nbrs}.  It is easy to verify that in each, after the addition of a gadget, no points previously covered by the HC become uncovered.  It is also easy to verify that in each, the points of the newly added gadget (always depicted in light blue) which were not already included in the border of a previous gadget become covered by the HC.  This ensures that all points have been covered after the addition of the $(n+1)$th gadget.  Finally, it can be seen by inspection that whenever a newly added gadget causes a new neighboring location, which could potentially receive a gadget in the future, to become adjacent to the gadgets of $S_G$, the boundary which is adjacent to that location contains either a straight edge or a ``V''.  (It is important to note that, oftentimes, some boundaries exposed to adjacent locations contain neither of those patterns.  However, whenever that is the case, it is also the case that some other gadget in $S_G$ was already adjacent to that location and must have shared such a boundary.  It is never the case that this previously existing boundary is switched to some other configuration, and thus the boundary invariant is maintained.)  This proves that the induction holds, and thus a Hamiltonian cycle is correctly generated through the scaled and rotated points of $S$.

\subsection{Algorithms for the proof of Theorem~\ref{thm:hard-coded}}\label{sec:finite-shapes-algorithms}

In this section, Figure~\ref{fig:numbering-schemes} gives a visual representation of the numbering schemes for the neighbors of points and gadgets.  Then, the algorithms used to calculate an ordering for the points of an input shape, as well as to scale and rotate it, are presented.

\begin{figure}[ht]
\centering
	\begin{subfigure}[t]{1.5in}
	\centering
	\includegraphics[width=1.3in]{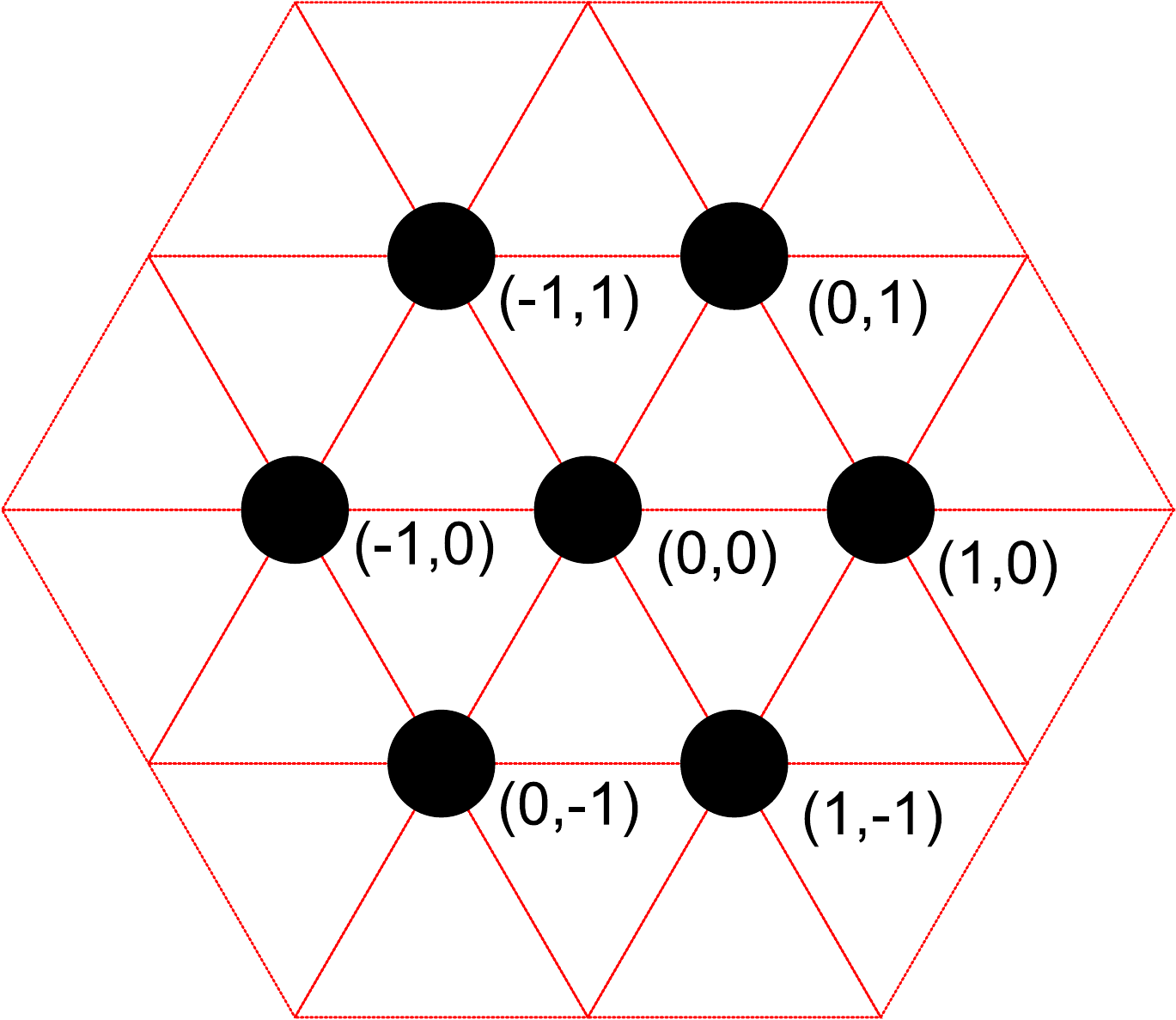}
	\caption{}
	\label{fig:coord-nbrs}
	\end{subfigure}
\quad
	\begin{subfigure}[t]{1.5in}
	\centering
	\includegraphics[width=1.3in]{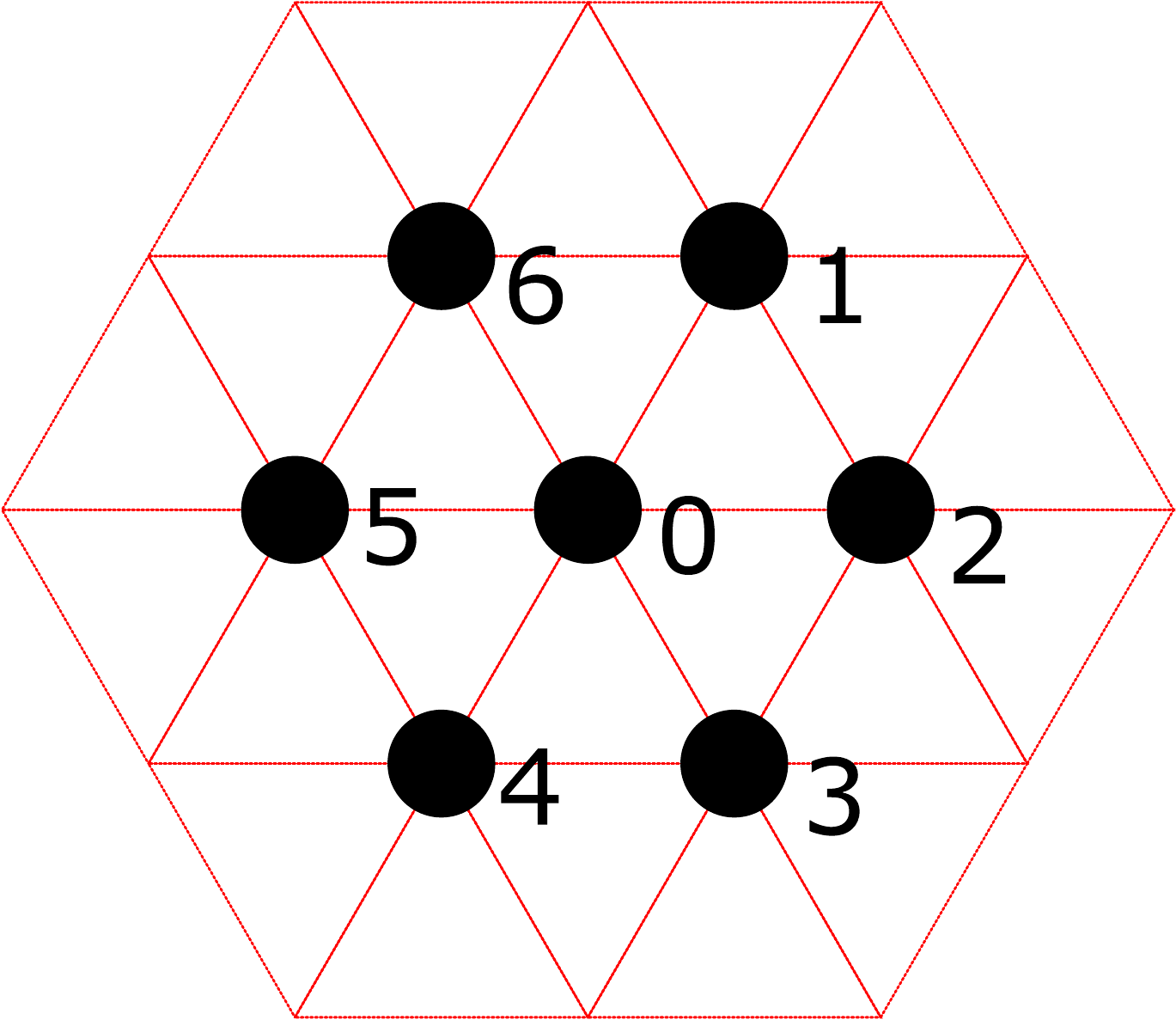}
	\caption{}
	\label{fig:point-nbrs}
	\end{subfigure}
\quad
	\begin{subfigure}[t]{2.5in}
	\centering
	\includegraphics[width=2.3in]{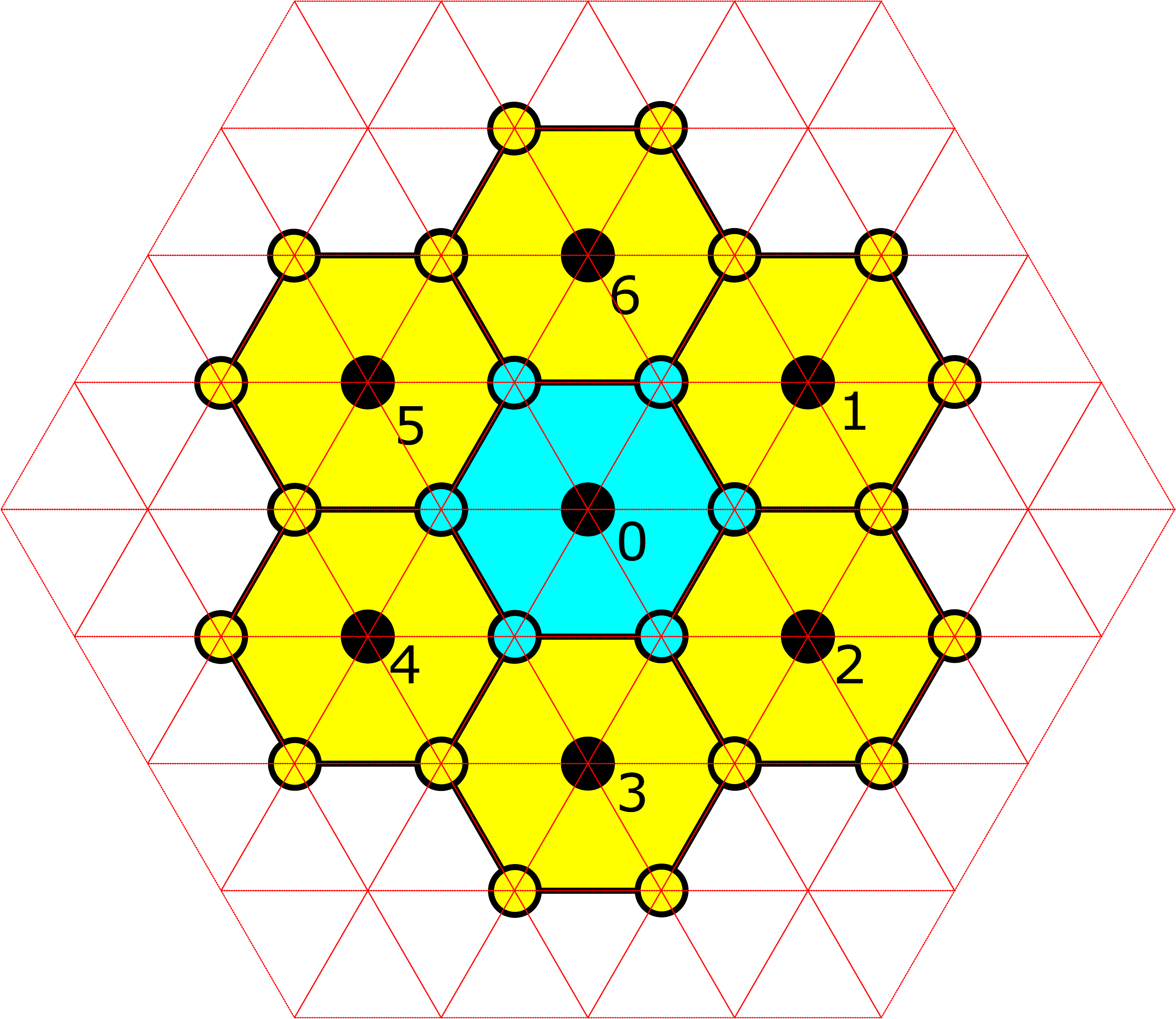}
	\caption{}
	\label{fig:gadget-nbrs}
	\end{subfigure}
	\caption{(a) The coordinate offsets of the neighbors of the point $(0,0)$, (b) the numbering scheme for the neighbors of a point, (c) the numbering for scaled and rotated points replaced by point gadgets.}
	\label{fig:numbering-schemes}
\end{figure}

\begin{algorithm}[H]
\caption{A procedure to assign an ordering to the points in a shape}
\label{alg:order-points}
\begin{algorithmic}[1]
\Procedure{\texttt{ORDER-POINTS}}{$S$} \Comment{Takes a set of points $S$}
    \State $t = \texttt{GET-TOP-LEFT-POINT}(S)$
    \State $L = \{t\}$
    \State $i = 0$
    \State $n = |L|$
    \While {$i < n$}
        \State $p = L[i]$
        %\If {$p \not \in D$}
        \State $L = L + \texttt{GET-TOP-NBRS}(p,L,S)$
        \State $L = L + \texttt{GET-RIGHT-NBR}(p,L,S)$
        \State $L = L + \texttt{GET-BOTTOM-NBRS}(p,L,S)$
        \State $L = L + \texttt{GET-LEFT-NBR}(p,L,S)$
        \State $i = i + 1$
        \State $n = |L|$
        %\EndIf
    \EndWhile
%\EndWhile
\State \textbf{return} $L$
\EndProcedure
\end{algorithmic}
\end{algorithm}

\begin{algorithm}[H]
\caption{A procedure to get the leftmost of the top points of a shape}
\label{alg:get-top-left-point}
\begin{algorithmic}[1]
\Procedure{\texttt{GET-TOP-LEFT-POINT}}{$S$} \Comment{Takes a set of points $S$}
    \State $S' = \{\}$
    \State $p = NULL$
    \ForAll {$q \in S$}
        \If {$q == NULL$}
            \State $p = q$
        \Else
            \If {$q_y > p_y$}
                \State $p = q$
            \Else
                \If {$q_y == p_y$ and $q_x < p_x$}
                    \State $p = q$
                \EndIf
            \EndIf
        \EndIf
    \EndFor
\State \textbf{return} $p$
\EndProcedure
\end{algorithmic}
\end{algorithm}

\begin{algorithm}[H]
\caption{A procedure to get the specified neighbor of a point if it exists within the definition of a shape}
\label{alg:get-nbr}
\begin{algorithmic}[1]
\State $NBRS = [(0,1),(1,0),(1,-1),(0,-1),(-1,0),(-1,1)]$
\Procedure{\texttt{GET-NBR}}{$p, i, S$} \Comment{Takes a point $p$, a neighbor index $0 \le i < 6$, and set of points $S$}
    \State $n = NBRS[i]$
    \State $q = (p_x + n_x, p_y + n_y)$
    \ForAll {$r \in S$}
        \If {$q == r$}
            \State \textbf{return} $q$
        \EndIf
    \EndFor
\State \textbf{return} $NULL$
\EndProcedure
\end{algorithmic}
\end{algorithm}

\begin{algorithm}[H]
\caption{A procedure to find the neighbor immediately left of a given point}
\label{alg:get-left-nbr}
\begin{algorithmic}[1]
\Procedure{\texttt{GET-LEFT-NBR}}{$p, L, S$} \Comment{Takes a point $p$, an ordered list $L$, and a set of points $S$}
\State $L_{ret} = []$
\State $p_1 = \texttt{GET-NBR}(p,S,5)$
\If {$p_1 \neq NULL$ and $p_1 \not \in L$}
    \State $L_{ret} = L_{ret} + [p_1]$
\EndIf
\State \textbf{return} $L_{ret}$
\EndProcedure
\end{algorithmic}
\end{algorithm}

\begin{algorithm}[H]
\caption{A procedure to find the neighbor immediately right of a given point}
\label{alg:get-right-nbr}
\begin{algorithmic}[1]
\Procedure{\texttt{GET-RIGHT-NBR}}{$p, L, S$} \Comment{Takes a point $p$, an ordered list $L$, and a set of points $S$}
\State $L_{ret} = []$
\State $p_1 = \texttt{GET-NBR}(p,S,2)$
\If {$p_1 \neq NULL$ and $p_1 \not \in L$}
    \State $L_{ret} = L_{ret} + [p_1]$
\EndIf
\State \textbf{return} $L_{ret}$
\EndProcedure
\end{algorithmic}
\end{algorithm}

\begin{algorithm}[H]
\caption{A procedure to find the set of neighbors immediately above a given point}
\label{alg:get-top-nbrs}
\begin{algorithmic}[1]
\Procedure{\texttt{GET-TOP-NBRS}}{$p, L, S$} \Comment{Takes a point $p$, an ordered list $L$, and a set of points $S$}
\State $L_{ret} = []$
\State $p_1 = \texttt{GET-NBR}(p,S,0)$
\If {$p_1 \neq NULL$ and $p_1 \not \in L$}
    \State $L_{ret} = L_{ret} + [p_1]$
\EndIf
\State $p_2 = \texttt{GET-NBR}(p,S,1)$
\If {$p_2 \neq NULL$ and $p_2 \not \in L$}
    \State $L_{ret} = L_{ret} + [p_2]$
\EndIf
\State \textbf{return} $L_{ret}$
\EndProcedure
\end{algorithmic}
\end{algorithm}

\begin{algorithm}[H]
\caption{A procedure to find the set of neighbors immediately below a given point}
\label{alg:get-bottom-nbrs}
\begin{algorithmic}[1]
\Procedure{\texttt{GET-BOTTOM-NBRS}}{$p, L, S$} \Comment{Takes a point $p$, an ordered list $L$, and a set of points $S$}
\State $L_{ret} = []$
\State $p_1 = \texttt{GET-NBR}(p,S,4)$
\If {$p_1 \neq NULL$ and $p_1 \not \in L$}
    \State $L_{ret} = L_{ret} + [p_1]$
\EndIf
\State $p_2 = \texttt{GET-NBR}(p,S,3)$
\If {$p_2 \neq NULL$ and $p_2 \not \in L$}
    \State $L_{ret} = L_{ret} + [p_2]$
\EndIf
\State \textbf{return} $L_{ret}$
\EndProcedure
\end{algorithmic}
\end{algorithm}

\begin{algorithm}[H]
\caption{A procedure to scale and rotate the points in a shape}
\label{alg:scale-and-rotate-points}
\begin{algorithmic}[1]
\Procedure{\texttt{SCALE-AND-ROTATE-POINTS}}{$L$} \Comment{Takes an ordered of points $L$}
    \State $p = L[0]$
    \State $L' = [p]$
    \For {$0 < i < |L|$}
        \State $q = L[i]$
        \State $d_x = q_x - p_x$
        \State $d_y = q_y - p_y$
        \State $s_x = (2 * d_x) + d_y$
        \State $s_y = d_y - d_x$
        \State $r = (p_x + s_x, p_y + s_y)$
        \State $L' = L' \cup \{r\}$
    \EndFor
\State \textbf{return} $L'$
\EndProcedure
\end{algorithmic}
\end{algorithm}

\begin{algorithm}[H]
\caption{A procedure to replace all points of an input shape with point gadgets, returning the set of points and a Hamiltonian cycle through them.}
\label{alg:add-gadgets}
\begin{algorithmic}[1]
\Procedure{\texttt{ADD-GADGETS}}{$L$} \Comment{Takes a list of points $L$}
    \State $S_G = \emptyset$
    \State $HC = \emptyset$
    \For {$0 < i < |L|$}
        \State $S_G = S_G \cup \{L[i]\}$
        \For {$0 \le i < 6$}
            \State $n = NBRS[i]$
            \State $q = (p_x + n_x, p_y + n_y)$
            \If {$q \not \in S_G$}
                \State $S_G = S_G \cup \{q\}$
            \EndIf
        \EndFor
        \State $HC = \texttt{EXTEND-HC}(L[i],S_G,HC)$
    \EndFor
\State \textbf{return} $(S_G, HC)$
\EndProcedure
\end{algorithmic}
\end{algorithm}

Explicit pseudocode is not provided for the $\texttt{EXTEND-HC}$ procedure due to its greater complexity\footnote{However, a version which has been implemented in Python can be downloaded from \url{http://www.self-assembly.net}}, but its general functionality is to first inspect the nodes of $S_G$ and the current edges of the HC to determine the pattern of the edges of any gadgets neighboring the newly added gadget.  After determining the number of neighbors, their relative arrangement, and the configuration of their edges, it simply compares the pattern to those seen in Figures\ref{fig:point-scaling-first},\ref{fig:point-scaling-1nbr}, and ~\ref{fig:point-scaling-2nbrs}-\ref{fig:point-scaling-6nbrs}.  Once it finds a match (perhaps after rotation and/or reflection), which it is guaranteed to find since those figures depict all possible scenarios which are possible due to the way the points are added and the HC is extended, it then extends the HC as depicted in the matching figure.  After gadgets have been extended to all points in $L'$, the HC will be correctly completed.

\subsection{Extension of the HC for the proof of Theorem~\ref{thm:hard-coded}}\label{sec:finite-shapes-HC-gadgets}

In this section, we provide graphical representations of the methods for extending the HC into gadgets as they are added in order.  Figures~\ref{fig:point-scaling-6nbrs}-\ref{fig:point-scaling-2nbrs} show gadgets being added into neighborhoods with 2,3,4,5, and 6 existing neighbor gadgets.

\begin{figure}[htp]
\centering
\includegraphics[width=6.0in]{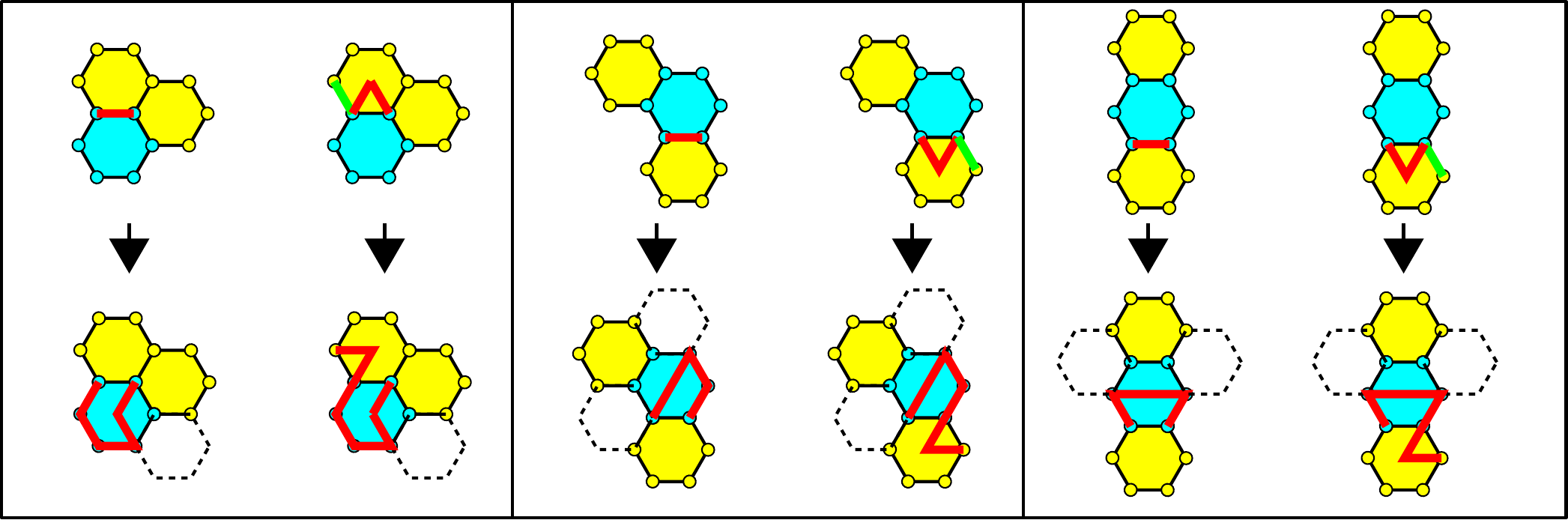}
\caption{The extension of the Hamiltonian cycle through a newly added point gadget (blue) which has exactly two neighboring point gadgets (yellow) when it is added.  (The same principles and manipulations are used as for the cases in Figure~\ref{fig:point-scaling-1nbr}.) For the adjacency configuration in the left box, if one of the adjacent gadgets has a solid edge on the boundary on the edge of the blue gadget, the left option is taken (symmetrically if it is the upper right neighbor).  Otherwise, one of them must have a $V$ adjacent and the right option is taken.  This is also how each of the other two possible adjacency configurations are handled.  Note that all necessary points are covered and line segment end points are maintained in all scenarios, so we must now verify that the boundary invariant is maintained.  Adjacent locations for potential future neighboring gadgets are shown outlined with dashed boundaries if the perimeter of the newly placed blue gadget does not contain either a straight edge or a $V$ on its boundary.  However, for each such location, before the addition of the new gadget (blue), that location was already adjacent to a point gadget contained within the shape, and thus by the induction hypothesis it must already have adjacent to it one of the necessary edge configurations (none of which were modified during the current gadget addition).  Therefore, the boundary invariant is maintained because at least one edge of each such location will have edges with the necessary configuration.}
\label{fig:point-scaling-2nbrs}
\end{figure}

\begin{figure}[htp]
\centering
\includegraphics[width=4.5in]{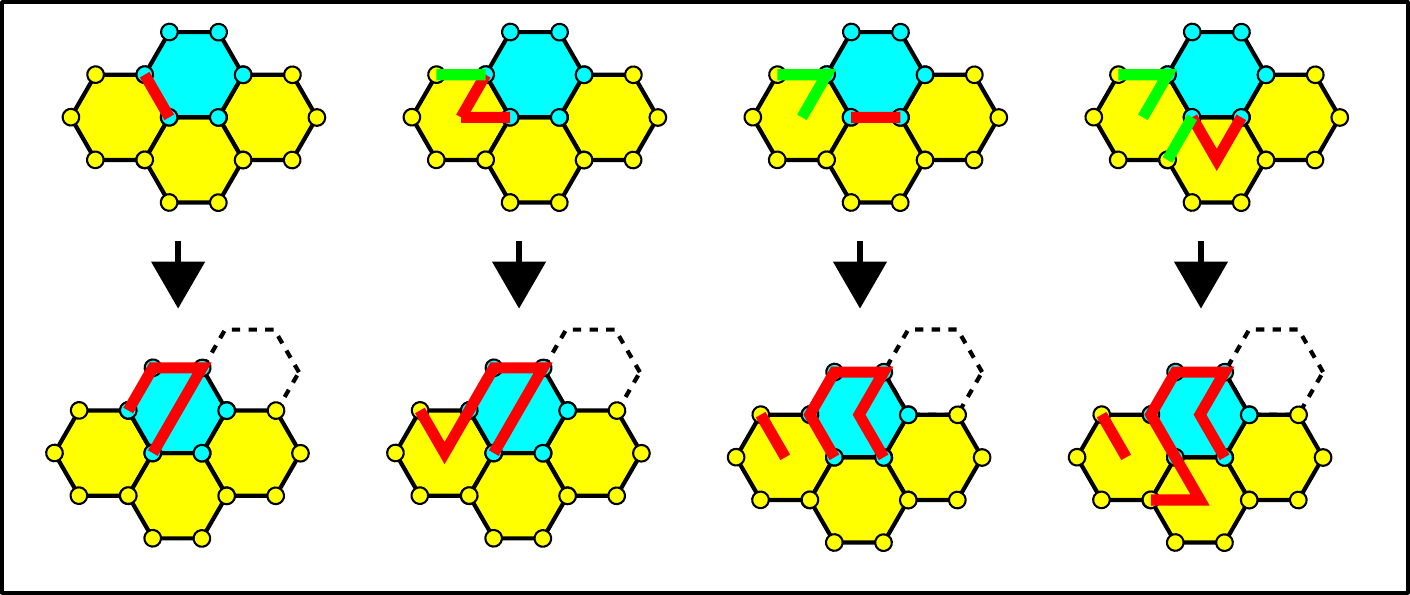}
\caption{The extension of the Hamiltonian cycle through a newly added point gadget (blue) which has exactly three neighboring point gadgets (yellow), in the first of three possible configurations, when it is added.  (The same principles and manipulations are used as for the cases in Figure~\ref{fig:point-scaling-1nbr}.) It is important to note that the gadget used in these scenarios are selected in preference from left to right. Note that this guarantees the existence of the green segments in the third and fourth column.}
\label{fig:point-scaling-3nbrs-1}
\end{figure}

\begin{figure}[htp]
\centering
\includegraphics[width=2in]{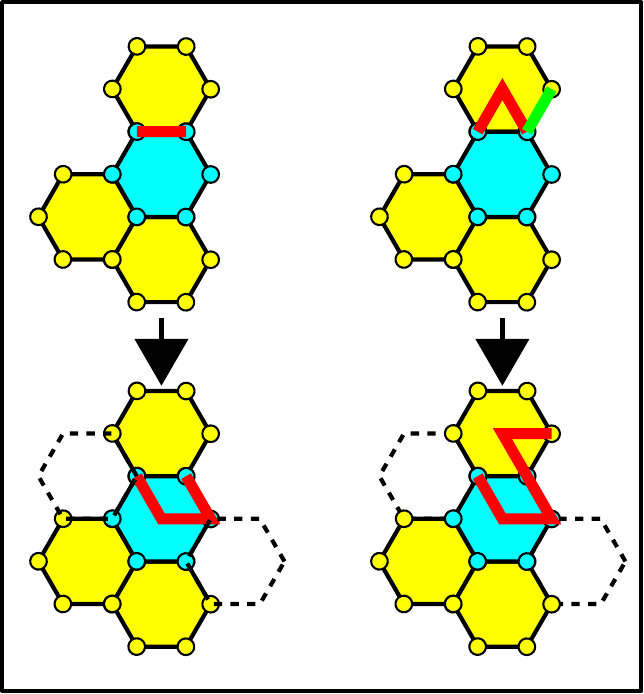}
\caption{The extension of the Hamiltonian cycle through a newly added point gadget (blue), which has exactly three neighboring point gadgets (yellow), in the second of three possible configurations, when it is added.  (The same principles and manipulations are used as for the cases in Figure~\ref{fig:point-scaling-1nbr}.) Note that it suffices to only consider these two case for the following reason. Assume that the yellow gadget at the top does not share a straight edge of the HC with the blue gadget. Since the yellow gadget does not have a gadget to its southeast, the edge between its southeast point and central point must be in the HC. By the same argument, the absence of a gadget to its southwest implies the inclusion of the edge between its southwest point and central point in the HC. Consequently, the yellow gadget provides a ``V'' towards the blue gadget. Hence we need only consider the two cases depicted here. }
\label{fig:point-scaling-3nbrs-2}
\end{figure}

\begin{figure}[htp]
\centering
\includegraphics[width=2.5in]{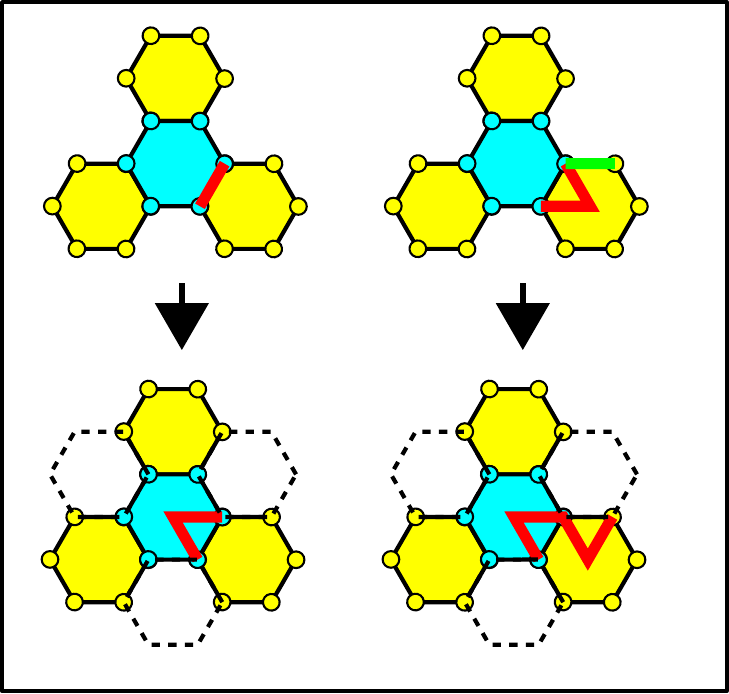}
\caption{The extension of the Hamiltonian cycle through a newly added point gadget (blue) which has exactly three neighboring point gadgets (yellow), in the third of three possible configurations, when it is added.  (The same principles and manipulations are used as for the cases in Figure~\ref{fig:point-scaling-1nbr}.)}
\label{fig:point-scaling-3nbrs-3}
\end{figure}

\begin{figure}[htp]
\centering
\includegraphics[width=4.5in]{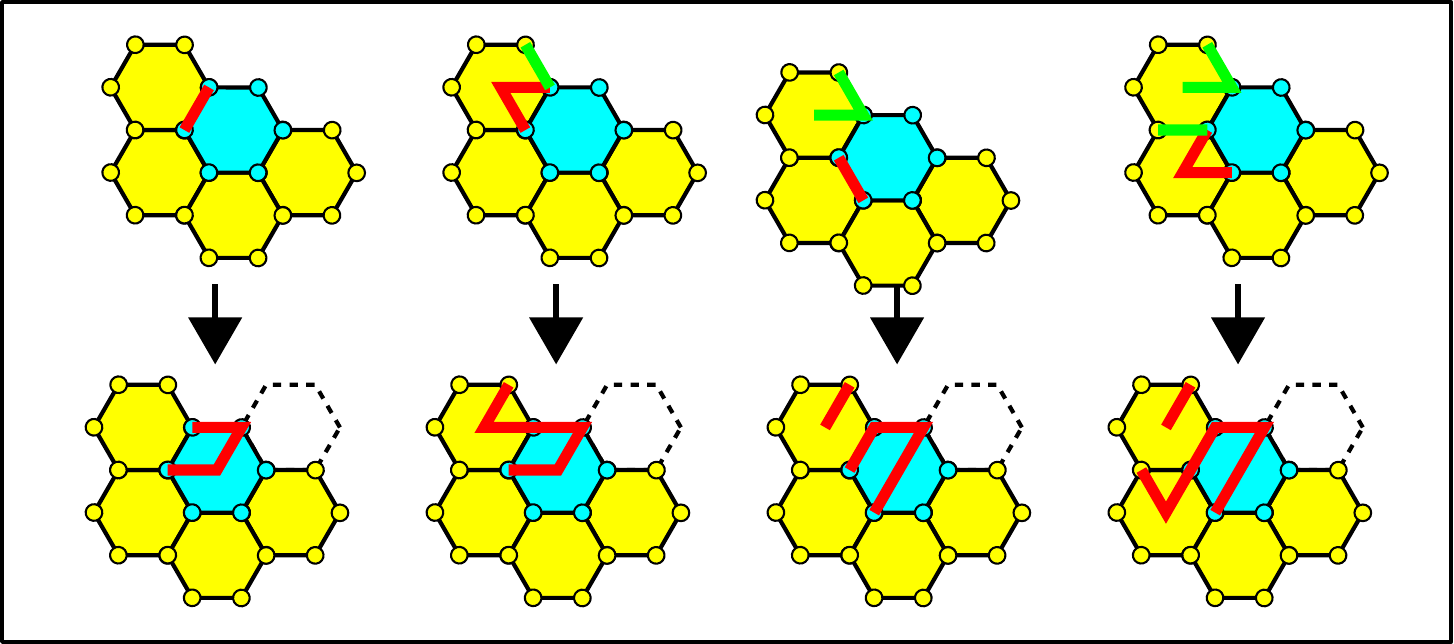}
\caption{The extension of the Hamiltonian cycle through a newly added point gadget (blue) which has exactly four neighboring point gadgets (yellow), in the first of three possible configurations, when it is added.  Again, it is important to note that the gadget used in these scenarios are selected in preference from left to right. Note that this guarantees the existence of the green segments in the third and fourth column.}
\label{fig:point-scaling-4nbrs-1}
\end{figure}

\begin{figure}[htp]
\centering
\includegraphics[width=2.1in]{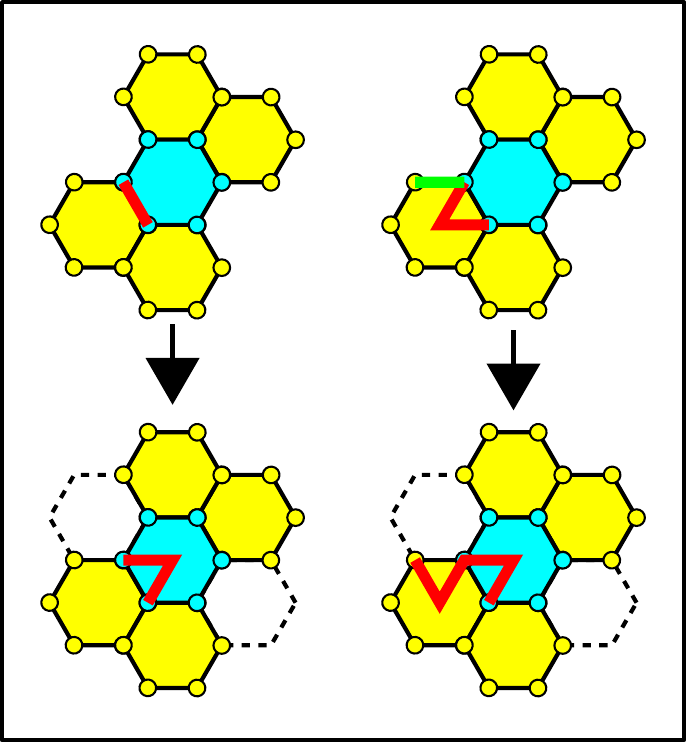}
\caption{The extension of the Hamiltonian cycle through a newly added point gadget (blue) which has exactly four neighboring point gadgets (yellow), in the second of three possible configurations, when it is added.  (The same principles and manipulations are used as for the cases in Figure~\ref{fig:point-scaling-1nbr}.)}
\label{fig:point-scaling-4nbrs-2}
\end{figure}

\begin{figure}[htp]
\centering
\includegraphics[width=2.0in]{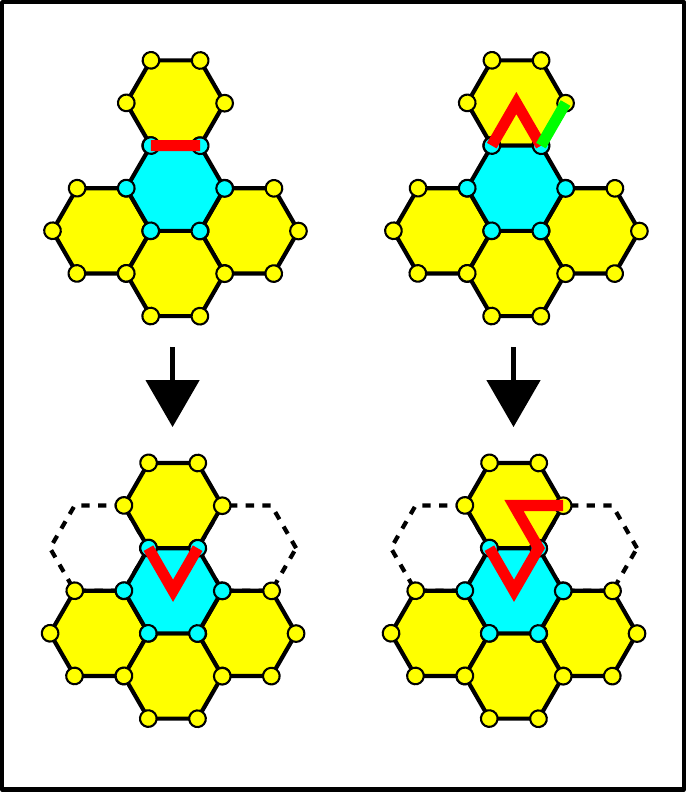}
\caption{The extension of the Hamiltonian cycle through a newly added point gadget (blue) which has exactly four neighboring point gadgets (yellow), in the third of three possible configurations, when it is added.  (The same principles and manipulations are used as for the cases in Figure~\ref{fig:point-scaling-4nbrs-1}.) Note that it suffices to only consider the two cases depicted here by the same argument given in caption of Figure~\ref{fig:point-scaling-3nbrs-2}.}
\label{fig:point-scaling-4nbrs-3}
\end{figure}

\begin{figure}[htp]
\centering
\includegraphics[width=4.8in]{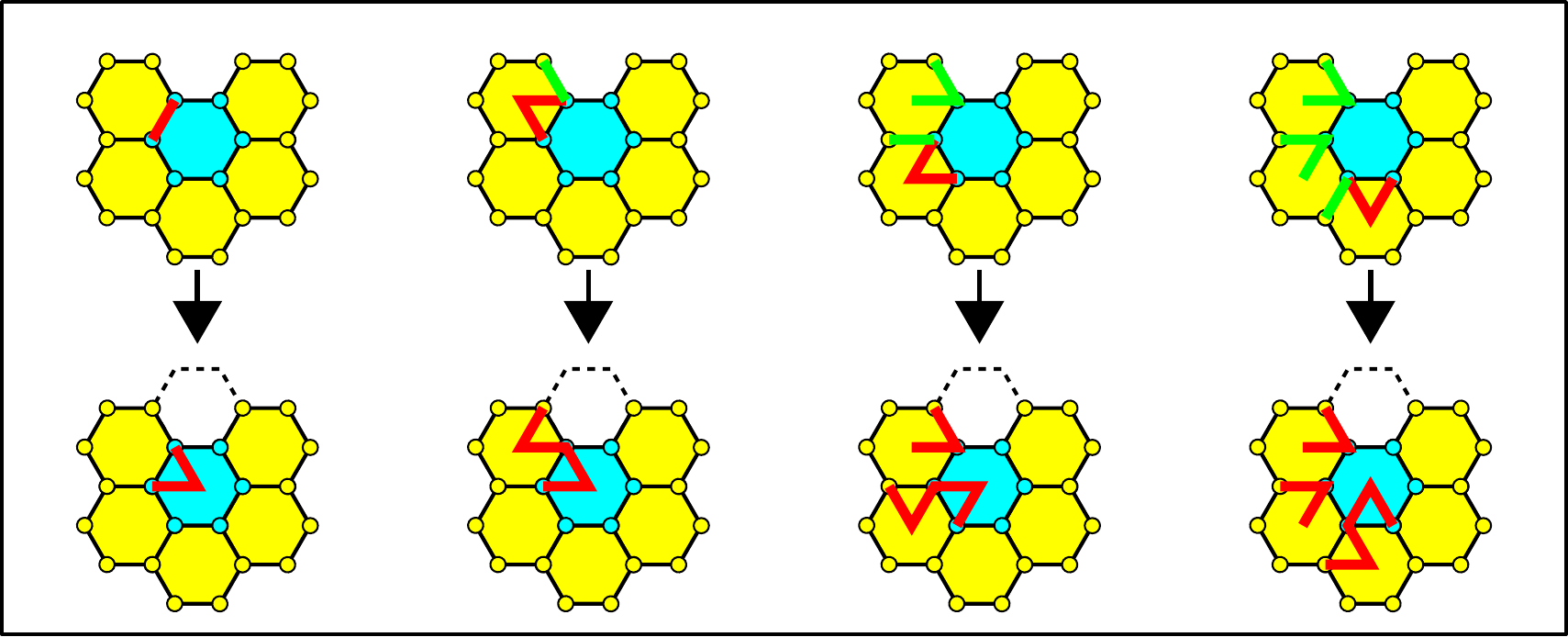}
\caption{The extension of the Hamiltonian cycle through a newly added point gadget (blue) which has exactly five neighboring point gadgets (yellow), in the only possible configuration, when it is added. It is important to note that the gadget used in these scenarios are selected in preference from left to right. Note that this guarantees the existence of the green segments in the third and fourth column. (The same principles and manipulations are used as for the cases in Figure~\ref{fig:point-scaling-4nbrs-1}.)}
\label{fig:point-scaling-5nbrs}
\end{figure}

\begin{figure}[htp]
\centering
\includegraphics[width=2.3in]{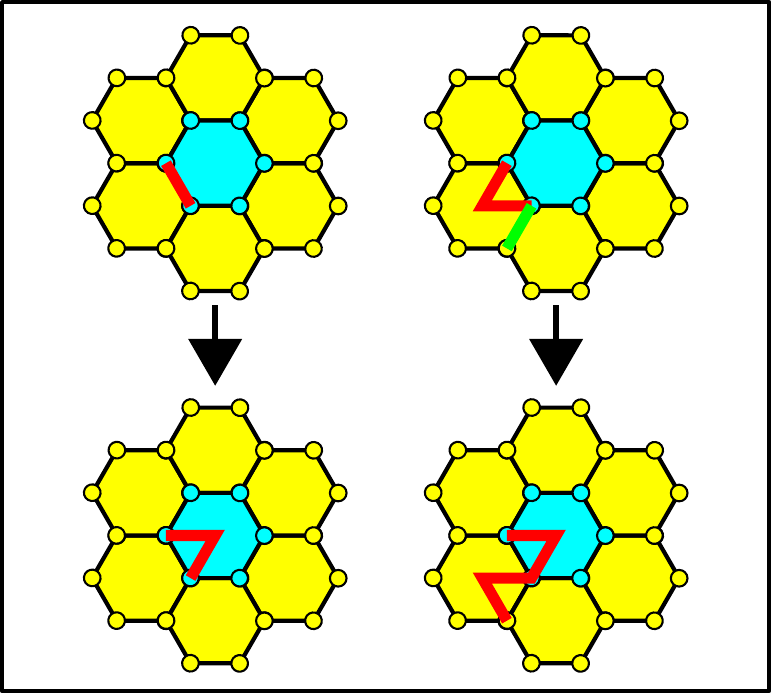}
\caption{The extension of the Hamiltonian cycle through a newly added point gadget (blue) which has exactly six neighboring point gadgets (yellow) when it is added. If any of the adjacent gadgets shares a straight edge of the HC on its boundary, the left option is chosen.  If none of the adjacent gadgets share a straight edge, then at least one must have a $V$ facing the new gadget. Furthermore, in at least one such gadget with a $V$ facing the new location, an edge along the boundary of that gadget (modulo symmetry, as shown in green on the top right) must also be included in the HC. This is because otherwise, to avoid including such an edge, each $V$ would have to connect to another $V$ with that pattern continuing completely around the blue gadget and creating a cycle which only includes those $V$s, which contradicts the fact that a single HC existing before the addition of the blue gadget.  Therefore, in this scenario the red and green edges must be included in the HC and can be modified as shown to extend the HC into the single new additional point.}
\label{fig:point-scaling-6nbrs}
\end{figure}

We now show that an HC can similarly be created through $S$ at scaling \scaling B2.

\begin{lemma}\label{lem:scaling-B2}
For any finite shape $S$ in the triangular grid graph, there exists a scaling \scaling B2 of $S$, say $S'$, such that there exists a Hamiltonian cycle through the points of $S'$.
\end{lemma}

The proof of Lemma~\ref{lem:scaling-B2} is a trivial modification of the proof of Lemma~\ref{lem:scaling} which replaces all cases of extending the HC into a new cell with the three cases shown in Figure~\ref{fig:scaling2B-HC}.  Since the sides of cells do not share points, it is much easier to add new cells while maintaining the HC, and the only cases to be considered are handled as shown in that figure.

\begin{figure}[htp]
\centering
\includegraphics[width=4.8in]{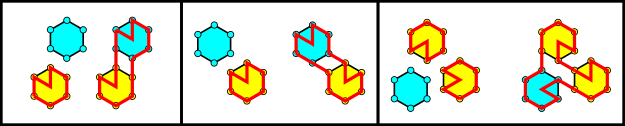}
\caption{The extension of the Hamiltonian cycle through a newly added point gadget (blue) in scaling \scaling B2.  (left) The first choice is taken if at least one neighbor of a newly added pixel gadget exposes an adjacent flat side.  Simply rotate the new pixel gadget so that flat walls face each other and connect them through the four points of those sides. (middle) If all neighbors have adjacent sides exposing ``V''s but there are no two which are adjacent to each other, perform the extension shown which changes the exposed side of the existing neighbor which is closest from a flat side into a ``V''.  However, since there was not a mutual neighbor for that gadget and the newly added gadget, that shared adjacent location must be empty, and if a pixel gadget is ever added there later, it can connect via the flat side of this newly added gadget. (right) If all neighbors expose ``V''s and two of them are adjacent to each other, extend the HC as shown, which modifies no other exposed sides of the existing pixel gadgets.}
\label{fig:scaling2B-HC}
\end{figure}

Finally, we show that an HC can similarly be created through $S$ at scaling \scaling C2.

\begin{lemma}\label{lem:scaling-C2}
For any finite shape $S$ in the triangular grid graph, there exists a scaling \scaling C2 of $S$, say $S'$, such that there exists a Hamiltonian cycle through the points of $S'$.
\end{lemma}

The proof of Lemma~\ref{lem:scaling-C2} is an even more trivial modification of the proof of Lemma~\ref{lem:scaling} which replaces all cases of extending the HC into a new cell with the simple observation that every pixel gadget can be of the shape shown in Figure~\ref{fig:scaling2C-HC}, and that every new pixel gadget is able to place a flat side of its pattern adjacent to that of a flat side of a neighbor, allowing the HC to be extended into the new gadget by simply extending two parallel edges between the gadgets through the four points of those sides. 

\begin{figure}[htp]
\centering
\includegraphics[width=2.0in]{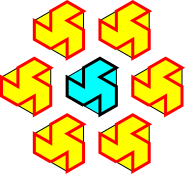}
\caption{The extension of the Hamiltonian cycle through a newly added point gadget (blue) in scaling \scaling C2 is trivial since all gadgets can mai.}
\label{fig:scaling2C-HC}
\end{figure}

\subsection{Time Complexity}

We note that by inspection of the algorithms given in Section~\ref{sec:finite-shapes-algorithms} that the algorithm to produce $S'$ and the HC runs in time $O(|S|^2)$.

\subsection{Self-Assembling Finite Shapes At Small Scale And Arbitrary Delay: Technical Details}\label{sec:hard-coded-details}

We prove Theorem \ref{thm:hard-coded} by splitting each scale factor into its own lemma and proving each separately.

\begin{lemma}\label{lem:hard-coded-2}
Let $S$ be an arbitrary shape such that $|S| < \infty$.  There exists OS $\mathcal{O}_S$ such that $\mathcal{O}_S$ self-assembles $S$ at scaling \scaling A2.
\end{lemma}

\begin{proof}
We prove Lemma~\ref{lem:hard-coded-2} by construction.  Therefore, assume $S$ is an arbitrary shape such that $|S| < \infty$, and let $S'$ be a $2$-scaling of $S$, produced following the algorithm for the proof of Lemma~\ref{lem:scaling}. Let $H$ be the Hamiltonian cycle (HC) through $S'$ found by that algorithm, and for $n = |S'|$ (the number of locations in $S'$), let $P = p_0, p_1, \ldots,p_{n-1}$ be an ordering of the points of $H$ such that $p_0$,$p_1$, and $p_2$ are points $6$, $1$, and $2$, respectively, (as the points of a pixel gadget are labeled in Figure~\ref{fig:point-nbrs}, i.e. NW, NE, and E) of the first pixel gadget added by the algorithm.\footnote{The only requirement for these three points is simply that they are $3$ consecutive connected points in the same pixel gadget in $H$, and by the definition of the algorithm which creates $H$, these three points are guaranteed to be such since the first pixel gadget represents the leftmost of the top points of $S$ and therefore there are no neighboring points which could cause those edges to be removed as future pixel gadgets are added.} We will break $H$ between $p_{n-1}$ and $p_0$ to form our transcript sequence. We now define OS $\mathcal{O}_{S'} = (B, w,\heart,\delay,\alpha)$
such that $\mathcal{O}_{S'}$ self-assembles $S'$. The set of bead types $B$ will contain a unique bead type for each point in $H$, and thus $|B| = n$. The seed $\sigma$ will be the first three bead types in locations $p_0$, $p_1$, and $p_2$. The transcript $w$ will be a finite transcript, with $|w| = n-3$, which enumerates the bead types $p_3,p_4,\ldots,p_{n-1}$, in the order of $P$. The delay factor $\delay$ will be $n-3$ (i.e. the full length of the transcript), and arity $\alpha = 4$.  The rule set $\heart$ is defined by inspecting $H$ overlaid with the ordering of bead types $P$, and adding a pair of bead types to $\heart$ for every pair of adjacent beads which are not connected via the routing, therefore allowing a bond to form between every pair of adjacent beads not already connected by the routing.

\begin{figure}[htb]
\centering
\includegraphics[width=4.0in]{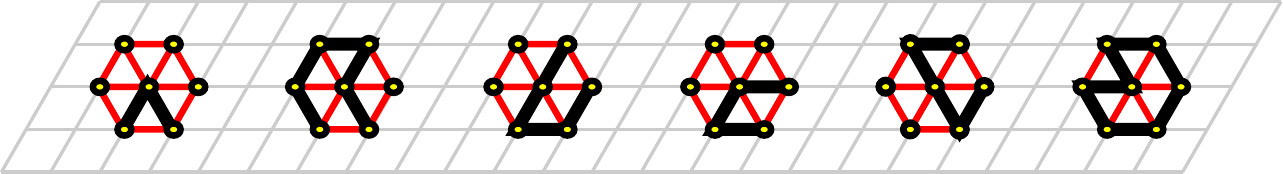}
\caption{The beads of pixel gadgets in various configurations, shown with portions of routing (black) connecting some of them and bonds (red) between them.}\label{fig:gadget-configs}
\end{figure}

We now prove that $\mathcal{O}_{S'}$ is a deterministic system whose single terminal configuration has shape $S'$.  First, it is obvious by the design of $\mathcal{O}_{S'}$ that the beads can be placed into the shape of $S'$ by exactly tracing the HC through $S'$ with the $n$ beads corresponding to the $n$ locations in $S'$.  We call this the \emph{designed configuration} and note that $\mathcal{O}_{S'}$, specifically the rule set $\heart$, is designed so that when the beads are laid out in this configuration, every neighboring pair of beads which is not connected by the transcript can form a bond, and furthermore, no bead can form a bond with any other bead other than those which are neighboring in this configuration.  Therefore, the designed configuration contains the maximal number of bonds which can be formed.

To complete the proof, we must simply show that there is no configuration other than the designed configuration which can contain as many bonds.  We prove this by first noting that in $S'$, there are $7$ points which form the pixel gadget corresponding to each point of $S$ (note that points other than the center may be shared by adjacent pixel gadgets), and proving the following series of claims about the beads of each pixel gadget.

Let $k = |S|$ and $0 \le i < k$, then $g_i$ is the pixel gadget of $S'$ corresponding to point $i$ of $|S|$. In the designed configuration of $\mathcal{O}_{S'}$, there are $7$ beads which correspond to each $g_i$.

\begin{claim}\label{claim:gadget-configs}
There is exactly one configuration of the beads of each $g_i$, modulo rotation and reflection, which allows for the formation of the maximum number of bonds which can be formed among those beads, and that is the subset of the designed configuration corresponding to those beads, modulo rotation and reflection.
\end{claim}

\begin{proof}
To prove this claim, we first note that the bead in the center of the designed configuration of $g_i$ must be (1) connected to two other beads of $g_i$ by transcript connections (by the definition of the transcript) and, in order to form the maximum number of possible bonds, it must form bonds with the other $4$ beads of $g_i$, or (2) in the case of $g_0$ it may be connected via the transcript to only one other bead of $g_0$ due to the location where the HC was broken (to form the path of the transcript), but will then be able to form bonds with the other $5$ beads of $g_0$. In order for this single bead to have all of these connections, it must be situated in the center of a hexagon with those beads surrounding it.  This guarantees that the $7$ beads of the $g_i$ must be arranged in the shape of a pixel gadget (i.e. a hexagon) with the bead in the center matching the center bead of the designed configuration.  To prove that the $6$ beads around the perimeter are in the same relative locations as in the designed configuration (and can't be reordered), we consider the transcript routing and/or bonds between them.  Depending on the arrangement and ordering of pixel gadgets in locations neighboring $g_i$, there may or may not be a connection between a pair of beads on $g_i$'s perimeter formed by the transcript.  However, for any pair that are neighbors in the designed configuration for which there is not such a connection, those two beads can form a bond. Therefore, for the transcript ordering to be maintained as well as for the maximum number of bonds to be formed, the set of all pairs of neighboring beads on the perimeter must match the set of all pairs on the perimeter of the designed configuration, which fixes the ordering of the beads (modulo rotation and reflection).  This, along with the fact that the center bead matches that of the designed configuration, proves the claim that the beads of $\mathcal{O}_{S'}$ corresponding to $g_i$ must match the designed configuration, modulo rotation and reflection.
\end{proof}

\begin{claim}\label{claim:seed-orient}
The beads corresponding to the pixel gadget $g_0$ (which contains the seed), must have the same orientation as $g_0$ in $S'$.
\end{claim}

\begin{proof}
The proof of this claim follows immediately from the fact that the placement of the first three beads of $g_0$ are fixed by the definition of the seed.  Given that $g_0$ must have the same configuration as the corresponding beads in the designed configuration (by the previous claim), which matches that portion of $S'$, the fixed location of the first $3$ beads forces its orientation, i.e. rotation and reflection, to match that of $S'$.
\end{proof}

\begin{claim}\label{claim:gadget-orients}
For each $0 \le i < k$, the beads corresponding to $g_k$ must have the same orientation as $g_k$ in $S'$.
\end{claim}

\begin{proof}
We prove this claim by induction on $g_i$.  Our inductive hypothesis is that, given that the beads corresponding to $g_i$ have the same orientation as $g_i$ in $S'$, then the same must hold for those of $g_{i+1}$.  Our base case is $g_0$, which holds by the previous claim.  Given that the beads corresponding to $g_i$ are oriented to match $S'$, and that by definition of $H$, $g_{i+1}$ must share $2$ beads with $g_i$ and that the beads corresponding to $g_{i+1}$ must be in the configuration matching $S'$ (by the first claim), then the only possible orientation for the beads of $g_{i+1}$ is that which matches $S'$.
\end{proof}

Finally, the proof of Lemma~\ref{lem:hard-coded-2} follows from the fact that $\mathcal{O}_{S'}$ results in a configuration with exactly as many beads as points in $S'$ and the previous three claims which prove that those beads must fold into a configuration in shape $S'$.
\end{proof}

We now provide the statements of the lemmas for the remaining two scalings, \scaling B2 and \scaling C2, and since they are proved by constructions nearly identical to that for scaling \scaling A2, we just refer to that construction.

\begin{lemma}\label{lem:hard-coded-2B}
Let $S$ be an arbitrary shape such that $|S| < \infty$.  There exists OS $\mathcal{O}_S$ such that $\mathcal{O}_S$ self-assembles $S$ at scaling \scaling B2.
\end{lemma}

The proof of Lemma~\ref{lem:hard-coded-2B} follows immediately from Lemma~\ref{lem:scaling-B2} (which shows it is possible to form a Hamiltonian cycle $H$ through $S'$, which is the shape $S$ at scaling \scaling B2) and the observation that an OS nearly identical to that constructed for the proof of Lemma~\ref{lem:hard-coded-2} can be created to self-assemble $S'$.

\begin{lemma}\label{lem:hard-coded-2C}
Let $S$ be an arbitrary shape such that $|S| < \infty$.  There exists OS $\mathcal{O}_S$ such that $\mathcal{O}_S$ self-assembles $S$ at scaling \scaling C2.
\end{lemma}

The proof of Lemma~\ref{lem:hard-coded-2C} follows immediately from Lemma~\ref{lem:scaling-C2} (which shows it is possible to form a Hamiltonian cycle $H$ through $S'$, which is the shape $S$ at scaling \scaling C2) and the observation that an OS nearly identical to that constructed for the proof of Lemma~\ref{lem:hard-coded-2} can be created to self-assemble $S'$.

By Lemmas \ref{lem:hard-coded-2}, \ref{lem:hard-coded-2B}, and \ref{lem:hard-coded-2C}, Theorem~\ref{thm:hard-coded} is proved.

%% file: shapes-scaling-algo-appendix.tex
\section{Omitted contents from Section~\ref{sec:scaling:algo}: Self-assembling finite shapes at scale $n\geq 3$ with delay $1$}

\subsection{Omitted contents from Subsection~\ref{sec:scaling:algo:univ:OS}: Universal tight OS}

\begin{omittedproof}{Theorem}{thm:scaling:algo:universal:bead:type}
Let $\inter m = \{0,\ldots,m-1\}$. Let $\dirSet = \allDir$ be the set of all directions in $\Tlat$. Consider the following \emph{affine $19$-coloring} of the vertices $(i,j)$ of~$\Tlat$:
$$
\vertexColor(i,j) = (2i+3j) \mod 19.
$$
For each $d\in\dirSet$, let $\Delta_d$ be the difference of the colors (modulo $19$) of a vertex and its $d$-neighbor (as the coloring is affine, $\Delta_d$ only depends on $d$): ${\Delta_\dirSE = -\Delta_\dirNW = 5}$, ${\Delta_\dirSW = -\Delta_\dirNE = 3}$, ${\Delta_\dirE = -\Delta_\dirW = 2}$. One can check (see Fig.~\ref{fig:scaling:algo:color}) that every of the $19$ vertices of any translation of the hexagon $H_2$ gets a distinct color. 
   
\begin{figure}[h]
\centerline{\includegraphics[width=4cm]{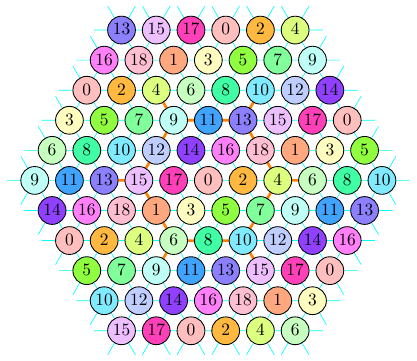}}
\caption{Affine coloring of $\Tlat$. Note that every vertex in any translation of the hexagon $H_3$ receives a distinct color in $\inter{19}$. The neighbor of a given color in a given direction, always receives the same color.} 
\label{fig:scaling:algo:color}
\end{figure}

%For $d\in\dirSet$, let $\vec d$ be the unit vector pointing in direction $d$ in $\Tlat$: ${\vec\dirNW = (-1,-1)}$, ${\vec\dirNW = (-1,-1)}$, ${\vec\dirNE = (0,-1)}$, ${\vec\dirE = (1,0)}$, ${\vec\dirSE = (1,1)}$, ${\vec\dirSW = (0,1)}$ and ${\vec\dirW = (-1,0)}$. 

We consider the bead type set ${U=\{(k,d): k\in\inter{19} \text{ and } d\in\dirSet\}}$ together with the symmetric rule $\heart$ defined by: for all $(k,d)$ and $(k',d')$ in $U$,  
$$
(k,d)\heart(k',d')
~\Leftrightarrow~
%\quad
%\text{if and only if}
%\quad
k' = (k + \Delta_d) \mod 19
\text{ or }
k = (k' + \Delta_{d'}) \mod 19
$$
that is to say, if and only if $k'$ is the neighboring color of $k$ in direction $d$, or $k$ is the neighboring color of $k'$ in direction $d'$. 

Let $\TMO$ be a tight folding. Let us consider the routing $r$ of the result of the folding of $\TMO$ starting from an arbitrary vertex in $\Tlat$. We assign to each vertex $(i,j)$ of $r$, the bead type $(k,d)$ where $k = \vertexColor(i,j)$ and $d$ is the direction of the unique bond it makes when it is placed during the folding $\TMO$. By construction, the transcript obtained by reading the bead types along the routing $r$ will exactly fold into the same shape: indeed, as the delay is $1$, the to-be-placed beads might only get in touch with beads at distance at most $2$ from the last placed beads; as every bead within radius $2$ gets a different color, the unique location where the to-be-placed bead can make its bond, is uniquely defined by the color of the bead it will connect to, which is in turn uniquely characterized by the color of the to-be-placed bead and the direction of the bond it can make, i.e. by its bead type in $U$. \qed
%\todoi{Clarify this argument by explaining which bead a bead can touch and how color identifies a unique neighbor and how direction is enough to identify a given neighbor - Maybe start by assigning a couple of color to each bead and then compress it with the direction?}
\end{omittedproof}

%We say that a bond is tight between two beads if there is only one position for the second bead to make this bond with the first. We say that the molecule has a tight bonding if all its bonds are tight. For a tight bonding molecule, one can give to every bead the beadtype $(\vertexColor(p), \vertexColor(q))$ to every bead that will be placed at position $p$ and making a bond with a bead at position $q$. Using the rule $(c,c')\heart (c',c'')$ for all $c,c',c''$, this would ensure a proper folding of the molecule. Now, using the fact that the coloring is affine, we can replace the couple $(c,c')$ by $(c,d)$ where $d\in\{\dirNW, \dirNE, \dirE, \dirSE, \dirSW, \dirW \}$ as the color of the neighboring vertex in direction $d$ from a vertex of color $c$ is fully determined by $c$ and $d$. This allows to reduce the number of bead type to $19\times 6= 114$ with rule: $(c,d)\heart(c',d')$ for all $d'$ and where $c'$ is the color in direction $d$ from $c$.  

\paragraph{Bead type representation in the figures.} 
A bead with bead type $(k,d)$ will be represented as a small ball of color $k$ inside a link of the same color as the small ball, of color ${k'= (k + \Delta_d) \mod 19}$, of its neighbor in direction $d$ it is pointing to. The routing is shown as a thick translucent black line.

\begin{figure}[H]
\hfill
\includegraphics[width=2cm]{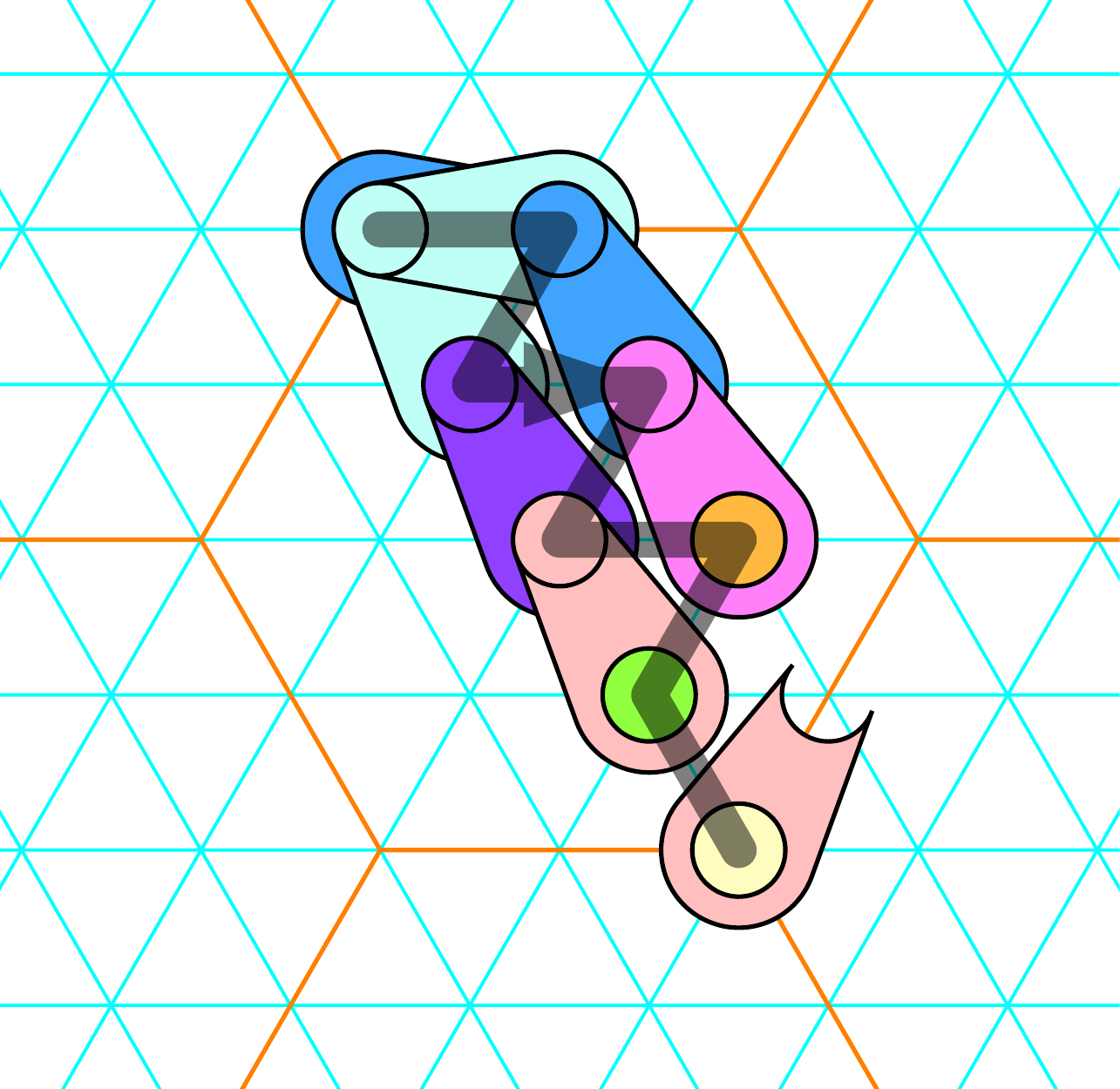}
\hfill
\includegraphics[width=2cm]{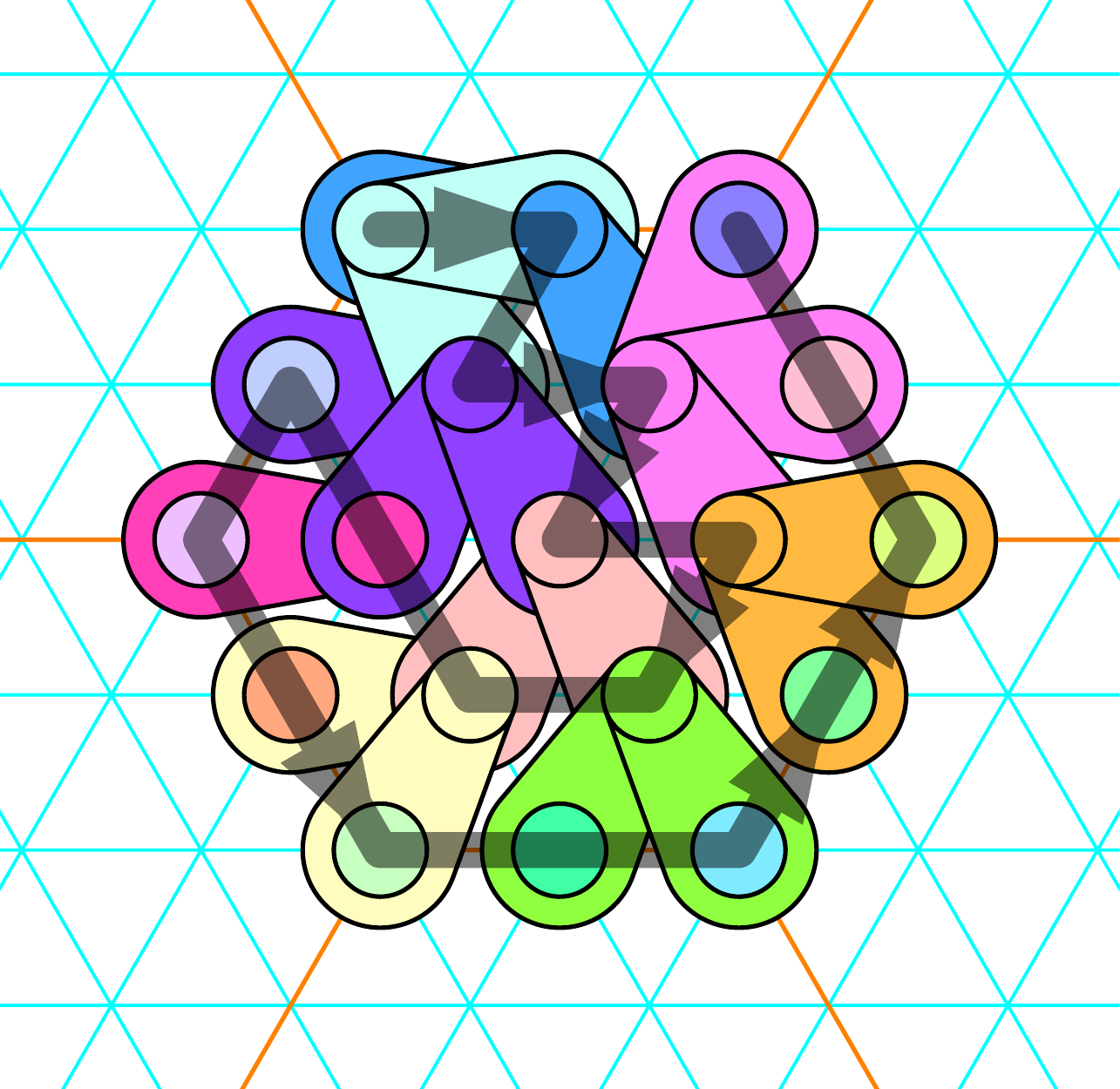}
\hfill\,
\caption{\captionpar{Bead type representation. Left:} a bead looking for its final position; \captionpar{Right:} the fully folded transcript.}
\end{figure}

%%%
%%%
%%%
\afterpage\clearpage
%%%
%%%
%%%

\subsection{Omitted contents from Subsection~\ref{sec:scaling:algo:cover:hex}: Pseudohexagon routing}

Note that we must have: ${a+b = d+e}$, ${b+c = e+f}$ and ${a+f = c+d}$, since ${a\,\vec\dirSE} + {b\,\vec\dirSW} + {c\,\vec\dirW} + {d\,\vec\dirNW} + {e\,\vec\dirNE} + {f\,\vec\dirW} = {(a-c-d+f) \vec\dirE + (a+b-d-e) \vec\dirSW} =  0$. 

\begin{omittedproof}{sketch of Theorem}{thm:scaling:algo:cover:hex}
The algorithm proceeds by covering 6 areas numbered from $A$ to $F$. As illustrated on Fig.~\ref{fig:scaling:algo:cover:hex}, there are four cases depending on the parity if the \dirSW- and \dirS-side lengths $c$ and $d$. 
\begin{figure}%[tb]
    \centering
    \includegraphics[width=0.9\textwidth]{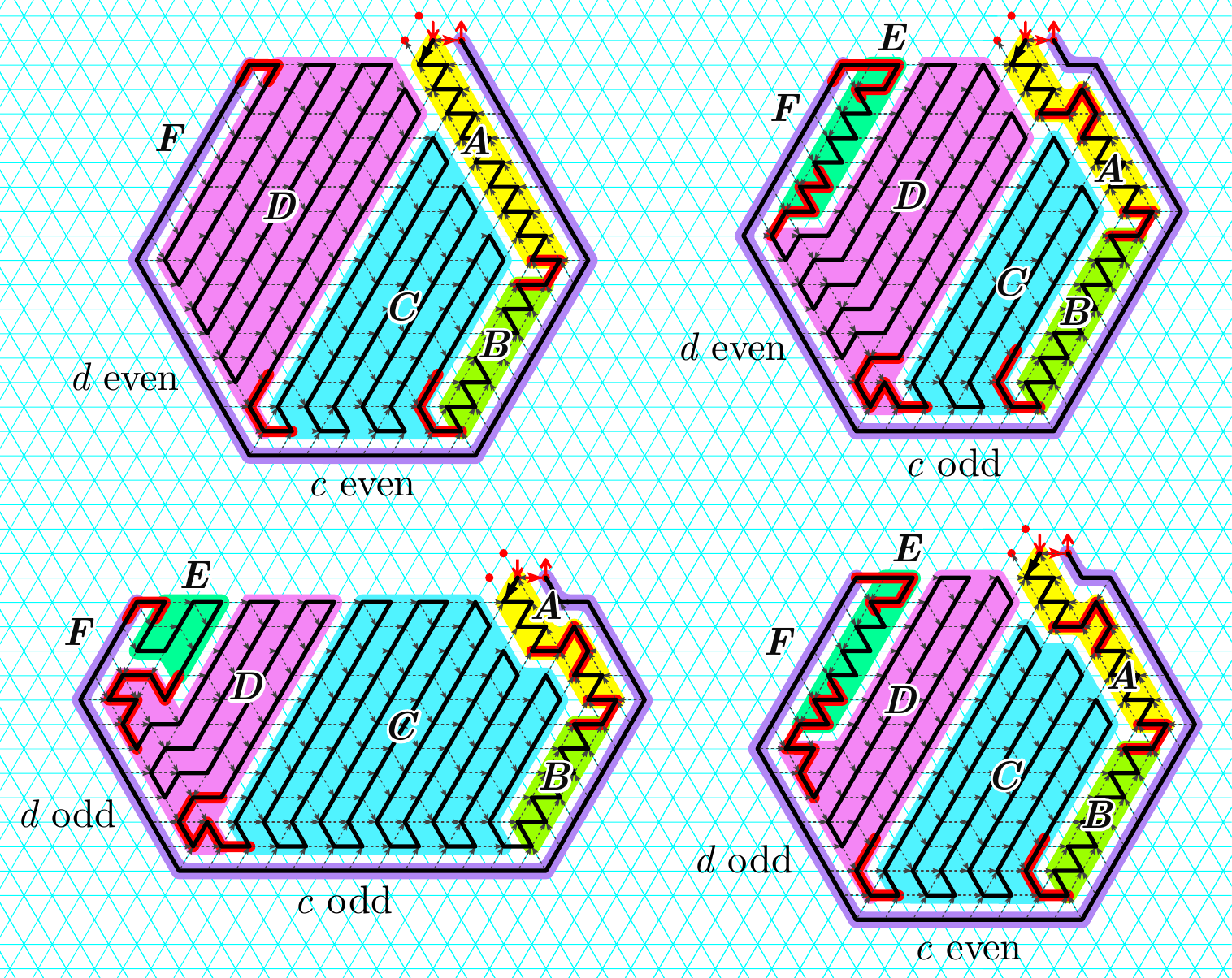}
    \caption{The four cases to design of the self-supported tight path for large enough pseudo-hexagons. Note that the clean edge in ref is not part of the (convex) pseudo-hexagon, but is the edge upon which the pseudo-hexagon is folded.}
    \label{fig:scaling:algo:cover:hex}
\end{figure}
Area $A$ consists in a simple \dirSE-zigzag pattern, or a \dirSE-zigzag pattern with a shift depending on the relative position of the supporting clean edge. Area $B$ consists in a simple zigzag. Area $C$ consists in long \dirNE/\dirSW-switchbacks. The junction to area $D$ is either a simple edge ($c$~even) or a ``$\lambda$''-shape ($c$ odd). Area $D$ consists in long \dirNE/\dirSW-switchbacks that stick along the shape of area $C$'s switchbacks. The junction to the next area is either a simple edge ($c$ and $d$ even) or a ``$\lambda$''-shape ($c$ or $d$ odd). Area $E$ does not exist if $c$ and $d$ are even. If $c$ or $d$ are odd, then area $E$ is either a long \dirNE/\dirSW-switchback ($c$ and $d$ odd), or a \dirNE-zigzag ($c$ and $d$ of opposite parity). Then area $F$ consists in a simple counterclockwise tour of the \cdirNW- to \cdirNE-sides. 
\qed
\end{omittedproof}

Note that as all the routings computed by \algoCoverHex\/ are self-supported and tight, Theorem~\ref{thm:scaling:algo:universal:bead:type} applies and provides an OS with $114$ bead types in linear time that folds each of them correctly. Note also that in the routings generated by this algorithm, all the edges on the five sides \cdirNW\/ to \cdirNE\/ are clean, except for the five edges originating at a corner.

%%%
%%%
%%%
\afterpage\clearpage
%%%
%%%
%%%

%%%
%%% Routing Bn
%%%

\subsection{Omitted contents from Subsection~\ref{sec:scaling:algo:Bn}: Scale \scaling Bn and \scaling Cn with $n\geq 3$}

\begin{figure}%[tb]
    \subfigB4
    \hfill 
    \subfigB5
    \\[2mm]
    \subfigB6
    \hfill 
    \subfigB7
    \\[2mm]
    \subfigB8
    \hfill 
    \subfigB9
    \caption{\captionpar{The self-supported tight routing extensions for scales \scaling B4 to \scaling B9:} in purple, the clean edge used to extend the routing in this cell; in red, the ready-to-use new clean edges in every direction; highlighted in orange, the seed.}
    \label{fig:scaling:algo:routing:B4-9}
\end{figure}

%%%
%%% Flower B3
%%%

\begin{figure}
    \resizebox{\textwidth}{!}{\tabMovieThirteenFrames{.15\textwidth}{flowerB}}
    \caption{The step-by-step construction of a routing folding into a flower shape at scale \scaling B3 according to the algorithm in Section~\ref{sec:scaling:algo:B}}
    \label{fig:scaling:algo:movie:B3}
\end{figure}

%%%
%%%
%%%
\afterpage\clearpage
%%%
%%%
%%%

%%%
%%% Routing Cn
%%%

\paragraph{Scale \scaling Cn}
\label{sec:scaling:algo:C}
is anisotropic. Thus, there are three cases to consider up to rotations: either the cell is the first, or it will plug onto a neighboring clean edge that belongs to either a larger or a smaller side. In \scaling Cn, the clean edges that we will plug onto, are (1) the \emph{counterclockwise-most} of each smaller side, and (2) the \emph{second clockwise-most} of each larger side, of an neighboring occupied cell. For $n\geq 7$, we rely on Theorem~\ref{thm:scaling:algo:cover:hex} to construct such a routing. The tight and self-supported routings for \scaling C3 are listed in Fig.~\ref{fig:scaling:algo:routing:C3} (see Fig.~\ref{fig:scaling:algo:routing:C4-7} for $n=4\ldots7$).
\NSomitted{(see Fig.~\ref{fig:scaling:algo:routing:C4-7} and~\ref{fig:scaling:algo:routing:C8-10} for $n=4\ldots10$).
}
\begin{figure}%[tb]
    \centering
    \tablePatsXFive{1}{3}
    \caption{\captionpar{The self-supported tight routing extensions for scale \protect\scaling C3:} in purple, the clean edge used to extend the routing in this cell; in red, the ready-to-use new clean edges in every direction; highlighted in orange, the seed.}
    \label{fig:scaling:algo:routing:C3}
\end{figure}

\begin{figure}%[tb]
    \subfigC4
    \\[2mm]
    \subfigC5
    \\[2mm]
    \subfigC6
    \\[2mm]
    \subfigC7
    \caption{\captionpar{Extensions for scales \scaling C4 to \scaling C7:} in purple, the clean edge used to extend the routing in this cell; in red, the ready-to-use new clean edges in every direction; highlighted in orange, the seed.}
    \label{fig:scaling:algo:routing:C4-7}
\end{figure}
\NSomitted{
\begin{figure}%[tb]
    \subfigC8
    \\[2mm]
    \subfigC9
    \\[2mm]
    \subfigC{10}
    \caption{\captionpar{The self-supported tight routing extensions for scales \scaling C8 to \scaling C{10}:} in purple, the clean edge used to extend the routing in this cell; in red, the ready-to-use new clean edges in every direction; highlighted in orange, the seed.}
    \label{fig:scaling:algo:routing:C8-10}
\end{figure}
}%NSomitted
Fig.~\ref{fig:scaling:algo:movie:C3} (p.~\pageref{fig:scaling:algo:movie:C3}) presents a step-by-step execution of the routing extension algorithm at scale \scaling C3.
Theorem~\ref{thm:scaling:algo:routing:Cn}
 follows, as above, from Theorem~\ref{thm:scaling:algo:universal:bead:type}.
 %

%%%
%%% Flower B3
%%%

\begin{figure}
    \resizebox{\textwidth}{!}{\tabMovieThirteenFrames{.15\textwidth}{flowerC}}
    \caption{The step-by-step construction of a routing folding into a flower shape at scale \scaling C3 according to the algorithm in Section~\ref{sec:scaling:algo:C}}
    \label{fig:scaling:algo:movie:C3}
\end{figure}

%%%
%%%
%%%
\afterpage\clearpage
%%%
%%%
%%%

%%%
%%% Routing A>=5
%%%

\subsection{Omitted contents from Subsection~\ref{sec:scaling:algo:An>=5}: scale \scaling An with $n\geq 5$}

\emph{We will always root the signature of an empty cell on the clockwise side of a segment}. With this convention, The two least significant bits of the signature of an empty cell with at least one and at most 5 neighboring occupied cells is always \texttt{01}. We are then left with the following possible signatures for an empty cell, sorted by the number of segments around this cell (see Fig.~\ref{fig:scaling:algo:all:CWmost:signatures} p.~\pageref{fig:scaling:algo:all:CWmost:signatures}):
\begin{description}
    \item[No segment:] \NSpattern{0}
    \item[1 segment:] \NSpattern{1}, \NSpattern{100001}, \NSpattern{110001}, \NSpattern{111001}, \NSpattern{111101}, \NSpattern{111111}.
    \item[2 segments:] \NSpattern{101}, \NSpattern{1001}, \NSpattern{10001}, when both have length $1$; \NSpattern{100101}, \NSpattern{101001}, \NSpattern{1101}, \NSpattern{11001}, when their lengths are $1$ and $2$; \NSpattern{101101} when both have length $2$; \NSpattern{110101}, \NSpattern{11101} when one has length $3$.
    \item[3 segments:] \NSpattern{10101}.
\end{description}
\begin{figure}
    \centering
    \includegraphics[width=\textwidth]{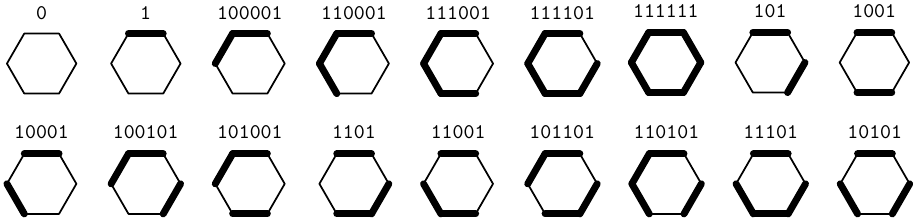}
    \caption{List of the 18 possible signatures rooted on the clockwise-most side of a segment, placed on the \protect\cdirN-side.}
    \label{fig:scaling:algo:all:CWmost:signatures}
\end{figure}

\fitfigurePatsXZero5 
\fitfigurePatsXZero6
\fitfigurePatsXZero7
\fitfigurePatsXZero8

\NSomitted{
\fitfigurePatsXZero6
\fitfigurePatsXZero7
\fitfigurePatsXZero8
\fitfigurePatsXZero9
\fitfigurePatsXZero{10}
\fitfigurePatsXZero{11}
}

%%%
%%%
%%%
\afterpage\clearpage
%%%
%%%
%%%

%%%
%%% Routing A4
%%%

\newpage

% ------------------------------
\subsection{Scale \scaling A4}
% ------------------------------
\label{sec:scaling:algo:A4}

At scale \scaling A4, the side are $3$ edges-long and the support of a clean edge may not belong to the same cell as the edge. We must then need to pay attention to the timing of the clean edges in the path. We solve this issue by always rooting the signature on the side with the \emph{latest} clean edge, where the time of a clean edge is the time of its origin. We are then guaranteed (by an immediate induction) that the support of this clean edge will always have been placed by the folding before the clean edge is folded. 

We can however no more freely rotate the signature and must then design the 33 routing extensions. At scale \scaling A4, the clean edges are located at the second clockwise-most edge of every available side of an occupied cell. The 33 routings are shown on Fig.~\ref{fig:scaling:algo:routing:A4}.
\begin{figure}[tbh]
    \resizebox{\textwidth}{!}{\tabFigFourZero
{0.8}}
    \caption{\captionpar{The 33 tight routing extensions for \scaling A4:} in purple, the latest edge, which is clean and can thus be used to extend the path in this cell; in red, the new potential clean edges available to extend the path for the neighboring cells (only the latest one around the empty cell might be clean); highlighted in orange, the seed for signature \NSpattern0.}
    \label{fig:scaling:algo:routing:A4}
\end{figure}
One can check that by immediate induction:

\begin{lemma}
At every step, the computed routing is self-supported and tight, covers all the cells inserted, and contains a potential clean edge on every available side with the exception of the \dirN-side of the initial cell $\cell(p_1)$. Furthermore, the potential clean edge of the latest available side of every empty cell is always clean. 
\end{lemma}

Theorem~\ref{thm:scaling:algo:universal:bead:type} allows then to conclude that:

\begin{theorem}
Any shape $S$ can be folded by a tight OS at scale \scaling A4.
\end{theorem}

%%%
%%%
%%%

%%%
%%%
%%%
\afterpage\clearpage
%%%
%%%
%%%

\newpage

\subsection{Omitted contents of subsection~\ref{sec:scaling:algo:A3}: Scale \scaling A3}
\label{sec:scaling:algo:A3:apx}

We say that an available (occupied) side of an empty cell $\cell(p)$ \emph{flows clockwise} (see Fig.~\ref{fig:scaling:algo:A3:flows:CW})if its 3 vertices $a,c,d$ (taken in clockwise order around $\cell(p)$) plus the vertex $b$ neighboring $a$ and $c$ inside the occupied neighboring cell, appear in clockwise order in the current routing, i.e. if
$$
\max(\rtime(a),\rtime(b))<\rtime(c)<\rtime(d)
$$
($a$ and $b$ may appear in any relative order). 
\begin{figure}[t]
    \centering
    \includegraphics[scale=1.5]{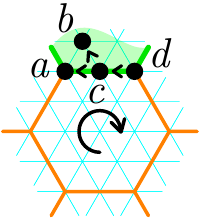}
    \caption{At all time, the vertices $a,b,c,d$ on and inside an available side verify the clockwise order property around an empty cell:\\ 
    \hspace*{3cm} ${\max(\rtime(a),\rtime(b))<\rtime(c)<\rtime(d)}$.} 
    \label{fig:scaling:algo:A3:flows:CW}
\end{figure}
This property ensures that the edge $c\rightarrow d$ is clean ($a$ and $b$ being resp. its support and bouncer) whenever it belongs to the routing, and can thus be used to extend the path (recall Fig.~\ref{fig:scaling:algo:grow:from:clean}).

\smallskip

Regardless of the algorithm, the following invariants are valid for any sequence of cell insertions/anomaly fixing made according to Fig.~\ref{fig:scaling:algo:A3:routing}. These are proved by inspecting the extension patterns together with an immediate induction: 

\begin{invariant}\upshape
\label{inv:algo:A3}
After any sequence of cell insertions with patterns according to Fig.~\ref{fig:scaling:algo:A3:routing},
\begin{enumerate}[itemsep=2mm]
    \item
    \label{inv:cover:once:forever}
    every vertex covered once remains covered;
    \item 
    \label{inv:side:covered:CW}
    the empty sides of an empty cell are covered by the routing one after the other in clockwise order starting from the counter-clockwise side to the clockwise side of the insertion side; 
    \item 
    \label{inv:relative:order:unchanged}
    as the routing is extended by inserting a pattern on an edge, the relative order in the routing of the vertices outside the newly covered cell is unchanged by a insertion a new cell or fixing an anomaly (we consider that fixing an anomaly according to Fig.~\ref{fig:scaling:algo:A3:fix:anomalies} as a new cell insertion here);
    \item 
    \label{inv:side:without:anomaly:flows:CW}
    every available (occupied) side of an empty cell that is not marked as a time-anomaly, flows clockwise.
    \item 
    \label{inv:enters:CCW:leaves:CW}
    the routing enters the first time and leaves a cell for the last time from the same cell side (the latest side of the cell at the step of its insertion): it enters at its middle vertex and exits at the clockwise-most vertex;
\end{enumerate}
\end{invariant}

In the patterns listed in \figPatAAA, some vertices are marked with yellow or red dots, these are respectively path- and time-anomalies and they require special attention in Algorithm~\ref{algo:scaling:algo:ext:A3}.

After step of the algorithm, the current routing covers a connected set of cells. An \empty{empty area} is a connected component of not-fully-covered cells. Every empty area has a boundary which is a cycle made of neighboring cell sides (see Fig.~\ref{fig:scaling:algo:A3:boundary}).
\begin{figure}[t]
    \centering
    \includegraphics[width=.75\textwidth]{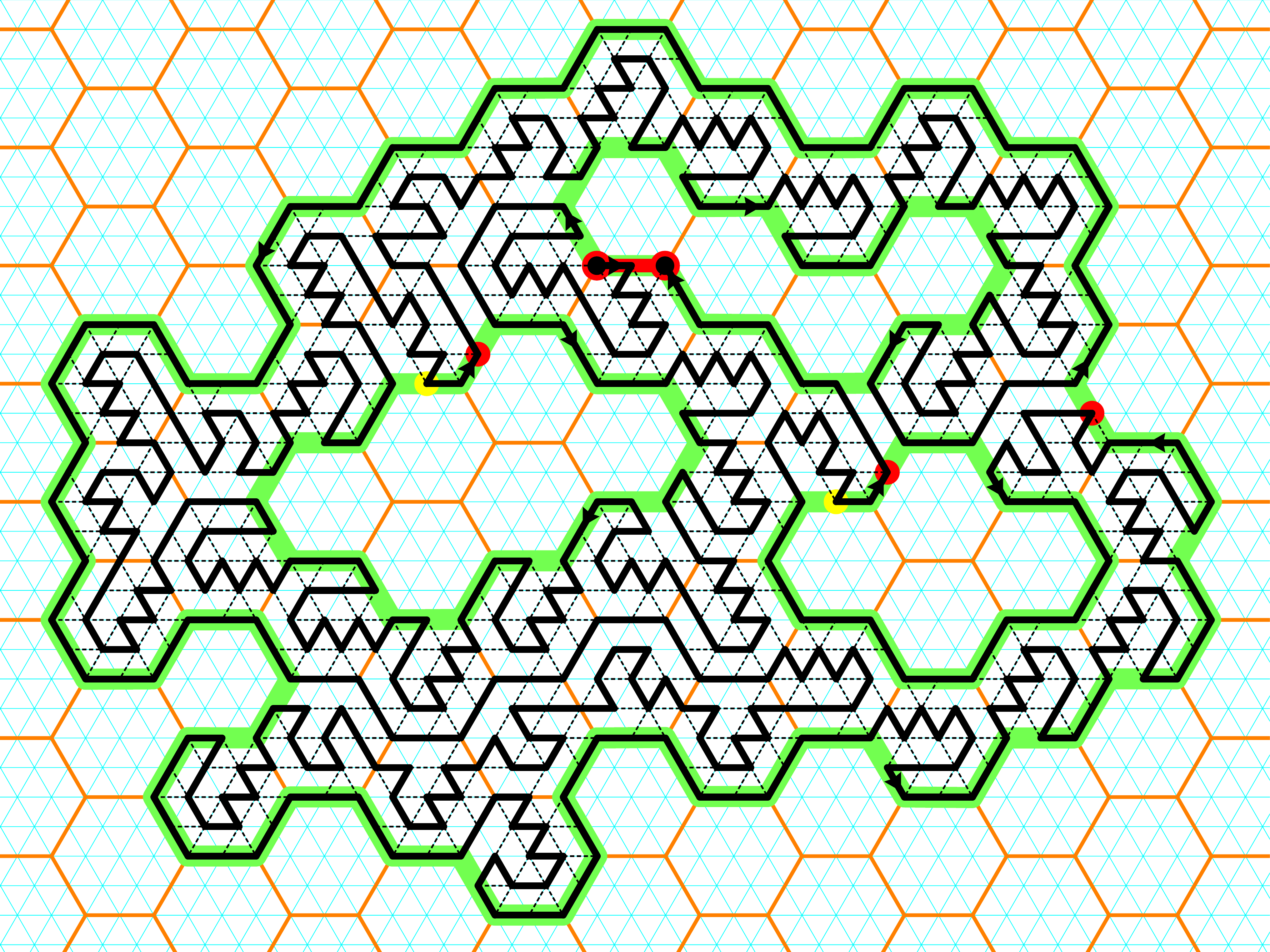}
    \caption{\captionpar{A configuration with four empty areas:} their boundaries (outlined in green) contain exactly one time-anomaly each (the red dots) (recall that we consider the originating side of the routing (in red) as one time-anomaly).} 
    \label{fig:scaling:algo:A3:boundary}
\end{figure}
The topological lemma~\ref{lem:scaling:algo:one:boundary} page~\pageref{lem:scaling:algo:one:boundary} is the key to our result.

\begin{omittedproof}{Lemma}{lem:scaling:algo:one:boundary}
The proof relies on applying Jordan's theorem to the current routing. Consider an empty area $A$. Because its boundary is a cycle, it must contain at least one time-anomaly: indeed, according to invariant~\refInvAAA{inv:side:without:anomaly:flows:CW}, time increases clockwise along the sides of an empty area without time-anomaly; as it must decrease at some point, it must contain at least one time-anomaly. Now, consider a time-anomaly $a$ and its clockwise and earlier neighbor $e$ on the boundary. $a$ was produces by one of the patterns in \figPatAAA. We illustrate the proof with pattern~\NSpattern{1101} here (see Fig.~\ref{fig:scaling:algo:A3:one:anomaly}); the proof works identically with all patterns containing a time-anomaly, as they are all topologically identical w.r.t. this result. 
\begin{figure}[t]
    \begin{subfigure}{.42\textwidth}
    \centering
    \includegraphics[scale=.8]{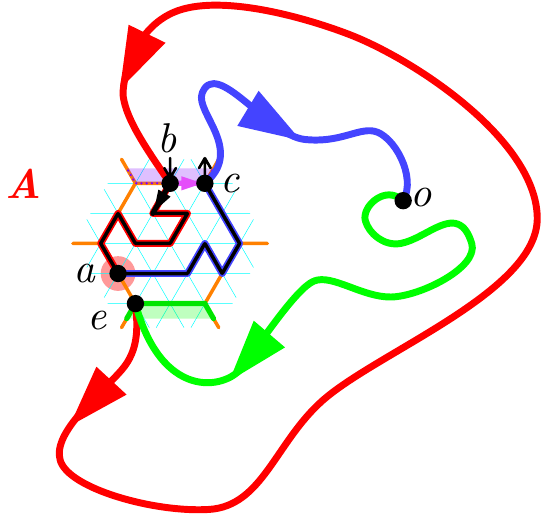}
    \caption{The routing from $e$ to $a$ goes to the right} 
    \label{fig:scaling:algo:A3:one:anomaly:jordan1}
    \end{subfigure}
    \hfill
    \begin{subfigure}{.55\textwidth}
    \centering
    \includegraphics[scale=.8]{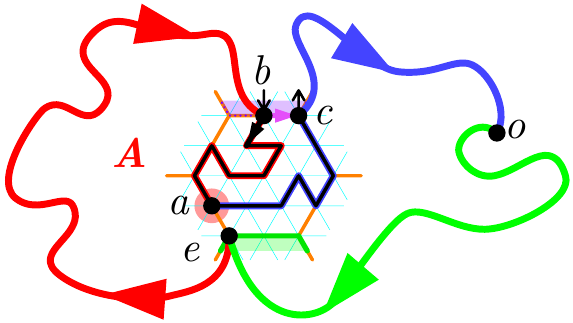}
    \caption{The routing from $e$ to $a$ goes to the left} 
    \label{fig:scaling:algo:A3:one:anomaly:jordan2}
    \end{subfigure}
    \caption{In both cases, all vertices on the boundary of the empty area $A$ must be earlier than the time-anomaly~$a$.} 
    \label{fig:scaling:algo:A3:one:anomaly}
\end{figure}
As $e$ is earlier in the routing than $a$, the routing connects $e$ to $a$ by a self-avoiding path (in red on Fig.~\ref{fig:scaling:algo:A3:one:anomaly}) that goes either (a) to the left or (b) to the right. As $a$ and $e$ are next to each other, together with the path in the pattern, they both ''seal'' this path, which thus encloses the empty area $A$ (in its outside in (a); in its inside in (b)). According to the pattern~\NSpattern{1101}, the routing must continue to the right after passing through $a$, to get back to the origin. By Jordan's theorem, the part of the routing after $a$ (in blue) is thus entirely isolated from the empty area by the red path, and the origin must lie there as well. It follows again by Jordan's theorem, that the part of the routing connecting the origin to $e$ (in green) is also isolated from the empty area by the red path. It follows that the only occupied vertices exposed at the boundary of $A$ belong to the red path. \emph{All} the vertices at the boundary of $A$ are thus \emph{earlier} than $a$. There can thus not be any other anomaly on this boundary; otherwise both anomalies would be earlier than each other.  
\qed
\end{omittedproof}

The key lemma implies that after the while loop, the empty cell is only surrounded by ``regular" edges of the boundary, which all flow clockwise by the invariant~\ref{inv:algo:A3}
 above. It follows that invariant~\ref{inv:scaling:algo:A3:CW:last} is now verified and applying the pattern corresponding to the new empty cell signature extends the routing to cover the cell while ensuring its foldability. Indeed:

\begin{omittedproof}{Corollary}{cor:scaling:algo:CW:clean}
First, observe by invariant~\refInvAAA{inv:side:without:anomaly:flows:CW}, that every edge which is not a time-anomaly flows clockwise. By the key topological lemma~\ref{lem:scaling:algo:one:boundary}, this implies that every side on the boundary of an empty cell is clean, unless the empty cell is neighboring the only time-anomaly on that boundary. Furthermore, as time increases clockwise, it also implies that the latest edge around the empty cell is not only clean but located at the clockwise end of a segment.

Let us thus first focus on time-anomalies. One can observe in the patterns in Fig.~\ref{fig:scaling:algo:A3:fix:anomalies}(b-d) that the pattern fixing a time-anomaly moves the anomaly outside the sides of the empty cell, while preserving its connectivity with an already covered and earlier cell side, ensuring that we are back to the case treated above. 

We can now assume that for the empty we want to cover, its latest side clean and is  the clockwise most of the segment. If this side is not a path-anomaly, a simple inspection of the patterns shows that any pattern can be plugged into this side safely while preserving the tight foldability of the routing. However, if the side contains a path-anomaly, then plugging the pattern, as is, might prevent the routing from folding or create an unfixable time-anomalie (see the \dirSW-path-anomaly in pattern \NSpattern{101} for instance). But, one can observe by inspecting carefully Fig.~\ref{fig:scaling:algo:A3:fix:anomalies}(b-d) that fixing a path-anomaly according to the prescribed patterns, ensures that: 
\begin{itemize}
    \item if the fixed edge is plugged into \emph{immediately afterwards}, then the routing will be foldable. For instance, observe the pattern \NSpattern{101\NSnx\dirSW} fixing the \dirSW-path-anomaly in pattern \NSpattern{101}: it cannot be folded as is, but will be foldable if an other pattern is plugged to the \dirSW opening.
    \item the fixed edge will be plugged immediately afterwards by the algorithm, because it is the latest edge of the to-be-inserted empty cell (otherwise it would not have been fixed)
\end{itemize}
It follows that fixing anomalies requires fixing at most 2 anomalies and that the resulting path is foldable and tight.
\qed
\end{omittedproof}

\begin{omittedproof}{Theorem}{thm:scaling:algo:A3}
It follows from corollary~\ref{cor:scaling:algo:CW:clean}, that every steps requires constant time computation, \emph{once we know the rank of each vertex in the current routing}. This can be maintained using your favorite balanced search tree in $O(\log(n))$ time per vertex query after $n$ cell insertions. 

Algorithm~\ref{algo:scaling:algo:ext:A3} outputs then a tight routing covering any shape in time $O(n\log n)$. This tight routing is then transformed into a delay-1 OS with seed of size 3 using the universal 114 bead types by Theorem~\ref{thm:scaling:algo:universal:bead:type} in linear time, which concludes the proof. 
\end{omittedproof}

\subsection{Examples of step-by-step construction of the routing of a shape}

%%%
%%% Flower A3
%%%

\begin{figure}
    \resizebox{\textwidth}{!}{\tabMovieThirteenFrames{.15\textwidth}{flowerA}}
    \caption{The step-by-step construction of a routing folding into a flower shape at scale \scaling A3 according to the algorithm in Section~\ref{sec:scaling:algo:A3}}
    \label{fig:scaling:algo:movie:A3}
\end{figure}

\NSomitted{
\begin{figure}
    \resizebox{\textwidth}{!}{\tabMovieThirteenFrames{.15\textwidth}{flowerA2}}
    \caption{Same as Fig.~\ref{fig:scaling:algo:movie:A3} with a different bead types representation.}
    \label{fig:scaling:algo:movie:A3:other}
\end{figure}
}%NSomitted

%%%
%%% Fixing anomalies A3
%%%

\begin{figure}
    \resizebox{\textwidth}{!}{\tabMovieFourteenFrames{.15\textwidth}{anomalyA1}}
    \caption{The step-by-step construction of a routing folding into a shape at scale \scaling A3 according to Algorithm~\ref{algo:scaling:algo:ext:A3}, involving solving anomalies $\NSpattern{101} \rightarrow \NSpattern{101\protect\NSnx\dirNW} \rightarrow \NSpattern{101\protect\NSnx\dirNW\protect\NSnx\dirS}\rightarrow \NSpattern{101\protect\NSnx\dirNW\protect\NSnx\dirS\protect\NSnx\dirSW}$ in the four last steps.}
    \label{fig:scaling:algo:movie:A3:anomaly}
\end{figure}
\begin{figure}
    \resizebox{\textwidth}{!}{\tabMovieFourteenFrames{.15\textwidth}{anomalyA2}}
    \caption{The step-by-step construction of a routing folding into a shape at scale \scaling A3 according to Algorithm~\ref{algo:scaling:algo:ext:A3}, involving solving anomalies $\NSpattern{101} \rightarrow \NSpattern{101\protect\NSnx\dirSW} \rightarrow \NSpattern{101\protect\NSnx\dirSW\protect\NSnx\dirS}$ and $\NSpattern{11101} \rightarrow \NSpattern{11101\protect\NSnx\dirNW} = \NSpattern{111101}$ (rotated clockwise) in the four last steps.}
    \label{fig:scaling:algo:movie:A3:anomalyB}
\end{figure}

\NSomitted{
\begin{figure}
    \resizebox{\textwidth}{!}{\tabMovieFourteenFrames{.15\textwidth}{anomalyA2B}}
    \caption{Same as Fig.~\ref{fig:scaling:algo:movie:A3:anomalyB} with a different bead types representation.}
    \label{fig:scaling:algo:movie:A3:anomalyB:other}
\end{figure}
}%NSomitted

%% file: smallDelayWeak.tex
\section{Details for a shape which can be assembled at delay $\delta$ but not $<\delta$} \label{sec:sdw}
We now present the full details of the proof of Theorem~\ref{thm:small-delay-weak}.  To best fit the figures to the page, in this section we discuss configurations which are situated on a rotated triangular grid relative to $\mathbb{T}$.

\begin{figure}[htp]
\centering
\includegraphics[width=4.0in]{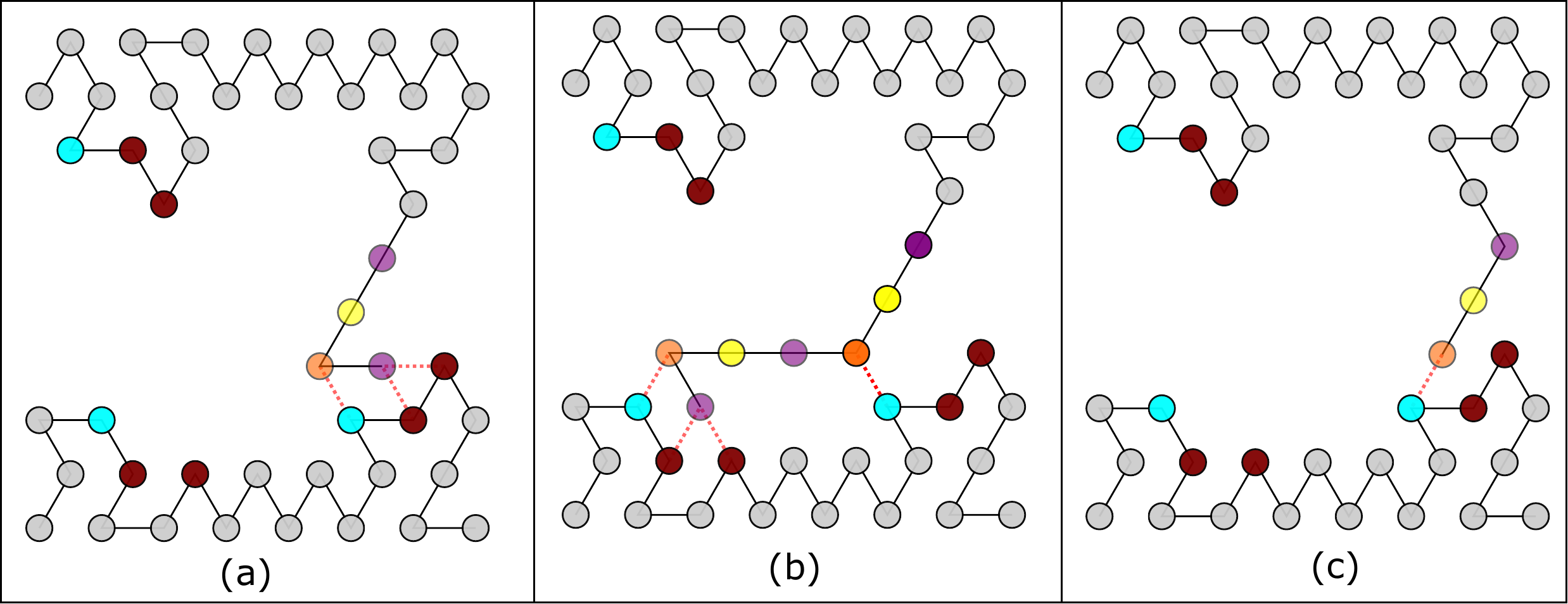}
\caption{Parts (a) and (b) show an example of how the rule set of $\Xi_{\delta}^*$ is used to assemble $R_{\delta}'$ when $\delta=4$.  Part (c) shows an incorrect configuration that results if only visible bonds are used to stabilize the purple bead.  Transparent beads represent nascent beads and opaque beads represent stabilized beads.}
\label{fig:sw-main-line}
\end{figure}

\subsection{Full description of $S_{\delta}$}
To formally describe $S_{\delta}$, we first describe a finite routing $R_{\delta}'$.  The routing is defined so that the points in the routing are a subset of $S_{\delta}$.  We then show that there exists a deterministic OS $\Xi_{\delta}^*$ whose terminal configuration $C$ has $R_{\delta}'$ as a routing.  Since there exists an OS which can trivially assemble any shape with a Hamiltonian path for any $\delta$ (by encoding it as a seed), we show a straight forward way to extend $\Xi_{\delta}^*$ to an infinite system $\Xi_{\delta}$ which assembles an infinite number of copies of the routing side-by-side.  The shape $S_{\delta}$ is then defined to be the domain of the terminal configuration of $\Xi_{\delta}$

\subsubsection{Description of $R_{\delta}'$} \label{sec:rdelta}
Algorithm~\ref{alg:build-routing} creates the routing $R_{\delta}'$ by calling a number of subroutines.  It begins with an empty routing $R$ and adds beads to the routing by calling subroutines.  Each subroutine returns a routing which is then added to $R$.  We say that the routing returned by each subroutine is a \emph{gadget}.  For example, we call the routing returned by the \texttt{LEFT-WALL} subroutine the \texttt{LEFT-WALL} gadget.  The algorithm begins by adding three beads to the routing $R$.  We call the routing of these three beads the \texttt{SEED} gadget due to the fact that it will act as the seed for the OS which assembles this routing.

Algorithm~\ref{alg:build-routing} and its subroutines make use of a few primitive subroutines which we now describe.  The \texttt{PATH} subroutine defined in Algorithm~\ref{alg:path} takes as input some length ``l'' and outputs a routing of width $2$ and ``length'' l.  An example of the output of the \texttt{PATH} subroutine is shown as the beads with a blue outline in Figure~\ref{fig:swGrowthGadgets}.  The \texttt{SMALL-BUMP} and \texttt{BIG-BUMP} subroutines are not explicitly defined due to their simple nature.  They do not take any arguments and the output of the routines is shown in Figure~\ref{fig:bumps}.  We refer to the output of either of these routines as a \texttt{BUMP} gadget.  An example of the \texttt{BUMP} gadgets are shown in Figure~\ref{fig:swGrowthGadgets} as beads with a red outline.

Algorithm~\ref{alg:build-routing} then calls the \texttt{LEFT-WALL} subroutine.  This subroutine adds $\lceil \frac{\delta(5(\delta-1)+1)}{4} \rceil$ repeating units to $R$ which consist of a \texttt{SMALL-BUMP} and a \texttt{BIG-BUMP} (the \texttt{SMALL-BUMP} and \texttt{BIG-BUMP} gadgets are shown in Figure~\ref{fig:bumps}) spaced out using \texttt{PATH} gadgets by some amount dependent on $\delta$.  The \texttt{LEFT-WALL} subroutine passes arguments to \texttt{SMALL-BUMP} and \texttt{BIG-BUMP} so that the labels of the beads in Figure~\ref{fig:bumps} have the symbol $x$ replaced by $l$ and $i \in \mathbb{N} \cap [1, \frac{6\delta}{4}]$. More specifically, these \texttt{BUMP} gadgets are spaced out so that the Euclidean distance between the point to the northeast of the $lc_i$ bead in the \texttt{BIG-BUMP} gadget and the point to the southeast of the $la_i$ bead in the \texttt{SMALL-BUMP} gadget is $\delta-1$.  The yellow beads in Figure~\ref{fig:swGrowthGadgets} show an example of the \texttt{LEFT-WALL} gadget when $\delta=4$.  Notice that the $lc_i$ and $la_i$ beads in this example correspond to the beads which are adjacent to the blue beads.  The spacing between the \texttt{BIG-BUMP} and \texttt{SMALL-BUMP} is such that $\delta$ beads can be placed in a straight line beginning from the northeast of the $lc_i$ bead and ending at the southeast of the $la_i$ bead.

The next subroutine to be called from Algorithm~\ref{alg:build-routing} is the \texttt{WIDE-TURN} subroutine.  This subroutine adds two paths to $R$ which form a wide ``V'' shape in relation to each other (as shown by the red beads in the example in Figure~\ref{fig:swGrowthGadgets}.  The length of these two paths is dependent on $\delta$.  We selected the lengths of these two paths so that after all gadgets assemble, the line of beads which stretch from the \texttt{LEFT-WALL} gadget to the \texttt{RIGHT-WALL} gadget (described below) consist of $\delta$ beads.  For an example notice in Figure~\ref{fig:swGrowthGadgets} how the length of the \texttt{WIDE-TURN} gadget (the beads shown in red) allows for a line of $4$ beads to stretch from the \texttt{LEFT-WALL} to the \texttt{RIGHT-WALL}.

After the \texttt{WIDE-TURN} gadget is added to $R$, the \texttt{SCAFFOLD} (shown in Figure~\ref{fig:swGrowthGadgets} as purple beads) and \texttt{UTURN}(shown in Figure~\ref{fig:swGrowthGadgets} as green beads) gadgets are added to $R$.  The \texttt{SCAFFOLD} gadget just consists of a long path whose length is determine by $\delta$.  The purpose of the \texttt{SCAFFOLD} and \texttt{UTURN} gadgets is to position the \texttt{RIGHT-WALL} gadget so that the \texttt{BUMP} gadgets in the \texttt{RIGHT-WALL} gadget lie in a particular position relative to the \texttt{BUMP} gadgets in the \texttt{LEFT-WALL} gadget.  The purpose for this is to allow particular beads in the \texttt{BEAD-LINE} gadget (described below) to be adjacent to exactly one bead in the \texttt{BUMP} gadgets.

The next subroutine to be called from Algorithm~\ref{alg:build-routing} is the \texttt{RIGHT-WALL} subroutine.  This subroutine adds a reflected version of the \texttt{LEFT-WALL} gadget to $R$ with a couple of caveats.   The \texttt{RIGHT-WALL} subroutine passes arguments to \texttt{SMALL-BUMP} and \texttt{BIG-BUMP} so that the labels of the beads in Figure~\ref{fig:bumps} have the symbol $x$ replaced by $r$ and $i \in \mathbb{N} \cap [1, \frac{6\delta}{4}]$.  The \texttt{RIGHT-WALL} gadget has a \texttt{PATH} gadget of length $\delta$ tacked onto the last repeating unit with a fixed size triangle adjoined to it.  An example of a \texttt{RIGHT-WALL} gadget is shown in Figure~\ref{fig:swGrowthGadgets} as grey beads.

The subroutine \texttt{BEAD-LINE} is the next subroutine to extend the routing $R$.  The routing returned by this subroutine is $\lceil \frac{\delta(5(\delta-1)+1)}{4} \rceil$ repeating units consisting of 4 lines of beads.  The first line of beads in a unit is a line of beads of length $\delta-1$ which grows to the northwest.  Attached to that is another line of beads of length $\delta-1$ which grows to the north. Next, a line of $\delta-1$ beads grow to the northeast.  Finally, another line of length $\delta-1$ beads grows to the north.  The set of blue colored beads in Figure~\ref{fig:swGrowthGadgets} depicts an example of the \texttt{BEAD-LINE} gadget when $\delta=4$.  Notice that the gadgets have been designed so that every $\delta-1$ beads in the \texttt{BEAD-LINE} gadget there is a bead which lies adjacent to exactly one \texttt{BUMP} gadget.

The last subroutine to be called by Algorithm~\ref{alg:build-routing} is the \texttt{SPACER} subroutine.  This subroutine returns a \texttt{SPACER} gadget which grows over the \texttt{UTURN} gadget using a series of \texttt{PATH} and \texttt{TURN} gadgets, and then it grows a path of length $2\delta + 4$ to the southeast.  The purpose for this \texttt{SPACER} gadget is to allow us to attach $R_{\delta}'$ routings to each other in a side-by-side manner.

\subsubsection{Algorithm~\ref{alg:build-routing} and its subroutines}

A detailed description of the shape $S_{\delta}$ is now provided.  We begin by providing some useful notation and a list of straightforward auxiliary methods.  For a finite directed path $P$, we use $|P|$ to denote the index of the last element in the sequence.  Also, $P(i)$ denotes the $i^{th}$ element in the sequence which in our case corresponds to a point in $\mathbb{Z}_{\Delta}$.  We now provide a list of simple auxiliary methods.  We define $D= (N, NE, SE, S, SW, NW)$ to be the set of \emph{directions}.  Given a direction $d$ and a point $p$, we define the \emph{$d$ neighbor of $p$} to be the point corresponding to the direction $d$ in Figure~\ref{fig:dirs}.  We denote the $d$ neighbor of $p$ by $d(p)$.  In this section, given an element of a routing $r_i = (p_i, b_i)$, we use $dom(r_i)$ to denote $p_i$, that is $dom(r_i)=p_i$.

\begin{figure}[htp]
\centering
\includegraphics[width=1.0in]{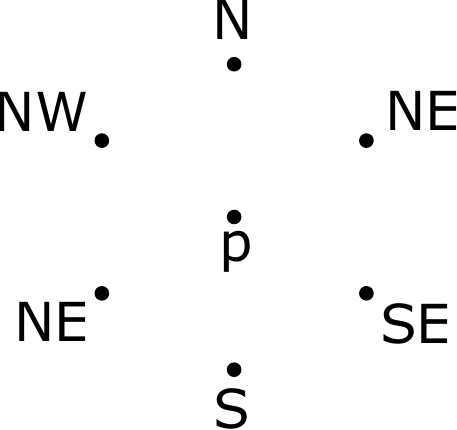}
\caption{A point $p$ and all of its $d$ neighbors for $d \in D$.}
\label{fig:dirs}
\end{figure}

Let $P$ be a directed path in $\mathbb{T}$.  We call a point $p \in \mathbb{T}$ \emph{empty with respect to $P$} provided that for all $i \in [1, n], p_i \neq p$.  The method \texttt{SHARED-NHBR($P, i, j$)} takes as input a directed path in $\mathbb{T}$, $P$, and two indices $i, j \in \mathbb{N}$.  If $1 \leq i, j \leq |P|$ and $P[i]$ and $P[j]$ have exactly one shared neighbor which is empty with respect to $P$, \texttt{SHARED-NHBR} returns this point.  Otherwise, \texttt{SHARED-NHBR} returns \texttt{null}.

We now describe some subroutines used in Algorithm~\ref{alg:build-routing} which are not explicitly defined.  The subroutine \texttt{UTURN} in Algorithm~\ref{alg:build-routing} takes as input some routing and it returns a new routing.  As shown in Figure~\ref{fig:swGrowthGadgets}, Algorithm~\ref{alg:build-routing} is designed so that the subroutine \texttt{UTURN} will always receive a routing which consists exactly of the seed (darkly shaded), the left wall (yellow), the wide turn (red) and the scaffold (purple) parts of the routing.  Note that the last position in this routing is the north most bead in the scaffold. The subroutine \texttt{UTURN} creates the routing shown in green (where the bead types are some generic labels which are unique) and returns this routing.

The \texttt{SMALL-BUMP} and \texttt{BIG-BUMP} routines called in Algorithm~\ref{alg:left-wall} and Algorithm~\ref{alg:right-wall} each take three arguments: 1) a routing $R$, 2) a symbol $x$ and 3) a number $i$.  The result of these two routines is shown in Figure~\ref{fig:bumps}.  The darkly shaded beads in both (a) and (b) of Figure~\ref{fig:bumps} represent beads that are in the routing $R$ which is passed as an argument.  The two routines extend the routing $R$ by adding the routing indicated by the lightly shaded beads.  All non-labeled beads in Figure~\ref{fig:bumps} are unique bead types with some generic labels.  The labeled beads are given bead types corresponding to their label.  For example, the routine call $\texttt{SMALL-BUMP}(R, ``r'', 1)$ would return a routing where the labeled bead types would be $ra_1$, $rb_1$ and $rb_1'$.

\begin{figure}[htp]
\centering
\includegraphics[width=3.0in]{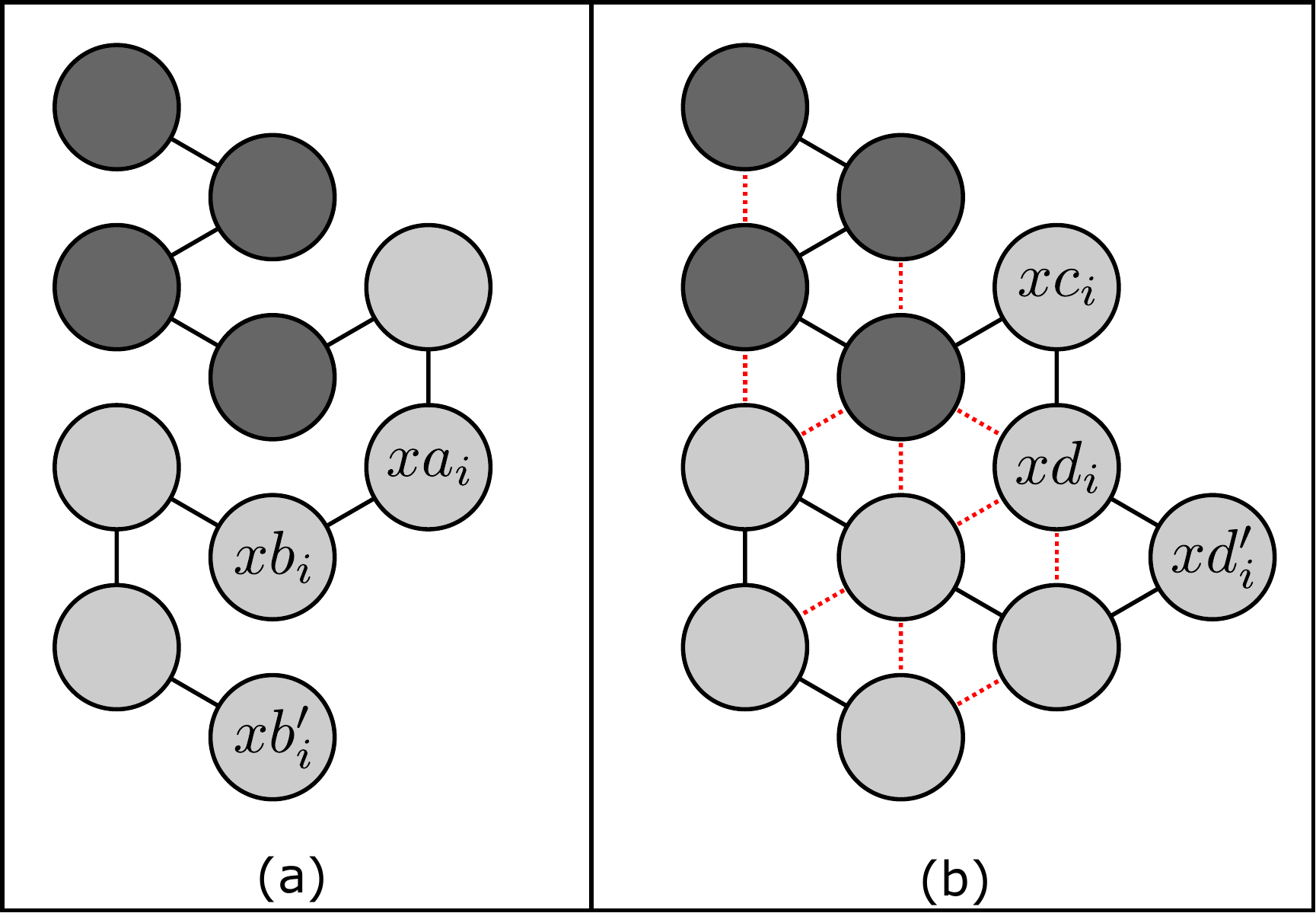}
\caption{The lightly shaded beads of (a) show the routing returned from the routine \texttt{SMALL-BUMP} and the specially labeled beads.  The lightly shaded beads of (b) show the routing returned from the routine \texttt{BIG-BUMP} and the specially labeled beads.  The red dotted lines represent attraction rules between the beads which allow them to assemble at any delay greater than two.}
\label{fig:bumps}
\end{figure}

\begin{algorithm}[H]
\caption{A procedure to build the routing $R_{\delta}'$}
\label{alg:build-routing}
\begin{algorithmic}[1]
\Procedure{\texttt{BUILD-ROUTING}}{$\delta$} \Comment{Takes $\delta \in \mathbb{N}$}
    \LineComment{Initialization}
    \State $R =\{ \}$ \Comment{An ordered list which holds the routing}
    \LineComment{Add seed to the routing}
    \State $R = R + ((0,0), b_1)$
    \State $R = R + ((\frac{-1}{2}, \frac{-\sqrt{3}}{2}), b_{2})$
    \State $R = R + ((0, -1), b_{3})$

    \LineComment{Add the yellow bead portion of the routing shown in Figure~\ref{fig:swGrowthGadgets}}
    \State $R = R + \texttt{LEFT-WALL}(R, \delta)$

    \LineComment{Add the red bead portion of the routing shown in Figure~\ref{fig:swGrowthGadgets}}
    \State $R = R + \texttt{WIDE-TURN}(R, \delta)$

    \LineComment{Add the purple portion of the routing shown in Figure~\ref{fig:swGrowthGadgets}.  We call this part of the routing the scaffold.}
    \State $R = R + \texttt{PATH}(R, NI(2(\delta-3)+4+(\delta-1))+\delta)$

    \LineComment{Add the green portion of the routing shown in Figure~\ref{fig:swGrowthGadgets}}
    \State $R = R + \texttt{UTURN}(R)$ \Comment{This routine is not explicitly defined. See text for details.}

    \LineComment{Add the grey bead portion of the routing shown in Figure~\ref{fig:swGrowthGadgets}}
    \State $R = R + \texttt{RIGHT-WALL}(R, \delta)$

    \LineComment{Add the blue portion of the routing shown in Figure~\ref{fig:swGrowthGadgets}}
    \State $R = R + \texttt{BEAD-LINE}(R, \delta)$

    \LineComment{Add the orange bead portion of the routing shown in Figure~\ref{fig:swGrowthGadgets}}
    \State $R = R + \texttt{SPACER}(R, \delta)$

\State \textbf{return} $R$
\EndProcedure
\end{algorithmic}
\end{algorithm}

\begin{algorithm}[H]
\caption{A procedure to build the \texttt{LEFT-WALL} gadget of $R_{\delta}'$ (shown as the yellow filled beads in Figure~\ref{fig:swGrowthGadgets}). }
\label{alg:left-wall}
\begin{algorithmic}[1]
\Procedure{\texttt{LEFT-WALL}}{$R$, $\delta$} \Comment{Takes a routing $R$ and $\delta \in \mathbb{N}$}
    \State $S = \{ \}$ \Comment{Routing to be returned from this procedure.}
    \State $sc = \lceil \frac{\delta(5(\delta-1)+1)}{4} \rceil$
    \For {$sc > 0$}
        \State $S = \texttt{PATH}(R, 2(\delta-3))$
        \State $S = S + \texttt{SMALL-BUMP}(S, \text{``l''}, sc)$ \Comment{This routine is not explicitly defined. See text for details.}
        \State $S = S + \texttt{PATH}(S, \delta-1)$
        \State $S = S + \texttt{BIG-BUMP}(S, \text{``l''}, sc)$ \Comment{This routine is not explicitly defined. See text for details.}
    \EndFor
\State \textbf{return} $S$
\EndProcedure
\end{algorithmic}
\end{algorithm}

\begin{algorithm}[H]
\caption{A procedure to build the \texttt{WIDE-TURN} gadget of $R_{\delta}'$ (shown as the red filled beads in Figure~\ref{fig:swGrowthGadgets}). }
\label{alg:wide-turn}
\begin{algorithmic}[1]
\Procedure{\texttt{WIDE-TURN}}{$R$, $\delta$} \Comment{Takes a routing $R$ and $\delta \in \mathbb{N}$}
    \State $S = \{ \}$ \Comment{Routing to be returned from this procedure.}
    \State $S = \texttt{PATH}(R, 2)$
    \State $S = S + \texttt{TURN}(S)$
    \State $S = S + \texttt{PATH}(S, \delta+1)$
    \State $S = S + \texttt{TURN}(S)$
    \State $S = S + \texttt{PATH}(S, \delta+1)$
    \State $S = S + \texttt{TURN}(S)$
\State \textbf{return} $S$
\EndProcedure
\end{algorithmic}
\end{algorithm}

\begin{algorithm}[H]
\caption{A procedure to build the \texttt{RIGHT-WALL} of $R_{\delta}'$ (shown as the grey filled beads in Figure~\ref{fig:swGrowthGadgets}). }
\label{alg:right-wall}
\begin{algorithmic}[1]
\Procedure{\texttt{RIGHT-WALL}}{$R$, $\delta$} \Comment{Takes a routing $R$ and $\delta \in \mathbb{N}$}
    \State $S = \{ \}$ \Comment{Routing to be returned from this procedure.}
    \State $sc = \lceil \frac{\delta(5(\delta-1)+1)}{4} \rceil$
    \For {$sc > 0$}
        \State $S = \texttt{PATH}(R, 2(\delta-3))$
        \State $S = S + \texttt{SMALL-BUMP}(S, \text{``r''}, sc)$ \Comment{This routine is not explicitly defined. See text for details.}
        \State $S = S + \texttt{PATH}(S, \delta-1)$
        \State $S = S + \texttt{BIG-BUMP}(S, \text{``l''}, sc)$ \Comment{This routine is not explicitly defined. See text for details.}
    \EndFor
    \State $S = S + \texttt{PATH}(S, \delta)$
    \State $S = S + \texttt{TURN}(S)$
    \State $S = S + \texttt{PATH}(S, 1)$
\State \textbf{return} $S$
\EndProcedure
\end{algorithmic}
\end{algorithm}

\begin{algorithm}[H]
\caption{A procedure to build the \texttt{BEAD-LINE} gadget of $R_{\delta}'$ (shown as the blue filled beads in Figure~\ref{fig:swGrowthGadgets}). }
\label{alg:bead-line}
\begin{algorithmic}[1]
\Procedure{\texttt{BEAD-LINE}}{$R$, $\delta$} \Comment{Takes a routing $R$ and $\delta \in \mathbb{N}$}
    \State $S = R[|R|]$ \Comment{$S$ gets the last bead in the routing $R$}
    \State $NI = \lceil \frac{\delta(5(\delta-1)+1)}{4} \rceil$
    \State $S = S + (NW(dom(S[|S|])), ca_{i})$
    \For{$0 < i \leq NI$}
        \For{$0 \leq j < \delta -3$}
            \State $S = S + (NW(dom(S[|S|])), l_{|S|+1})$ \Comment{Recall this denotes the point to the northwest of the last position in R}
        \EndFor

        \State $S = S + (NW(dom(S[|S|])), cb_{i})$
        \State $S = S + (N(dom(S[|S|])), cc_{i})$

        \For{$0 \leq j < \delta -3$}
            \State $S = S + (N(dom(S[|S|])), l_{|S|+1})$
        \EndFor

        \State $S = S + (NW(dom(S[|S|])), cd_{i})$
        \State $S = S + (N(S[|S|]), ce_{i})$

        \For{$0 \leq j < \delta -3$}
            \State $S = S + (NE(dom(S[|S|])), l_{|S|+1})$
        \EndFor

        \State $S = S + (NW(dom(S[|S|]), cf_{i})$
        \State $S = S + (N(dom(S[|S|])), cg_{i})$

        \For{$0 \leq j < \delta -2$}
            \State $S = S + (N(dom(S[|S|])), l_{|S|+1})$
        \EndFor
        \State $S = S + (N(dom(S[|S|])), ch_{i})$
    \EndFor
\State \textbf{return} $S$
\EndProcedure
\end{algorithmic}
\end{algorithm}

\begin{algorithm}[H]
\caption{A procedure to build the \texttt{SPACER} gadget of $R_{\delta}'$ (shown as the blue filled beads in Figure~\ref{fig:swGrowthGadgets}) to the routing. }
\label{alg:spacer}
\begin{algorithmic}[1]
\Procedure{\texttt{SPACER}}{$R$, $\delta$} \Comment{Takes a routing $R$ and $\delta \in \mathbb{N}$}
    \State $S = \{\}$ \Comment{Routing to be returned from this procedure.}
    \State $S = (N(dom(R[|R|])), s_1)$
    \State $S = S + (N(dom(R[|R|])), s_2)$
    \State $S = S + (N(dom(S[|S|])), s_3)$
    \State $S = S + (NE(dom(S[|S|])), s_4)$
    \State $S = S + \texttt{PATH}(S, 2(\delta-3) + 2)$
    \State $S = S + \texttt{TURN}(S)$
    \State $S = S + \texttt{PATH}(S, 2)$
    \State $S = S + \texttt{TURN}(S)$
    \State $S = S + \texttt{PATH}(S , 2\delta + 4)$
    \State $S = S + \texttt{TURN}(S)$
\State \textbf{return} $S$
\EndProcedure
\end{algorithmic}
\end{algorithm}

\begin{algorithm}[H]
\caption{A procedure to build the \texttt{PATH} gadgets of $R_{\delta}'$ (shown as the beads with a blue outline in Figure~\ref{fig:swGrowthGadgets}). }
\label{alg:path}
\begin{algorithmic}[1]
\Procedure{\texttt{PATH}}{$R$, l} \Comment{Takes a routing $R$ and a length $l \in \mathbb{N}$}
    \State $S = \{\}$ \Comment{Routing to be returned from this procedure.}
    \State $S = S + R[|R|-1]$ \Comment{Add the second to the last bead in $R$ to $S$}
    \State $S = S + R[|R|]$   \Comment{Add the last bead in $R$ to $S$}

    \For{$0 \leq i < l$}
        \State $S = S + (\texttt{SHARED-NHBR}(S, dom(S[|S|-1]), dom(S[|S|])), b_{|R|+|S|+1})$ \Comment{Add a bead with a unique generic type in the location next to the previous two beads.}
        \State $S = S + (\texttt{SHARED-NHBR}(S, dom(S[|S|-1]), dom(S[|S|])), b_{|R|+|S|+1})$ \Comment{Add a bead with a unique generic type in the location next to the previous two beads.}
    \EndFor

\State \textbf{return} $S$
\EndProcedure
\end{algorithmic}
\end{algorithm}

\begin{algorithm}[H]
\caption{A procedure to build the \texttt{TURN} gadgets of $R_{\delta}'$ (shown as the beads with a green outline in Figure~\ref{fig:swGrowthGadgets}). }
\label{alg:turn}
\begin{algorithmic}[1]
\Procedure{\texttt{TURN}}{$R$} \Comment{Takes a routing $R$}
\State \textbf{return} $(\texttt{SHARED-NHBR}(R, dom(R[|S|-2]), dom(R[|R|])), b_{|R|+1})$
\EndProcedure
\end{algorithmic}
\end{algorithm}

\subsubsection{A system at delay $\delta$ which assembles $R_{\delta}'$}\label{sec:finite-system}
\begin{figure}[htp]
\centering
\includegraphics[width=3.0in]{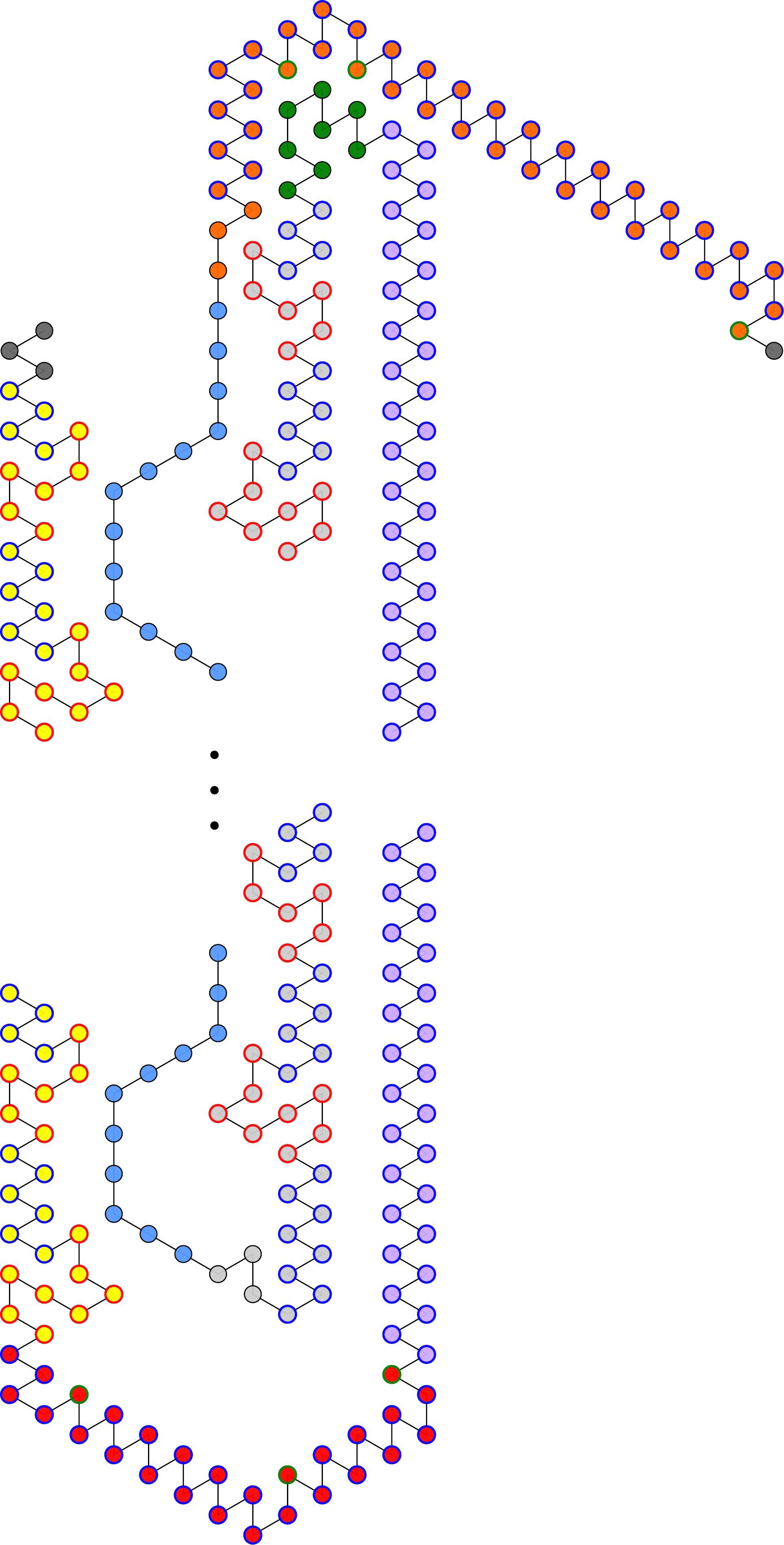}
\caption{An example of the routing $R_{\delta}'$ when $\delta=4$.  The beads are colored according to the routines from which they were returned.}
\label{fig:swGrowthGadgets}
\end{figure}

We now show that there exists an OS at delay $\delta$ which assembles the routing $R_{\delta}'$ described in Section~\ref{sec:rdelta}.  Formally, we say that a system $\Xi$ assembles a routing $R$ if for every $C \in \termasm{\Xi}$ the routing of $C$ is $R$.  Likewise, we say that a system $\Xi$ assembles a directed path $P'$ provided that for every $C = (P, w, H) \in \termasm{\Xi}$, $P' = P$.

\begin{lemma} \label{lem:shapeConditions}
Let $P = p_1 p_2 ... p_n$ be a  finite directed path in $\mathbb{T}$ with the property that there exists $j \in [4, n]$ such that for all $k \geq j$, $p_k$ has two neighbors $p_r, p_s \in \{p_1, ..., p_{k-1} \}$ such that there is exactly one empty point with respect to the directed path $p_1 p_2 ...  p_{k-1}$ which is adjacent to both $p_r$ and $p_s$.  Then, for every $\delta \in \mathbb{N}$, there exists an oritatami system with a seed of size $j$ and delay $\delta$ which assembles $P$.
\end{lemma}

The intuition behind the proof is as follows.  We construct an OS $\Xi = (\Sigma, w, \calH, \delta, \alpha)$ with seed $\sigma$ and $\alpha=5$ based on $P$ by first creating a hard-coded sequence of beads (i.e. every bead in the sequence is unique).  We then construct $\calH$ so that if $p_r$ or $p_s$ are equal to $p_{i-1}$, WLOG let's assume $p_r=p_{i+1}$, we add $\{p_k, p_s\}$ to $\calH$.  If that's not the case we add the rules $\{p_k, p_s\}$ and $\{p_k, p_r\}$ to $\calH$.  We argue that this assembles by inductively showing that each configuration $C_{i+1}$ in the assembly sequence stabilizes bead type $w[i+1]$ in the correct position.  To see this, note that there exists a favorable elongation of $C_i$ which stabilizes bead type $w[i+1]$ in the correct position and this elongation makes every bond possible between the beads.  Any elongation which places the bead type $w[i+1]$ in the incorrect position must ``break a bond'' to do so, and, consequently must not be a favorable configuration.

\begin{proof}
Let $P = p_1 p_2 ... p_n$ be a  finite directed path in $\mathbb{T}$ with the property that there exists $j \in [4, n]$ such that for all $k \geq j$, $p_k$ has two neighbors $p_r, p_s \in \{p_1, ..., p_{k-1} \}$ such that there is exactly one empty point with respect to the directed path $p_1 p_2 ...  p_{k-1}$ which is adjacent to both $p_r$ and $p_s$.  Also, let $\delta \geq 2$.

We now describe a system $\Xi$ which we claim can assemble $P$.  Let $\Xi = (\Sigma, w, \calH, \delta, \beta, \alpha, \sigma)$ where
\begin{itemize}
  \item $\Sigma = \{b_i | i \in [1, |P|] \}$,
  \item $w = (b_i)_{i=j}^{|P|}$,
  \item $\beta = 1$,
  \item $\alpha = 5$,
  \item $\sigma$ is the configuration $((p_i)_{i=1}^{j-1}, (b_i)_{i=1}^{j-1}, \emptyset)$.
\end{itemize}
To generate $\calH$ we iterate over $i \in [j, |P|]$ and consider two cases for each $i$.  The first case we consider is that there exists $p_{i-1}$ and $p_l$ for some $l \in [1, i-2]$ such that they are both adjacent to $p_i$ and there is exactly one empty point with respect to the directed path $p_1 p_2 ...  p_{i-1}$ which is adjacent to both $p_{i-1}$ and $p_l$. If more than one point in $P$ satisfies this condition for $p_l$, we chose the one with the lowest index as a convention.  In this case, we add the rule $(b_{i-1}, b_l)$ to $\calH$\footnote{Though we do not explicitly state it, when we add rule $(a,b)$ to the set $\calH$, we also add the rule $(b,a)$ to ensure that $\calH$ defines a symmetric relation.}.  In the case that no such $p_l$ exists, then it must be the case that exists $p_r, p_s \in \{p_1, ..., p_{k-1} \}$ such that there is exactly one empty point with respect to the directed path $p_1 p_2 ...  p_{k-1}$ which is adjacent to both $p_r$ and $p_s$.  If multiple such $p_r$ and $p_s$ exist we chose the pair of indices with the lowest lexicographical ordering as a convention.  In this case, we add the rules $(b_{i}, b_s)$ and $(b_i, b_r)$ to $\calH$.  See Figure~\ref{fig:con-shape-cases} for an example of the two cases we consider.

\begin{figure}[htp]
\centering
\includegraphics[width=4.0in]{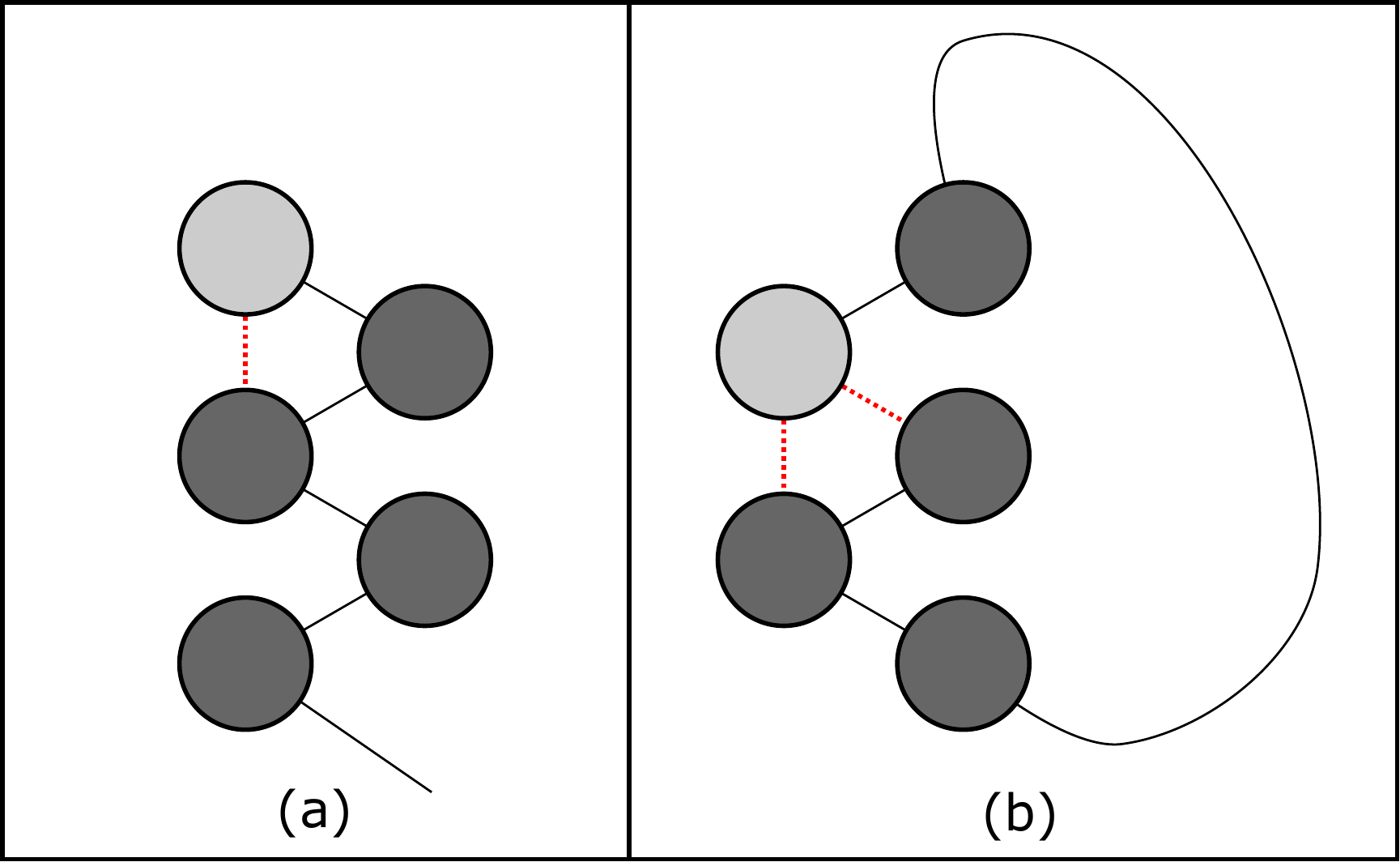}
\caption{The two cases we consider in the proof of Lemma~\ref{lem:shapeConditions}.  Part (a) corresponds to the first case in the proof and part (b) corresponds to the second case of the proof.}
\label{fig:con-shape-cases}
\end{figure}

Let $\alpha \in [1,5]$, let $H$ be a rule set, let $\Sigma$ be a set of bead types and let $C = (P, w, H)$ be an $\calH$-valid configuration where $w \in \Sigma^*$.  Also, let $t \in \Sigma^*$.  Recall that $\calP_{\calH, \alpha}^{\leq t}(C)$ is the set of $\calH$-valid elongations by prefixes of $t$.  We call a configuration $C' = (P', w', H') \in \calP_{\calH, \alpha}^{\leq t}(C)$ a \emph{saturated $\calH$-valid $\alpha$-$\delta$-elongation of $C$ by $t$} provided that for all other configurations $C^* = (P^*, w^*, H^*) \in \calP_{\calH, \alpha}^{\leq t}(C)$, $H^* \subseteq H'$.  Intuitively, a configuration $C$ is saturated provided that even if we made a ``configuration'' $C'$ by allowing bonds to form between nascent beads which are not adjacent (that is we remove geometry), $C'$ would have the same bonds as $C$.

\begin{observation} \label{obs:sat}
Suppose that there exists a saturated $\calH$-valid $\alpha$-$\delta$-elongation of $C$ by $t$.  Then every favorable $\alpha$-$\delta$-elongation of $C$ is saturated.
\end{observation}

To prove $\Xi$ assembles $P$ we inductively show that at each step in the folding process, the configuration which stabilizes the next bead in the correct position is the one and only favorable configuration.  Let $C \in \termasm{\Xi}$, and let $\vec{C}=(C_i)_{i=0}^{l}$ be a foldable sequence such that $res(\vec{C})=C$.  For the base case, note that $C_0=(P_0, w_0, \emptyset)$, which is the seed, is such that $P_0$ is a prefix of $P$.  For the inductive step, assume $C_i=(P_i, w_i, H_i)$ is such that $P_i$ is a prefix of $P$.  We now show that the configuration $C_{i+1} = (P_{i+1}, w_{i+1}, H_{i+1})$ stabilizes bead type $w[i+1]$ at position $p[i+1]$.  Consequently, this means $P_{i+1}$ is a prefix of $P$.  To see this, first note that by construction of $\Xi$, there exists $C^* = (P^*, w^*, H^*) \in \calF_{\calH, \alpha}^{\leq w[i ..i+\delta-1]}$ such that $P^*$ is a prefix of $P$.  This is due to the fact that the set $H_{sat} = \{\{l, h\} \mid \{w[l], w[h]\} \in \calH \text{ and } l \in [i, i+\delta-1] \text{ or } h \in [i, i+\delta-1]\}$ is a subset of $H^*$.  In other words, the beads in the ``nascent portion'' of $C^*$ make every bond that they possibly can.  Thus, $C^*$ is a favorable configuration.  Also, note that the fact $H_{sat} \subseteq H^*$ implies that $C^*$ is saturated.

We now show that any configuration $\bar{C}$ which does not stabilize bead type $w[i+1]$ at location $p_{i+1}$ is not saturated.  It then follows from Observation~\ref{obs:sat} that $\bar{C}$ cannot be a favorable configuration (since there exists a saturated elongation).  Thus, the only favorable configurations are those which stabilize $w[i+1]$ at location $p_{i+1}$.

Let $\bar{C}=(\bar{P}, \bar{w}, \bar{H}) \in \calP_{\calH, \alpha}^{\leq w[i .. i+ \delta-1]}(C)$ be a configuration such that $\bar{P}(i+1 - k) \neq p_{i+1}$ (the $-k$ expression appears due to the offset caused by the seed).  If $p_{i+1}$ falls into the first case listed above, then there was a single rule $\{b_{i}, b_l\}$ added to $\calH$ for some $l \in [1,i-2]$.  Note that in any configuration it is necessary that bead $b_i$ is adjacent to bead $b_{i+1}$ due to the fact they are next to each other on the transcript. Observe that two adjacent points in $\calT$ have exactly two common neighbors.  It now follows from the assumption that there is exactly one empty point with respect to the directed path $p_1 p_2 ...  p_{i-1}$ which is adjacent to both $p_{i-1}$ and $p_l$ that if the bead $b_{i+1}$ is stabilized in the incorrect position, it must be stabilized in a way such that it is not adjacent to the bead type $b_l$.  Consequently, $\{w[i+1], w[l] \} \notin \bar{H}$.  Hence, the configuration $\bar{C}$ is not saturated.  A similar argument shows that the configuration $\bar{C}$ is not saturated in the event $p_{i+1}$ falls into the second case mentioned above.
\end{proof}

\begin{lemma}
There exists a deterministic oritatami system with delay $\delta$ which assembles $R_{\delta}'$.
\end{lemma}

\begin{proof}
Let $R_{\delta}' = (P_{\delta}', w_{\delta}')$ be the routing constructed in Algorithm~\ref{alg:build-routing}.  We construct an OS $\Xi_{\delta}^*$ and argue that it has a single terminal configuration which has the routing $R_{\delta}'$.  Let $\Xi_{\delta}^* = (\Sigma_{\delta}, w_{\delta}^*, \calH_{\delta}^*, \delta, \alpha, \sigma)$ where
\begin{itemize}
  \item $\Sigma_{\delta} = \{w_{\delta}'[i] \mid i \in [1, |w_{\delta}'|]\}$,
  \item $w_{\delta}^* = w_{\delta}'[4 .. |w_{\delta}'|]$ (we skip the first $3$ bead types since they are the seed),
  \item $\alpha = 5$, and
  \item $\sigma$ is the configuration created from the routing defined in lines 5-7 of Algorithm~\ref{alg:build-routing} along with the empty set.
\end{itemize}
To generate $\calH_{\delta}^*$, for every portion of $R_{\delta}$ which is not created by the \texttt{BEAD-LINE} or \texttt{BIG-BUMP} routine, we add the rules generated by the implicit algorithm in the proof of Lemma~\ref{lem:shapeConditions}.  In particular, for each of these portions of the routing, treat the preceding portion of the routing as the seed and generate the rule set to build the new portion of the routing using the algorithm which is implicit in the proof of Lemma~\ref{lem:shapeConditions}.  We note that by the construction of Algorithm~\ref{alg:build-routing}, these portions of $R_{\delta}'$ meet the criteria listed in the lemma statement.  Indeed, all these portions of the routing are made by placing beads in a position relative to beads currently in the routing using the subroutine \texttt{SHARED-NHBR}.

We now discuss the rules which must be added to assemble the \texttt{BEAD-LINE} and \texttt{BIG-BUMP} gadgets.  The \texttt{BIG-BUMP} gadget can be assembled by adding interaction rules so that the bonds shown in part (b) Figure~\ref{fig:bumps} form.   Recall that the routing created in the $i^{th}$ iteration of routines \texttt{LEFT-WALL}, \texttt{RIGHT-WALL}, and \texttt{BEAD-LINE} are translations of the routings shown in Figure~\ref{fig:labeled-sub} (which shows an example when $\delta=4$).  In order to allow the \texttt{BEAD-LINE} gadget to assemble in $\Xi_{\delta}^*$ we add the interaction rules $(cc_i, ld_i)$, $(cc_i, ld_i')$, $(cb_i, lc_i)$, $(ce_i, lb_i)$, $(ce_i, lb_i')$, $(cd_i, la_i)$, $(cg_i, rd_i)$, $(cg_i, rd_i')$, $(cf_i, rc_i)$, $(ca_{i+1}, rb_i)$, $(ca_{i+1}, rb_i')$, $(ch_i, ra_i)$ to $\calH_{\delta}^*$.

\begin{figure}[htp]
\centering
\includegraphics[width=4.0in]{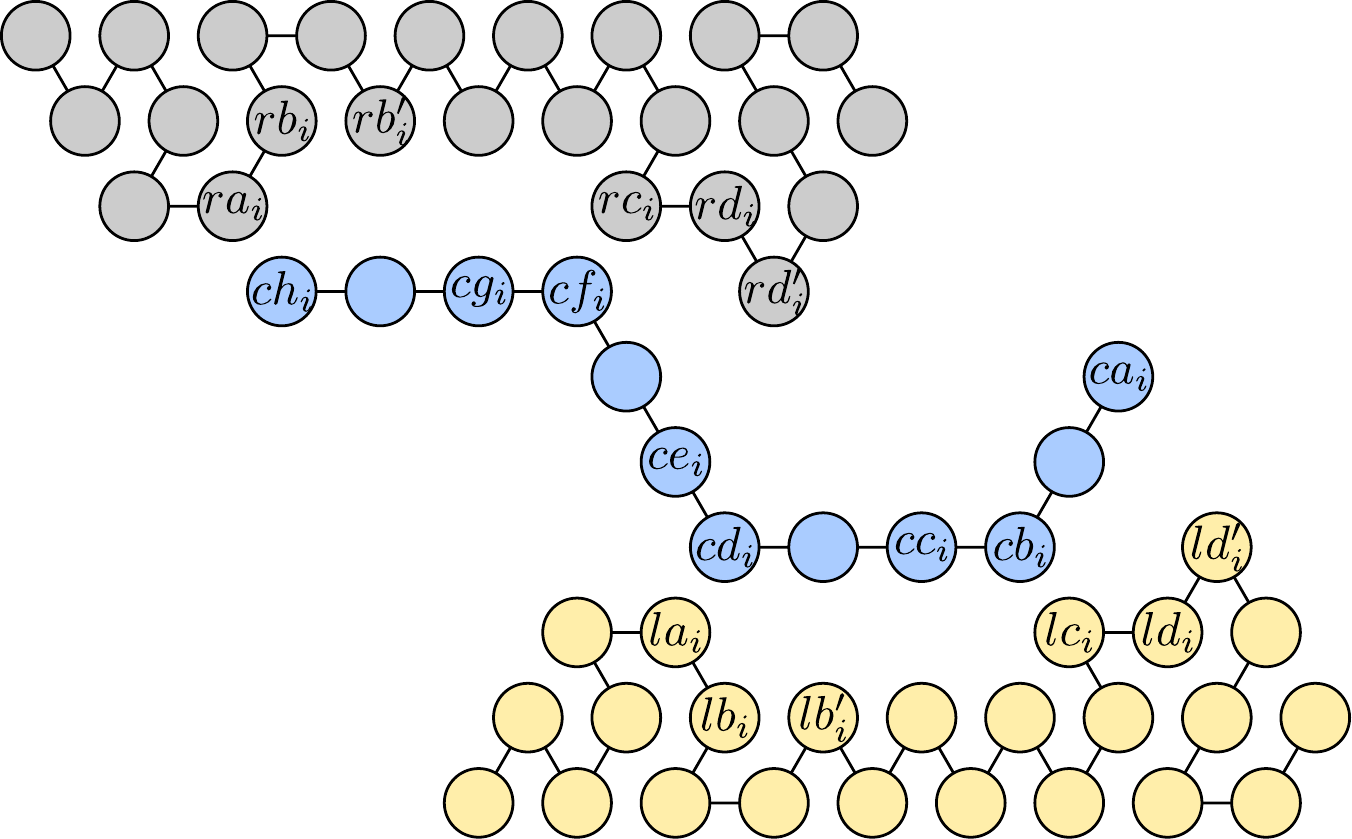}
\caption{This image is a $90^\circ$ rotation from the actual configuration. An example of the routings created in the $i^{th}$ iteration of routines \texttt{LEFT-WALL}, \texttt{RIGHT-WALL}, and \texttt{BEAD-LINE} when $\delta=4$.  Only the important bead types are labeled.}
\label{fig:labeled-sub}
\end{figure}

To see that $\Xi_{\delta}^*$ does indeed assemble the routing $R_{\delta}'$, first note that it follows from the proof of Lemma~\ref{lem:shapeConditions} that the only terminal configuration $\Xi_{\delta}$ folds is $R_{\delta}'$ provided that the \texttt{BIG-BUMP} and \texttt{BEAD-LINE} gadgets are assembled correctly.  It's easy to check these interactions allow the \texttt{BIG-BUMP} gadget to form at any delay.  To see that the \texttt{BEAD-LINE} gadget can assemble correctly, note that the added rules allow $\Xi_{\delta}^*$ to assemble the \texttt{BEAD-LINE} gadget as shown in Figure~\ref{fig:sw-main-line}.
\end{proof}

\subsubsection{An infinite version of $\Xi_{\delta}^*$}
\begin{figure}[htp]
\centering
\includegraphics[width=4.0in]{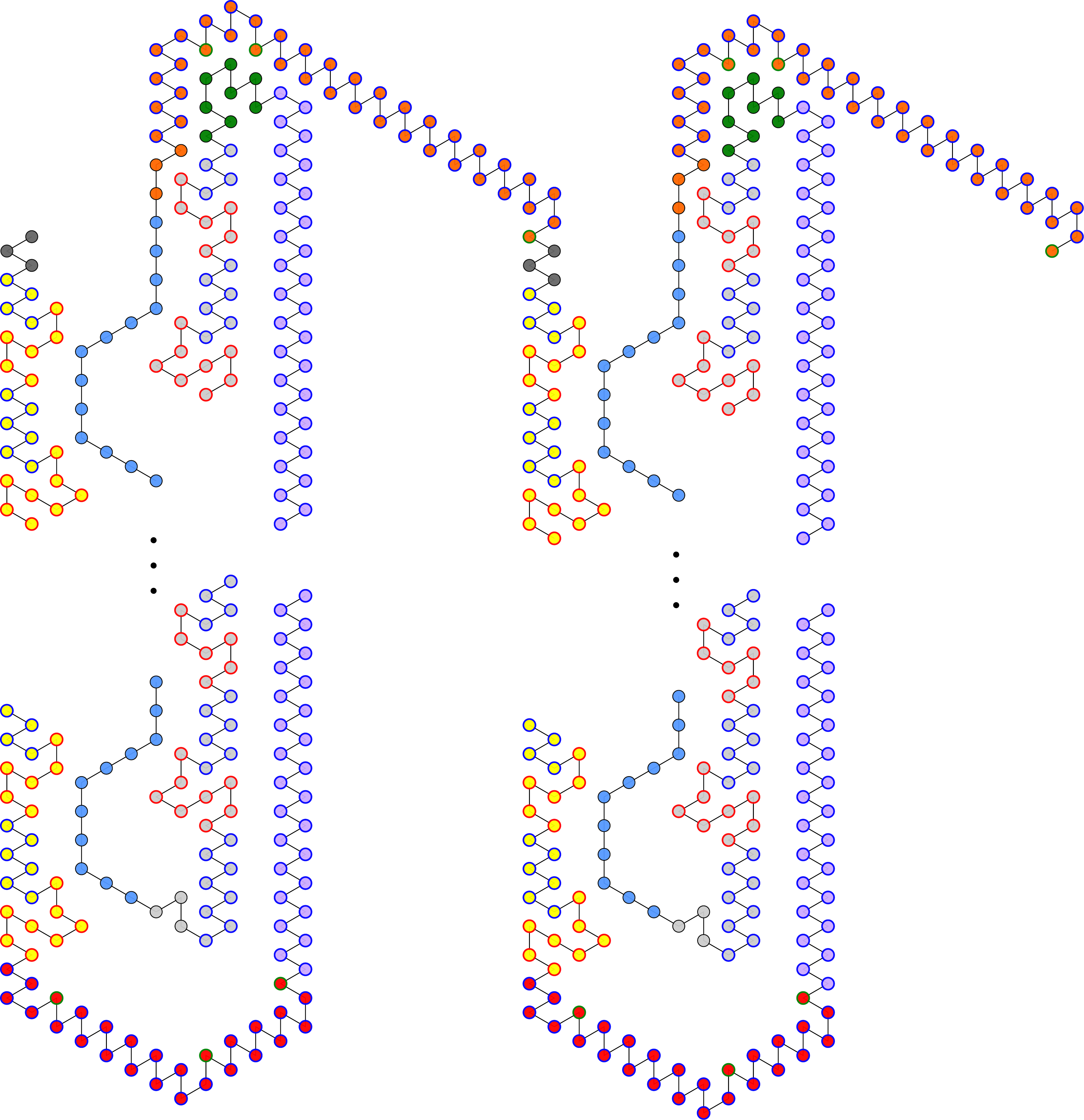}
\caption{A portion of the routing of the terminal configuration of $\Xi_{\delta}$ when $\delta=4$.}
\label{fig:sw-stacked}
\end{figure}

We define an infinite version of $\Xi_{\delta}^*$, which we call $\Xi_{\delta}$, so that $\Xi_{\delta}$ assembles an infinite number of copies of $R_{\delta}'$ stacked on top of each other as shown in Figure~\ref{fig:sw-stacked}.  Let $\Xi_{\delta}^* = (\Sigma_{\delta}, w_{\delta}^*, \calH_{\delta}^*, \delta, \alpha, \sigma)$ be the OS defined in Section~\ref{sec:finite-system}.  Also, let $t_{\delta} = w_{\delta}^* \cdot s_1 \cdot s_2 \cdot s_3$, that is $t_{\delta}$ is the bead sequence $w_{\delta}^*$ with the three bead types in the seed concatenated onto it. Let $\Xi_{\delta}$ be the oritatami system defined by $\Xi_{\delta} = (\Sigma_{\delta}, w_{\delta}, \calH_{\delta}, \delta, \alpha, \sigma)$ where $w_{\delta}$ is the infinite bead sequence defined by $w_{\delta}(i) = t( i \mod |t| + 1)$ and $\calH_{\delta}$ is the rule set $\calH_{\delta}^*$ with rules added so that the bead types $s_1, s_2$, and $s_3$ (the beads in the seed $\sigma$) assemble in a position relative to the \texttt{SPACER} gadget as shown in Figure~\ref{fig:sw-stacked}.
To see that $\Xi_{\delta}$ is deterministic recall that the \texttt{SPACER} gadget grows an ``arm'' where the last bead in the arm is greater than $\delta$ away from other gadgets in the routing.  This means that beads can not ``accidently'' interact with beads in other copies of the routing.  This combined with the fact that $\Xi_{\delta}^*$ is deterministic implies $\Xi_{\delta}$ is deterministic.  We denote the terminal assembly of $\Xi_{\delta}$ by $C_{\delta}$.  For $\delta > 2$, we define the shape $S_{\delta}$ by $S_{\delta} = dom(C_{\delta}$).

\input{SDW-impossible}

%% file: SDW-impossible.tex
\subsection{$S_{\delta}$ cannot be assembled by any system with delay $< \delta$}
Let $C_{\delta}$ be the terminal configuration of $\TMO_{\delta}$ and let $C^*$ be the terminal configuration of $\TMO_{\delta}^*$ (Recall $\TMO_{\delta}^*$ was the system which assembled the finite routing $R_{\delta}'$ constructed in Algorithm~\ref{alg:build-routing}.  We call a set of points $D_{iter} \subseteq dom(C_{\delta})$ an \emph{iteration} if there exists $\vec{v} \in \mathbb{R}^2$ such that $D_{iter} = \{ \vec{p} \mid \vec{p} = \vec{v} + \vec{x} \text{ for some } x \in \dom(C^*) \}$.  In other words, an iteration is just a translation of the set of points in $R_{\delta}'$.  We call the points added to the routing in the \texttt{BEAD-LINE} routine, (Algorithm~\ref{alg:bead-line}), the \emph{bead-line points}. Let $B$ be the set of bead-line points.  Let $\vec{v}$ be such that $D_{iter} = \{ \vec{p} \mid \vec{p} = \vec{v} + \vec{x} \text{ for some } x \in \dom(C^*) \}$.  The \emph{bead-line points of iteration $D_{iter}$} is the set $B_{iter} = \{x \mid x = \vec{v} + \vec{y} \text{ for some } y \in B\}$.  The locations of the all the beads in Figure~\ref{fig:swGrowthGadgets} is an example of an iteration when $\delta=4$.  Furthermore, the points of the blue beads in the figure make up the set of bead-line points in the iteration.

Let $R_{\delta}$ be the routing of $C_{\delta}$.  For convenience, we let $p(b) \subset \mathbb{T}$ be the set of points defined by $p(b) = \{\vec{x} \mid (b, \vec{x}) \in R_{\delta}$\}.  That is $p(b)$ is the set of points where bead type $b$ is located in the configuration $C_{\delta}$.  Given a specific iteration $D$, we define $p_D(b) = \vec{x}$ where $x \in p(b)$.  Note this is well defined since the bead types placed at points in an iteration in $R_{\delta}$ are unique.

\begin{lemma}
Let $\delta > 2$.  There does not exist any system $\TMO' = (\Sigma', w', \calH', \delta', \alpha', \sigma')$ with $\delta' < \delta$ such that $\TMO'$ assembles $S_{\delta}$.
\end{lemma}

\begin{proof}
For the sake of contradiction, suppose that there exists a system $\TMO' = (\Sigma', w', \calH', \delta', \alpha', \sigma')$ with $\delta' < \delta$ which assembles $S_{\delta}$.  Let $C' = (P', w', H') \in \termasm{\TMO'}$ (note that $\TMO'$ isn't necessarily deterministic, so there may be more than one terminal configuration) and let $\vec{C'}=(C_i')_{i=0}$, where $C_i' = (P_i', w_i', H_i')$, be the foldable sequence of $\TMO'$ such that $res(\vec{C'}) = C'$.  By assumption, $dom(C') = S_{\delta}$.

Intuitively, the next claim states that in at least one of the iterations, $P'$ must pass through the set of bead-line points in a contiguous manner.  That is, the path doesn't ``exit'' the set of bead line points and then ``re-enter''.  Indeed, the point at which $P'$ exits must have an adjacent neighbor with which it doesn't share an edge with in $P'$.  Consequently, that neighbor only shares an edge with one other point in $P'$ and, consequently, it is an endpoint.  Since a directed path can only have two endpoints and there are an infinite number of iterations, the claim is proven.

\begin{claim} \label{clm:route}
There exists an iteration $D$ such that the subsequence of $P'$ consisting of exactly the bead-line points of $D$ is a contiguous subsequence of $P'$, and no point in $D$ is contained in $\dom(\sigma')$.
\end{claim}

\begin{proof}
Before we prove this claim we introduce the notion of an edge in a directed path.  We say that there is an edge between $p_i$ and $p_j$ in a directed path $P$ provided that $|i-j| = 1$.  Note that for a single element $p$ of a directed path, there are at most two elements such that there is an edge between those elements and $p$. And, if an element only has one edge in a directed path, then it is an endpoint.

To prove this claim, we show that any iteration $D_{iter}$ where the subsequence of $P'$ consisting of exactly the bead-line points of $D_{iter}$ is not a contiguous subsequence of $P'$ must contain an endpoint of the routing of $C'$.  To see this, let $D_{iter}$ be such an iteration.  Let $P^*$ be the minimal contiguous subsequence of $P'$ which contains all the bead-line points of $D_{iter}$.  By assumption, $P^*$ contains points which are not in the bead-line points of $D_{iter}$.  We assume the first point in $P^*$ is adjacent to a point which is in $S_{\delta}$ but not in the set of bead-line points of $D_{iter}$.  Otherwise, it would be the case that the first point of $P^*$ is part of the seed $\sigma'$ which would imply that $D_{iter}$ contains an endpoint of the routing of the configuration $C'$.

We now consider the case where the first point in $P^*$ which is not a point in the bead-line points of $D_{iter}$ is in the set $\cup_{i \in [1, \lceil \frac{\delta(5(\delta-1)+1)}{4} \rceil]} \{p(ra_i), p(rc_i), p(la_i), p(lc_i)\}$ (the beads located at these points are shown in Figure~\ref{fig:labeled-sub}).  Without loss of generality, we assume the first point in $P^*$ which is not a bead-line point is $p(la_i)$ for some $i \in [1, \lceil \frac{\delta(5(\delta-1)+1)}{4} \rceil]$.  This means that there is not an edge in the directed path between $p(cd_i)$ and one of the neighbors which is adjacent to it in the set of bead-line points of $D_{iter}$, which we denote by $\vec{t}$, since it must share an edge with the point which directly preceded it and it must share an edge with $p(la_i)$.  Consequently, $\vec{t}$ must be an endpoint of the directed path of the configuration $C'$ since $\dom(C') = S_{\delta}$, $\vec{t}$ only has two neighbors in $S_{\delta}$, and one of $\vec{t}$'s neighbors does not share an edge with $\vec{t}$.

The second case we consider is that the first point in $P^*$ which is not a point in the bead-line point of $D_{iter}$ is not in the set $\cup_{i \in [1, \lceil \frac{\delta(5(\delta-1)+1)}{4} \rceil]} \{p(ra_i), p(rc_i), p(la_i), p(lc_i)\}$.  This means it is either the point in the path of the \texttt{LEFT-WALL} gadget which is adjacent to $ca_1$ or the point in the path of the \texttt{SPACER} gadget which is adjacent to the bead-line points of $D_{iter}$.  In either case, it must be the case that the first bead in $P^*$ is adjacent to one of the points in $\cup_{i \in [1, \lceil \frac{\delta(5(\delta-1)+1)}{4} \rceil]} \{p(ra_i), p(rc_i), p(la_i), p(lc_i)\}$ (by the assumption $P^*$ isn't a contiguous subsequence of $P'$).  A similar argument to the first case we considered shows that one of the beads adjacent to the first bead in $P^*$ must be an end point of the path of $C'$.

Since every iteration $D_{iter}$ where the subsequence of $P'$ consisting of exactly the bead-line points of $D_{iter}$ is not a contiguous subsequence of $P'$ must contain an endpoint of $P'$, there can be at most two iterations where the subsequence of $P'$ consisting of exactly the bead-line points of $D_{iter}$ is not a contiguous subsequence of $P'$.  This along with the fact $C'$ is infinite and $\sigma'$ is finite shows that there exists an iteration $D_{iter}$ such that the subsequence of $P'$ consisting of exactly the bead-line points of $D_{iter}$, denoted by $\bar{P}$, is a contiguous subsequence of $P'$, and no point in $D_{iter}$ is contained in $\dom(\sigma')$.  This concludes the proof of the claim.
\end{proof}

To reduce notation, we drop the $\alpha$, $\delta$ and $\calH$ when referring to elongations, favorable elongations and $\xrightarrow[\calH, t]{\alpha, \delta}$ since they are clear from context.  That is, for the rest of this section, these terms are implicitly referring to the parameters of $\TMO'$.  Another convenient piece of terminology we use is we say that a configuration $C_i'$ in $\vec{C'}$ \emph{stabilizes} a bead if that bead appears in $C_i'$ but not in $C_{i-1}'$.  Similarly, we say a configuration $C_i'$ stabilizes a bond provided that bond appear in $C_i'$ but not in $C_{i-1}'$.  Let $BL$ be a set of bead-line points in an iteration $D$ with a contiguous routing.  We know such a $D$ exists due to Claim~\ref{clm:route}.  For the rest of this section we shorten the notation $p_D(b)$ to just $p(b)$ since $D$ is clear from context.

Let $C_a = (P_a, w_a, H_a) \in \prodasm{\TMO'}$ be an element of the assembly sequence $\vec{C'}$ and let $C_b=(P_b, w_b, H_b)$ be an $\calH$-valid $\alpha$-$\delta$-favorable elongation of $C_a$.  We call the set of bonds in the set $H_b \setminus (H_b \cap H')$ the set of \emph{phantom bonds of $C_a$} and we denote this set by $\calP \calH (C_b)$.  When $H_b = \calP\calH(C_b)$ we say that $C_b$ is \emph{stabilized by only phantom bonds}.  We call $H_b \setminus \calP \calH (C_b)$ the set of \emph{visible bonds of $C_b$}.  Intuitively, the phantom bonds of a favorable elongation are the bonds which help to stabilize a bead in the next configuration in the foldable sequence, but do not show up in the terminal configuration.  The set of visible bonds do show up in the terminal configuration.  Note here that $C'$ is fixed and we are always talking about the phantom bonds with respect to $C'$.  Figure~\ref{fig:sw-main-line} shows an example of both phantom and visible bonds.  Note that the bonds between the purple bead and the maroon beads are phantom bonds (since they do not appear in the terminal configuration), and the bond between the orange and aqua bead is a visible bond.  Note that since any elongation can have at most $\delta$ nascent beads and a bead can have at most $5$ bonds, there can be at most $5\delta$ phantom bonds in any favorable elongation.  We denote the total number of bonds made by the nascent beads in an elongation $C_e$ by $\calNB (C_e)$.

Let $R'$ be the routing of $C'$. Without loss of generality, we assume that the first bead to be stabilized by $\TMO'$ in $BL$ is at location $p(ca_1)$ (the right most bead in Figure~\ref{fig:sw-main-full-example}. Let $u$ be the subsequence (which is not contiguous) of beads in $R'$ constructed by adding a bead $(\vec{x}, b)$ to $u$ if and only if $\vec{x} \in  \cup_{i \in [1, \lceil \frac{\delta(5(\delta-1)+1)}{4} \rceil]} \{p(ca_i), p(cc_i), p(ce_i), p(cg_i)\}$ (these correspond to the points where the purple beads are located in the Figure~\ref{fig:sw-main-full-example}.  Let $o$ be the subsequence of beads in $R'$ constructed by adding a bead  $(\vec{x}, b)$ to $o$ if and only if $\vec{x} \in \cup_{i \in [1, \lceil \frac{\delta(5(\delta-1)+1)}{4} \rceil]} \{p(ra_i), p(rc_i), p(la_i), p(lc_i)\}$ (these correspond to the points where the orange beads are located in the Figure~\ref{fig:sw-main-full-example}.  Now, let $\vec{C_u}$ be the subsequence of configurations in $\vec{C'}$ constructed by adding a configuration $C_i'$ to $\vec{C_u}$ provided that $C_i'$ stabilizes a bead in the subsequence $u$.  Similarly, let $\vec{C_o}$ be the subsequence of configurations in $\vec{C'}$ constructed by adding a configuration $C_i$ to $\vec{C_o}$ provided that $C_i$ stabilizes a bead in the subsequence $o$.  Recall that given a bead $b = (\vec{x}, a)$ (which is a bead type along with a point), $\dom(b) = \vec{x}$.  We define the sequence $\dom(u)$ and $\dom(o)$ to be sequences in $\mathbb{T}$ such that $\dom(u) = (\vec{x}_i)_{i=1}$ where $\vec{x}_i = \dom(u(i))$ for all $i$ and $\dom(o) = (\vec{x}_i)_{i=1}$ where $\vec{x}_i = \dom(o(i))$ for all $i$

Let $\vec{C^*} = (C_i^*)_{i=0}$, where $C_i^* = (P_i^*, w_i^*, H_i^*)$, be a sequence of configurations such that $C_i^*$ is a favorable elongation of $C_i'$ and $C_{i+1} \sqsubseteq C_i^*$.  We define $\vec{C_o^*}$ to be the subsequence of configurations in $\vec{C^*}$ such that $C_i^*$ is in $\vec{C_u^*}$ if and only if $C_{i-1}'$ is in $\vec{C_o}$.  Similarly, we define $\vec{C_u^*}$ to be the subsequence of configurations in $\vec{C^*}$ such that $C_i^*$ is in $\vec{C_u^*}$ if and only if $C_{i-1}'$ is in $\vec{C_p}$.  Intuitively, these are the favorable elongations which are projected to stabilize beads in $o$ and $u$.

\begin{claim} \label{obs:no-vis-bonds}
Suppose that $i$ is such that $H_{i+1}' = H_i'$.  Then $\calNB (C_{i+1}^*) \geq \calNB (C_{i}^*)$.
\end{claim}

Suppose $t \in \mathbb{N}$ is such that $C_{i}^*$ is an elongation of $C_i'$ by $w[t ... t+\delta-1]$.  Then it must be the case that all bonds in $C_i^*$ occur between some bead type and a bead type in $w[t+1 .. t+\delta-1]$ (otherwise during the projection the bond would be added to $C_i'$ and consequently $H_{i+1}'$ would not equal $H_i'$).  Now, consider the favorable elongations of $C_{i+1}'$.  Since, by definition, the favorable elongations include elongations of length shorter than $\delta$, it also includes the favorable elongations by $w[t+1 .. t+\delta-1]$.  Now, observe that there exists a configuration $\bar{C} = (\bar{P}, \bar{w}, \bar{H})$ with $H_i' \subset \bar{H}$ such that $\bar{C}$ is an elongation of $C_{i+1}'$ by $w[t+1 .. t+\delta-1]$.  Consequently, any favorable elongation of $C_{i+1}'$ by $w[t+1 .. t+\delta]$ $C''$ must be such that $\calNB (C'') \geq \calNB (\bar{C}) \geq \calNB (C_{i}^*)$.

\begin{claim} \label{clm:nbgeqpb}
For any $i \in \mathbb{N}$, $\calNB(C_{i+1}^*) \geq |\calPB(C_i^*)|$.
\end{claim}

Suppose $t \in \mathbb{N}$ is such that $C_{i}^*$ is an elongation of $C_i'$ by $w[t ... t+\delta-1]$.  Then it must be the case that all bonds in $\calPB(C_i^*)$ occur between some bead type and a bead type in $w[t+1 .. t+\delta-1]$ (otherwise during the projection the bond would be added to $C_i'$ and consequently the bond would not be a phantom bond).  Now, consider the favorable elongations of $C_{i+1}'$.  Since, by definition, the favorable elongations can include elongations of length shorter than $\delta$, it also considers the favorable elongations by $w[t+1 .. t+\delta-1]$.  Now, observe that there exists a configuration $\bar{C} = (\bar{P}, \bar{w}, \bar{H})$ with $\calPB(C_{i}^*) \subset \bar{H}$ such that $\bar{C}$ is an elongation of $C_{i+1}'$ by $w[t+1 .. t+\delta-1]$.  Consequently, any favorable elongation $C''$ of $C_{i+1}'$ by $w[t+1 .. t+\delta]$ must be such that $\calNB (C'') \geq \calNB (\bar{C}) \geq |\calPB(C_{i}^*)|$.

Since $S_{\delta}$ was constructed so that the Euclidean distance between point $\dom(o(i))$ and point $\dom(o(i+1))$ is $\delta-1$ and the delay factor of $\TMO$ is assumed to be $\leq \delta-1$,  the only way for phantom bonds to form in $C_{k-1}$ is if the following observation holds.

\begin{observation} \label{obs:naughty-bead}
Let $h \in \lceil \frac{\delta(5(\delta-1)+1)}{4} \rceil$, and suppose $C_k'$ stabilizes $u(h)$.  Then $o(h) \notin \dom(C_{k-1}^*)$.
\end{observation}

If this observation didn't hold, then it would not be possible for any bead in the nascent portion to be adjacent to a bead with which it could bind since the nascent portion would be ``fully stretched out'' as shown in part (a) of Figure~\ref{fig:swBad}.  A similar argument also allows us to see that $C_{k-1}^*$ must be stabilized by only phantom bonds.

\begin{observation} \label{obs:bad-stab}
Let $h \in \lceil \frac{\delta(5(\delta-1)+1)}{4} \rceil$, and suppose $C_k'$ stabilizes $u(h)$.  Then $C_{k-1}^*$ is stabilized by only phantom bonds.
\end{observation}

\begin{claim} \label{clm:jacobs}
Let $\delta > 2$ and $\delta' < \delta$.  For every sequence $a = (a_i)$ of length $\delta$ or greater where $a_i \leq \delta'$,  there exists $k,l \in \mathbb{N}$ such that $\delta \cdot k < \sum_{i=1}^l(a_i) < \sum_{i=1}^{l+1}(a_i) \leq \delta \cdot (k+1)$.
\end{claim}

\begin{proof}
To see this claim, note that since $a_i \leq \delta' < \delta$, $\sum_{i=1}^{\delta}(a_i) \leq \delta(\delta -1) = \delta^2 - \delta$. Now consider the $\delta - 1$ many sets given by $[j\delta, (j+1)\delta)$ for each $j\in \N$ such that $0\leq j \leq \delta-1$. Note that there are $\delta$ many sums $\sum_{i=1}^{m}(a_i)$ for each $m\in \N$ such that $1\leq m \leq \delta$. By the pigeonhole principle, there exists a $k$ such that at least two such sums must be numbers in the set $(j \delta, (j-1)\delta]$ for $j = k+1$. Let $l$ be such that $\sum_{i=1}^{l}(a_i)$ is the first of these two sums. The existence of $k$ and $l$ prove the claim.
\end{proof}

\begin{observation} \label{obs:1}
If $C_i'$ stabilizes a bead which has a position in $BL$ except for the last $\delta'$ beads in $BL$, $C_i^*$ must have phantom bonds.
\end{observation}

\begin{observation} \label{obs:2}
Observation~\ref{obs:1} implies that if $C_i'$ stabilizes a bead which has position in $BL$, $C_i^*$ at most $\delta'-1$ (notice $\delta'-1 \leq \delta-2$) nascent beads can be in the correct position.  In other words, the last bead in the nascent portion of $C_i^*$ must always be in the incorrect position.
\end{observation}

This final claim allows us to say that every $\delta$ orange beads the bonds required by a favorable configuration to stabilize an orange increases.  It does this by showing that there exists an ``intermediate configuration'' $C_k'$ between configurations which stabilize orange beads in the foldable sequence such that $C_k^*$ does not have a nascent orange bead in the proper position.  To show that such a $C_k'$ exists, we consider the case where $C_k^*$ has a visible bond, but it is not in the right position.  An example where there is a visible bond, but the orange bead isn't in the correct position is shown in part (b) Figure~\ref{fig:swBad}.  In this case, the configuration which stabilizes this bead, must use an ``extra bond'' to cause the bead to be stabilized in the proper position.  Consequently, this configuration must have one more nascent bond than the configuration which stabilized the previous orange bead.  In the case where all configurations in the foldable sequence have elongations which place the orange bead in the correct position, we show that since only $\delta-2$ beads can be stabilized in the correct positions (by Observation~\ref{obs:1}, there comes a point where an extra intermediate configuration $C_k'$ occurs between configurations which stabilize orange beads in the foldable sequence.  This configuration $C_k'$ is required to only be stabilized using phantom bonds.  Thus, the next configuration to stabilize an orange bead must ``overcome'' those phantom bonds by using one more nascent bond than the elongation which stabilized the previous orange bead.

\begin{claim} \label{clm:final}
For every $j \in [1, 5(\delta-1)+1]$, there exists $C_k'$ such that $C_o(j\delta) \rightarrow C_k'$, $C_k' \rightarrow C_o((j+1)\delta)$ and $\calPB(C_o^*(j\delta)) < \calPB(C_k^*) < \calNB(C_o^*((j+1)\delta))$.
\end{claim}
\begin{proof}
Let $j \in [1, 5(\delta-1)+1]$.  Also, let $\vec{B}$ be the configuration sequence from $C_o(j \delta )$ to $C_o(j(\delta + 1))$.  Due to the constraints placed on how $\TMO'$ can assemble $BL$ in Observation~\ref{clm:route}, we know that only one such sequence exists.  Without loss of generality, we assume that the first configuration $C_i'$ in $\vec{C'}$ such that $\dom(C_i') \cap BL \neq \emptyset$ has $p(ca_1) \in \dom(C_i)$.  That is, the system assembles the bead line points shown as blue beads in Figure~\ref{fig:swGrowthGadgets} starting from the bottom.  We consider two cases: 1) for every configuration $C_i'$ in $\vec{B}$ there exists a favorable elongation $C_i^*$ of $C_i'$ such that every visible bond in $C_i^*$ involves a bead which has a position in $BL$ and 2) there exists some configuration $C_i'$ of $\vec{B}$ such that all favorable elongations of $C_i'$ have a visible bond which involves a bead which has a position not in $BL$.

For the first case, we can assume that for every $C_i^*$ which is an elongation of some configuration in $\vec{B}$, there is no bead in $C_i^*$ which is outside $BL$ and has a visible bond.  We first show that there exists a configuration $C_l'$ such that $C_o(j\delta) \rightarrow C_l'$, $C_l' \rightarrow C_o((j+1)\delta)$ and $C_l'$ is stabilized by only phantom bonds.  Let $d, h \in \mathbb{N}$ be such that $C_d' = C_o(j\delta)$ and  $C_{d + h}' = C_o((j+1)\delta)$.  Note that for $d','h\in \mathbb{N}$ if $R'(d')$ contains $\dom(o(i))$ and $R'(h')$ contains $\dom(o(i+1))$, $h'-d' = \delta-1$.  This means that $h=\delta(\delta-1)=\delta^2-\delta$ which implies $\dom(C_{d+h}') \setminus \dom(C_d') = \delta^2 - \delta$.

Define $C^{**}$ to be the subsequence of $C^*$ such that $C^{**} = C_d^*$ and $C_i^*$ is in $C^{**}$ if and only if $d < i \leq d+h$ and $C_i^*$ is not an elongation of any favorable elongation of  $C_{i-1}'$.  Intuitively, $C^{**}$ includes a configuration $C$ if that configuration is a favorable elongation which forces a bead that was in one position in a previous favorable elongation to switch to a new position due to phantom bonds, .  Note that since every configuration $C$ which stabilizes a new bead in $BL$ must ``incorrectly stabilize'' the last bead (per Observation~\ref{obs:2}), this set is not empty since another configuration $C'$ will stabilize that bead in the correct position and consequently there will not be any favorable elongation of $C'$ which is an elongation of $C$.  In fact, from Observation~\ref{obs:2} it follows that $C^{**}$ has at least $\delta$ configurations since $\frac{\delta^2-\delta}{\delta - 2} \geq \frac{\delta (\delta-1)}{(\delta - 1)} = \delta$.

Observe that in order for $C_i^{**}$ to contain a visible bond, it must be the case that there exist $r$ such that $\dom(o(r)) \in \dom(C_i^{**})$ (by the assumption of case 1).  As noted above if $R'(d')$ contains a point $\dom(o(j))$ and $R'(h')$ contain $\dom(o(j+1))$, then $h'-d' = \delta-1$.  Hence, there exists $\vec{c}$ such that $R'(c + (\delta-1)i')$ contains a point in $o$ for all of $i' \in[1, 5(\delta-1)+1]$.  Let $b=(b_i)$ be the sequence in $\mathbb{N}$ such that $b_i$ is number of beads in $C_i^{**}$ which are correctly stabilized.  Recall that $C^{**}$ has at least $\delta$ elements which implies $b$ does as well, and note that it follows from Observation~\ref{obs:2} that for all $i$, $b_i < \delta-1$.  It now follows from Claim~\ref{clm:jacobs} that there exists $l', k \in \mathbb{N}$ such that $c + (\delta-1)k < \Sigma_{i=1}^{l'} b_i < \Sigma_{i=1}^{l'+1} b_i \leq c + (\delta-1)(k+1)$.  But, this means that the elongation $C_{l'+1}^{**}$ does not contain any visible bonds since it can only potentially stabilize the beads $R'(\vec{c} + \Sigma_{i=1}^{l} (b_i) ), R'(\vec{c} + \Sigma_{i=1}^{l} (b_i) + 1), ..., R'(\vec{c} + \Sigma_{i=1}^{l+1} (b_i) - 1)$ in the correct positions and this does not include any bead which has a position in $o$.

Let $l = l' + 1$, and let $k$ be such that $C_k' = \rho_{\delta}(C_l^{**})$.  So far, we have established that there exists $C_l^{**}$ such that $C_o(j\delta) \rightarrow C_{k}'$, $C_{k}' \rightarrow C_o((j+1)\delta)$ and $C_k^*$ is stabilized by only phantom bonds.  Let $t$ be such that $C_o(t) \rightarrow C_k'$ and $C_o(t)$ is maximal.  It follows from the fact that $C_k^*$ is in the sequence $C^{**}$ that $|\calPB(C_o^*(t))| < \calNB(C_k^*)$.  Indeed, suppose for the sake of contradiction that $|\calPB(C_o^*(t))| \geq \calNB(C_k^*)$. Then there exists a favorable elongation $C_{lb}$ of $C_k'$ where the bonds in $\calPB(C_o^*(t))$ ``dominate'' the elongation.  This implies that there exists an elongation of $C_o(t)$ such that $C_k^*$ is an elongation of $C_o(t)$.  This contradicts the way we constructed $C^{**}$.

For the second case, let $C_r'$ be the configuration such that all elongations of $C_r'$ have a nascent bead involved in a visible bond which has a position outside of $BL$.  Furthermore, let $i$ be such that $C_o(i)$ is the minimal element such that $C_r' \rightarrow C_o(i)$.  Intuitively, $C_r^*$ makes a visible bond by placing an orange bead in a position which is next to an aqua bead, but not in $S_{\delta}$ as shown in part (b) of Figure~\ref{fig:swBad}.  The configuration $C_o(i)$ is a configuration where this same orange bead is stabilized in the ``correct'' position.  Now, we prove that $\calPB(C_o^*(i-1)) \leq \calNB(C_u^*(i))$, $\calNB(C_p^*(i)) < \calNB(C_r^*)$, and $\calNB(C_r^*) < \calNB(C_o^*(i))$ (otherwise, the orange bead could be incorrectly stabilized).

First, note that $\calPB(C_o^*(i-1)) \leq \calNB(C_u^*(i))$ follows directly from Claim~\ref{clm:nbgeqpb}.  To see that $\calNB(C_p^*(i)) < \calNB(C_r^*)$ note that if this was not the case there would exist an elongation where the bonds of $C_p^*(i)$ dominate and there would exist an elongation of $C_r'$ with no visible bond (since $C_p^*(i)$ can't have a visible bond).  To see that $\calNB(C_r^*) < \calNB(C_o^*(i))$ note that anything otherwise would mean that a malformed configuration could result.  This implies $\calPB(C_o^*(i-1)) < \calNB(C_r^*) < \calNB(C_o^*(i))$.
\end{proof}

Since all bond strengths are natural numbers, Claim~\ref{clm:final} implies that for every $j$, $\calNB(C_o^*(j\delta)) = |\calPB(C_o^*(j\delta))| + 1 < \calNB(C_o^*((j+1)\delta))$.  It follows from the construction of $S_{\delta}$ that there are at least $\delta \times (5 (\delta-1) + 1)$ elements in $o$.  Thus, $\calNB(C_o^*(\delta \times (5 (\delta-1) + 1))) \geq 5 (\delta-1) + 1$.  But, this contradicts the fact that the total number of bonds involving nascent beads in a $\delta'$-elongation of any configuration is $5(\delta')$ (since $\delta' < \delta$).
\end{proof}

\begin{figure}[htp]
\centering
\includegraphics[width=4.0in]{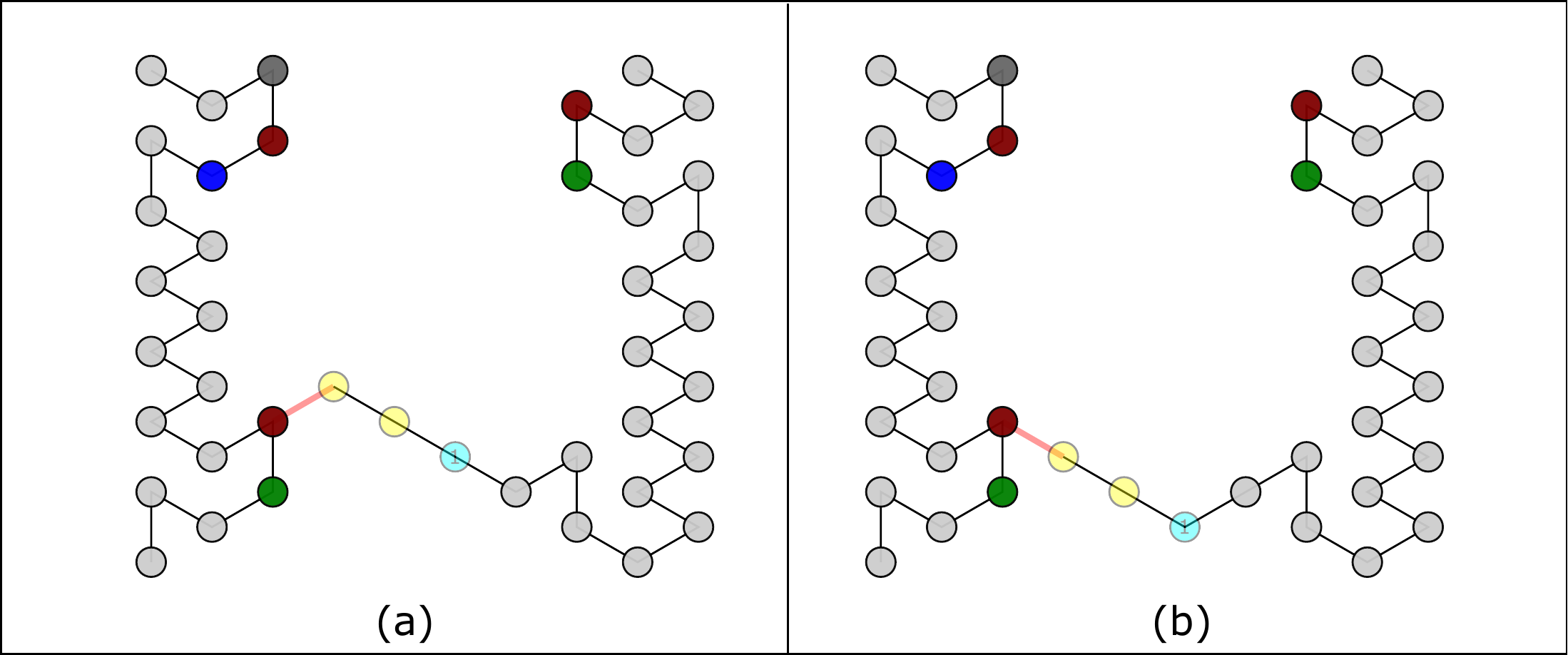}
\caption{Any system which uses delay less than $4$ must use only phantom bonds to stabilize the first bead otherwise something bad can happen.}
\label{fig:swBad}
\end{figure}

%% file: d1a1_det_finiteness.tex
	\section{Finiteness of delay-1, arity-1 deterministic oritatami systems}\label{sec:d1a1}
%-------------------------------------------------------------------

\tikzstyle{mol} = [fill, circle, inner sep=1.25pt]
\tikzstyle{point} = [fill, circle, inner sep=0.5pt]

\begin{figure}[h]
\centering
\begin{minipage}{0.3\linewidth}
\centering
\scalebox{0.9}{\begin{tikzpicture}

\draw[red, thick] (0, 0) node[mol]{} node[above] {$a$}
-- ++(300:1) node[mol]{} 
-- ++(240:1) node[mol]{} node[below] {$\overline{b}$}
-- ++(0:1) node[mol]{} node[below] {$b$}
-- ++(60:1) node[mol]{} 
-- ++(120:1) node[mol]{} node[above] {$\overline{a}$}
;
\draw[-latex, thick] (1, 0) -- ++(0:1) node[mol]{} node[above] {$a$}
-- ++(300:1) node[mol]{} 
-- ++(240:1) node[mol]{} node[below] {$\overline{b}$}
-- ++(0:1) node[mol]{} node[below] {$b$}
-- ++(60:1) node[mol]{} 
-- ++(120:1) node[mol]{} node[above] {$\overline{a}$}
-- ++(0:1)
;
\draw[dotted, thick, red] (0, 0) -- ++(0:1); 
\draw[dotted, thick] (2, 0) -- ++(0:1) ++(240:2) -- ++(180:1);

\draw(4,0)++(300:1) node {\Large $\cdots$};

\end{tikzpicture}}
\end{minipage}
\begin{minipage}{0.025\linewidth}
\ \\
\end{minipage}
\begin{minipage}{0.3\linewidth}
\centering
\scalebox{0.9}{\begin{tikzpicture}

\draw[red, thick] (0, 0) node[mol]{} node[above] {$a$}
-- ++(300:1) node[mol]{} node[left] {$\overline{b}$}
-- ++(240:1) node[mol]{} node[below] {$\overline{b}$}
-- ++(0:1) node[mol]{} node[below] {$b$}
-- ++(60:1) node[mol]{} node[left] {$\overline{a}$}
-- ++(120:1) node[mol]{} node[above] {$\overline{a}$}
;
\draw[-latex, thick] (1, 0) -- ++(0:1) node[mol]{} node[above] {$a$}
-- ++(300:1) node[mol]{} node[left] {$\overline{b}$}
-- ++(240:1) node[mol]{} node[below] {$\overline{b}$}
-- ++(0:1) node[mol]{} node[below] {$b$}
-- ++(60:1) node[mol]{} node[left] {$\overline{a}$}
-- ++(120:1) node[mol]{} node[above] {$\overline{a}$}
-- ++(0:1)
;
\draw[dotted, thick, red] (0, 0) -- ++(0:1) ++(240:1) -- ++(300:1);
\draw[dotted, thick] (0,0) ++(300:2) -- ++(0:1) ++(120:1) -- ++(60:1) -- ++(0:1) ++(240:1) -- ++(300:1);

\draw(4,0)++(300:1) node {\Large $\cdots$};

\end{tikzpicture}}
\end{minipage}
\begin{minipage}{0.025\linewidth}
\ \\
\end{minipage}
\begin{minipage}{0.3\linewidth}
\centering
\scalebox{0.9}{\begin{tikzpicture}

\draw[red, thick] (0, 0) node[mol]{} node[below] {$a$}
-- ++(60:1) node[mol]{} node[above] {$b$}
-- ++(300:1) node[mol]{} node[below] {$\overline{a}$}
;
\draw[-latex, thick] (1, 0)
-- ++(60:1) node[mol]{} node[above] {$\overline{b}$}
-- ++(300:1) node[mol]{} node[below] {$a$}
-- ++(60:1) node[mol]{} node[above] {$b$}
-- ++(300:1) node[mol]{} node[below] {$\overline{a}$}
-- ++(60:1) node[mol]{} node[above] {$\overline{b}$}
-- ++(300:1) node[mol]{} node[below] {$a$}
-- ++(60:1) node[mol]{} node[above] {$b$}
-- ++(300:0.5)
;

\draw(5,0)++(60:0.5) node {\Large $\cdots$};

\draw[dotted, thick] (0, 0) -- ++(0:4) (0, 0)++(60:1) -- ++(0:4);
\end{tikzpicture}}
\end{minipage}

\caption{Deterministically foldable infinite shapes: 
(Left) A glider at delay-3 and arity-1; 
(Middle) A glider at delay-2 and arity-2, and 
(Right) A zigzag at delay-1 and arity-2.
Seeds are colored in red. 
The common bead type set consists of four letters $a, \overline{a}, b, \overline{b}$ and the rule set used in common is complementary: $a$ with $\overline{a}$ and $b$ with $\overline{b}$. 
}
\label{fig:infinite_shapes}
\end{figure}

In this section, we prove that oritatami systems cannot yield any infinite terminal conformation at delay~1 and arity~1 deterministically. 
The finiteness stems essentially from the particular setting of values to delay and arity. 
The \textit{glider} is a well-known infinite conformation foldable deterministically by an oritatami system at delay~3 and arity~1; see Figure~\ref{fig:infinite_shapes} (Left). 
The glider can be ``widened'' in order to be folded deterministically at arbitrarily longer delays. 
The glider can be ``reinforced'' with more bonds so that it folds at a shorter delay~2 with arity~2 as suggested in Figure~\ref{fig:infinite_shapes} (Middle). 
Even at the shortest possible delay, that is, 1, arity being 2 enables oritatami systems to fold an infinite structure deterministically, as exemplified by a zigzag conformation shown in Figure~\ref{fig:infinite_shapes} (Right). 
These infinite conformations leave just two possible settings of delay and arity under which infinite conformations cannot be folded deterministically: arity is set to 1 and delay is set to either 1 or 2. 
We will show the finiteness of deterministic folding in the first case in the rest of this paper, and leave the case of delay~2, arity~1 open. 
Note that even at these settings infinite shapes can be folded nondeterministically; an infinite transcript of inert beads folds into an arbitrary non-self-intersecting path at an arbitrary delay, and arity does not matter because beads are inert. 

\begin{figure}[tb]
\centering
\begin{tikzpicture}

\draw[-latex, thick] (0, 0) node[mol] {} node[below] {$a_{i-1}$} 
-- ++(0:1) node[mol] {} node[below] {$a_i$}
;
\draw (1,0)++(60:1) node[mol] {} node[above] {$a_j$};
\draw[dashed, thick] (1, 0) -- ++(60:1);

\foreach \x in {5, 9, 13} {
	\draw (\x, 0)++(60:1) node[mol] {} node[above] {$a_{j_1}$};
	\draw (\x, 0)++(120:1) node[mol] {} node[above] {$a_{j_4}$};
	\draw (\x, 0)++(240:1) node[mol] {} node[below] {$a_{j_3}$};
	\draw (\x, 0)++(300:1) node[mol] {} node[below] {$a_{j_2}$};
}

\draw (4, 0) node[point] {} ++(0:1) node[point]{} node[above] {$p$} ++(0:1) node[point]{};

\draw (7, 0) node {$\Rightarrow$};

\draw[-latex, thick] (8, 0) node[mol] {} node[below] {$a_{i-2}$} 
-- ++(0:1) node[mol] {} node[below] {$a_{i-1}$};
\draw (10, 0) node[point] {} node[above] {};

\draw (11, 0) node {$\Rightarrow$};

\draw[-latex, thick] (12, 0) node[mol] {} node[below] {$a_{i-2}$} 
-- ++(0:1) node[mol] {} node[below] {$a_{i-1}$}
-- ++(0:1) node[mol] {} node[below] {$a_i$}
;

\end{tikzpicture}
\caption{The two ways for a bead to get stabilized in oritatami systems at delay 1 and arity 1: 
(Left) by being bound to a bead $a_j$ for some $j \le i-2$; and 
(Right) through a tunnel section formed by the four beads $a_{j_1}, a_{j_2}, a_{j_3}, a_{j_4}$. 
}
\label{fig:stabilization_2ways}
\end{figure}

\begin{figure}[tb]
\centering
\scalebox{1}{\begin{tikzpicture}

\draw[-latex, thick] (-1, 0) node[mol]{} node[below] {$a_{i-2}$} -- ++(0:1) node[mol]{} node[above] {$a_{i-1}$} -- ++(240: 1) node[mol]{}  node[below] {$a_i$};
\draw (0, 0) ++(120:1) node {$\times$};% node[above] {$p_0$}
\draw (0, 0) ++(60:1) node {$\times$};% node[above] {$p_1$};
\draw (0, 0) ++(0:1) node {$\times$};% node[above] {$p_2$};
\draw (0, 0) ++(300:1) node {$\times$};% node[below] {$p_3$};

\draw[-latex, thick] (2.5, 0) node[mol]{} node[below] {$a_{i-2}$} -- ++(0:1) node[mol]{} node[above] {$a_{i-1}$} -- ++(300: 1) node[mol]{} node[below] {$a_i$};
\draw (3.5, 0) ++(120:1) node {$\times$};% node[above] {$p_0$};
\draw (3.5, 0) ++(60:1) node {$\times$};% node[above] {$p_1$};
\draw (3.5, 0) ++(0:1) node {$\times$};% node[above] {$p_2$};
\draw (3.5, 0) ++(240:1) node {$\times$};% node[below] {$p_4$};

\draw[-latex, thick] (6, 0) node[mol]{} node[below] {$a_{i-2}$} -- ++(0:1) node[mol]{} node[above] {$a_{i-1}$} -- ++(0: 1) node[mol]{} node[above] {$a_i$};
\draw (7, 0) ++(120:1) node {$\times$};% node[above] {$p_0$};
\draw (7, 0) ++(60:1) node {$\times$};% node[above] {$p_1$};
\draw (7, 0) ++(300:1) node {$\times$};% node[below] {$p_3$};
\draw (7, 0) ++(240:1) node {$\times$};% node[below] {$p_4$};

\draw[-latex, thick] (9.5, 0) node[mol]{} node[below] {$a_{i-2}$} -- ++(0:1) node[mol]{} node[below] {$a_{i-1}$} -- ++(60: 1) node[mol]{} node[above] {$a_i$};
\draw (10.5, 0) ++(120:1) node {$\times$};% node[above] {$p_0$};
\draw (10.5, 0) ++(0:1) node {$\times$};% node[above] {$p_2$};
\draw (10.5, 0) ++(300:1) node {$\times$};% node[below] {$p_3$};
\draw (10.5, 0) ++(240:1) node {$\times$};% node[below] {$p_4$};

\draw[-latex, thick] (13, 0) node[mol]{} node[below] {$a_{i-2}$} -- ++(0:1) node[mol]{} node[below] {$a_{i-1}$} -- ++(120: 1) node[mol]{} node[above] {$a_i$};
\draw (14, 0) ++(60:1) node {$\times$};% node[above] {$p_1$};
\draw (14, 0) ++(0:1) node {$\times$};% node[above] {$p_2$};
\draw (14, 0) ++(300:1) node {$\times$};% node[below] {$p_3$};
\draw (14, 0) ++(240:1) node {$\times$};% node[below] {$p_4$};

\draw (0, 0)++(300:1)++(240:1) node {$t_{-120}$};
\draw (3.5, 0)++(300:1)++(240:1) node {$t_{-60}$};
\draw (7, 0)++(300:1)++(240:1) node {$t_{0}$};
\draw (10.5, 0)++(300:1)++(240:1) node {$t_{+60}$};
\draw (14, 0)++(300:1)++(240:1) node {$t_{+120}$};

\end{tikzpicture}}
\caption{All possible tunnel sections: acute right turn $t_{-120}$, obtuse right turn $t_{-60}$, straight $t_0$, obtuse left turn $t_{+60}$, and acute left turn $t_{+120}$.
}
\label{fig:tunnel_sections}
\end{figure}

Let $\Xi$ be a deterministic oritatami system of delay~1 and arity~1. 
Assume its seed $\sigma$ consists of $n$ beads for some $n \ge 1$, and along its primary structure, we index these $n$ beads as $a_{-n+1}, a_{-n+2}, \ldots, a_{-1}, a_0$. 
Let us denote its transcript by $w = a_1 a_2 a_3 \cdots$ for some $a_1, a_2, a_3, \ldots \in \Sigma$.
For $i \ge 0$, let $C_i$ be the unique elongation of $\sigma$ by $w[1..i]$ that is foldable by $\Xi$. 
Hence, $C_0 = \sigma$. 
We assume that the directed path of $C_i$ is indexed rather as $-n+1, -n+2, \ldots, 0, 1, \ldots, i$. 

Let us consider the stabilization of the $i$-th bead $a_i$ upon $C_{i-1}$. 
The bead cannot collaborate with any succeeding beads $a_{i+1}, a_{i+2}, \ldots$ at delay~1. 
There are just two ways to get stabilized at delay~1. 
One way is to be bound to another bead, as shown in Figure~\ref{fig:stabilization_2ways} (Left).
The other way is through a \textit{tunnel section}. 
A tunnel section consists of four beads that occupy four neighbors of a point. 
See Figure~\ref{fig:stabilization_2ways} (Right). 
Assume that four of the six neighbors of a point $p$ are occupied by beads $a_{j_1}, a_{j_2}, a_{j_3}, a_{j_4}$ with $-n+1 \le j_1 < j_2 < j_3 < j_4 < i-2$ while the other two are not occupied. 
If the beads $a_{i-2}$ and $a_{i-1}$ are stabilized respectively at one of the two free neighbors and at $p$ one after another, then the next bead $a_i$ cannot help but be stabilized at the other free neighbor. 
In this way, $a_i$ can get stabilized without being bound. 

Let us now formalize the tunnel section. 
Four beads $a_{j_1}, a_{j_2}, a_{j_3}, a_{j_4}$ with $-n+1 \le j_1 < j_2 < j_3 < j_4$ \textit{form a tunnel section around a point $p$} if there exist an index $k \ge j_4+2$ and the foldable configuration $C_{k} = (P, u, H) \in \mathcal{A}(\Xi)$ such that 
\begin{enumerate}
\item For all $s \in \{j_1, j_2, j_3, j_4\}$, $(P[s], p) \in E_\Delta$;
\item $(P[k-1], p) \in E_\Delta$; and 
\item $P[k] = p$. 
\end{enumerate}
We call the four beads $a_{j_1}, a_{j_2}, a_{j_3}, a_{j_4}$ \textit{walls} of this tunnel. 
The walls and the bead $a_{k-1}$ leave at most one of the neighbors of $p$ free in $C_k$. 
If the neighbor is not free, $C_k$ is terminal. 
Otherwise, $a_{k+1}$ is to be stabilized at the neighbor and yields $C_{k+1}$. 
The bead $a_{j_4}$ can be regarded the newest wall because of $j_1, j_2, j_3 < j_4$. 
If $j_4 \ge 1$, that is, if it is transcribed, then we say the tunnel section is \textit{created by the bead $a_{j_4}$}. 
Otherwise, we say it is \textit{equipped in the seed}. 
%
%Then we say that the four beads at $p_0, p_1, p_3, p_4$ form a tunnel around $p$. 
%The bead $a_{i-1}$ has been stabilized at their center. 
%This bead, its predecessor $a_{i-2}$, and the tunnel section cooperatively leave only one point, $p_2$ in the figure, at which the next bead $a_i$ can be placed. 
%In this way, a bead can get stabilized without consuming the binding capability of any other bead (at longer delays, other ways of non-binding stabilization are possible due to so-called ``hidden rule'', which never appears in any terminal conformation but indispensable, as argued in \cite{OtaSeki2017}). 
%On the other hand, if $p_2$ has been also occupied, there is no room for $a_i$ to be transcribed so that the current conformation $C_{i-1}$ is terminal and the system halts. 
%Our argument will not have to consider such dead ends. 
Figure~\ref{fig:tunnel_sections} exhibits all the five kinds of tunnel sections depending on which neighbors are walls (indicated by $\times$'s), modulo types and indices of wall beads. 

Tunnel sections and unbound beads, or more precisely, their \textit{one-time} capability of binding, are the resources for beads to get stabilized deterministically at delay~1 and arity~1 (at longer delays, other ways of non-binding stabilization are possible due to so-called ``hidden rule,'' which never appears in any terminal conformation but indispensable, as argued in \cite{OtaSeki2017}). 
The seed $\sigma$ of $\Xi$, consisting of $n$ beads, provides at most $n$ binding capabilities, one per bead. 
Claim that it can be equipped with at most $n$ tunnel sections. 
If $n < 4$, it cannot be equipped with any tunnel section. 
For larger $n$, any bead $a_j$ has its predecessor or successor or both. 
By definition, it cannot be a wall of any tunnel around the point where its predecessor or successor is. 
Therefore, the first bead $a_{-n+1}$ and the last bead $a_0$ can be a wall of at most five tunnel sections, whereas any other bead $a_j$ with $-n+2 \le j \le -1$ can be a wall of at most four tunnel sections. 
One tunnel section consists of four beads. 
Therefore, the seed can be equipped with at most $\lfloor (4n+2)/4 \rfloor = n$ tunnel sections. 

Being bound for stabilization, a bead will not be able to bind to another bead later due to arity~1. 
In contrast, if it is stabilized through a tunnel section, then it can provide one-time binding capability and create tunnel sections. 

\begin{theorem}\label{thm:d1a1_det_finiteness}
	Let $\Xi$ be an oritatami system of delay 1 and arity 1 whose seed consists of $n$ beads, and let $w$ be the transcript of $\Xi$. 
	If $\Xi$ is deterministic, then $|w| \le 9n$. 
\end{theorem}
\begin{proof}
	Assume $\Xi$ is deterministic. 
	Let us represent its transcript $w$ as $w = a_1 a_2 a_3 \cdots$ for beads $a_1, a_2, a_3, \ldots \in \Sigma$. 
	Each of these beads is stabilized either by being bound or through a tunnel section (or by both). 
	How they are stabilized can be described by a binary sequence $S$ of $b$'s (bound) and $t$'s (tunnel section); priority is given to $t$, that is, $S[i] = t$ if the $i$-th bead $a_i$ is stabilized not only by being bound but also through a tunnel section. 
	For $\ell \ge 1$, we call a factor $bt^\ell b$ of $S$ a \textit{tunnel of length $\ell$}. 
	See Figure~\ref{fig:stabilization} (right) for a tunnel of length 3, where $S[i-3 .. i+1] = btttb$; observe that the bead $a_{i-2}$ is stabilized both by both ways but due to the priority, $S[i-2] = t$. 

\begin{figure}[tb]
\begin{minipage}{0.3\linewidth}
\centering
\begin{tikzpicture}
\draw[-latex, thick] (0, 0) node[mol]{} node[below] {$a_{i-1}$} -- ++(0:1) node[mol]{} -- ++(60:1) node[mol]{} node[above] {$a_{i+1}$};

\draw[very thick, dotted] (1, 0) -- ++(120:1) node{$\times$};

\draw[-latex, dashed] (1, 0) -- ++(240: 0.75);
\draw[-latex, dashed] (1, 0) -- ++(300: 0.75);
\draw[-latex, dashed] (1, 0) -- ++(0: 0.75);

\draw (1,0)++(240:1) node[point] {};
\draw (1,0)++(300:1) node[point] {};
\draw (1,0)++(0:1) node[point] {};

\draw (2,0)++(60:1) node{$\times$} -- ++(300:1) node{$\times$} -- ++(240:1) node{$\times$} -- ++(240:1) node{$\times$} -- ++(180:1) node{$\times$} -- ++(180:1) node{$\times$};

\end{tikzpicture}
\end{minipage}
\begin{minipage}{0.05\linewidth}
\ \\
\end{minipage}
\begin{minipage}{0.6\linewidth}
\centering
\begin{tikzpicture}

\draw (0, 0) ++(60:1) node[point]{};
\draw (0, 0) ++(300:1) node[point]{};

\foreach \x in {1, 2, 3, 4} {
\draw (\x, 0) ++(60:1) node{$\times$};
\draw (\x, 0) ++(300:1) node{$\times$};
}
\draw (1,0)++(60:1) -- ++(0:3);
\draw (1,0)++(300:1) -- ++(0:3);

\draw[-latex, thick] (0, 0) node[mol]{} node[below] {$a_{i-5}$} -- ++(0:1) node[mol]{} node[below] {$a_{i-4}$} -- ++(0:1) node[mol]{} node[below] {$a_{i-3}$} -- ++(0:1) node[mol]{} node[below] {$a_{i-2}$} -- ++(0:1) node[mol] {} node[above] {$a_{i-1}$} -- ++(0:1) node[mol]{} node[right] {$a_i$} -- ++(60:1) node[mol]{} node[above left] {$a_{i+1}$} -- ++(120:1) node[mol] {} node[left] {$a_{i+2}$};

\draw[very thick, dotted] (2, 0) -- ++(60:1);
\draw[very thick, dotted] (4, 0) -- ++(300:1);
\draw[very thick, dotted] (4, 0)++(60:1) -- ++(0:1);

\draw (5, 0)++(300:1) node[point]{} ++(60:1) node[point]{} ++(60:1) node[point]{} ++(120:1) node[point]{} ++(180:1);

\draw[-latex, dashed] (5, 0) -- ++(300:0.75);
\draw[-latex, dashed] (5, 0)++(60:1) -- ++(300:0.75);
\draw[-latex, dashed] (5, 0)++(60:1) -- ++(0:0.75);
\draw[-latex, dashed] (5, 0)++(60:1) -- ++(60:0.75);

\draw (4,0)++(60:3) node{$\times$}
-- ++(0:1) node{$\times$}
-- ++(300:1) node{$\times$}
-- ++(300:1) node{$\times$}
-- ++(240:1) node{$\times$}
-- ++(240:1) node{$\times$}
-- ++(240:1) node{$\times$}
;

\end{tikzpicture}
\end{minipage}
\caption{
Stabilization of a bead $a_i$ (Left) by being bound, and (Right) through a tunnel of length 3.
The symbol $\times$ indicates that the point is occupied, while the small dot means that the point is free. 
A dashed arrow indicates that the bead at its origin creates a tunnel at the pointed free point. 
}
\label{fig:stabilization}
\end{figure}

If $S[i] = b$, that is, if the bead $a_i$ is stabilized not through a tunnel section but by being bound, then it can be involved in three separate tunnel sections as a wall but no more. 
%Figure~\ref{fig:stabilization} (Left) suggests that if $a_i$ is stabilized by being bound, it can be involved in three separate tunnel sections, and moreover, it can create all of them, but no more. 
Indeed, two of the six neighbors of the point at which $a_i$ is stabilized have been already occupied by its predecessor $a_{i-1}$ and by the bead to which $a_i$ is bound. 
It cannot be a wall of a tunnel section around the point where the successor $a_{i+1}$ will be stabilized. 
A tunnel is of the form $b t^+ b$ by definition, that is, it consumes two binding capabilities: one for the transcript to enter it and one for the transcript to decide which way to go after exit; while only the bead stabilized by its last tunnel section can provide a new binding capability. 
That is, a tunnel consumes at least one binding capability in total. 
For instance, in Figure~\ref{fig:stabilization} (Right), $a_{i-3}$ enters a tunnel of length 3 by being bound, its three successors $a_{i-2}, a_{i-1}, a_i$ are stabilized by the tunnel, and $a_{i+1}$ is also bound, while $a_i$ provides a new binding capability. 
Since $a_{i-1}$ wastes one unnecessary binding capability so that this tunnel consumes two binding capabilities in total. 
A tunnel can let the transcript create at most 4 tunnel sections, as suggested in Figure~\ref{fig:stabilization} (Right). 

If the sequence $S$ is free from any subsequence of the form $bt^+bt^+b$, then it can factorize as $S = u_1 u_2 u_3 \cdots$ for some $u_1, u_2, u_3, \ldots \in \{b\} \cup bt^+ b$. 
As argued above, each of these factors $u_1, u_2, \ldots$ consumes at least one binding capability. 
Since the seed can provide at most $n$ binding capabilities, there exists $m \le n$ such that $S = u_1 u_2 \cdots u_m$. 
Let $m_1$ be the number of tunnels among the $m$ factors. 
The $m_1$ tunnels can create at most $4m_1$ tunnel sections in total and the remaining $m-m_1$ factors, which correspond to beads that are bound for stabilization, can create at most $3(m-m_1)$ tunnel sections. 
The seed is equipped with no more than $n$ tunnel sections. 
The sum of the length of the $m_1$ tunnels is hence at most $n + 4m_1 + 3(m-m_1) = n+3m+m_1$. 
Consequently, $|S| \le n+3m+m_1 + m - m_1 = n+4m \le 5n$.  

\begin{figure}[tb]
\centering
\begin{tikzpicture}

\draw[-latex, thick] (0, 0) node[mol] {} -- ++(0:1) node[mol]{} node[above] {$a_{i-1}$} -- ++(0:1) node[mol]{} node[right]{$a_i$} -- ++(60:1) node[mol]{} node[left] {$a_{i+1}$} -- ++(0:1) node[mol] {} -- ++(0:1) node[mol] (ai3) {} node[above] {$a_{i+3}$} -- ++(300:1) node[mol] (ai4) {} node[right] {$a_{i+4}$};

\draw (0,0)++(60:1) node{$\times$} --++(0:1) node{$\times$} --++(60:1) node{$\times$} --++(0:1) node{$\times$} --++(0:1) node{$\times$};
\draw (0,0)++(300:1) node{$\times$} --++(0:1) node{$\times$} ++(60:1) ++(0:1) node{$\times$} --++(0:1) node{$\times$};

\draw[thick, dotted] (1, 0) -- ++(300:1);
\draw[thick, dotted] (2, 0)++(60:1) -- ++(60:1);
\draw[thick, dotted] (ai4) -- ++(180:1);

\draw (2, 0)++(300:1) node[mol]{} node (goal) {} node[below] {$a_j$};

\draw[-latex, dashed, thick] (ai4) to [out=300, in=0] (goal);

\draw[dotted, thick, -latex] (2, 0) -- ++(300:0.75);
\draw[dotted, thick, -latex] (ai3) -- ++(0:0.75);

\end{tikzpicture}
\caption{A tandem of two tunnels.}
\label{fig:tunnel_tandem}
\end{figure}

Now we have to handle a subsequence of the form $bt^ibt^jb$ of $S$ for $i, j \ge 1$, which is a tandem of tunnels. 
Figure~\ref{fig:tunnel_tandem} shows two tunnels in tandem. 
The transcript diverts the binding capability which it uses to exit the first tunnel in order to enter the second. 
Moreover, the beads $a_i$ and $a_{i+3}$, which are the last beads stabilized by the first and second tunnels, respectively, can provide one binding capability each. 
Thus, these two tunnels appear to lose only one binding capability \textit{in total} by forming a tandem. 
This argument is incorrect unless the second tunnel turns right acutely (see Figure~\ref{fig:tunnel_sections}). 
Unless turning right acutely, the second tunnel is provided with a right wall. 
The index of a bead that serves as a right wall must be smaller than $i-2$, and by definition, the bead is connected to $a_{i-1}$ by a primary structure of $C_i$. 
If $a_i$ provided a binding capability for a future bead, say $a_j$ ($j > i$), then their bond once formed would close the curve along the transcript from $a_i$ to $a_j$ and isolate a region including the right wall from the rest of the plane, which includes $a_{i-1}$ due to the Jordan curve theorem. 
This is contradictory because the primary structure of a conformation is defined to be non-self-interacting. 
The second tunnel should turn right acutely or $a_i$ cannot provide any binding capability (in order for $a_i$ to provide a binding capability rather to the left of the transcript, then the second tunnel is required to turn rather left acutely). 

\begin{figure}[tb]
\centering
\begin{minipage}{0.4\linewidth}
\centering
\begin{tikzpicture}

\draw[-latex, thick] (0, 0) node[mol]{} 
-- ++(0:1) node[mol] (ai-1) {} node[above] {$a_{i-1}$} 
-- ++(0:1) node[mol] (ai) {} node[right] {$a_i$} 
-- ++(60:1) node[mol] (ai+1) {} node[left] {$a_{i+1}$} 
-- ++(300:1) node[mol] (ai+2) {} %node[above] {$a_{i+2}$} 
-- ++(0:1) node[mol] (ai+3) {}
-- ++(60:0.5)
;

\draw[thick, dotted] (ai-1) -- ++(300:1);
\draw[thick, dotted] (ai+1) -- ++(60:1);
\draw[thick, dotted] (ai+3) -- ++(120:1);

\draw[-latex, thick, dotted] (ai) -- ++(300:0.75);
\draw[-latex, thick, dotted] (ai+2) -- ++(300:0.75);

\draw (0,0)++(60:1) node{$\times$} --++(0:1) node{$\times$} --++(60:1) node{$\times$} --++(0:1) node{$\times$} --++(300:1) node{$\times$};
\draw (0,0)++(300:1) node{$\times$} --++(0:1) node{$\times$};

\end{tikzpicture}
\end{minipage}
\begin{minipage}{0.05\linewidth}
\ \\
\end{minipage}
\begin{minipage}{0.5\linewidth}
\centering
\begin{tikzpicture}

\draw[-latex, thick] (0, 0) node[mol]{} 
-- ++(0:1) node[mol] (ai-1) {} node[above] {$a_{i-1}$} 
-- ++(0:1) node[mol] (ai) {} node[right] {$a_i$} 
-- ++(60:1) node[mol] (ai+1) {} node[left] {$a_{i+1}$} 
-- ++(300:1) node[mol] (ai+2) {} %node[above] {$a_{i+2}$} 
-- ++(0:1) node[mol] (ai+3) {}
-- ++(240:1) node[mol] {}
-- ++(300:1) node[mol] (ai+5) {}
-- ++(0:0.5)
;

\draw[thick, dotted] (ai-1) -- ++(300:1);
\draw[thick, dotted] (ai+1) -- ++(60:1);
\draw[thick, dotted] (ai+3) -- ++(0:1);
\draw[thick, dotted] (ai+5) -- ++(60:1);

\draw[-latex, thick, dotted] (ai) -- ++(300:0.75);
%\draw[-latex, thick, dotted] (ai+2) -- ++(300:0.75);

\draw (0,0)++(60:1) node{$\times$} 
--++(0:1) node{$\times$} 
--++(60:1) node{$\times$} 
--++(0:1) node{$\times$} 
--++(300:1) node{$\times$}
--++(0:1) node{$\times$}
--++(300:1) node{$\times$}
--++(240:1) node{$\times$}
;
\draw (0,0)++(300:1) node{$\times$} --++(0:1) node{$\times$};

\end{tikzpicture}
\end{minipage}
\caption{A tandem of two tunnels can save binding capability but the third one just wastes a bond.}
\label{fig:tandem_beneficial}
\end{figure}

If the second tunnel turns right acutely, the tandem can provide two binding capabilities at the cost of three as shown in Figure~\ref{fig:tandem_beneficial}. 
We cannot improve this ratio further even if another tunnel is concatenated to this tandem. 
Not turning right acutely, the tunnel makes the binding capabilities provided by $a_i$ or $a_{i+2}$ useless, as discussed above based on the Jordan curve theorem. 
Turning right acutely, on the other hand, the third tunnel just narrows a binding region of the second tunnel so that these three tunnels in tandem can provide at most two binding capabilities at the cost of four. 
Therefore, the number of binding capabilities decrements every two tunnels. 
Let $m$ be the number of tunnels in the sequence $S$, that is, $m \le 2n$. 
We denote their length by $\ell_1, \ell_2, \ldots, \ell_m$. 
These tunnels consume at least $\lceil m/2 \rceil$ binding capabilities. 
Hence, at most $n - \lceil m/2 \rceil$ beads can be stabilized by being bound. 
These beads can create at most $3(n-\lceil m/2 \rceil)$ tunnels. 
The $m$ tunnels can stabilize $\sum_{i=1}^m \ell_m$ beads and create at most $4m$ tunnels. 
Initially, the system can have at most $n$ tunnels. 
Combining all of these together, the number of beads that the system can stabilize deterministically is at most 
\[
	n - \left\lceil \frac{m}{2} \right\rceil + \sum_{i=1}^m \ell_i 
		\le n - \left\lceil \frac{m}{2} \right\rceil + n + 3 \left( n- \left\lceil \frac{m}{2} \right\rceil \right) + 4m \le 5n+2m \le 9n.
\]
Thus, the transcript can be of length at most $9n.$ 
\end{proof}

%% file: Oritatami-shapes.bbl
\providecommand{\bysame}{\leavevmode\hbox to3em{\hrulefill}\thinspace}
\providecommand{\MR}{\relax\ifhmode\unskip\space\fi MR }
\providecommand{\MRhref}[2]{%
  \href{http://www.ams.org/mathscinet-getitem?mr=#1}{#2}
}
\providecommand{\href}[2]{#2}
\begin{thebibliography}{10}

\bibitem{Arkin}
E.~M. Arkin, S.~P. Fekete, K.~Islam, H.~Meijer, J.~S.~B. Mitchell, Y.~N\'u\ nez
  Rodr\'iguez, V.~Polishchuk, D.~Rappaport, and H.~Xiao, \emph{Not being
  (super)thin or solid is hard: A study of grid {H}amiltonicity}, Comp.
  Geom.-Theor. Appl. \textbf{42} (2009), no.~6--7, 582--605.

\bibitem{ChaoKanChao1995}
M.~Y. Chao, M.-C. Kan, and S.~Lin-Chao, \emph{{RNAII} transcribed by
  {IPTG}-induced {T7} {RNA} polymerase is non-functional as a replication
  primer for {ColE1}-type plasmids in {\it escherichia coli}}, Nucleic Acids
  Res. \textbf{23} (1995), 1691--1695.

\bibitem{DDFIRSS07}
E.~D. Demaine, M.~L. Demaine, S.~P. Fekete, M.~Ishaque, E.~Rafalin, R.~T.
  Schweller, and D.~L. Souvaine, \emph{Staged self-assembly: nanomanufacture of
  arbitrary shapes with ${O}(1)$ glues}, Natural Computing \textbf{7} (2008),
  no.~3, 347--370.

\bibitem{RNAPods}
E.~D. Demaine, M.~J. Patitz, R.~T. Schweller, and S.~M. Summers,
  \emph{{Self-Assembly of Arbitrary Shapes Using RNAse Enzymes: Meeting the
  Kolmogorov Bound with Small Scale Factor (extended abstract)}}, STACS 2011,
  LIPIcs, vol.~9, Schloss Dagstuhl--Leibniz-Zentrum fuer Informatik, 2011,
  pp.~201--212.

\bibitem{derakhshandeh2016universal}
Z.~Derakhshandeh, R.~Gmyr, A.~W. Richa, G.~Scheideler, and T.~Strothmann,
  \emph{Universal shape formation for programmable matter}, SPAA 2016, ACM,
  2016, pp.~289--299.

\bibitem{MolBiolRNA}
D.~Elliott and M.~Ladomery, \emph{Molecular biology of {RNA}}, 2nd ed., Oxford
  University Press, 2016.

\bibitem{Elonen2016}
A.~Elonen, \emph{Molecular folding and computation}, Bachelor Thesis, Aalto
  University, 2016.

\bibitem{FeynmanLecComp}
R.~P. Feynman, \emph{Feynman lectures on computation}, Westview Press, 1996.

\bibitem{GeMeScSe2017}
C.~Geary, P.-E. Meunier, N.~Schabanel, and S.~Seki, \emph{Folding {T}uring is
  hard but feasible}, arXiv:1508.00510v2.

\bibitem{GeMeScSe2016}
\bysame, \emph{Programming biomolecules that fold greedily during
  transcription}, MFCS 2016, LIPIcs, vol.~58, 2016, pp.~43:1--43:14.

\bibitem{GeRoAn2014}
C.~Geary, P.~W.~K. Rothemund, and E.~S. Andersen, \emph{A single-stranded
  architecture for cotranscriptional folding of {RNA} nanostructures}, Science
  \textbf{345} (2014), no.~6198, 799--804.

\bibitem{HanKim2017}
Y.-S. Han and H.~Kim, \emph{Ruleset optimization on isomorphic oritatami
  systems}, DNA 23, LNCS 10467, Springer, 2017, pp.~33--45.

\bibitem{HanKim2018DNA}
Y-S. Han and H.~Kim, \emph{Construction of geometric structure by oritatami
  system}, DNA24, 2018.

\bibitem{Isambert2009}
H.~Isambert, \emph{The jerky and knotty dynamics of {RNA}}, Methods \textbf{49}
  (2009), 189--196.

\bibitem{LeMaReNi1993}
B.~T.~U. Lewicki, T.~Margus, J.~Remme, and K.~H. Nierhaus, \emph{Coupling of
  {rRNA} transcription and ribosomal assembly {\it in vivo}: Formation of
  active ribosomal subunits in {\it escherichia coli} requires transcription of
  {rRNA} genes by host {RNA} polymerase which cannot be replaced by
  bacteriophage {T7} {RNA} polymerase}, J. Mol. Biol. \textbf{231} (1993),
  no.~3, 581--593.

\bibitem{MasudaSekiUbukata}
Y.~Masuda, S.~Seki, and Y.~Ubukata, \emph{Towards the algorithmic molecular
  self-assembly of fractals by cotranscriptional folding}, CIAA, vol. LNCS
  10977, 2018.

\bibitem{MerkhoferHuJohnson2014}
E.~C. Merkhofer, P.~Hu, and T.~L. Johnson, \emph{Introduction to
  cotranscriptional {RNA} splicing}, Spliceosomal Pre-m{RNA} Splicing: Methods
  and Protocols, vol. 1126, Springer, 2014, pp.~83--96.

\bibitem{OtaSeki2017}
M.~Ota and S.~Seki, \emph{Rule set design problems for oritatami systems},
  Theor. Comput. Sci. \textbf{671} (2017), 26--35.

\bibitem{ReWiRaScRiSt1999}
D.~Repsilber, S.~Wiese, M.~Rachen, A.~W. Schr\"{o}der, D.~Riesner, and
  G.~Steger, \emph{Formation of metastable {RNA} structures by sequential
  folding during transcription: Time-resolved structural analysis of potato
  spindle tuber viroid (-)-stranded {RNA} by temperature-gradient gel
  electrophoresis}, RNA \textbf{5} (1999), 574--584.

\bibitem{RogersSeki2017}
T.~A. Rogers and S.~Seki, \emph{Oritatami system: A survey and impossibility of
  simple simulation at small delays}, Fund. Inform. \textbf{154} (2017),
  359--372.

\bibitem{Seki2017}
S.~Seki, \emph{Cotranscriptional folding: A frontier in molecular engineering
  -- a challenge for computer scientists}, SIAM News \textbf{50} (2017), no.~4.

\bibitem{SolWin07}
D.~Soloveichik and E.~Winfree, \emph{Complexity of self-assembled shapes}, SIAM
  J. Comput. \textbf{36} (2007), no.~6, 1544--1569.

\bibitem{WaStYuLiLu2016}
K.~E. Watters, E.~J Strobel, A.~M. Yu, J.~T. Lis, and J.~B. Lucks,
  \emph{Cotranscriptional folding of a riboswitch at nucleotide resolution},
  Nat. Struct. Mol. Biol. \textbf{23} (2016), no.~12, 1124--1133.

\bibitem{WongSosnickPan2007}
T.~N. Wong, T.~R. Sosnick, and T.~Pan, \emph{Folding of noncoding {RNA}s during
  transcription facilitated by pausing-induced nonnative structures}, PNAS
  \textbf{104} (2007), no.~46, 17995--18000.

\bibitem{Nubots}
D.~Woods, H.-L. Chen, S.~Goodfriend, N.~Dabby, E.~Winfree, and P.~Yin,
  \emph{Active self-assembly of algorithmic shapes and patterns in
  polylogarithmic time}, ITCS '13, ACM, 2013, pp.~353--354.

\end{thebibliography}
